\documentclass[11pt]{article} %%%%%%%%%%%%%%%%%%%%%%%%%%%%%%%%%%%%%%%%%%%%%%%%%

\usepackage{textcomp, amssymb}  % additional symbols (there are more packages)
\usepackage{anysize}            % margin package sets tighter margins
\usepackage{amsfonts}
\usepackage{amsmath}
\usepackage{dsfont}
\usepackage{MnSymbol}
\usepackage{mathtools}
\usepackage{graphicx}
\usepackage{epsfig}
\usepackage{epstopdf}
\usepackage{lmerge}
\usepackage{tikz}
\usepackage{amsthm}
\usepackage{ifpdf}
\usetikzlibrary{calc,decorations.pathmorphing,shapes,graphs}

\newcounter{sarrow}
\newcommand\xrsquigarrow[1]{%
\stepcounter{sarrow}%
\mathrel{\begin{tikzpicture}[baseline= {( $ (current bounding box.south) + (0,-0.5ex) $ )}]
\node[inner sep=.5ex] (\thesarrow) {$\scriptstyle #1$};
\path[draw,<-,decorate,
  decoration={zigzag,amplitude=0.7pt,segment length=1.2mm,pre=lineto,pre length=4pt}]
    (\thesarrow.south east) -- (\thesarrow.south west);
    \end{tikzpicture}}%
}

\newcounter{sarrow1}
\newcommand\xnrsquigarrow[1]{%
\stepcounter{sarrow1}%
\mathrel{\begin{tikzpicture}[baseline= {( $ (current bounding box.south) + (0,-0.5ex) $ )}]
\node[inner sep=.5ex] (\thesarrow) {$\scriptstyle #1$};
\path[draw,<-,decorate,
  decoration={zigzag,amplitude=0.7pt,segment length=1.2mm,pre=lineto,pre length=4pt}]
    (\thesarrow1.south east) -- (\thesarrow1.south west);
    $\slashedarrowfill@\relbar\relbar/$
    \end{tikzpicture}}%
}

\makeatletter
\def\slashedarrowfill@#1#2#3#4#5{%
  $\m@th\thickmuskip0mu\medmuskip\thickmuskip\thinmuskip\thickmuskip
   \relax#5#1\mkern-7mu%
   \cleaders\hbox{$#5\mkern-2mu#2\mkern-2mu$}\hfill
   \mathclap{#3}\mathclap{#2}%
   \cleaders\hbox{$#5\mkern-2mu#2\mkern-2mu$}\hfill
   \mkern-7mu#4$%
}
\def\rightslashedarrowfillb@{%
  \slashedarrowfill@\relbar\relbar/\rightarrow}
\newcommand\xnrightarrow[2][]{%
  \ext@arrow 0055{\rightslashedarrowfillb@}{#1}{#2}}

\def\rightslashedarrowfille@{%
  \slashedarrowfill@\relbar\relbar/\twoheadrightarrow}
\newcommand\xntworightarrow[2][]{%
  \ext@arrow 0055{\rightslashedarrowfille@}{#1}{#2}}

\def\rightslashedarrowfillg@{%
  \slashedarrowfill@\relbar\relbar{\raisebox{.12em}{}}\twoheadrightarrow}
\newcommand\xtworightarrow[2][]{%
  \ext@arrow 0055{\rightslashedarrowfillg@}{#1}{#2}}

\def\rightslashedarrowfillx@{%
  \slashedarrowfill@\Relbar\Relbar/\rightrightarrows}
\newcommand\xnTworightarrow[2][]{%
  \ext@arrow 0055{\rightslashedarrowfillx@}{#1}{#2}}

\def\rightslashedarrowfilly@{%
  \slashedarrowfill@\Relbar\Relbar{\raisebox{.12em}{}}\rightrightarrows}
\newcommand\xTworightarrow[2][]{%
  \ext@arrow 0055{\rightslashedarrowfilly@}{#1}{#2}}

\pgfdeclareshape{slash underlined}
{
  \inheritsavedanchors[from=rectangle] % this is nearly a circle
  \inheritanchorborder[from=rectangle]
  \inheritanchor[from=rectangle]{north}
  \inheritanchor[from=rectangle]{north west}
  \inheritanchor[from=rectangle]{north east}
  \inheritanchor[from=rectangle]{center}
  \inheritanchor[from=rectangle]{west}
  \inheritanchor[from=rectangle]{east}
  \inheritanchor[from=rectangle]{mid}
  \inheritanchor[from=rectangle]{mid west}
  \inheritanchor[from=rectangle]{mid east}
  \inheritanchor[from=rectangle]{base}
  \inheritanchor[from=rectangle]{base west}
  \inheritanchor[from=rectangle]{base east}
  \inheritanchor[from=rectangle]{south}
  \inheritanchor[from=rectangle]{south west}
  \inheritanchor[from=rectangle]{south east}
  \inheritanchorborder[from=rectangle]
  \foregroundpath{
    \southwest \pgf@xa=\pgf@x \pgf@ya=\pgf@y
    \northeast \pgf@xb=\pgf@x \pgf@yb=\pgf@y
    \pgf@xc=\pgf@xa
    \advance\pgf@xc by .5\pgf@xb
    \pgf@yc=\pgf@ya
    \advance\pgf@xc by -1.3pt
    \advance\pgf@yc by -1.8pt
    \pgfpathmoveto{\pgfqpoint{\pgf@xc}{\pgf@yc}}
    \advance\pgf@xc by  2.6pt
    \advance\pgf@yc by  3.6pt
    \pgfpathlineto{\pgfqpoint{\pgf@xc}{\pgf@yc}}
    \pgfpathmoveto{\pgfqpoint{\pgf@xa}{\pgf@ya}}
    \pgfpathlineto{\pgfqpoint{\pgf@xb}{\pgf@ya}}
 }
}
\tikzset{nomorepostaction/.code=\let\tikz@postactions\pgfutil@empty}

\newtheorem{theorem}{Theorem}[section]
\newtheorem{definition}[theorem]{Definition}

\begin{document}
%%% tile page %%%%%%%%%%%%%%%%%%%%%%%%%%%%%%%%%%%%%%%%%%%%%%%%%%%%%%%%%%%%%%%%%

\begin{titlepage}
\thispagestyle{empty}

\hrule
\begin{center}
{\bf\LARGE Verification of Distributed Quantum Protocols}

\vspace{0.7cm}
--- Yong Wang ---

\vspace{2cm}
\begin{figure}[!htbp]
 \centering
 \includegraphics[width=1.0\textwidth]{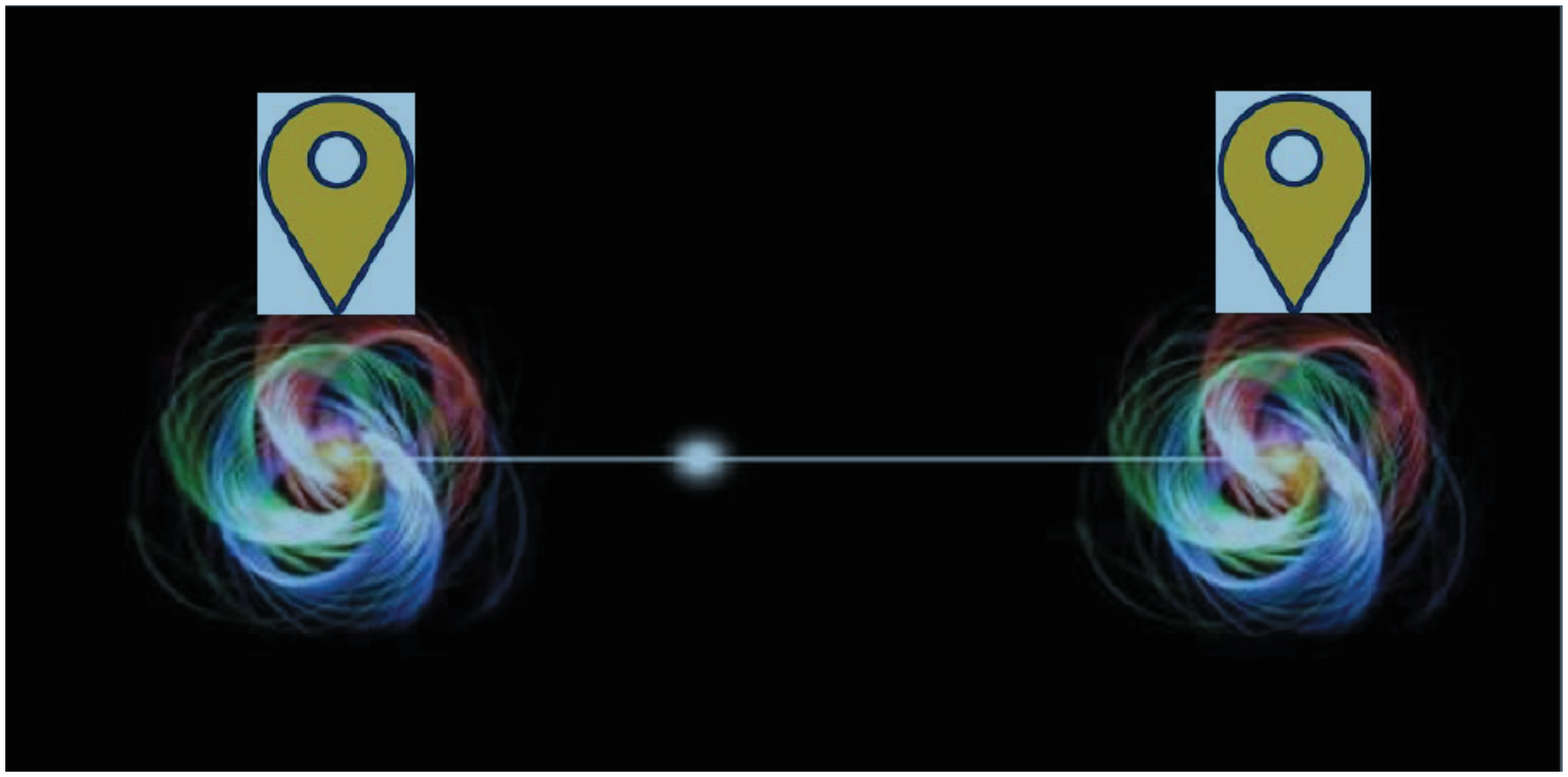}
\end{figure}

\end{center}
\end{titlepage}

\newpage %%%%%%%%%%%%%%%%%%%%%%%%%%%%%%%%%%%%%%%%%%%%%%%%%%%%%%%%%%%%%%%%%%%%%%

\setcounter{page}{1}\pagenumbering{roman}

\tableofcontents

\newpage

\setcounter{page}{1}\pagenumbering{arabic}

        \section{Introduction}

Truly concurrent process algebras are generalizations to the traditional process algebras for true concurrency, CTC \cite{CTC} to CCS \cite{CC} \cite{CCS}, APTC \cite{ATC} to ACP \cite{ACP},
$\pi_{tc}$ \cite{PITC} to $\pi$ calculus \cite{PI1} \cite{PI2}, APPTC \cite{APPTC} to probabilistic process algebra \cite{PPA} \cite{PPA2} \cite{PPA3}, APTC with localities \cite{LOC1} \cite{LOC2} to process algebra
with localities \cite{LOC}.

In quantum process algebras, there are several well-known work \cite{PSQP} \cite{QPA} \cite{QPA2} \cite{CQP} \cite{CQP2} \cite{qCCS} \cite{BQP} \cite{PSQP} \cite{SBQP}, and we ever
did some work \cite{QPA11} \cite{QPA12} \cite{QPA13} to unify quantum and classical computing under the framework of ACP \cite{ACP} and probabilistic process algebra \cite{PPA}.

Now, it is the time to utilize truly concurrent process algebras with localities \cite{LOC1} \cite{LOC2} to model quantum computing and unify quantum and classical computing in this book. Since
this work is with localities, it is suitable to verify the distribution of quantum communication protocols.
This book is organized as follows. In chapter \ref{bg}, we introduce the preliminaries. In chapter \ref{qaptcl} and \ref{aqaptcl}, we introduce the utilization of APTC with localities to unify quantum
and classical computing and its usage in verification of distributed quantum communication protocols. In chapter \ref{qapptcl2} and \ref{aqapptcl}, we introduce the utilization of APPTC 
with localities to unifying quantum and classical computing and its usage in verification of distributed quantum communication protocols.

\newpage\section{Backgrounds}\label{bg}

To make this book self-satisfied, we introduce some preliminaries in this chapter, including some introductions on operational semantics, proof techniques, truly concurrent process algebra
\cite{ATC} \cite{CTC} \cite{PITC} which is based on truly concurrent operational semantics, and also probabilistic truly concurrent process algebra and probabilistic truly concurrent
operational semantics, and also operational semantics for quantum computing.

\subsection{Operational Semantics}\label{OS}

The semantics of $ACP$ is based on bisimulation/rooted branching bisimulation equivalences, and the modularity of $ACP$ relies on the concept of conservative extension, for the
conveniences, we introduce some concepts and conclusions on them.

\begin{definition}[Bisimulation]
A bisimulation relation $R$ is a binary relation on processes such that: (1) if $p R q$ and $p\xrightarrow{a}p'$ then $q\xrightarrow{a}q'$ with $p' R q'$; (2) if $p R q$ and
$q\xrightarrow{a}q'$ then $p\xrightarrow{a}p'$ with $p' R q'$; (3) if $p R q$ and $pP$, then $qP$; (4) if $p R q$ and $qP$, then $pP$. Two processes $p$ and $q$ are bisimilar,
denoted by $p\sim_{HM} q$, if there is a bisimulation relation $R$ such that $p R q$.
\end{definition}

\begin{definition}[Congruence]
Let $\Sigma$ be a signature. An equivalence relation $R$ on $\mathcal{T}(\Sigma)$ is a congruence if for each $f\in\Sigma$, if $s_i R t_i$ for $i\in\{1,\cdots,ar(f)\}$, then
$f(s_1,\cdots,s_{ar(f)}) R f(t_1,\cdots,t_{ar(f)})$.
\end{definition}

\begin{definition}[Branching bisimulation]
A branching bisimulation relation $R$ is a binary relation on the collection of processes such that: (1) if $p R q$ and $p\xrightarrow{a}p'$ then either $a\equiv \tau$ and $p' R q$ or there is a sequence of (zero or more) $\tau$-transitions $q\xrightarrow{\tau}\cdots\xrightarrow{\tau}q_0$ such that $p R q_0$ and $q_0\xrightarrow{a}q'$ with $p' R q'$; (2) if $p R q$ and $q\xrightarrow{a}q'$ then either $a\equiv \tau$ and $p R q'$ or there is a sequence of (zero or more) $\tau$-transitions $p\xrightarrow{\tau}\cdots\xrightarrow{\tau}p_0$ such that $p_0 R q$ and $p_0\xrightarrow{a}p'$ with $p' R q'$; (3) if $p R q$ and $pP$, then there is a sequence of (zero or more) $\tau$-transitions $q\xrightarrow{\tau}\cdots\xrightarrow{\tau}q_0$ such that $p R q_0$ and $q_0P$; (4) if $p R q$ and $qP$, then there is a sequence of (zero or more) $\tau$-transitions $p\xrightarrow{\tau}\cdots\xrightarrow{\tau}p_0$ such that $p_0 R q$ and $p_0P$. Two processes $p$ and $q$ are branching bisimilar, denoted by $p\approx_{bHM} q$, if there is a branching bisimulation relation $R$ such that $p R q$.
\end{definition}

\begin{definition}[Rooted branching bisimulation]
A rooted branching bisimulation relation $R$ is a binary relation on processes such that: (1) if $p R q$ and $p\xrightarrow{a}p'$ then $q\xrightarrow{a}q'$ with $p'\approx_{bHM} q'$;
(2) if $p R q$ and $q\xrightarrow{a}q'$ then $p\xrightarrow{a}p'$ with $p'\approx_{bHM} q'$; (3) if $p R q$ and $pP$, then $qP$; (4) if $p R q$ and $qP$, then $pP$. Two processes $p$ and $q$ are rooted branching bisimilar, denoted by $p\approx_{rbHM} q$, if there is a rooted branching bisimulation relation $R$ such that $p R q$.
\end{definition}

\begin{definition}[Conservative extension]
Let $T_0$ and $T_1$ be TSSs (transition system specifications) over signatures $\Sigma_0$ and $\Sigma_1$, respectively. The TSS $T_0\oplus T_1$ is a conservative extension of $T_0$ if
the LTSs (labeled transition systems) generated by $T_0$ and $T_0\oplus T_1$ contain exactly the same transitions $t\xrightarrow{a}t'$ and $tP$ with $t\in \mathcal{T}(\Sigma_0)$.
\end{definition}

\begin{definition}[Source-dependency]
The source-dependent variables in a transition rule of $\rho$ are defined inductively as follows: (1) all variables in the source of $\rho$ are source-dependent; (2) if
$t\xrightarrow{a}t'$ is a premise of $\rho$ and all variables in $t$ are source-dependent, then all variables in $t'$ are source-dependent. A transition rule is source-dependent if
all its variables are. A TSS is source-dependent if all its rules are.
\end{definition}

\begin{definition}[Freshness]
Let $T_0$ and $T_1$ be TSSs over signatures $\Sigma_0$ and $\Sigma_1$, respectively. A term in $\mathbb{T}(T_0\oplus T_1)$ is said to be fresh if it contains a function symbol from
$\Sigma_1\setminus\Sigma_0$. Similarly, a transition label or predicate symbol in $T_1$ is fresh if it does not occur in $T_0$.
\end{definition}

\begin{theorem}[Conservative extension]\label{TCE}
Let $T_0$ and $T_1$ be TSSs over signatures $\Sigma_0$ and $\Sigma_1$, respectively, where $T_0$ and $T_0\oplus T_1$ are positive after reduction. Under the following conditions,
$T_0\oplus T_1$ is a conservative extension of $T_0$. (1) $T_0$ is source-dependent. (2) For each $\rho\in T_1$, either the source of $\rho$ is fresh, or $\rho$ has a premise of the
form $t\xrightarrow{a}t'$ or $tP$, where $t\in \mathbb{T}(\Sigma_0)$, all variables in $t$ occur in the source of $\rho$ and $t'$, $a$ or $P$ is fresh.
\end{theorem}

\subsection{Proof Techniques}\label{PT}

In this subsection, we introduce the concepts and conclusions about elimination, which is very important in the proof of completeness theorem.

\begin{definition}[Elimination property]
Let a process algebra with a defined set of basic terms as a subset of the set of closed terms over the process algebra. Then the process algebra has the elimination to basic terms
property if for every closed term $s$ of the algebra, there exists a basic term $t$ of the algebra such that the algebra$\vdash s=t$.
\end{definition}

\begin{definition}[Strongly normalizing]
A term $s_0$ is called strongly normalizing if does not an infinite series of reductions beginning in $s_0$.
\end{definition}

\begin{definition}
We write $s>_{lpo} t$ if $s\rightarrow^+ t$ where $\rightarrow^+$ is the transitive closure of the reduction relation defined by the transition rules of an algebra.
\end{definition}

\begin{theorem}[Strong normalization]\label{SN}
Let a term rewriting system (TRS) with finitely many rewriting rules and let $>$ be a well-founded ordering on the signature of the corresponding algebra. If $s>_{lpo} t$ for each
rewriting rule $s\rightarrow t$ in the TRS, then the term rewriting system is strongly normalizing.
\end{theorem}

\subsection{APTC with Localities}

\subsubsection{Operational Semantics}

\begin{definition}[Prime event structure with silent event]\label{PES}
Let $\Lambda$ be a fixed set of labels, ranged over $a,b,c,\cdots$ and $\tau$. A ($\Lambda$-labelled) prime event structure with silent event $\tau$ is a tuple
$\mathcal{E}=\langle \mathbb{E}, \leq, \sharp, \lambda\rangle$, where $\mathbb{E}$ is a denumerable set of events, including the silent event $\tau$. Let
$\hat{\mathbb{E}}=\mathbb{E}\backslash\{\tau\}$, exactly excluding $\tau$, it is obvious that $\hat{\tau^*}=\epsilon$, where $\epsilon$ is the empty event.
Let $\lambda:\mathbb{E}\rightarrow\Lambda$ be a labelling function and let $\lambda(\tau)=\tau$. And $\leq$, $\sharp$ are binary relations on $\mathbb{E}$,
called causality and conflict respectively, such that:

\begin{enumerate}
  \item $\leq$ is a partial order and $\lceil e \rceil = \{e'\in \mathbb{E}|e'\leq e\}$ is finite for all $e\in \mathbb{E}$. It is easy to see that
  $e\leq\tau^*\leq e'=e\leq\tau\leq\cdots\leq\tau\leq e'$, then $e\leq e'$.
  \item $\sharp$ is irreflexive, symmetric and hereditary with respect to $\leq$, that is, for all $e,e',e''\in \mathbb{E}$, if $e\sharp e'\leq e''$, then $e\sharp e''$.
\end{enumerate}

Then, the concepts of consistency and concurrency can be drawn from the above definition:

\begin{enumerate}
  \item $e,e'\in \mathbb{E}$ are consistent, denoted as $e\frown e'$, if $\neg(e\sharp e')$. A subset $X\subseteq \mathbb{E}$ is called consistent, if $e\frown e'$ for all
  $e,e'\in X$.
  \item $e,e'\in \mathbb{E}$ are concurrent, denoted as $e\parallel e'$, if $\neg(e\leq e')$, $\neg(e'\leq e)$, and $\neg(e\sharp e')$.
\end{enumerate}
\end{definition}

\begin{definition}[Configuration]
Let $\mathcal{E}$ be a PES. A (finite) configuration in $\mathcal{E}$ is a (finite) consistent subset of events $C\subseteq \mathcal{E}$, closed with respect to causality
(i.e. $\lceil C\rceil=C$). The set of finite configurations of $\mathcal{E}$ is denoted by $\mathcal{C}(\mathcal{E})$. We let $\hat{C}=C\backslash\{\tau\}$.
\end{definition}

A consistent subset of $X\subseteq \mathbb{E}$ of events can be seen as a pomset. Given $X, Y\subseteq \mathbb{E}$, $\hat{X}\sim \hat{Y}$ if $\hat{X}$ and $\hat{Y}$ are
isomorphic as pomsets. In the following of the paper, we say $C_1\sim C_2$, we mean $\hat{C_1}\sim\hat{C_2}$.

\begin{definition}[Pomset transitions and step]
Let $\mathcal{E}$ be a PES and let $C\in\mathcal{C}(\mathcal{E})$, and $\emptyset\neq X\subseteq \mathbb{E}$, if $C\cap X=\emptyset$ and $C'=C\cup X\in\mathcal{C}(\mathcal{E})$,
then $C\xrightarrow{X} C'$ is called a pomset transition from $C$ to $C'$. When the events in $X$ are pairwise concurrent, we say that $C\xrightarrow{X}C'$ is a step.
\end{definition}

\begin{definition}[Posetal product]
Given two PESs $\mathcal{E}_1$, $\mathcal{E}_2$, the posetal product of their configurations, denoted $\mathcal{C}(\mathcal{E}_1)\overline{\times}\mathcal{C}(\mathcal{E}_2)$,
is defined as

$$\{(C_1,f,C_2)|C_1\in\mathcal{C}(\mathcal{E}_1),C_2\in\mathcal{C}(\mathcal{E}_2),f:C_1\rightarrow C_2 \textrm{ isomorphism}\}.$$

A subset $R\subseteq\mathcal{C}(\mathcal{E}_1)\overline{\times}\mathcal{C}(\mathcal{E}_2)$ is called a posetal relation. We say that $R$ is downward closed when for any
$(C_1,f,C_2),(C_1',f',C_2')\in \mathcal{C}(\mathcal{E}_1)\overline{\times}\mathcal{C}(\mathcal{E}_2)$, if $(C_1,f,C_2)\subseteq (C_1',f',C_2')$ pointwise and $(C_1',f',C_2')\in R$,
then $(C_1,f,C_2)\in R$.

For $f:X_1\rightarrow X_2$, we define $f[x_1\mapsto x_2]:X_1\cup\{x_1\}\rightarrow X_2\cup\{x_2\}$, $z\in X_1\cup\{x_1\}$,(1)$f[x_1\mapsto x_2](z)=
x_2$,if $z=x_1$;(2)$f[x_1\mapsto x_2](z)=f(z)$, otherwise. Where $X_1\subseteq \mathbb{E}_1$, $X_2\subseteq \mathbb{E}_2$, $x_1\in \mathbb{E}_1$, $x_2\in \mathbb{E}_2$.
\end{definition}

Let $Loc$ be the set of locations, and $u,v\in Loc^*$. Let $\ll$ be the sequential ordering on $Loc^*$, we call $v$ is an extension or a sublocation of $u$ in $u\ll v$; and if $u\nll v$
$v\nll u$, then $u$ and $v$ are independent and denoted $u\diamond v$.

\begin{definition}[Consistent location association]
A relation $\varphi\subseteq (Loc^*\times Loc^*)$ is a consistent location association (cla), if $(u,v)\in \varphi \&(u',v')\in\varphi$, then $u\diamond u'\Leftrightarrow v\diamond v'$.
\end{definition}

\begin{definition}[Static location pomset, step bisimulation]
Let $\mathcal{E}_1$, $\mathcal{E}_2$ be PESs. A static location pomset bisimulation is a relation $R_{\varphi}\subseteq\mathcal{C}(\mathcal{E}_1)\times\mathcal{C}(\mathcal{E}_2)$, such that if
$(C_1,C_2)\in R_{\varphi}$, and $C_1\xrightarrow[u]{X_1}C_1'$ then $C_2\xrightarrow[v]{X_2}C_2'$, with $X_1\subseteq \mathbb{E}_1$, $X_2\subseteq \mathbb{E}_2$, $X_1\sim X_2$ and
$(C_1',C_2')\in R_{\varphi\cup\{(u,v)\}}$,
and vice-versa. We say that $\mathcal{E}_1$, $\mathcal{E}_2$ are static location pomset bisimilar, written $\mathcal{E}_1\sim_p^{sl}\mathcal{E}_2$, if there exists a static location pomset bisimulation $R_{\varphi}$, such that
$(\emptyset,\emptyset)\in R_{\varphi}$. By replacing pomset transitions with steps, we can get the definition of static location step bisimulation. When PESs $\mathcal{E}_1$ and $\mathcal{E}_2$ are static location step
bisimilar, we write $\mathcal{E}_1\sim_s^{sl}\mathcal{E}_2$.
\end{definition}

\begin{definition}[Static location (hereditary) history-preserving bisimulation]
A static location history-preserving (hp-) bisimulation is a posetal relation $R_{\varphi}\subseteq\mathcal{C}(\mathcal{E}_1)\overline{\times}\mathcal{C}(\mathcal{E}_2)$ such that if $(C_1,f,C_2)\in R_{\varphi}$,
and $C_1\xrightarrow[u]{e_1} C_1'$, then $C_2\xrightarrow[v]{e_2} C_2'$, with $(C_1',f[e_1\mapsto e_2],C_2')\in R_{\varphi\cup\{(u,v)\}}$, and vice-versa. $\mathcal{E}_1,\mathcal{E}_2$ are static location history-preserving
(hp-)bisimilar and are written $\mathcal{E}_1\sim_{hp}^{sl}\mathcal{E}_2$ if there exists a static location hp-bisimulation $R_{\varphi}$ such that $(\emptyset,\emptyset,\emptyset)\in R_{\varphi}$.

A static location hereditary history-preserving (hhp-)bisimulation is a downward closed static location hp-bisimulation. $\mathcal{E}_1,\mathcal{E}_2$ are static location hereditary history-preserving (hhp-)bisimilar and are
written $\mathcal{E}_1\sim_{hhp}^{sl}\mathcal{E}_2$.
\end{definition}

\begin{definition}[Weak static location pomset, step bisimulation]
Let $\mathcal{E}_1$, $\mathcal{E}_2$ be PESs. A weak static location pomset bisimulation is a relation $R_{\varphi}\subseteq\mathcal{C}(\mathcal{E}_1)\times\mathcal{C}(\mathcal{E}_2)$, such that if
$(C_1,C_2)\in R_{\varphi}$, and $C_1\xRightarrow[u]{X_1}C_1'$ then $C_2\xRightarrow[v]{X_2}C_2'$, with $X_1\subseteq \hat{\mathbb{E}_1}$, $X_2\subseteq \hat{\mathbb{E}_2}$, $X_1\sim X_2$ and
$(C_1',C_2')\in R_{\varphi\cup\{(u,v)\}}$, and vice-versa. We say that $\mathcal{E}_1$, $\mathcal{E}_2$ are weak static location pomset bisimilar, written $\mathcal{E}_1\approx_p^{sl}\mathcal{E}_2$, if there exists a weak static location pomset
bisimulation $R_{\varphi}$, such that $(\emptyset,\emptyset)\in R_{\varphi}$. By replacing weak pomset transitions with weak steps, we can get the definition of weak static location step bisimulation. When PESs
$\mathcal{E}_1$ and $\mathcal{E}_2$ are weak static location step bisimilar, we write $\mathcal{E}_1\approx_s^{sl}\mathcal{E}_2$.
\end{definition}

\begin{definition}[Weak static location (hereditary) history-preserving bisimulation]
A weak static location history-preserving (hp-) bisimulation is a weakly posetal relation $R_{\varphi}\subseteq\mathcal{C}(\mathcal{E}_1)\overline{\times}\mathcal{C}(\mathcal{E}_2)$ such that if
$(C_1,f,C_2)\in R_{\varphi}$, and $C_1\xRightarrow[u]{e_1} C_1'$, then $C_2\xRightarrow[v]{e_2} C_2'$, with $(C_1',f[e_1\mapsto e_2],C_2')\in R_{\varphi\cup\{(u,v)\}}$, and vice-versa. $\mathcal{E}_1,\mathcal{E}_2$ are weak
static location history-preserving (hp-)bisimilar and are written $\mathcal{E}_1\approx_{hp}^{sl}\mathcal{E}_2$ if there exists a weak static location hp-bisimulation $R_{\varphi}$ such that $(\emptyset,\emptyset,\emptyset)\in R_{\varphi}$.

A weak static location hereditary history-preserving (hhp-)bisimulation is a downward closed weak static location hp-bisimulation. $\mathcal{E}_1,\mathcal{E}_2$ are weak static location hereditary history-preserving
(hhp-)bisimilar and are written $\mathcal{E}_1\approx_{hhp}^{sl}\mathcal{E}_2$.
\end{definition}

\begin{definition}[Branching static location pomset, step bisimulation]
Assume a special termination predicate $\downarrow$, and let $\surd$ represent a state with $\surd\downarrow$. Let $\mathcal{E}_1$, $\mathcal{E}_2$ be PESs. A branching static location pomset
bisimulation is a relation $R_{\varphi}\subseteq\mathcal{C}(\mathcal{E}_1)\times\mathcal{C}(\mathcal{E}_2)$, such that:
 \begin{enumerate}
   \item if $(C_1,C_2)\in R_{\varphi}$, and $C_1\xrightarrow[u]{X}C_1'$ then
   \begin{itemize}
     \item either $X\equiv \tau^*$, and $(C_1',C_2)\in R_{\varphi}$;
     \item or there is a sequence of (zero or more) $\tau$-transitions $C_2\xrightarrow{\tau^*} C_2^0$, such that $(C_1,C_2^0)\in R_{\varphi}$ and $C_2^0\xRightarrow[v]{X}C_2'$ with
     $(C_1',C_2')\in R_{\varphi\cup\{(u,v)\}}$;
   \end{itemize}
   \item if $(C_1,C_2)\in R_{\varphi}$, and $C_2\xrightarrow[v]{X}C_2'$ then
   \begin{itemize}
     \item either $X\equiv \tau^*$, and $(C_1,C_2')\in R_{\varphi}$;
     \item or there is a sequence of (zero or more) $\tau$-transitions $C_1\xrightarrow{\tau^*} C_1^0$, such that $(C_1^0,C_2)\in R_{\varphi}$ and $C_1^0\xRightarrow[u]{X}C_1'$ with
     $(C_1',C_2')\in R_{\varphi\cup\{(u,v)\}}$;
   \end{itemize}
   \item if $(C_1,C_2)\in R_{\varphi}$ and $C_1\downarrow$, then there is a sequence of (zero or more) $\tau$-transitions $C_2\xrightarrow{\tau^*}C_2^0$ such that $(C_1,C_2^0)\in R_{\varphi}$
   and $C_2^0\downarrow$;
   \item if $(C_1,C_2)\in R_{\varphi}$ and $C_2\downarrow$, then there is a sequence of (zero or more) $\tau$-transitions $C_1\xrightarrow{\tau^*}C_1^0$ such that $(C_1^0,C_2)\in R_{\varphi}$
   and $C_1^0\downarrow$.
 \end{enumerate}

We say that $\mathcal{E}_1$, $\mathcal{E}_2$ are branching static location pomset bisimilar, written $\mathcal{E}_1\approx_{bp}^{sl}\mathcal{E}_2$, if there exists a branching static location pomset bisimulation $R_{\varphi}$,
such that $(\emptyset,\emptyset)\in R_{\varphi}$.

By replacing pomset transitions with steps, we can get the definition of branching static location step bisimulation. When PESs $\mathcal{E}_1$ and $\mathcal{E}_2$ are branching static location step bisimilar,
we write $\mathcal{E}_1\approx_{bs}^{sl}\mathcal{E}_2$.
\end{definition}

\begin{definition}[Rooted branching static location pomset, step bisimulation]
Assume a special termination predicate $\downarrow$, and let $\surd$ represent a state with $\surd\downarrow$. Let $\mathcal{E}_1$, $\mathcal{E}_2$ be PESs. A rooted branching static location pomset
bisimulation is a relation $R_{\varphi}\subseteq\mathcal{C}(\mathcal{E}_1)\times\mathcal{C}(\mathcal{E}_2)$, such that:
 \begin{enumerate}
   \item if $(C_1,C_2)\in R_{\varphi}$, and $C_1\xrightarrow[u]{X}C_1'$ then $C_2\xrightarrow[v]{X}C_2'$ with $C_1'\approx_{bp}^{sl}C_2'$;
   \item if $(C_1,C_2)\in R_{\varphi}$, and $C_2\xrightarrow[v]{X}C_2'$ then $C_1\xrightarrow[u]{X}C_1'$ with $C_1'\approx_{bp}^{sl}C_2'$;
   \item if $(C_1,C_2)\in R_{\varphi}$ and $C_1\downarrow$, then $C_2\downarrow$;
   \item if $(C_1,C_2)\in R_{\varphi}$ and $C_2\downarrow$, then $C_1\downarrow$.
 \end{enumerate}

We say that $\mathcal{E}_1$, $\mathcal{E}_2$ are rooted branching static location pomset bisimilar, written $\mathcal{E}_1\approx_{rbp}^{sl}\mathcal{E}_2$, if there exists a rooted branching static location pomset
bisimulation $R_{\varphi}$, such that $(\emptyset,\emptyset)\in R_{\varphi}$.

By replacing pomset transitions with steps, we can get the definition of rooted branching static location step bisimulation. When PESs $\mathcal{E}_1$ and $\mathcal{E}_2$ are rooted branching static location step
bisimilar, we write $\mathcal{E}_1\approx_{rbs}^{sl}\mathcal{E}_2$.
\end{definition}

\begin{definition}[Branching static location (hereditary) history-preserving bisimulation]
Assume a special termination predicate $\downarrow$, and let $\surd$ represent a state with $\surd\downarrow$. A branching static location history-preserving (hp-) bisimulation is a posetal
relation $R_{\varphi}\subseteq\mathcal{C}(\mathcal{E}_1)\overline{\times}\mathcal{C}(\mathcal{E}_2)$ such that:

 \begin{enumerate}
   \item if $(C_1,f,C_2)\in R$, and $C_1\xrightarrow[u]{e_1}C_1'$ then
   \begin{itemize}
     \item either $e_1\equiv \tau$, and $(C_1',f[e_1\mapsto \tau],C_2)\in R_{\varphi}$;
     \item or there is a sequence of (zero or more) $\tau$-transitions $C_2\xrightarrow{\tau^*} C_2^0$, such that $(C_1,f,C_2^0)\in R_{\varphi}$ and $C_2^0\xrightarrow[v]{e_2}C_2'$ with
     $(C_1',f[e_1\mapsto e_2],C_2')\in R_{\varphi\cup\{(u,v)\}}$;
   \end{itemize}
   \item if $(C_1,f,C_2)\in R_{\varphi}$, and $C_2\xrightarrow[v]{e_2}C_2'$ then
   \begin{itemize}
     \item either $X\equiv \tau$, and $(C_1,f[e_2\mapsto \tau],C_2')\in R_{\varphi}$;
     \item or there is a sequence of (zero or more) $\tau$-transitions $C_1\xrightarrow{\tau^*} C_1^0$, such that $(C_1^0,f,C_2)\in R_{\varphi}$ and $C_1^0\xrightarrow[u]{e_1}C_1'$ with
     $(C_1',f[e_2\mapsto e_1],C_2')\in R_{\varphi\cup\{(u,v)\}}$;
   \end{itemize}
   \item if $(C_1,f,C_2)\in R_{\varphi}$ and $C_1\downarrow$, then there is a sequence of (zero or more) $\tau$-transitions $C_2\xrightarrow{\tau^*}C_2^0$ such that $(C_1,f,C_2^0)\in R_{\varphi}$
   and $C_2^0\downarrow$;
   \item if $(C_1,f,C_2)\in R_{\varphi}$ and $C_2\downarrow$, then there is a sequence of (zero or more) $\tau$-transitions $C_1\xrightarrow{\tau^*}C_1^0$ such that $(C_1^0,f,C_2)\in R_{\varphi}$
   and $C_1^0\downarrow$.
 \end{enumerate}

$\mathcal{E}_1,\mathcal{E}_2$ are branching static location history-preserving (hp-)bisimilar and are written $\mathcal{E}_1\approx_{bhp}^{sl}\mathcal{E}_2$ if there exists a branching static location hp-bisimulation
$R_{\varphi}$ such that $(\emptyset,\emptyset,\emptyset)\in R_{\varphi}$.

A branching static location hereditary history-preserving (hhp-)bisimulation is a downward closed branching static location hhp-bisimulation. $\mathcal{E}_1,\mathcal{E}_2$ are branching static location hereditary history-preserving
(hhp-)bisimilar and are written $\mathcal{E}_1\approx_{bhhp}^{sl}\mathcal{E}_2$.
\end{definition}

\begin{definition}[Rooted branching static location (hereditary) history-preserving bisimulation]
Assume a special termination predicate $\downarrow$, and let $\surd$ represent a state with $\surd\downarrow$. A rooted branching static location history-preserving (hp-) bisimulation is a posetal relation $R_{\varphi}\subseteq\mathcal{C}(\mathcal{E}_1)\overline{\times}\mathcal{C}(\mathcal{E}_2)$ such that:

 \begin{enumerate}
   \item if $(C_1,f,C_2)\in R_{\varphi}$, and $C_1\xrightarrow[u]{e_1}C_1'$, then $C_2\xrightarrow[v]{e_2}C_2'$ with $C_1'\approx_{bhp}^{sl}C_2'$;
   \item if $(C_1,f,C_2)\in R_{\varphi}$, and $C_2\xrightarrow[v]{e_2}C_2'$, then $C_1\xrightarrow[u]{e_1}C_1'$ with $C_1'\approx_{bhp}^{sl}C_2'$;
   \item if $(C_1,f,C_2)\in R_{\varphi}$ and $C_1\downarrow$, then $C_2\downarrow$;
   \item if $(C_1,f,C_2)\in R_{\varphi}$ and $C_2\downarrow$, then $C_1\downarrow$.
 \end{enumerate}

$\mathcal{E}_1,\mathcal{E}_2$ are rooted branching static location history-preserving (hp-)bisimilar and are written $\mathcal{E}_1\approx_{rbhp}^{sl}\mathcal{E}_2$ if there exists a rooted branching
static location hp-bisimulation $R_{\varphi}$ such that $(\emptyset,\emptyset,\emptyset)\in R_{\varphi}$.

A rooted branching static location hereditary history-preserving (hhp-)bisimulation is a downward closed rooted branching static location hp-bisimulation. $\mathcal{E}_1,\mathcal{E}_2$ are rooted branching
static location hereditary history-preserving (hhp-)bisimilar and are written $\mathcal{E}_1\approx_{rbhhp}^{sl}\mathcal{E}_2$.
\end{definition}

\subsubsection{BATC with Localities}

Let $Loc$ be the set of locations, and $loc\in Loc$, $u,v\in Loc^*$, $\epsilon$ is the empty location. A distribution allocates a location $u\in Loc*$ to an action $e$ denoted
$u::e$ or a process $x$ denoted $u::x$.

%\subsubsection{Axiom System of BATC with static localities}

In the following, let $e_1, e_2, e_1', e_2'\in \mathbb{E}$, and let variables $x,y,z$ range over the set of terms for true concurrency, $p,q,s$ range over the set of closed terms.
The set of axioms of BATC with static localities ($BATC^{sl}$) consists of the laws given in Table \ref{AxiomsForBATC21}.

\begin{center}
    \begin{table}
        \begin{tabular}{@{}ll@{}}
            \hline No. &Axiom\\
            $A1$ & $x+ y = y+ x$\\
            $A2$ & $(x+ y)+ z = x+ (y+ z)$\\
            $A3$ & $x+ x = x$\\
            $A4$ & $(x+ y)\cdot z = x\cdot z + y\cdot z$\\
            $A5$ & $(x\cdot y)\cdot z = x\cdot(y\cdot z)$\\
            $L1$ & $\epsilon::x=x$\\
            $L2$ & $u::(x\cdot y)=u::x\cdot u::y$\\
            $L3$ & $u::(x+ y)=u::x+ u::y$\\
            $L4$ & $u::(v::x)=uv::x$\\
        \end{tabular}
        \caption{Axioms of BATC with static localities}
        \label{AxiomsForBATC21}
    \end{table}
\end{center}

%\subsubsection{Properties of BATC with static localities}

\begin{definition}[Basic terms of BATC with static localities]\label{BTBATC21}
The set of basic terms of BATC with static localities, $\mathcal{B}(BATC^{sl})$, is inductively defined as follows:
\begin{enumerate}
  \item $\mathbb{E}\subset\mathcal{B}(BATC^{sl})$;
  \item if $u\in Loc^*, t\in\mathcal{B}(BATC^{sl})$ then $u::t\in\mathcal{B}(BATC^{sl})$;
  \item if $e\in \mathbb{E}, t\in\mathcal{B}(BATC^{sl})$ then $e\cdot t\in\mathcal{B}(BATC^{sl})$;
  \item if $t,s\in\mathcal{B}(BATC^{sl})$ then $t+ s\in\mathcal{B}(BATC^{sl})$.
\end{enumerate}
\end{definition}

\begin{theorem}[Elimination theorem of BATC with static localities]\label{ETBATC21}
Let $p$ be a closed BATC with static localities term. Then there is a basic BATC with static localities term $q$ such that $BATC^{sl}\vdash p=q$.
\end{theorem}

In this subsection, we will define a term-deduction system which gives the operational semantics of BATC with static localities. We give the operational transition rules for operators
$\cdot$ and $+$ as Table \ref{SETRForBATC21} shows. And the predicate $\xrightarrow[u]{e}\surd$ represents successful termination after execution of the event $e$ at the location $u$.

\begin{center}
    \begin{table}
        $$\frac{}{e\xrightarrow[\epsilon]{e}\surd}\quad \frac{}{loc::e\xrightarrow[loc]{e}\surd}$$
        $$\frac{x\xrightarrow[u]{e}x'}{loc::x\xrightarrow[loc\ll u]{e}loc::x'}$$
        $$\frac{x\xrightarrow[u]{e}\surd}{x+ y\xrightarrow[u]{e}\surd} \quad\frac{x\xrightarrow[u]{e}x'}{x+ y\xrightarrow[u]{e}x'} \quad\frac{y\xrightarrow[u]{e}\surd}{x+ y\xrightarrow[u]{e}\surd} \quad\frac{y\xrightarrow[u]{e}y'}{x+ y\xrightarrow[u]{e}y'}$$
        $$\frac{x\xrightarrow[u]{e}\surd}{x\cdot y\xrightarrow[u]{e} y} \quad\frac{x\xrightarrow[u]{e}x'}{x\cdot y\xrightarrow[u]{e}x'\cdot y}$$
        \caption{Single event transition rules of BATC with static localities}
        \label{SETRForBATC21}
    \end{table}
\end{center}

\begin{theorem}[Congruence of BATC with static localities with respect to static location pomset bisimulation equivalence]
Static location pomset bisimulation equivalence $\sim_p^{sl}$ is a congruence with respect to BATC with static localities.
\end{theorem}

\begin{theorem}[Soundness of BATC with static localities modulo static location pomset bisimulation equivalence]\label{SBATCPBE21}
Let $x$ and $y$ be BATC with static localities terms. If $BATC^{sl}\vdash x=y$, then $x\sim_p^{sl} y$.
\end{theorem}

\begin{theorem}[Completeness of BATC with static localities modulo static location pomset bisimulation equivalence]\label{CBATCPBE21}
Let $p$ and $q$ be closed BATC with static localities terms, if $p\sim_p^{sl} q$ then $p=q$.
\end{theorem}

\begin{theorem}[Congruence of BATC with static localities with respect to static location step bisimulation equivalence]
Static location step bisimulation equivalence $\sim_s^{sl}$ is a congruence with respect to BATC with static localities.
\end{theorem}

\begin{theorem}[Soundness of BATC with static localities modulo static location step bisimulation equivalence]\label{SBATCSBE21}
Let $x$ and $y$ be BATC with static localities terms. If $BATC^{sl}\vdash x=y$, then $x\sim_s^{sl} y$.
\end{theorem}

\begin{theorem}[Completeness of BATC with static localities modulo static location step bisimulation equivalence]\label{CBATCSBE21}
Let $p$ and $q$ be closed BATC with static localities terms, if $p\sim_s^{sl} q$ then $p=q$.
\end{theorem}

\begin{theorem}[Congruence of BATC with static localities with respect to static location hp-bisimulation equivalence]
Static location hp-bisimulation equivalence $\sim_{hp}^{sl}$ is a congruence with respect to BATC with static localities.
\end{theorem}

\begin{theorem}[Soundness of BATC with static localities modulo static location hp-bisimulation equivalence]\label{SBATCHPBE21}
Let $x$ and $y$ be BATC with static localities terms. If $BATC\vdash x=y$, then $x\sim_{hp}^{sl} y$.
\end{theorem}

\begin{theorem}[Completeness of BATC with static localities modulo static location hp-bisimulation equivalence]\label{CBATCHPBE21}
Let $p$ and $q$ be closed BATC with static localities terms, if $p\sim_{hp}^{sl} q$ then $p=q$.
\end{theorem}

\begin{theorem}[Congruence of BATC with static localities with respect to static location hhp-bisimulation equivalence]
Static location hhp-bisimulation equivalence $\sim_{hhp}^{sl}$ is a congruence with respect to BATC with static localities.
\end{theorem}

\begin{theorem}[Soundness of BATC with static localities modulo static location hhp-bisimulation equivalence]\label{SBATCHHPBE21}
Let $x$ and $y$ be BATC with static localities terms. If $BATC\vdash x=y$, then $x\sim_{hhp}^{sl} y$.
\end{theorem}

\begin{theorem}[Completeness of BATC with static localities modulo static location hhp-bisimulation equivalence]\label{CBATCHHPBE21}
Let $p$ and $q$ be closed BATC with static localities terms, if $p\sim_{hhp}^{sl} q$ then $p=q$.
\end{theorem}

\subsubsection{APTC with Localities}

We give the transition rules of APTC with static localities as Table \ref{TRForAPTC21} shows.

\begin{center}
    \begin{table}
        $$\frac{x\xrightarrow[u]{e_1}\surd\quad y\xrightarrow[v]{e_2}\surd}{x\parallel y\xrightarrow[u\diamond v]{\{e_1,e_2\}}\surd} \quad\frac{x\xrightarrow[u]{e_1}x'\quad y\xrightarrow[v]{e_2}\surd}{x\parallel y\xrightarrow[u\diamond v]{\{e_1,e_2\}}x'}$$
        $$\frac{x\xrightarrow[u]{e_1}\surd\quad y\xrightarrow[v]{e_2}y'}{x\parallel y\xrightarrow[u\diamond v]{\{e_1,e_2\}}y'} \quad\frac{x\xrightarrow[u]{e_1}x'\quad y\xrightarrow[v]{e_2}y'}{x\parallel y\xrightarrow[u\diamond v]{\{e_1,e_2\}}x'\between y'}$$
        $$\frac{x\xrightarrow[u]{e_1}\surd\quad y\xrightarrow[v]{e_2}\surd \quad(e_1\leq e_2)}{x\leftmerge y\xrightarrow[u\diamond v]{\{e_1,e_2\}}\surd} \quad\frac{x\xrightarrow[u]{e_1}x'\quad y\xrightarrow[v]{e_2}\surd \quad(e_1\leq e_2)}{x\leftmerge y\xrightarrow[u\diamond v]{\{e_1,e_2\}}x'}$$
        $$\frac{x\xrightarrow[u]{e_1}\surd\quad y\xrightarrow[v]{e_2}y' \quad(e_1\leq e_2)}{x\leftmerge y\xrightarrow[u\diamond v]{\{e_1,e_2\}}y'} \quad\frac{x\xrightarrow[u]{e_1}x'\quad y\xrightarrow[v]{e_2}y' \quad(e_1\leq e_2)}{x\leftmerge y\xrightarrow[u\diamond v]{\{e_1,e_2\}}x'\between y'}$$
        $$\frac{x\xrightarrow[u]{e_1}\surd\quad y\xrightarrow[v]{e_2}\surd}{x\mid y\xrightarrow[u\diamond v]{\gamma(e_1,e_2)}\surd} \quad\frac{x\xrightarrow[u]{e_1}x'\quad y\xrightarrow[v]{e_2}\surd}{x\mid y\xrightarrow[u\diamond v]{\gamma(e_1,e_2)}x'}$$
        $$\frac{x\xrightarrow[u]{e_1}\surd\quad y\xrightarrow[v]{e_2}y'}{x\mid y\xrightarrow[u\diamond v]{\gamma(e_1,e_2)}y'} \quad\frac{x\xrightarrow[u]{e_1}x'\quad y\xrightarrow[v]{e_2}y'}{x\mid y\xrightarrow[u\diamond v]{\gamma(e_1,e_2)}x'\between y'}$$
        $$\frac{x\xrightarrow[u]{e_1}\surd\quad (\sharp(e_1,e_2))}{\Theta(x)\xrightarrow[u]{e_1}\surd} \quad\frac{x\xrightarrow[u]{e_2}\surd\quad (\sharp(e_1,e_2))}{\Theta(x)\xrightarrow[u]{e_2}\surd}$$
        $$\frac{x\xrightarrow[u]{e_1}x'\quad (\sharp(e_1,e_2))}{\Theta(x)\xrightarrow[u]{e_1}\Theta(x')} \quad\frac{x\xrightarrow[u]{e_2}x'\quad (\sharp(e_1,e_2))}{\Theta(x)\xrightarrow[u]{e_2}\Theta(x')}$$
%        $$\frac{x\xrightarrow{e_1}\surd\quad (\sharp_{\pi}(e_1,e_2))}{\Theta(x)\xrightarrow{e_1}\surd} \quad\frac{x\xrightarrow{e_2}\surd\quad (\sharp_{\pi}(e_1,e_2))}{\Theta(x)\xrightarrow{e_2}\surd}$$
%        $$\frac{x\xrightarrow{e_1}x'\quad (\sharp_{\pi}(e_1,e_2))}{\Theta(x)\xrightarrow{e_1}\Theta(x')} \quad\frac{x\xrightarrow{e_2}x'\quad (\sharp_{\pi}(e_1,e_2))}{\Theta(x)\xrightarrow{e_2}\Theta(x')}$$
        $$\frac{x\xrightarrow[u]{e_1}\surd \quad y\nrightarrow^{e_2}\quad (\sharp(e_1,e_2))}{x\triangleleft y\xrightarrow[u]{\tau}\surd}
        \quad\frac{x\xrightarrow[u]{e_1}x' \quad y\nrightarrow^{e_2}\quad (\sharp(e_1,e_2))}{x\triangleleft y\xrightarrow[u]{\tau}x'}$$
        $$\frac{x\xrightarrow[u]{e_1}\surd \quad y\nrightarrow^{e_3}\quad (\sharp(e_1,e_2),e_2\leq e_3)}{x\triangleleft y\xrightarrow[u]{e_1}\surd}
        \quad\frac{x\xrightarrow[u]{e_1}x' \quad y\nrightarrow^{e_3}\quad (\sharp(e_1,e_2),e_2\leq e_3)}{x\triangleleft y\xrightarrow[u]{e_1}x'}$$
        $$\frac{x\xrightarrow[u]{e_3}\surd \quad y\nrightarrow^{e_2}\quad (\sharp(e_1,e_2),e_1\leq e_3)}{x\triangleleft y\xrightarrow[u]{\tau}\surd}
        \quad\frac{x\xrightarrow[u]{e_3}x' \quad y\nrightarrow^{e_2}\quad (\sharp(e_1,e_2),e_1\leq e_3)}{x\triangleleft y\xrightarrow[u]{\tau}x'}$$
%        $$\frac{x\xrightarrow{e_1}\surd \quad y\nrightarrow^{e_2}\quad (\sharp_{\pi}(e_1,e_2))}{x\triangleleft y\xrightarrow{\tau}\surd}
%        \quad\frac{x\xrightarrow{e_1}x' \quad y\nrightarrow^{e_2}\quad (\sharp_{\pi}(e_1,e_2))}{x\triangleleft y\xrightarrow{\tau}x'}$$
%        $$\frac{x\xrightarrow{e_1}\surd \quad y\nrightarrow^{e_3}\quad (\sharp_{\pi}(e_1,e_2),e_2\leq e_3)}{x\triangleleft y\xrightarrow{e_1}\surd}
%        \quad\frac{x\xrightarrow{e_1}x' \quad y\nrightarrow^{e_3}\quad (\sharp_{\pi}(e_1,e_2),e_2\leq e_3)}{x\triangleleft y\xrightarrow{e_1}x'}$$
%        $$\frac{x\xrightarrow{e_3}\surd \quad y\nrightarrow^{e_2}\quad (\sharp_{\pi}(e_1,e_2),e_1\leq e_3)}{x\triangleleft y\xrightarrow{\tau}\surd}
%        \quad\frac{x\xrightarrow{e_3}x' \quad y\nrightarrow^{e_2}\quad (\sharp_{\pi}(e_1,e_2),e_1\leq e_3)}{x\triangleleft y\xrightarrow{\tau}x'}$$
        \caption{Transition rules of APTC with static localities}
        \label{TRForAPTC21}
    \end{table}
\end{center}

In the following, we show that the elimination theorem does not hold for truly concurrent processes combined the operators $\cdot$, $+$ and $\leftmerge$. Firstly, we define the basic
terms for APTC with static localities.

\begin{definition}[Basic terms of APTC with static localities]\label{BTAPTC21}
The set of basic terms of APTC with static localities, $\mathcal{B}(APTC^{sl})$, is inductively defined as follows:
\begin{enumerate}
  \item $\mathbb{E}\subset\mathcal{B}(APTC^{sl})$;
  \item if $u\in Loc^*, t\in\mathcal{B}(APTC^{sl})$ then $u::t\in\mathcal{B}(APTC^{sl})$;
  \item if $e\in \mathbb{E}, t\in\mathcal{B}(APTC^{sl})$ then $e\cdot t\in\mathcal{B}(APTC^{sl})$;
  \item if $t,s\in\mathcal{B}(APTC^{sl})$ then $t+ s\in\mathcal{B}(APTC^{sl})$;
  \item if $t,s\in\mathcal{B}(APTC^{sl})$ then $t\leftmerge s\in\mathcal{B}(APTC^{sl})$.
\end{enumerate}
\end{definition}

\begin{theorem}[Congruence theorem of APTC with static localities]
Static location truly concurrent bisimulation equivalences $\sim_p^{sl}$, $\sim_s^{sl}$, $\sim_{hp}^{sl}$ and $\sim_{hhp}^{sl}$ are all congruences with respect to APTC with static
localities.
\end{theorem}

So, we design the axioms of parallelism in Table \ref{AxiomsForParallelism21}, including algebraic laws for parallel operator $\parallel$, communication operator $\mid$, conflict
elimination operator $\Theta$ and unless operator $\triangleleft$, and also the whole parallel operator $\between$. Since the communication between two communicating events in different
parallel branches may cause deadlock (a state of inactivity), which is caused by mismatch of two communicating events or the imperfectness of the communication channel. We introduce a
new constant $\delta$ to denote the deadlock, and let the atomic event $e\in \mathbb{E}\cup\{\delta\}$.

\begin{center}
    \begin{table}
        \begin{tabular}{@{}ll@{}}
            \hline No. &Axiom\\
            $A6$ & $x+ \delta = x$\\
            $A7$ & $\delta\cdot x =\delta$\\
            $P1$ & $x\between y = x\parallel y + x\mid y$\\
            $P2$ & $x\parallel y = y \parallel x$\\
            $P3$ & $(x\parallel y)\parallel z = x\parallel (y\parallel z)$\\
            $P4$ & $x\parallel y = x\leftmerge y + y\leftmerge x$\\
            $P5$ & $(e_1\leq e_2)\quad e_1\leftmerge (e_2\cdot y) = (e_1\leftmerge e_2)\cdot y$\\
            $P6$ & $(e_1\leq e_2)\quad (e_1\cdot x)\leftmerge e_2 = (e_1\leftmerge e_2)\cdot x$\\
            $P7$ & $(e_1\leq e_2)\quad (e_1\cdot x)\leftmerge (e_2\cdot y) = (e_1\leftmerge e_2)\cdot (x\between y)$\\
            $P8$ & $(x+ y)\leftmerge z = (x\leftmerge z)+ (y\leftmerge z)$\\
            $P9$ & $\delta\leftmerge x = \delta$\\
            $C1$ & $e_1\mid e_2 = \gamma(e_1,e_2)$\\
            $C2$ & $e_1\mid (e_2\cdot y) = \gamma(e_1,e_2)\cdot y$\\
            $C3$ & $(e_1\cdot x)\mid e_2 = \gamma(e_1,e_2)\cdot x$\\
            $C4$ & $(e_1\cdot x)\mid (e_2\cdot y) = \gamma(e_1,e_2)\cdot (x\between y)$\\
            $C5$ & $(x+ y)\mid z = (x\mid z) + (y\mid z)$\\
            $C6$ & $x\mid (y+ z) = (x\mid y)+ (x\mid z)$\\
            $C7$ & $\delta\mid x = \delta$\\
            $C8$ & $x\mid\delta = \delta$\\
            $CE1$ & $\Theta(e) = e$\\
            $CE2$ & $\Theta(\delta) = \delta$\\
            $CE3$ & $\Theta(x+ y) = \Theta(x)\triangleleft y + \Theta(y)\triangleleft x$\\
            $CE4$ & $\Theta(x\cdot y)=\Theta(x)\cdot\Theta(y)$\\
            $CE5$ & $\Theta(x\parallel y) = ((\Theta(x)\triangleleft y)\parallel y)+ ((\Theta(y)\triangleleft x)\parallel x)$\\
            $CE6$ & $\Theta(x\mid y) = ((\Theta(x)\triangleleft y)\mid y)+ ((\Theta(y)\triangleleft x)\mid x)$\\
            $U1$ & $(\sharp(e_1,e_2))\quad e_1\triangleleft e_2 = \tau$\\
            $U2$ & $(\sharp(e_1,e_2),e_2\leq e_3)\quad e_1\triangleleft e_3 = e_1$\\
            $U3$ & $(\sharp(e_1,e_2),e_2\leq e_3)\quad e3\triangleleft e_1 = \tau$\\
            $U4$ & $e\triangleleft \delta = e$\\
            $U5$ & $\delta \triangleleft e = \delta$\\
            $U6$ & $(x+ y)\triangleleft z = (x\triangleleft z)+ (y\triangleleft z)$\\
            $U7$ & $(x\cdot y)\triangleleft z = (x\triangleleft z)\cdot (y\triangleleft z)$\\
            $U8$ & $(x\leftmerge y)\triangleleft z = (x\triangleleft z)\leftmerge (y\triangleleft z)$\\
            $U9$ & $(x\mid y)\triangleleft z = (x\triangleleft z)\mid (y\triangleleft z)$\\
            $U10$ & $x\triangleleft (y+ z) = (x\triangleleft y)\triangleleft z$\\
            $U11$ & $x\triangleleft (y\cdot z)=(x\triangleleft y)\triangleleft z$\\
            $U12$ & $x\triangleleft (y\leftmerge z) = (x\triangleleft y)\triangleleft z$\\
            $U13$ & $x\triangleleft (y\mid z) = (x\triangleleft y)\triangleleft z$\\
            $L5$ & $u::(x\between y) = u::x\between u:: y$\\
            $L6$ & $u::(x\parallel y) = u::x\parallel u:: y$\\
            $L7$ & $u::(x\mid y) = u::x\mid u:: y$\\
            $L8$ & $u::(\Theta(x)) = \Theta(u::x)$\\
            $L9$ & $u::(x\triangleleft y) = u::x\triangleleft u:: y$\\
            $L10$ & $u::\delta=\delta$\\
        \end{tabular}
        \caption{Axioms of parallelism}
        \label{AxiomsForParallelism21}
    \end{table}
\end{center}

%\subsubsection{Properties of Parallelism}

Based on the definition of basic terms for APTC with static localities (see Definition \ref{BTAPTC21}) and axioms of parallelism (see Table \ref{AxiomsForParallelism21}), we can prove the
elimination theorem of parallelism.

\begin{theorem}[Elimination theorem of parallelism]\label{ETParallelism21}
Let $p$ be a closed APTC with static localities term. Then there is a basic APTC with static localities term $q$ such that $APTC^{sl}\vdash p=q$.
\end{theorem}

\begin{theorem}[Generalization of APTC with static localities with respect to BATC with static localities]
APTC with static localities is a generalization of BATC with static localities.
\end{theorem}

\begin{theorem}[Soundness of APTC with static localities modulo static location pomset bisimulation equivalence]\label{SPPBE21}
Let $x$ and $y$ be APTC with static localities terms. If $APTC^{sl}\vdash x=y$, then $x\sim_p^{sl} y$.
\end{theorem}

\begin{theorem}[Completeness of APTC with static localities modulo static location pomset bisimulation equivalence]\label{CPPBE21}
Let $p$ and $q$ be closed APTC with static localities terms, if $p\sim_p^{sl} q$ then $p=q$.
\end{theorem}

\begin{theorem}[Soundness of APTC with static localities modulo static location step bisimulation equivalence]\label{SPSBE21}
Let $x$ and $y$ be APTC with static localities terms. If $APTC^{sl}\vdash x=y$, then $x\sim_s^{sl} y$.
\end{theorem}

\begin{theorem}[Completeness of APTC with static localities modulo static location step bisimulation equivalence]\label{CPSBE21}
Let $p$ and $q$ be closed APTC with static localities terms, if $p\sim_s^{sl} q$ then $p=q$.
\end{theorem}

\begin{theorem}[Soundness of APTC with static localities modulo static location hp-bisimulation equivalence]\label{SPHPBE21}
Let $x$ and $y$ be APTC with static localities terms. If $APTC^{sl}\vdash x=y$, then $x\sim_{hp}^{sl} y$.
\end{theorem}

\begin{theorem}[Completeness of APTC with static localities modulo static location hp-bisimulation equivalence]\label{CPHPBE21}
Let $p$ and $q$ be closed APTC with static localities terms, if $p\sim_{hp}^{sl} q$ then $p=q$.
\end{theorem}

\begin{theorem}[Soundness of APTC with static localities modulo static location hhp-bisimulation equivalence]\label{SPHPBE21}
Let $x$ and $y$ be APTC with static localities terms. If $APTC^{sl}\vdash x=y$, then $x\sim_{hhp}^{sl} y$.
\end{theorem}

\begin{theorem}[Completeness of APTC with static localities modulo static location hhp-bisimulation equivalence]\label{CPHPBE21}
Let $p$ and $q$ be closed APTC with static localities terms, if $p\sim_{hhp}^{sl} q$ then $p=q$.
\end{theorem}

The transition rules of encapsulation operator $\partial_H$ are shown in Table \ref{TRForEncapsulation21}.

\begin{center}
    \begin{table}
        $$\frac{x\xrightarrow[u]{e}\surd}{\partial_H(x)\xrightarrow[u]{e}\surd}\quad (e\notin H)\quad\quad\frac{x\xrightarrow[u]{e}x'}{\partial_H(x)\xrightarrow[u]{e}\partial_H(x')}\quad(e\notin H)$$
        \caption{Transition rules of encapsulation operator $\partial_H$}
        \label{TRForEncapsulation21}
    \end{table}
\end{center}

Based on the transition rules for encapsulation operator $\partial_H$ in Table \ref{TRForEncapsulation21}, we design the axioms as Table \ref{AxiomsForEncapsulation21} shows.

\begin{center}
    \begin{table}
        \begin{tabular}{@{}ll@{}}
            \hline No. &Axiom\\
            $D1$ & $e\notin H\quad\partial_H(e) = e$\\
            $D2$ & $e\in H\quad \partial_H(e) = \delta$\\
            $D3$ & $\partial_H(\delta) = \delta$\\
            $D4$ & $\partial_H(x+ y) = \partial_H(x)+\partial_H(y)$\\
            $D5$ & $\partial_H(x\cdot y) = \partial_H(x)\cdot\partial_H(y)$\\
            $D6$ & $\partial_H(x\leftmerge y) = \partial_H(x)\leftmerge\partial_H(y)$\\
            $L11$ & $u::\partial_H(x) = \partial_H(u::x)$\\
        \end{tabular}
        \caption{Axioms of encapsulation operator}
        \label{AxiomsForEncapsulation21}
    \end{table}
\end{center}

\begin{theorem}[Congruence theorem of encapsulation operator $\partial_H$]
Static location truly concurrent bisimulation equivalences $\sim_p^{sl}$, $\sim_s^{sl}$, $\sim_{hp}^{sl}$ and $\sim_{hhp}^{sl}$ are all congruences with respect to encapsulation
operator $\partial_H$.
\end{theorem}

\begin{theorem}[Elimination theorem of APTC with static localities]\label{ETEncapsulation21}
Let $p$ be a closed APTC with static localities term including the encapsulation operator $\partial_H$. Then there is a basic APTC with static localities term $q$ such that
$APTC\vdash p=q$.
\end{theorem}

\begin{theorem}[Soundness of APTC with static localities modulo static location pomset bisimulation equivalence]\label{SAPTCPBE21}
Let $x$ and $y$ be APTC with static localities terms including encapsulation operator $\partial_H$. If $APTC^{sl}\vdash x=y$, then $x\sim_p^{sl} y$.
\end{theorem}

\begin{theorem}[Completeness of APTC with static localities modulo static location pomset bisimulation equivalence]\label{CAPTCPBE21}
Let $p$ and $q$ be closed APTC with static localities terms including encapsulation operator $\partial_H$, if $p\sim_p^{sl} q$ then $p=q$.
\end{theorem}

\begin{theorem}[Soundness of APTC with static localities modulo static location step bisimulation equivalence]\label{SAPTCSBE21}
Let $x$ and $y$ be APTC with static localities terms including encapsulation operator $\partial_H$. If $APTC^{sl}\vdash x=y$, then $x\sim_s^{sl} y$.
\end{theorem}

\begin{theorem}[Completeness of APTC with static localities modulo static location step bisimulation equivalence]\label{CAPTCSBE21}
Let $p$ and $q$ be closed APTC with static localities terms including encapsulation operator $\partial_H$, if $p\sim_s^{sl} q$ then $p=q$.
\end{theorem}

\begin{theorem}[Soundness of APTC with static localities modulo static location hp-bisimulation equivalence]\label{SAPTCHPBE21}
Let $x$ and $y$ be APTC with static localities terms including encapsulation operator $\partial_H$. If $APTC^{sl}\vdash x=y$, then $x\sim_{hp}^{sl} y$.
\end{theorem}

\begin{theorem}[Completeness of APTC with static localities modulo static location hp-bisimulation equivalence]\label{CAPTCHPBE21}
Let $p$ and $q$ be closed APTC with static localities terms including encapsulation operator $\partial_H$, if $p\sim_{hp}^{sl} q$ then $p=q$.
\end{theorem}

\begin{theorem}[Soundness of APTC with static localities modulo static location hhp-bisimulation equivalence]\label{SAPTCHPBE21}
Let $x$ and $y$ be APTC with static localities terms including encapsulation operator $\partial_H$. If $APTC^{sl}\vdash x=y$, then $x\sim_{hhp}^{sl} y$.
\end{theorem}

\begin{theorem}[Completeness of APTC with static localities modulo static location hhp-bisimulation equivalence]\label{CAPTCHPBE}
Let $p$ and $q$ be closed APTC with static localities terms including encapsulation operator $\partial_H$, if $p\sim_{hhp}^{sl} q$ then $p=q$.
\end{theorem}

\subsubsection{Recursion}

In this section, we introduce recursion to capture infinite processes based on APTC with static localities. Since in APTC with static localities, there are four basic operators
$::$, $\cdot$, $+$ and $\leftmerge$, the recursion must be adapted this situation to include $\leftmerge$.

In the following, $E,F,G$ are recursion specifications, $X,Y,Z$ are recursive variables.

%\subsubsection{Guarded Recursive Specifications}

\begin{definition}[Recursive specification]
A recursive specification is a finite set of recursive equations

$$X_1=t_1(X_1,\cdots,X_n)$$
$$\cdots$$
$$X_n=t_n(X_1,\cdots,X_n)$$

where the left-hand sides of $X_i$ are called recursion variables, and the right-hand sides $t_i(X_1,\cdots,X_n)$ are process terms in APTC with static localities with possible
occurrences of the recursion variables $X_1,\cdots,X_n$.
\end{definition}

\begin{definition}[Solution]
Processes $p_1,\cdots,p_n$ are a solution for a recursive specification $\{X_i=t_i(X_1,\cdots,X_n)|i\in\{1,\cdots,n\}\}$ (with respect to static location truly concurrent bisimulation equivalences
$\sim_s^{sl}$($\sim_p^{sl}$, $\sim_{hp}^{sl}$, $\sim_{hhp}^{sl}$)) if $p_i\sim_s^{sl} (\sim_p^{sl}, \sim_{hp}^{sl}, \sim_{hhp}^{sl})t_i(p_1,\cdots,p_n)$ for $i\in\{1,\cdots,n\}$.
\end{definition}

\begin{definition}[Guarded recursive specification]
A recursive specification

$$X_1=t_1(X_1,\cdots,X_n)$$
$$...$$
$$X_n=t_n(X_1,\cdots,X_n)$$

is guarded if the right-hand sides of its recursive equations can be adapted to the form by applications of the axioms in APTC with static localities and replacing recursion variables
by the right-hand sides of their recursive equations,

$(u_{11}::a_{11}\leftmerge\cdots\leftmerge u_{1i_1}::a_{1i_1})\cdot s_1(X_1,\cdots,X_n)+\cdots+(u_{k1}::a_{k1}\leftmerge\cdots\leftmerge u_{ki_k}::a_{ki_k})\cdot s_k(X_1,\cdots,X_n)\\
+(v_{11}::b_{11}\leftmerge\cdots\leftmerge v_{1j_1}::b_{1j_1})+\cdots+(v_{1j_1}::b_{1j_1}\leftmerge\cdots\leftmerge v_{1j_l}::b_{lj_l})$

where $a_{11},\cdots,a_{1i_1},a_{k1},\cdots,a_{ki_k},b_{11},\cdots,b_{1j_1},b_{1j_1},\cdots,b_{lj_l}\in \mathbb{E}$, and the sum above is allowed to be empty, in which case it
represents the deadlock $\delta$.
\end{definition}

\begin{definition}[Linear recursive specification]\label{LRS}
A recursive specification is linear if its recursive equations are of the form

$(u_{11}::a_{11}\leftmerge\cdots\leftmerge u_{1i_1}::a_{1i_1})X_1+\cdots+(u_{k1}::a_{k1}\leftmerge\cdots\leftmerge u_{ki_k}::a_{ki_k})X_k\\
+(v_{11}::b_{11}\leftmerge\cdots\leftmerge v_{1j_1}::b_{1j_1})+\cdots+(v_{1j_1}::b_{1j_1}\leftmerge\cdots\leftmerge v_{1j_l}::b_{lj_l})$

where $a_{11},\cdots,a_{1i_1},a_{k1},\cdots,a_{ki_k},b_{11},\cdots,b_{1j_1},b_{1j_1},\cdots,b_{lj_l}\in \mathbb{E}$, and the sum above is allowed to be empty, in which case it
represents the deadlock $\delta$.
\end{definition}

For a guarded recursive specifications $E$ with the form

$$X_1=t_1(X_1,\cdots,X_n)$$
$$\cdots$$
$$X_n=t_n(X_1,\cdots,X_n)$$

the behavior of the solution $\langle X_i|E\rangle$ for the recursion variable $X_i$ in $E$, where $i\in\{1,\cdots,n\}$, is exactly the behavior of their right-hand sides
$t_i(X_1,\cdots,X_n)$, which is captured by the two transition rules in Table \ref{TRForGR21}.

\begin{center}
    \begin{table}
        $$\frac{t_i(\langle X_1|E\rangle,\cdots,\langle X_n|E\rangle)\xrightarrow[u]{\{e_1,\cdots,e_k\}}\surd}{\langle X_i|E\rangle\xrightarrow[u]{\{e_1,\cdots,e_k\}}\surd}$$
        $$\frac{t_i(\langle X_1|E\rangle,\cdots,\langle X_n|E\rangle)\xrightarrow[u]{\{e_1,\cdots,e_k\}} y}{\langle X_i|E\rangle\xrightarrow[u]{\{e_1,\cdots,e_k\}} y}$$
        \caption{Transition rules of guarded recursion}
        \label{TRForGR21}
    \end{table}
\end{center}

\begin{theorem}[Conservitivity of APTC with static localities and guarded recursion]
APTC with static localities and guarded recursion is a conservative extension of APTC with static localities.
\end{theorem}

\begin{theorem}[Congruence theorem of APTC with static localities and guarded recursion]
Static location truly concurrent bisimulation equivalences $\sim_p^{sl}$, $\sim_s^{sl}$, $\sim_{hp}^{sl}$ and $\sim_{hhp}^{sl}$ are all congruences with respect to APTC with static localities and guarded
recursion.
\end{theorem}

The $RDP$ (Recursive Definition Principle) and the $RSP$ (Recursive Specification Principle) are shown in Table \ref{RDPRSP21}.

\begin{center}
\begin{table}
  \begin{tabular}{@{}ll@{}}
\hline No. &Axiom\\
  $RDP$ & $\langle X_i|E\rangle = t_i(\langle X_1|E\rangle,\cdots,\langle X_n|E\rangle)\quad (i\in\{1,\cdots,n\})$\\
  $RSP$ & if $y_i=t_i(y_1,\cdots,y_n)$ for $i\in\{1,\cdots,n\}$, then $y_i=\langle X_i|E\rangle \quad(i\in\{1,\cdots,n\})$\\
\end{tabular}
\caption{Recursive definition and specification principle}
\label{RDPRSP21}
\end{table}
\end{center}

\begin{theorem}[Elimination theorem of APTC with static localities and linear recursion]\label{ETRecursion21}
Each process term in APTC with static localities and linear recursion is equal to a process term $\langle X_1|E\rangle$ with $E$ a linear recursive specification.
\end{theorem}

\begin{theorem}[Soundness of APTC with static localities and guarded recursion]\label{SAPTCR21}
Let $x$ and $y$ be APTC with static localities and guarded recursion terms. If $APTC\textrm{ with guarded recursion}\vdash x=y$, then
\begin{enumerate}
  \item $x\sim_s^{sl} y$;
  \item $x\sim_p^{sl} y$;
  \item $x\sim_{hp}^{sl} y$;
  \item $x\sim_{hhp}^{sl} y$.
\end{enumerate}
\end{theorem}

\begin{theorem}[Completeness of APTC with static localities and linear recursion]\label{CAPTCR21}
Let $p$ and $q$ be closed APTC with static localities and linear recursion terms, then,
\begin{enumerate}
  \item if $p\sim_s^{sl} q$ then $p=q$;
  \item if $p\sim_p^{sl} q$ then $p=q$;
  \item if $p\sim_{hp}^{sl} q$ then $p=q$;
  \item if $p\sim_{hhp}^{sl} q$ then $p=q$.
\end{enumerate}
\end{theorem}

\subsubsection{Abstraction}

To abstract away from the internal implementations of a program, and verify that the program exhibits the desired external behaviors, the silent step $\tau$ (and making $\tau$ distinct
by $\tau^e$) and abstraction operator $\tau_I$ are introduced, where $I\subseteq \mathbb{E}$ denotes the internal events. The silent step $\tau$ represents the internal
events, when we consider the external behaviors of a process, $\tau$ events can be removed, that is, $\tau$ events must keep silent. The transition rule of $\tau$ is shown in Table
\ref{TRForTau21}. In the following, let the atomic event $e$ range over $\mathbb{E}\cup\{\delta\}\cup\{\tau\}$, and let the communication function
$\gamma:\mathbb{E}\cup\{\tau\}\times \mathbb{E}\cup\{\tau\}\rightarrow \mathbb{E}\cup\{\delta\}$, with each communication involved $\tau$ resulting into $\delta$.

\begin{center}
    \begin{table}
        $$\frac{}{\tau\xrightarrow{\tau}\surd}$$
        \caption{Transition rule of the silent step}
        \label{TRForTau21}
    \end{table}
\end{center}

%\subsubsection{Guarded Linear Recursion}

The silent step $\tau$ as an atomic event, is introduced into $E$. Considering the recursive specification $X=\tau X$, $\tau s$, $\tau\tau s$, and $\tau\cdots s$ are all its solutions, that is, the solutions make the existence of $\tau$-loops which cause unfairness. To prevent $\tau$-loops, we extend the definition of linear recursive specification (Definition \ref{LRS}) to the guarded one.

\begin{definition}[Guarded linear recursive specification]\label{GLRS21}
A recursive specification is linear if its recursive equations are of the form

$(u_{11}::a_{11}\leftmerge\cdots\leftmerge u_{1i_1}::a_{1i_1})X_1+\cdots+(u_{k1}::a_{k1}\leftmerge\cdots\leftmerge u_{ki_k}::a_{ki_k})X_k\\
+(v_{11}::b_{11}\leftmerge\cdots\leftmerge v_{1j_1}::b_{1j_1})+\cdots+(v_{1j_1}::b_{1j_1}\leftmerge\cdots\leftmerge v_{1j_l}::b_{lj_l})$

where $a_{11},\cdots,a_{1i_1},a_{k1},\cdots,a_{ki_k},b_{11},\cdots,b_{1j_1},b_{1j_1},\cdots,b_{lj_l}\in \mathbb{E}\cup\{\tau\}$, and the sum above is allowed to be empty, in which case
it represents the deadlock $\delta$.

A linear recursive specification $E$ is guarded if there does not exist an infinite sequence of $\tau$-transitions
$\langle X|E\rangle\xrightarrow{\tau}\langle X'|E\rangle\xrightarrow{\tau}\langle X''|E\rangle\xrightarrow{\tau}\cdots$.
\end{definition}

\begin{theorem}[Conservitivity of APTC with static localities and silent step and guarded linear recursion]
APTC with static localities and silent step and guarded linear recursion is a conservative extension of APTC with static localities and linear recursion.
\end{theorem}

\begin{theorem}[Congruence theorem of APTC with static localities and silent step and guarded linear recursion]
Rooted branching static location truly concurrent bisimulation equivalences $\approx_{rbp}^{sl}$, $\approx_{rbs}^{sl}$ and $\approx_{rbhp}^{sl}$ are all congruences with respect to
APTC with static localities and silent step and guarded linear recursion.
\end{theorem}

We design the axioms for the silent step $\tau$ in Table \ref{AxiomsForTau21}.

\begin{center}
\begin{table}
  \begin{tabular}{@{}ll@{}}
\hline No. &Axiom\\
  $B1$ & $e\cdot\tau=e$\\
  $B2$ & $e\cdot(\tau\cdot(x+y)+x)=e\cdot(x+y)$\\
  $B3$ & $x\leftmerge\tau=x$\\
  $L13$ & $u::\tau=\tau$\
\end{tabular}
\caption{Axioms of silent step}
\label{AxiomsForTau21}
\end{table}
\end{center}

\begin{theorem}[Elimination theorem of APTC with static localities and silent step and guarded linear recursion]\label{ETTau21}
Each process term in APTC with static localities and silent step and guarded linear recursion is equal to a process term $\langle X_1|E\rangle$ with $E$ a guarded linear recursive
specification.
\end{theorem}

\begin{theorem}[Soundness of APTC with static localities and silent step and guarded linear recursion]\label{SAPTCTAU21}
Let $x$ and $y$ be APTC with static localities and silent step and guarded linear recursion terms. If APTC with static localities and silent step and guarded linear recursion
$\vdash x=y$, then
\begin{enumerate}
  \item $x\approx_{rbs}^{sl} y$;
  \item $x\approx_{rbp}^{sl} y$;
  \item $x\approx_{rbhp}^{sl} y$;
  \item $x\approx_{rbhhp}^{sl} y$.
\end{enumerate}
\end{theorem}

\begin{theorem}[Completeness of APTC with static localities and silent step and guarded linear recursion]\label{CAPTCTAU21}
Let $p$ and $q$ be closed APTC with static localities and silent step and guarded linear recursion terms, then,
\begin{enumerate}
  \item if $p\approx_{rbs}^{sl} q$ then $p=q$;
  \item if $p\approx_{rbp}^{sl} q$ then $p=q$;
  \item if $p\approx_{rbhp}^{sl} q$ then $p=q$;
  \item if $p\approx_{rbhhp}^{sl} q$ then $p=q$.
\end{enumerate}
\end{theorem}

The unary abstraction operator $\tau_I$ ($I\subseteq \mathbb{E}$) renames all atomic events in $I$ into $\tau$. APTC with static localities and silent step and abstraction operator is
called $APTC_{\tau}$ with static localities. The transition rules of operator $\tau_I$ are shown in Table \ref{TRForAbstraction21}.

\begin{center}
    \begin{table}
        $$\frac{x\xrightarrow[u]{e}\surd}{\tau_I(x)\xrightarrow[u]{e}\surd}\quad e\notin I
        \quad\quad\frac{x\xrightarrow[u]{e}x'}{\tau_I(x)\xrightarrow[u]{e}\tau_I(x')}\quad e\notin I$$

        $$\frac{x\xrightarrow[u]{e}\surd}{\tau_I(x)\xrightarrow{\tau}\surd}\quad e\in I
        \quad\quad\frac{x\xrightarrow[u]{e}x'}{\tau_I(x)\xrightarrow{\tau}\tau_I(x')}\quad e\in I$$
        \caption{Transition rule of the abstraction operator}
        \label{TRForAbstraction21}
    \end{table}
\end{center}

\begin{theorem}[Conservitivity of $APTC_{\tau}$ with static localities and guarded linear recursion]
$APTC_{\tau}$ with static localities and guarded linear recursion is a conservative extension of APTC with static localities and silent step and guarded linear recursion.
\end{theorem}

\begin{theorem}[Congruence theorem of $APTC_{\tau}$ with static localities and guarded linear recursion]
Rooted branching static location truly concurrent bisimulation equivalences $\approx_{rbp}^{sl}$, $\approx_{rbs}^{sl}$, $\approx_{rbhp}^{sl}$ and $\approx_{rbhhp}^{sl}$ are all
congruences with respect to $APTC_{\tau}$ with static localities and guarded linear recursion.
\end{theorem}

We design the axioms for the abstraction operator $\tau_I$ in Table \ref{AxiomsForAbstraction21}.

\begin{center}
\begin{table}
  \begin{tabular}{@{}ll@{}}
\hline No. &Axiom\\
  $TI1$ & $e\notin I\quad \tau_I(e)=e$\\
  $TI2$ & $e\in I\quad \tau_I(e)=\tau$\\
  $TI3$ & $\tau_I(\delta)=\delta$\\
  $TI4$ & $\tau_I(x+y)=\tau_I(x)+\tau_I(y)$\\
  $TI5$ & $\tau_I(x\cdot y)=\tau_I(x)\cdot\tau_I(y)$\\
  $TI6$ & $\tau_I(x\leftmerge y)=\tau_I(x)\leftmerge\tau_I(y)$\\
  $L14$ & $u::\tau_I(x)=\tau_I(u::x)$\\
  $L15$ & $e\notin I\quad \tau_I(u::e)=u::e$\\
  $L16$ & $e\in I\quad \tau_I(u::e)=\tau$\\
\end{tabular}
\caption{Axioms of abstraction operator}
\label{AxiomsForAbstraction21}
\end{table}
\end{center}

\begin{theorem}[Soundness of $APTC_{\tau}$ with static localities and guarded linear recursion]\label{SAPTCABS21}
Let $x$ and $y$ be $APTC_{\tau}$ with static localities and guarded linear recursion terms. If $APTC_{\tau}$ with static localities and guarded linear recursion $\vdash x=y$, then
\begin{enumerate}
  \item $x\approx_{rbs}^{sl} y$;
  \item $x\approx_{rbp}^{sl} y$;
  \item $x\approx_{rbhp}^{sl} y$;
  \item $x\approx_{rbhhp}^{sl} y$.
\end{enumerate}
\end{theorem}

Though $\tau$-loops are prohibited in guarded linear recursive specifications (see Definition \ref{GLRS21}) in a specifiable way, they can be constructed using the abstraction operator,
for example, there exist $\tau$-loops in the process term $\tau_{\{a\}}(\langle X|X=aX\rangle)$. To avoid $\tau$-loops caused by $\tau_I$ and ensure fairness, the concept of cluster
and $CFAR$ (Cluster Fair Abstraction Rule) \cite{CFAR} are still valid in true concurrency, we introduce them below.

\begin{definition}[Cluster]\label{CLUSTER21}
Let $E$ be a guarded linear recursive specification, and $I\subseteq \mathbb{E}$. Two recursion variable $X$ and $Y$ in $E$ are in the same cluster for $I$ iff there exist sequences of
transitions $\langle X|E\rangle\xrightarrow[u]{\{b_{11},\cdots, b_{1i}\}}\cdots[u]\xrightarrow{\{b_{m1},\cdots, b_{mi}\}}\langle Y|E\rangle$ and
$\langle Y|E\rangle\xrightarrow[v]{\{c_{11},\cdots, c_{1j}\}}\cdots\xrightarrow[v]{\{c_{n1},\cdots, c_{nj}\}}\langle X|E\rangle$, where
$b_{11},\cdots,b_{mi},c_{11},\cdots,c_{nj}\in I\cup\{\tau\}$.

$u_1::a_1\leftmerge\cdots\leftmerge u_k::a_k$ or $(u_1::a_1\leftmerge\cdots\leftmerge u_k::a_k) X$ is an exit for the cluster $C$ iff: (1) $u_1::a_1\leftmerge\cdots\leftmerge u_k::a_k$
or $(u_1::a_1\leftmerge\cdots\leftmerge u_k::a_k) X$ is a summand at the right-hand side of the recursive equation for a recursion variable in $C$, and (2) in the case of
$(u_1::a_1\leftmerge\cdots\leftmerge u_k::a_k) X$, either $a_l\notin I\cup\{\tau\}(l\in\{1,2,\cdots,k\})$ or $X\notin C$.
\end{definition}

\begin{center}
\begin{table}
  \begin{tabular}{@{}ll@{}}
\hline No. &Axiom\\
  $CFAR$ & If $X$ is in a cluster for $I$ with exits \\
           & $\{(u_{11}::a_{11}\leftmerge\cdots\leftmerge u_{1i}::a_{1i})Y_1,\cdots,(u_{m1}::a_{m1}\leftmerge\cdots\leftmerge u_{mi}::a_{mi})Y_m,$ \\
           & $v_{11}::b_{11}\leftmerge\cdots\leftmerge v_{1j}::b_{1j},\cdots,v_{n1}::b_{n1}\leftmerge\cdots\leftmerge v_{nj}::b_{nj}\}$, \\
           & then $\tau\cdot\tau_I(\langle X|E\rangle)=$\\
           & $\tau\cdot\tau_I((u_{11}::a_{11}\leftmerge\cdots\leftmerge u_{1i}::a_{1i})\langle Y_1|E\rangle+\cdots+(u_{m1}::a_{m1}\leftmerge\cdots\leftmerge u_{mi}::a_{mi})\langle Y_m|E\rangle$\\
           & $+v_{11}::b_{11}\leftmerge\cdots\leftmerge v_{1j}::b_{1j}+\cdots+v_{n1}::b_{n1}\leftmerge\cdots\leftmerge v_{nj}::b_{nj})$\\
\end{tabular}
\caption{Cluster fair abstraction rule}
\label{CFAR21}
\end{table}
\end{center}

\begin{theorem}[Soundness of $CFAR$]\label{SCFAR21}
$CFAR$ is sound modulo rooted branching truly concurrent bisimulation equivalences $\approx_{rbs}^{sl}$, $\approx_{rbp}^{sl}$, $\approx_{rbhp}^{sl}$ and $\approx_{rbhhp}^{sl}$.
\end{theorem}

\begin{theorem}[Completeness of $APTC_{\tau}$ with static localities and guarded linear recursion and $CFAR$]\label{CCFAR21}
Let $p$ and $q$ be closed $APTC_{\tau}$ with static localities and guarded linear recursion and $CFAR$ terms, then,
\begin{enumerate}
  \item if $p\approx_{rbs}^{sl} q$ then $p=q$;
  \item if $p\approx_{rbp}^{sl} q$ then $p=q$;
  \item if $p\approx_{rbhp}^{sl} q$ then $p=q$;
  \item if $p\approx_{rbhhp}^{sl} q$ then $p=q$.
\end{enumerate}
\end{theorem}

\subsection{APPTC with Localities}

\subsubsection{Operational Semantics}

Let $Loc$ be the set of locations, and $u,v\in Loc^*$. Let $\ll$ be the sequential ordering on $Loc^*$, we call $v$ is an extension or a sublocation of $u$ in $u\ll v$; and if $u\nll v$
$v\nll u$, then $u$ and $v$ are independent and denoted $u\diamond v$.

\begin{definition}[Probabilistic prime event structure with silent event]\label{PPES}
Let $\Lambda$ be a fixed set of labels, ranged over $a,b,c,\cdots$ and $\tau$. A ($\Lambda$-labelled) prime event structure with silent event $\tau$ is a quintuple
$\mathcal{E}=\langle \mathbb{E}, \leq, \sharp, \sharp_{\pi}, \lambda\rangle$, where $\mathbb{E}$ is a denumerable set of events, including the silent event $\tau$. Let
$\hat{\mathbb{E}}=\mathbb{E}\backslash\{\tau\}$, exactly excluding $\tau$, it is obvious that $\hat{\tau^*}=\epsilon$, where $\epsilon$ is the empty event.
Let $\lambda:\mathbb{E}\rightarrow\Lambda$ be a labelling function and let $\lambda(\tau)=\tau$. And $\leq$, $\sharp$, $\sharp_{\pi}$ are binary relations on $\mathbb{E}$,
called causality, conflict and probabilistic conflict respectively, such that:

\begin{enumerate}
  \item $\leq$ is a partial order and $\lceil e \rceil = \{e'\in \mathbb{E}|e'\leq e\}$ is finite for all $e\in \mathbb{E}$. It is easy to see that
  $e\leq\tau^*\leq e'=e\leq\tau\leq\cdots\leq\tau\leq e'$, then $e\leq e'$.
  \item $\sharp$ is irreflexive, symmetric and hereditary with respect to $\leq$, that is, for all $e,e',e''\in \mathbb{E}$, if $e\sharp e'\leq e''$, then $e\sharp e''$;
  \item $\sharp_{\pi}$ is irreflexive, symmetric and hereditary with respect to $\leq$, that is, for all $e,e',e''\in \mathbb{E}$, if $e\sharp_{\pi} e'\leq e''$, then $e\sharp_{\pi} e''$.
\end{enumerate}

Then, the concepts of consistency and concurrency can be drawn from the above definition:

\begin{enumerate}
  \item $e,e'\in \mathbb{E}$ are consistent, denoted as $e\frown e'$, if $\neg(e\sharp e')$ and $\neg(e\sharp_{\pi} e')$. A subset $X\subseteq \mathbb{E}$ is called consistent, if $e\frown e'$ for all
  $e,e'\in X$.
  \item $e,e'\in \mathbb{E}$ are concurrent, denoted as $e\parallel e'$, if $\neg(e\leq e')$, $\neg(e'\leq e)$, and $\neg(e\sharp e')$ and $\neg(e\sharp_{\pi} e')$.
\end{enumerate}
\end{definition}

\begin{definition}[Configuration]
Let $\mathcal{E}$ be a PES. A (finite) configuration in $\mathcal{E}$ is a (finite) consistent subset of events $C\subseteq \mathcal{E}$, closed with respect to causality
(i.e. $\lceil C\rceil=C$). The set of finite configurations of $\mathcal{E}$ is denoted by $\mathcal{C}(\mathcal{E})$. We let $\hat{C}=C\backslash\{\tau\}$.
\end{definition}

A consistent subset of $X\subseteq \mathbb{E}$ of events can be seen as a pomset. Given $X, Y\subseteq \mathbb{E}$, $\hat{X}\sim \hat{Y}$ if $\hat{X}$ and $\hat{Y}$ are
isomorphic as pomsets. In the following of the paper, we say $C_1\sim C_2$, we mean $\hat{C_1}\sim\hat{C_2}$.

\begin{definition}[Pomset transitions and step]
Let $\mathcal{E}$ be a PES and let $C\in\mathcal{C}(\mathcal{E})$, and $\emptyset\neq X\subseteq \mathbb{E}$, if $C\cap X=\emptyset$ and $C'=C\cup X\in\mathcal{C}(\mathcal{E})$,
then $C\xrightarrow[u]{X} C'$ is called a pomset transition from $C$ to $C'$. When the events in $X$ are pairwise concurrent, we say that $C\xrightarrow[u]{X}C'$ is a step.
\end{definition}

\begin{definition}[Probabilistic transitions]
Let $\mathcal{E}$ be a PES and let $C\in\mathcal{C}(\mathcal{E})$, the transition $C\xrsquigarrow{\pi} C^{\pi}$ is called a probabilistic transition from $C$ to $C^{\pi}$.
\end{definition}

\begin{definition}[Weak pomset transitions and weak step]
Let $\mathcal{E}$ be a PES and let $C\in\mathcal{C}(\mathcal{E})$, and $\emptyset\neq X\subseteq \hat{\mathbb{E}}$, if $C\cap X=\emptyset$ and
$\hat{C'}=\hat{C}\cup X\in\mathcal{C}(\mathcal{E})$, then $C\xRightarrow[u]{X} C'$ is called a weak pomset transition from $C$ to $C'$, where we define
$\xRightarrow[u']{e}\triangleq\xrightarrow{\tau^*}\xrightarrow[u']{e}\xrightarrow{\tau^*}$. And $\xRightarrow{X}\triangleq\xrightarrow{\tau^*}\xrightarrow[u']{e}\xrightarrow{\tau^*}$,
for every $e\in X$. When the events in $X$ are pairwise concurrent, we say that $C\xRightarrow[u]{X}C'$ is a weak step.
\end{definition}

\begin{definition}[Consistent location association]
A relation $\varphi\subseteq (Loc^*\times Loc^*)$ is a consistent location association (cla), if $(u,v)\in \varphi \&(u',v')\in\varphi$, then $u\diamond u'\Leftrightarrow v\diamond v'$.
\end{definition}

\begin{definition}[Probabilistic static location pomset, step bisimulation]
Let $\mathcal{E}_1$, $\mathcal{E}_2$ be PESs. A probabilistic static location pomset bisimulation is a relation $R_{\varphi}\subseteq\mathcal{C}(\mathcal{E}_1)\times\mathcal{C}(\mathcal{E}_2)$, such that (1) if
$(C_1,C_2)\in R_{\varphi}$, and $C_1\xrightarrow[u]{X_1}C_1'$ then $C_2\xrightarrow[v]{X_2}C_2'$, with $X_1\subseteq \mathbb{E}_1$, $X_2\subseteq \mathbb{E}_2$, $X_1\sim X_2$ and
$(C_1',C_2')\in R_{\varphi\cup\{(u,v)\}}$,
and vice-versa; (2) if $(C_1,C_2)\in R_{\varphi}$, and $C_1\xrsquigarrow{\pi}C_1^{\pi}$ then $C_2\xrsquigarrow{\pi}C_2^{\pi}$ and $(C_1^{\pi},C_2^{\pi})\in R_{\varphi}$, and vice-versa; (3) if $(C_1,C_2)\in R_{\varphi}$,
then $\mu(C_1,C)=\mu(C_2,C)$ for each $C\in\mathcal{C}(\mathcal{E})/R_{\varphi}$; (4) $[\surd]_{R_{\varphi}}=\{\surd\}$. We say that $\mathcal{E}_1$, $\mathcal{E}_2$ are probabilistic static location pomset bisimilar, written $\mathcal{E}_1\sim_{pp}^{sl}\mathcal{E}_2$, if there exists a probabilistic static location pomset bisimulation $R_{\varphi}$, such that
$(\emptyset,\emptyset)\in R_{\varphi}$. By replacing pomset transitions with steps, we can get the definition of probabilistic static location step bisimulation. When PESs $\mathcal{E}_1$ and $\mathcal{E}_2$ are probabilistic static location step
bisimilar, we write $\mathcal{E}_1\sim_{ps}^{sl}\mathcal{E}_2$.
\end{definition}

\begin{definition}[Probabilistic static location (hereditary) history-preserving bisimulation]
A probabilistic static location history-preserving (hp-) bisimulation is a posetal relation $R_{\varphi}\subseteq\mathcal{C}(\mathcal{E}_1)\overline{\times}\mathcal{C}(\mathcal{E}_2)$ such that (1) if $(C_1,f,C_2)\in R_{\varphi}$,
and $C_1\xrightarrow[u]{e_1} C_1'$, then $C_2\xrightarrow[v]{e_2} C_2'$, with $(C_1',f[e_1\mapsto e_2],C_2')\in R_{\varphi\cup\{(u,v)\}}$, and vice-versa; (2) if $(C_1,f,C_2)\in R_{\varphi}$, and
$C_1\xrsquigarrow{\pi}C_1^{\pi}$ then $C_2\xrsquigarrow{\pi}C_2^{\pi}$ and $(C_1^{\pi},f,C_2^{\pi})\in R_{\varphi}$, and vice-versa; (3) if $(C_1,f,C_2)\in R_{\varphi}$,
then $\mu(C_1,C)=\mu(C_2,C)$ for each $C\in\mathcal{C}(\mathcal{E})/R_{\varphi}$; (4) $[\surd]_{R_{\varphi}}=\{\surd\}$. $\mathcal{E}_1,\mathcal{E}_2$ are probabilistic static location history-preserving
(hp-)bisimilar and are written $\mathcal{E}_1\sim_{php}^{sl}\mathcal{E}_2$ if there exists a probabilistic static location hp-bisimulation $R_{\varphi}$ such that $(\emptyset,\emptyset,\emptyset)\in R_{\varphi}$.

A probabilistic static location hereditary history-preserving (hhp-)bisimulation is a downward closed probabilistic static location hp-bisimulation. $\mathcal{E}_1,\mathcal{E}_2$ are probabilistic static location hereditary history-preserving (hhp-)bisimilar and are
written $\mathcal{E}_1\sim_{phhp}^{sl}\mathcal{E}_2$.
\end{definition}

\begin{definition}[Weak probabilistic static location pomset, step bisimulation]
Let $\mathcal{E}_1$, $\mathcal{E}_2$ be PESs. A weak probabilistic static location pomset bisimulation is a relation $R_{\varphi}\subseteq\mathcal{C}(\mathcal{E}_1)\times\mathcal{C}(\mathcal{E}_2)$, such that (1) if
$(C_1,C_2)\in R_{\varphi}$, and $C_1\xRightarrow[u]{X_1}C_1'$ then $C_2\xRightarrow[v]{X_2}C_2'$, with $X_1\subseteq \hat{\mathbb{E}_1}$, $X_2\subseteq \hat{\mathbb{E}_2}$, $X_1\sim X_2$ and
$(C_1',C_2')\in R_{\varphi\cup\{(u,v)\}}$, and vice-versa; (2) if $(C_1,C_2)\in R_{\varphi}$, and $C_1\xrsquigarrow{\pi}C_1^{\pi}$ then $C_2\xrsquigarrow{\pi}C_2^{\pi}$ and $(C_1^{\pi},C_2^{\pi})\in R_{\varphi}$, and vice-versa; (3) if $(C_1,C_2)\in R_{\varphi}$,
then $\mu(C_1,C)=\mu(C_2,C)$ for each $C\in\mathcal{C}(\mathcal{E})/R_{\varphi}$; (4) $[\surd]_{R_{\varphi}}=\{\surd\}$. We say that $\mathcal{E}_1$, $\mathcal{E}_2$ are weak probabilistic static location pomset bisimilar, written $\mathcal{E}_1\approx_{pp}^{sl}\mathcal{E}_2$, if there exists a weak probabilistic static location pomset
bisimulation $R_{\varphi}$, such that $(\emptyset,\emptyset)\in R_{\varphi}$. By replacing weak pomset transitions with weak steps, we can get the definition of weak probabilistic static location step bisimulation. When PESs
$\mathcal{E}_1$ and $\mathcal{E}_2$ are weak probabilistic static location step bisimilar, we write $\mathcal{E}_1\approx_{ps}^{sl}\mathcal{E}_2$.
\end{definition}

\begin{definition}[Weak probabilistic static location (hereditary) history-preserving bisimulation]
A weak probabilistic static location history-preserving (hp-) bisimulation is a weakly posetal relation $R_{\varphi}\subseteq\mathcal{C}(\mathcal{E}_1)\overline{\times}\mathcal{C}(\mathcal{E}_2)$ such that (1) if
$(C_1,f,C_2)\in R_{\varphi}$, and $C_1\xRightarrow[u]{e_1} C_1'$, then $C_2\xRightarrow[v]{e_2} C_2'$, with $(C_1',f[e_1\mapsto e_2],C_2')\in R_{\varphi\cup\{(u,v)\}}$, and vice-versa; (2) if $(C_1,f, C_2)\in R_{\varphi}$, and
$C_1\xrsquigarrow{\pi}C_1^{\pi}$ then $C_2\xrsquigarrow{\pi}C_2^{\pi}$ and $(C_1^{\pi},f,C_2^{\pi})\in R_{\varphi}$, and vice-versa; (3) if $(C_1,f,C_2)\in R_{\varphi}$,
then $\mu(C_1,C)=\mu(C_2,C)$ for each $C\in\mathcal{C}(\mathcal{E})/R_{\varphi}$; (4) $[\surd]_{R_{\varphi}}=\{\surd\}$. $\mathcal{E}_1,\mathcal{E}_2$ are weak probabilistic
static location history-preserving (hp-)bisimilar and are written $\mathcal{E}_1\approx_{php}^{sl}\mathcal{E}_2$ if there exists a weak probabilistic static location hp-bisimulation $R_{\varphi}$ such that $(\emptyset,\emptyset,\emptyset)\in R_{\varphi}$.

A weak probabilistic static location hereditary history-preserving (hhp-)bisimulation is a downward closed weak probabilistic static location hp-bisimulation. $\mathcal{E}_1,\mathcal{E}_2$ are weak probabilistic static location hereditary history-preserving
(hhp-)bisimilar and are written $\mathcal{E}_1\approx_{phhp}^{sl}\mathcal{E}_2$.
\end{definition}

\begin{definition}[Branching probabilistic static location pomset, step bisimulation]
Assume a special termination predicate $\downarrow$, and let $\surd$ represent a state with $\surd\downarrow$. Let $\mathcal{E}_1$, $\mathcal{E}_2$ be PESs. A branching probabilistic static location pomset
bisimulation is a relation $R_{\varphi}\subseteq\mathcal{C}(\mathcal{E}_1)\times\mathcal{C}(\mathcal{E}_2)$, such that:
 \begin{enumerate}
   \item if $(C_1,C_2)\in R_{\varphi}$, and $C_1\xrightarrow[u]{X}C_1'$ then
   \begin{itemize}
     \item either $X\equiv \tau^*$, and $(C_1',C_2)\in R_{\varphi}$;
     \item or there is a sequence of (zero or more) probabilistic transitions and $\tau$-transitions $C_2\rightsquigarrow^*\xrightarrow{\tau^*} C_2^0$, such that $(C_1,C_2^0)\in R_{\varphi}$ and $C_2^0\xRightarrow[v]{X}C_2'$ with
     $(C_1',C_2')\in R_{\varphi\cup\{(u,v)\}}$;
   \end{itemize}
   \item if $(C_1,C_2)\in R_{\varphi}$, and $C_2\xrightarrow[v]{X}C_2'$ then
   \begin{itemize}
     \item either $X\equiv \tau^*$, and $(C_1,C_2')\in R_{\varphi}$;
     \item or there is a sequence of (zero or more) probabilistic transitions and $\tau$-transitions $C_2\rightsquigarrow^*\xrightarrow{\tau^*} C_1^0$, such that $(C_1^0,C_2)\in R_{\varphi}$ and $C_1^0\xRightarrow[u]{X}C_1'$ with
     $(C_1',C_2')\in R_{\varphi\cup\{(u,v)\}}$;
   \end{itemize}
   \item if $(C_1,C_2)\in R_{\varphi}$ and $C_1\downarrow$, then there is a sequence of (zero or more) probabilistic transitions and $\tau$-transitions $C_2\rightsquigarrow^*\xrightarrow{\tau^*} C_2^0$ such that $(C_1,C_2^0)\in R_{\varphi}$
   and $C_2^0\downarrow$;
   \item if $(C_1,C_2)\in R_{\varphi}$ and $C_2\downarrow$, then there is a sequence of (zero or more) probabilistic transitions and $\tau$-transitions $C_2\rightsquigarrow^*\xrightarrow{\tau^*} C_1^0$ such that $(C_1^0,C_2)\in R_{\varphi}$
   and $C_1^0\downarrow$;
   \item if $(C_1,C_2)\in R_{\varphi}$,then $\mu(C_1,C)=\mu(C_2,C)$ for each $C\in\mathcal{C}(\mathcal{E})/R_{\varphi}$;
   \item $[\surd]_{R_{\varphi}}=\{\surd\}$.
 \end{enumerate}

We say that $\mathcal{E}_1$, $\mathcal{E}_2$ are branching probabilistic static location pomset bisimilar, written $\mathcal{E}_1\approx_{pbp}^{sl}\mathcal{E}_2$, if there exists a branching probabilistic static location pomset bisimulation $R_{\varphi}$,
such that $(\emptyset,\emptyset)\in R_{\varphi}$.

By replacing pomset transitions with steps, we can get the definition of branching probabilistic static location step bisimulation. When PESs $\mathcal{E}_1$ and $\mathcal{E}_2$ are branching probabilistic static location step bisimilar,
we write $\mathcal{E}_1\approx_{pbs}^{sl}\mathcal{E}_2$.
\end{definition}

\begin{definition}[Rooted branching probabilistic static location pomset, step bisimulation]
Assume a special termination predicate $\downarrow$, and let $\surd$ represent a state with $\surd\downarrow$. Let $\mathcal{E}_1$, $\mathcal{E}_2$ be PESs. A rooted branching probabilistic static location pomset
bisimulation is a relation $R_{\varphi}\subseteq\mathcal{C}(\mathcal{E}_1)\times\mathcal{C}(\mathcal{E}_2)$, such that:
 \begin{enumerate}
   \item if $(C_1,C_2)\in R_{\varphi}$, and $C_1\rightsquigarrow C_1^{\pi}\xrightarrow[u]{X}C_1'$ then $C_2\rightsquigarrow C_2^{\pi}\xrightarrow[v]{X}C_2'$ with $C_1'\approx_{pbp}^{sl}C_2'$;
   \item if $(C_1,C_2)\in R_{\varphi}$, and $C_2\rightsquigarrow C_2^{\pi}\xrightarrow[v]{X}C_2'$ then $C_1\rightsquigarrow C_1^{\pi}\xrightarrow[u]{X}C_1'$ with $C_1'\approx_{pbp}^{sl}C_2'$;
   \item if $(C_1,C_2)\in R_{\varphi}$ and $C_1\downarrow$, then $C_2\downarrow$;
   \item if $(C_1,C_2)\in R_{\varphi}$ and $C_2\downarrow$, then $C_1\downarrow$.
 \end{enumerate}

We say that $\mathcal{E}_1$, $\mathcal{E}_2$ are rooted branching probabilistic static location pomset bisimilar, written $\mathcal{E}_1\approx_{prbp}^{sl}\mathcal{E}_2$, if there exists a rooted branching probabilistic static location pomset
bisimulation $R_{\varphi}$, such that $(\emptyset,\emptyset)\in R_{\varphi}$.

By replacing pomset transitions with steps, we can get the definition of rooted branching probabilistic static location step bisimulation. When PESs $\mathcal{E}_1$ and $\mathcal{E}_2$ are rooted branching probabilistic static location step
bisimilar, we write $\mathcal{E}_1\approx_{prbs}^{sl}\mathcal{E}_2$.
\end{definition}

\begin{definition}[Branching probabilistic static location (hereditary) history-preserving bisimulation]
Assume a special termination predicate $\downarrow$, and let $\surd$ represent a state with $\surd\downarrow$. A branching probabilistic static location history-preserving (hp-) bisimulation is a posetal
relation $R_{\varphi}\subseteq\mathcal{C}(\mathcal{E}_1)\overline{\times}\mathcal{C}(\mathcal{E}_2)$ such that:

 \begin{enumerate}
   \item if $(C_1,f,C_2)\in R$, and $C_1\xrightarrow[u]{e_1}C_1'$ then
   \begin{itemize}
     \item either $e_1\equiv \tau$, and $(C_1',f[e_1\mapsto \tau],C_2)\in R_{\varphi}$;
     \item or there is a sequence of (zero or more) (zero or more) probabilistic transitions and $\tau$-transitions $C_2\rightsquigarrow^*\xrightarrow{\tau^*} C_2^0$, such that $(C_1,f,C_2^0)\in R_{\varphi}$ and $C_2^0\xrightarrow[v]{e_2}C_2'$ with
     $(C_1',f[e_1\mapsto e_2],C_2')\in R_{\varphi\cup\{(u,v)\}}$;
   \end{itemize}
   \item if $(C_1,f,C_2)\in R_{\varphi}$, and $C_2\xrightarrow[v]{e_2}C_2'$ then
   \begin{itemize}
     \item either $X\equiv \tau$, and $(C_1,f[e_2\mapsto \tau],C_2')\in R_{\varphi}$;
     \item or there is a sequence of (zero or more) probabilistic transitions and $\tau$-transitions $C_1\rightsquigarrow^*\xrightarrow{\tau^*} C_1^0$, such that $(C_1^0,f,C_2)\in R_{\varphi}$ and $C_1^0\xrightarrow[u]{e_1}C_1'$ with
     $(C_1',f[e_2\mapsto e_1],C_2')\in R_{\varphi\cup\{(u,v)\}}$;
   \end{itemize}
   \item if $(C_1,f,C_2)\in R_{\varphi}$ and $C_1\downarrow$, then there is a sequence of (zero or more) probabilistic transitions and $\tau$-transitions $C_2\rightsquigarrow^*\xrightarrow{\tau^*} C_2^0$ such that $(C_1,f,C_2^0)\in R_{\varphi}$
   and $C_2^0\downarrow$;
   \item if $(C_1,f,C_2)\in R_{\varphi}$ and $C_2\downarrow$, then there is a sequence of (zero or more) probabilistic transitions and $\tau$-transitions $C_1\rightsquigarrow^*\xrightarrow{\tau^*} C_1^0$ such that $(C_1^0,f,C_2)\in R_{\varphi}$
   and $C_1^0\downarrow$;
   \item if $(C_1,C_2)\in R_{\varphi}$,then $\mu(C_1,C)=\mu(C_2,C)$ for each $C\in\mathcal{C}(\mathcal{E})/R_{\varphi}$;
   \item $[\surd]_{R_{\varphi}}=\{\surd\}$.
 \end{enumerate}

$\mathcal{E}_1,\mathcal{E}_2$ are branching probabilistic static location history-preserving (hp-)bisimilar and are written $\mathcal{E}_1\approx_{pbhp}^{sl}\mathcal{E}_2$ if there exists a branching probabilistic static location hp-bisimulation
$R_{\varphi}$ such that $(\emptyset,\emptyset,\emptyset)\in R_{\varphi}$.

A branching probabilistic static location hereditary history-preserving (hhp-)bisimulation is a downward closed branching probabilistic static location hhp-bisimulation. $\mathcal{E}_1,\mathcal{E}_2$ are branching probabilistic static location hereditary history-preserving
(hhp-)bisimilar and are written $\mathcal{E}_1\approx_{pbhhp}^{sl}\mathcal{E}_2$.
\end{definition}

\begin{definition}[Rooted branching probabilistic static location (hereditary) history-preserving bisimulation]
Assume a special termination predicate $\downarrow$, and let $\surd$ represent a state with $\surd\downarrow$. A rooted branching probabilistic static location history-preserving (hp-) bisimulation is a posetal relation $R_{\varphi}\subseteq\mathcal{C}(\mathcal{E}_1)\overline{\times}\mathcal{C}(\mathcal{E}_2)$ such that:

 \begin{enumerate}
   \item if $(C_1,f,C_2)\in R_{\varphi}$, and $C_1\rightsquigarrow C_1^{\pi}\xrightarrow[u]{e_1}C_1'$, then $C_2\rightsquigarrow C_2^{\pi}\xrightarrow[v]{e_2}C_2'$ with $C_1'\approx_{pbhp}^{sl}C_2'$;
   \item if $(C_1,f,C_2)\in R_{\varphi}$, and $C_2\rightsquigarrow C_2^{\pi}\xrightarrow[v]{e_2}C_2'$, then $C_1\rightsquigarrow C_1^{\pi}\xrightarrow[u]{e_1}C_1'$ with $C_1'\approx_{pbhp}^{sl}C_2'$;
   \item if $(C_1,f,C_2)\in R_{\varphi}$ and $C_1\downarrow$, then $C_2\downarrow$;
   \item if $(C_1,f,C_2)\in R_{\varphi}$ and $C_2\downarrow$, then $C_1\downarrow$.
 \end{enumerate}

$\mathcal{E}_1,\mathcal{E}_2$ are rooted branching probabilistic static location history-preserving (hp-)bisimilar and are written $\mathcal{E}_1\approx_{prbhp}^{sl}\mathcal{E}_2$ if there exists a rooted branching
probabilistic static location hp-bisimulation $R_{\varphi}$ such that $(\emptyset,\emptyset,\emptyset)\in R_{\varphi}$.

A rooted branching probabilistic static location hereditary history-preserving (hhp-)bisimulation is a downward closed rooted branching probabilistic static location hp-bisimulation. $\mathcal{E}_1,\mathcal{E}_2$ are rooted branching
probabilistic static location hereditary history-preserving (hhp-)bisimilar and are written $\mathcal{E}_1\approx_{prbhhp}^{sl}\mathcal{E}_2$.
\end{definition}

\subsubsection{BAPTC with Localities}

Let $Loc$ be the set of locations, and $loc\in Loc$, $u,v\in Loc^*$, $\epsilon$ is the empty location.

%\subsubsection{Axiom System of BATC with probabilistic static localities}

In the following, let $e_1, e_2, e_1', e_2'\in \mathbb{E}$, and let variables $x,y,z$ range over the set of terms for true concurrency, $p,q,s$ range over the set of closed terms.
The set of axioms of BATC with probabilistic static localities ($BAPTC^{sl}$) consists of the laws given in Table \ref{AxiomsForBATC22}.

\begin{center}
    \begin{table}
        \begin{tabular}{@{}ll@{}}
            \hline No. &Axiom\\
            $A1$ & $x+ y = y+ x$\\
            $A2$ & $(x+ y)+ z = x+ (y+ z)$\\
            $A3$ & $x+ x = x$\\
            $A4$ & $(x+ y)\cdot z = x\cdot z + y\cdot z$\\
            $A5$ & $(x\cdot y)\cdot z = x\cdot(y\cdot z)$\\
            $L1$ & $\epsilon::x=x$\\
            $L2$ & $u::(x\cdot y)=u::x\cdot u::y$\\
            $L3$ & $u::(x+ y)=u::x+ u::y$\\
            $L4$ & $u::(v::x)=uv::x$\\
            $PA1$ & $x\boxplus_{\pi} y=y\boxplus_{1-\pi} x$\\
            $PA2$ & $x\boxplus_{\pi}(y\boxplus_{\rho} z)=(x\boxplus_{\frac{\pi}{\pi+\rho-\pi\rho}}y)\boxplus_{\pi+\rho-\pi\rho} z$\\
            $PA3$ & $x\boxplus_{\pi}x=x$\\
            $PA4$ & $(x\boxplus_{\pi}y)\cdot z=x\cdot z\boxplus_{\pi}y\cdot z$\\
            $PA5$ & $(x\boxplus_{\pi}y)+z=(x+z)\boxplus_{\pi}(y+z)$\\
        \end{tabular}
        \caption{Axioms of BATC with probabilistic static localities}
        \label{AxiomsForBATC22}
    \end{table}
\end{center}

\begin{definition}[Basic terms of BATC with probabilistic static localities]\label{BTBATC22}
The set of basic terms of BATC with probabilistic static localities, $\mathcal{B}(BAPTC^{sl})$, is inductively defined as follows:
\begin{enumerate}
  \item $\mathbb{E}\subset\mathcal{B}(BAPTC^{sl})$;
  \item if $u\in Loc^*, t\in\mathcal{B}(BAPTC^{sl})$ then $u::t\in\mathcal{B}(BAPTC^{sl})$;
  \item if $e\in \mathbb{E}, t\in\mathcal{B}(BAPTC^{sl})$ then $e\cdot t\in\mathcal{B}(BAPTC^{sl})$;
  \item if $t,s\in\mathcal{B}(BAPTC^{sl})$ then $t+ s\in\mathcal{B}(BAPTC^{sl})$;
  \item if $t,s\in\mathcal{B}(BAPTC)$ then $t\boxplus_{\pi} s\in\mathcal{B}(BAPTC)$.
\end{enumerate}
\end{definition}

\begin{theorem}[Elimination theorem of BATC with probabilistic static localities]\label{ETBATC22}
Let $p$ be a closed BATC with probabilistic static localities term. Then there is a basic BATC with probabilistic static localities term $q$ such that $BAPTC^{sl}\vdash p=q$.
\end{theorem}

In this subsection, we will define a term-deduction system which gives the operational semantics of BATC with probabilistic static localities. We give the operational transition rules for operators
$\cdot$ and $+$ as Table \ref{SETRForBATC22} shows. And the predicate $\xrightarrow[u]{e}\surd$ represents successful termination after execution of the event $e$ at the location $u$. Like the way in \cite{PPA}, we also introduce the counterpart $\breve{e}$
of the event $e$, and also the set $\breve{\mathbb{E}}=\{\breve{e}|e\in\mathbb{E}\}$.

Firstly, we give the definition of PDFs in Table \ref{PDFBAPTC22}.

\begin{center}
    \begin{table}
        $$\mu(e,\breve{e})=1$$
        $$\mu(x\cdot y, x'\cdot y)=\mu(x,x')$$
        $$\mu(x+y,x'+y')=\mu(x,x')\cdot \mu(y,y')$$
        $$\mu(x\boxplus_{\pi}y,z)=\pi\mu(x,z)+(1-\pi)\mu(y,z)$$
        $$\mu(x,y)=0,\textrm{otherwise}$$
        \caption{PDF definitions of $BAPTC$}
        \label{PDFBAPTC22}
    \end{table}
\end{center}

\begin{center}
    \begin{table}
        $$\frac{}{e\rightsquigarrow\breve{e}}$$
        $$\frac{x\rightsquigarrow x'}{x\cdot y\rightsquigarrow x'\cdot y}$$
        $$\frac{x\rightsquigarrow x'\quad y\rightsquigarrow y'}{x+y\rightsquigarrow x'+y'}$$
        $$\frac{x\rightsquigarrow x'}{x\boxplus_{\pi}y\rightsquigarrow x'}\quad \frac{y\rightsquigarrow y'}{x\boxplus_{\pi}y\rightsquigarrow y'}$$
        $$\frac{}{\breve{e}\xrightarrow[\epsilon]{e}\surd}\quad \frac{}{loc::\breve{e}\xrightarrow[loc]{e}\surd}$$
        $$\frac{x\xrightarrow[u]{e}x'}{loc::x\xrightarrow[loc\ll u]{e}loc::x'}$$
        $$\frac{x\xrightarrow[u]{e}\surd}{x+ y\xrightarrow[u]{e}\surd} \quad\frac{x\xrightarrow[u]{e}x'}{x+ y\xrightarrow[u]{e}x'} \quad\frac{y\xrightarrow[u]{e}\surd}{x+ y\xrightarrow[u]{e}\surd} \quad\frac{y\xrightarrow[u]{e}y'}{x+ y\xrightarrow[u]{e}y'}$$
        $$\frac{x\xrightarrow[u]{e}\surd}{x\cdot y\xrightarrow[u]{e} y} \quad\frac{x\xrightarrow[u]{e}x'}{x\cdot y\xrightarrow[u]{e}x'\cdot y}$$
        \caption{Single event transition rules of BATC with probabilistic static localities}
        \label{SETRForBATC22}
    \end{table}
\end{center}

\begin{theorem}[Congruence of BATC with probabilistic static localities with respect to probabilistic static location pomset bisimulation equivalence]
Probabilistic static location pomset bisimulation equivalence $\sim_{pp}^{sl}$ is a congruence with respect to BATC with probabilistic static localities.
\end{theorem}

\begin{theorem}[Soundness of BATC with probabilistic static localities modulo probabilistic static location pomset bisimulation equivalence]\label{SBATCPBE22}
Let $x$ and $y$ be BATC with probabilistic static localities terms. If $BAPTC^{sl}\vdash x=y$, then $x\sim_{pp}^{sl} y$.
\end{theorem}

\begin{theorem}[Completeness of BATC with probabilistic static localities modulo probabilistic static location pomset bisimulation equivalence]\label{CBATCPBE22}
Let $p$ and $q$ be closed BATC with probabilistic static localities terms, if $p\sim_{pp}^{sl} q$ then $p=q$.
\end{theorem}

\begin{theorem}[Congruence of BATC with probabilistic static localities with respect to probabilistic static location step bisimulation equivalence]
Probabilistic static location step bisimulation equivalence $\sim_{ps}^{sl}$ is a congruence with respect to BATC with probabilistic static localities.
\end{theorem}

\begin{theorem}[Soundness of BATC with probabilistic static localities modulo probabilistic static location step bisimulation equivalence]\label{SBATCSBE22}
Let $x$ and $y$ be BATC with probabilistic static localities terms. If $BAPTC^{sl}\vdash x=y$, then $x\sim_{ps}^{sl} y$.
\end{theorem}

\begin{theorem}[Completeness of BATC with probabilistic static localities modulo probabilistic static location step bisimulation equivalence]\label{CBATCSBE22}
Let $p$ and $q$ be closed BATC with probabilistic static localities terms, if $p\sim_{ps}^{sl} q$ then $p=q$.
\end{theorem}

\begin{theorem}[Congruence of BATC with probabilistic static localities with respect to probabilistic static location hp-bisimulation equivalence]
Probabilistic static location hp-bisimulation equivalence $\sim_{php}^{sl}$ is a congruence with respect to BATC with probabilistic static localities.
\end{theorem}

\begin{theorem}[Soundness of BATC with probabilistic static localities modulo probabilistic static location hp-bisimulation equivalence]\label{SBATCHPBE22}
Let $x$ and $y$ be BATC with probabilistic static localities terms. If $BATC\vdash x=y$, then $x\sim_{php}^{sl} y$.
\end{theorem}

\begin{theorem}[Completeness of BATC with probabilistic static localities modulo probabilistic static location hp-bisimulation equivalence]\label{CBATCHPBE22}
Let $p$ and $q$ be closed BATC with probabilistic static localities terms, if $p\sim_{php}^{sl} q$ then $p=q$.
\end{theorem}

\begin{theorem}[Congruence of BATC with probabilistic static localities with respect to probabilistic static location hhp-bisimulation equivalence]
Probabilistic static location hhp-bisimulation equivalence $\sim_{phhp}^{sl}$ is a congruence with respect to BATC with probabilistic static localities.
\end{theorem}

\begin{theorem}[Soundness of BATC with probabilistic static localities modulo probabilistic static location hhp-bisimulation equivalence]\label{SBATCHHPBE22}
Let $x$ and $y$ be BATC with probabilistic static localities terms. If $BATC\vdash x=y$, then $x\sim_{phhp}^{sl} y$.
\end{theorem}

\begin{theorem}[Completeness of BATC with probabilistic static localities modulo probabilistic static location hhp-bisimulation equivalence]\label{CBATCHHPBE22}
Let $p$ and $q$ be closed BATC with probabilistic static localities terms, if $p\sim_{phhp}^{sl} q$ then $p=q$.
\end{theorem}

\subsubsection{APPTC with Localities}

Firstly, we give the definition of PDFs in Table \ref{PDFAPPTC22}.

\begin{center}
    \begin{table}
        $$\mu(\delta,\breve{\delta})=1$$
        $$\mu(x\between y,x'\parallel y'+x'\mid y')=\mu(x,x')\cdot\mu(y,y')$$
        $$\mu(x\parallel y,x'\leftmerge y+y'\leftmerge x)=\mu(x,x')\cdot \mu(y,y')$$
        $$\mu(x\leftmerge y, x'\leftmerge y)=\mu(x,x')$$
        $$\mu(x\mid y,x'\mid y')=\mu(x,x')\cdot \mu(y,y')$$
        $$\mu(\Theta(x),\Theta(x'))=\mu(x,x')$$
        $$\mu(x\triangleleft y, x'\triangleleft y)=\mu(x,x')$$
        $$\mu(x,y)=0,\textrm{otherwise}$$
        \caption{PDF definitions of $APPTC$}
        \label{PDFAPPTC22}
    \end{table}
\end{center}

We give the transition rules of APTC with probabilistic static localities as Table \ref{TRForAPPTC122} and \ref{TRForAPTC22} shows.

\begin{center}
    \begin{table}
        $$\frac{x\rightsquigarrow x'\quad y\rightsquigarrow y'}{x\between y\rightsquigarrow x'\parallel y'+x'\mid y'}$$
        $$\frac{x\rightsquigarrow x'\quad y\rightsquigarrow y'}{x\parallel y\rightsquigarrow x'\leftmerge y+y'\leftmerge x}$$
        $$\frac{x\rightsquigarrow x'}{x\leftmerge y\rightsquigarrow x'\leftmerge y}$$
        $$\frac{x\rightsquigarrow x'\quad y\rightsquigarrow y'}{x\mid y\rightsquigarrow x'\mid y'}$$
        $$\frac{x\rightsquigarrow x'}{\Theta(x)\rightsquigarrow \Theta(x')}$$
        $$\frac{x\rightsquigarrow x'}{x\triangleleft y\rightsquigarrow x'\triangleleft y}$$
        \caption{Probabilistic transition rules of APTC with probabilistic static localities}
        \label{TRForAPPTC122}
    \end{table}
\end{center}

\begin{center}
    \begin{table}
        $$\frac{x\xrightarrow[u]{e_1}\surd\quad y\xrightarrow[v]{e_2}\surd}{x\parallel y\xrightarrow[u\diamond v]{\{e_1,e_2\}}\surd} \quad\frac{x\xrightarrow[u]{e_1}x'\quad y\xrightarrow[v]{e_2}\surd}{x\parallel y\xrightarrow[u\diamond v]{\{e_1,e_2\}}x'}$$
        $$\frac{x\xrightarrow[u]{e_1}\surd\quad y\xrightarrow[v]{e_2}y'}{x\parallel y\xrightarrow[u\diamond v]{\{e_1,e_2\}}y'} \quad\frac{x\xrightarrow[u]{e_1}x'\quad y\xrightarrow[v]{e_2}y'}{x\parallel y\xrightarrow[u\diamond v]{\{e_1,e_2\}}x'\between y'}$$
        $$\frac{x\xrightarrow[u]{e_1}\surd\quad y\xrightarrow[v]{e_2}\surd \quad(e_1\leq e_2)}{x\leftmerge y\xrightarrow[u\diamond v]{\{e_1,e_2\}}\surd} \quad\frac{x\xrightarrow[u]{e_1}x'\quad y\xrightarrow[v]{e_2}\surd \quad(e_1\leq e_2)}{x\leftmerge y\xrightarrow[u\diamond v]{\{e_1,e_2\}}x'}$$
        $$\frac{x\xrightarrow[u]{e_1}\surd\quad y\xrightarrow[v]{e_2}y' \quad(e_1\leq e_2)}{x\leftmerge y\xrightarrow[u\diamond v]{\{e_1,e_2\}}y'} \quad\frac{x\xrightarrow[u]{e_1}x'\quad y\xrightarrow[v]{e_2}y' \quad(e_1\leq e_2)}{x\leftmerge y\xrightarrow[u\diamond v]{\{e_1,e_2\}}x'\between y'}$$
        $$\frac{x\xrightarrow[u]{e_1}\surd\quad y\xrightarrow[v]{e_2}\surd}{x\mid y\xrightarrow[u\diamond v]{\gamma(e_1,e_2)}\surd} \quad\frac{x\xrightarrow[u]{e_1}x'\quad y\xrightarrow[v]{e_2}\surd}{x\mid y\xrightarrow[u\diamond v]{\gamma(e_1,e_2)}x'}$$
        $$\frac{x\xrightarrow[u]{e_1}\surd\quad y\xrightarrow[v]{e_2}y'}{x\mid y\xrightarrow[u\diamond v]{\gamma(e_1,e_2)}y'} \quad\frac{x\xrightarrow[u]{e_1}x'\quad y\xrightarrow[v]{e_2}y'}{x\mid y\xrightarrow[u\diamond v]{\gamma(e_1,e_2)}x'\between y'}$$
        $$\frac{x\xrightarrow[u]{e_1}\surd\quad (\sharp(e_1,e_2))}{\Theta(x)\xrightarrow[u]{e_1}\surd} \quad\frac{x\xrightarrow[u]{e_2}\surd\quad (\sharp(e_1,e_2))}{\Theta(x)\xrightarrow[u]{e_2}\surd}$$
        $$\frac{x\xrightarrow[u]{e_1}x'\quad (\sharp(e_1,e_2))}{\Theta(x)\xrightarrow[u]{e_1}\Theta(x')} \quad\frac{x\xrightarrow[u]{e_2}x'\quad (\sharp(e_1,e_2))}{\Theta(x)\xrightarrow[u]{e_2}\Theta(x')}$$
        $$\frac{x\xrightarrow[u]{e_1}\surd\quad (\sharp_{\pi}(e_1,e_2))}{\Theta(x)\xrightarrow[u]{e_1}\surd} \quad\frac{x\xrightarrow[u]{e_2}\surd\quad (\sharp_{\pi}(e_1,e_2))}{\Theta(x)\xrightarrow[u]{e_2}\surd}$$
        $$\frac{x\xrightarrow[u]{e_1}x'\quad (\sharp_{\pi}(e_1,e_2))}{\Theta(x)\xrightarrow[u]{e_1}\Theta(x')} \quad\frac{x\xrightarrow[u]{e_2}x'\quad (\sharp_{\pi}(e_1,e_2))}{\Theta(x)\xrightarrow[u]{e_2}\Theta(x')}$$
        \caption{Action transition rules of APTC with probabilistic static localities}
        \label{TRForAPTC22}
    \end{table}
\end{center}

\begin{center}
    \begin{table}
        $$\frac{x\xrightarrow[u]{e_1}\surd \quad y\nrightarrow^{e_2}\quad (\sharp(e_1,e_2))}{x\triangleleft y\xrightarrow[u]{\tau}\surd}
        \quad\frac{x\xrightarrow[u]{e_1}x' \quad y\nrightarrow^{e_2}\quad (\sharp(e_1,e_2))}{x\triangleleft y\xrightarrow[u]{\tau}x'}$$
        $$\frac{x\xrightarrow[u]{e_1}\surd \quad y\nrightarrow^{e_3}\quad (\sharp(e_1,e_2),e_2\leq e_3)}{x\triangleleft y\xrightarrow[u]{e_1}\surd}
        \quad\frac{x\xrightarrow[u]{e_1}x' \quad y\nrightarrow^{e_3}\quad (\sharp(e_1,e_2),e_2\leq e_3)}{x\triangleleft y\xrightarrow[u]{e_1}x'}$$
        $$\frac{x\xrightarrow[u]{e_3}\surd \quad y\nrightarrow^{e_2}\quad (\sharp(e_1,e_2),e_1\leq e_3)}{x\triangleleft y\xrightarrow[u]{\tau}\surd}
        \quad\frac{x\xrightarrow[u]{e_3}x' \quad y\nrightarrow^{e_2}\quad (\sharp(e_1,e_2),e_1\leq e_3)}{x\triangleleft y\xrightarrow[u]{\tau}x'}$$
        $$\frac{x\xrightarrow[u]{e_1}\surd \quad y\nrightarrow^{e_2}\quad (\sharp_{\pi}(e_1,e_2))}{x\triangleleft y\xrightarrow[u]{\tau}\surd}
        \quad\frac{x\xrightarrow[u]{e_1}x' \quad y\nrightarrow^{e_2}\quad (\sharp_{\pi}(e_1,e_2))}{x\triangleleft y\xrightarrow[u]{\tau}x'}$$
        $$\frac{x\xrightarrow[u]{e_1}\surd \quad y\nrightarrow^{e_3}\quad (\sharp_{\pi}(e_1,e_2),e_2\leq e_3)}{x\triangleleft y\xrightarrow[u]{e_1}\surd}
        \quad\frac{x\xrightarrow[u]{e_1}x' \quad y\nrightarrow^{e_3}\quad (\sharp_{\pi}(e_1,e_2),e_2\leq e_3)}{x\triangleleft y\xrightarrow[u]{e_1}x'}$$
        $$\frac{x\xrightarrow[u]{e_3}\surd \quad y\nrightarrow^{e_2}\quad (\sharp_{\pi}(e_1,e_2),e_1\leq e_3)}{x\triangleleft y\xrightarrow[u]{\tau}\surd}
        \quad\frac{x\xrightarrow[u]{e_3}x' \quad y\nrightarrow^{e_2}\quad (\sharp_{\pi}(e_1,e_2),e_1\leq e_3)}{x\triangleleft y\xrightarrow[u]{\tau}x'}$$
        \caption{Action transition rules of APTC with probabilistic static localities (continuing)}
        \label{TRForAPTC222}
    \end{table}
\end{center}

In the following, we show that the elimination theorem does not hold for truly concurrent processes combined the operators $\cdot$, $+$, $\boxplus_{\pi}$ and $\leftmerge$. Firstly, we define the basic terms for APTC with probabilistic static localities.

\begin{definition}[Basic terms of APTC with probabilistic static localities]\label{BTAPTC22}
The set of basic terms of APTC with probabilistic static localities, $\mathcal{B}(APPTC^{sl})$, is inductively defined as follows:
\begin{enumerate}
  \item $\mathbb{E}\subset\mathcal{B}(APPTC^{sl})$;
  \item if $u\in Loc^*, t\in\mathcal{B}(APPTC^{sl})$ then $u::t\in\mathcal{B}(APPTC^{sl})$;
  \item if $e\in \mathbb{E}, t\in\mathcal{B}(APPTC^{sl})$ then $e\cdot t\in\mathcal{B}(APPTC^{sl})$;
  \item if $t,s\in\mathcal{B}(APPTC^{sl})$ then $t+ s\in\mathcal{B}(APPTC^{sl})$;
  \item if $t,s\in\mathcal{B}(APPTC^{sl})$ then $t\boxplus_{\pi} s\in\mathcal{B}(APPTC^{sl})$;
  \item if $t,s\in\mathcal{B}(APPTC^{sl})$ then $t\leftmerge s\in\mathcal{B}(APPTC^{sl})$.
\end{enumerate}
\end{definition}

\begin{theorem}[Congruence theorem of APTC with probabilistic static localities]
Probabilistic static location truly concurrent bisimulation equivalences $\sim_{pp}^{sl}$, $\sim_{ps}^{sl}$, $\sim_{php}^{sl}$ and $\sim_{phhp}^{sl}$ are all congruences with respect to APTC with probabilistic static
localities.
\end{theorem}

\begin{theorem}[Elimination theorem of parallelism]\label{ETParallelism22}
Let $p$ be a closed APTC with probabilistic static localities term. Then there is a basic APTC with probabilistic static localities term $q$ such that $APPTC^{sl}\vdash p=q$.
\end{theorem}

\begin{theorem}[Generalization of APTC with probabilistic static localities with respect to BATC with probabilistic static localities]
APTC with probabilistic static localities is a generalization of BATC with probabilistic static localities.
\end{theorem}

\begin{theorem}[Soundness of APTC with probabilistic static localities modulo probabilistic static location pomset bisimulation equivalence]\label{SPPBE22}
Let $x$ and $y$ be APTC with probabilistic static localities terms. If $APPTC^{sl}\vdash x=y$, then $x\sim_{pp}^{sl} y$.
\end{theorem}

\begin{theorem}[Completeness of APTC with probabilistic static localities modulo probabilistic static location pomset bisimulation equivalence]\label{CPPBE22}
Let $p$ and $q$ be closed APTC with probabilistic static localities terms, if $p\sim_{pp}^{sl} q$ then $p=q$.
\end{theorem}

\begin{theorem}[Soundness of APTC with probabilistic static localities modulo probabilistic static location step bisimulation equivalence]\label{SPSBE22}
Let $x$ and $y$ be APTC with probabilistic static localities terms. If $APPTC^{sl}\vdash x=y$, then $x\sim_{ps}^{sl} y$.
\end{theorem}

\begin{theorem}[Completeness of APTC with probabilistic static localities modulo probabilistic static location step bisimulation equivalence]\label{CPSBE22}
Let $p$ and $q$ be closed APTC with probabilistic static localities terms, if $p\sim_{ps}^{sl} q$ then $p=q$.
\end{theorem}

\begin{theorem}[Soundness of APTC with probabilistic static localities modulo probabilistic static location hp-bisimulation equivalence]\label{SPHPBE22}
Let $x$ and $y$ be APTC with probabilistic static localities terms. If $APPTC^{sl}\vdash x=y$, then $x\sim_{php}^{sl} y$.
\end{theorem}

\begin{theorem}[Completeness of APTC with probabilistic static localities modulo probabilistic static location hp-bisimulation equivalence]\label{CPHPBE22}
Let $p$ and $q$ be closed APTC with probabilistic static localities terms, if $p\sim_{php}^{sl} q$ then $p=q$.
\end{theorem}

\begin{theorem}[Soundness of APTC with probabilistic static localities modulo probabilistic static location hhp-bisimulation equivalence]\label{SPHPBE22}
Let $x$ and $y$ be APTC with probabilistic static localities terms. If $APPTC^{sl}\vdash x=y$, then $x\sim_{phhp}^{sl} y$.
\end{theorem}

\begin{theorem}[Completeness of APTC with probabilistic static localities modulo probabilistic static location hhp-bisimulation equivalence]\label{CPHPBE22}
Let $p$ and $q$ be closed APTC with probabilistic static localities terms, if $p\sim_{phhp}^{sl} q$ then $p=q$.
\end{theorem}

The transition rules of encapsulation operator $\partial_H$ are shown in Table \ref{TRForEncapsulation22}.

\begin{center}
    \begin{table}
        $$\frac{x\rightsquigarrow x'}{\partial_H(x)\rightsquigarrow\partial_H(x')}$$
        $$\frac{x\xrightarrow[u]{e}\surd}{\partial_H(x)\xrightarrow[u]{e}\surd}\quad (e\notin H)\quad\quad\frac{x\xrightarrow[u]{e}x'}{\partial_H(x)\xrightarrow[u]{e}\partial_H(x')}\quad(e\notin H)$$
        \caption{Transition rules of encapsulation operator $\partial_H$}
        \label{TRForEncapsulation22}
    \end{table}
\end{center}

Based on the transition rules for encapsulation operator $\partial_H$ in Table \ref{TRForEncapsulation22}, we design the axioms as Table \ref{AxiomsForEncapsulation22} shows.

\begin{center}
    \begin{table}
        \begin{tabular}{@{}ll@{}}
            \hline No. &Axiom\\
            $D1$ & $e\notin H\quad\partial_H(e) = e$\\
            $D2$ & $e\in H\quad \partial_H(e) = \delta$\\
            $D3$ & $\partial_H(\delta) = \delta$\\
            $D4$ & $\partial_H(x+ y) = \partial_H(x)+\partial_H(y)$\\
            $D5$ & $\partial_H(x\cdot y) = \partial_H(x)\cdot\partial_H(y)$\\
            $D6$ & $\partial_H(x\leftmerge y) = \partial_H(x)\leftmerge\partial_H(y)$\\
            $L11$ & $u::\partial_H(x) = \partial_H(u::x)$\\
            $PD1$ & $\partial_H(x\boxplus_{\pi}y)=\partial_H(x)\boxplus_{\pi}\partial_H(y)$\\
        \end{tabular}
        \caption{Axioms of encapsulation operator}
        \label{AxiomsForEncapsulation22}
    \end{table}
\end{center}

\begin{theorem}[Congruence theorem of encapsulation operator $\partial_H$]
Probabilistic static location truly concurrent bisimulation equivalences $\sim_{pp}^{sl}$, $\sim_{ps}^{sl}$, $\sim_{php}^{sl}$ and $\sim_{phhp}^{sl}$ are all congruences with respect to encapsulation
operator $\partial_H$.
\end{theorem}

\begin{theorem}[Elimination theorem of APTC with probabilistic static localities]\label{ETEncapsulation22}
Let $p$ be a closed APTC with probabilistic static localities term including the encapsulation operator $\partial_H$. Then there is a basic APTC with probabilistic static localities term $q$ such that
$APTC\vdash p=q$.
\end{theorem}

\begin{theorem}[Soundness of APTC with probabilistic static localities modulo probabilistic static location pomset bisimulation equivalence]\label{SAPTCPBE22}
Let $x$ and $y$ be APTC with probabilistic static localities terms including encapsulation operator $\partial_H$. If $APPTC^{sl}\vdash x=y$, then $x\sim_{pp}^{sl} y$.
\end{theorem}

\begin{theorem}[Completeness of APTC with probabilistic static localities modulo probabilistic static location pomset bisimulation equivalence]\label{CAPTCPBE22}
Let $p$ and $q$ be closed APTC with probabilistic static localities terms including encapsulation operator $\partial_H$, if $p\sim_{pp}^{sl} q$ then $p=q$.
\end{theorem}

\begin{theorem}[Soundness of APTC with probabilistic static localities modulo probabilistic static location step bisimulation equivalence]\label{SAPTCSBE22}
Let $x$ and $y$ be APTC with probabilistic static localities terms including encapsulation operator $\partial_H$. If $APPTC^{sl}\vdash x=y$, then $x\sim_{ps}^{sl} y$.
\end{theorem}

\begin{theorem}[Completeness of APTC with probabilistic static localities modulo probabilistic static location step bisimulation equivalence]\label{CAPTCSBE22}
Let $p$ and $q$ be closed APTC with probabilistic static localities terms including encapsulation operator $\partial_H$, if $p\sim_{ps}^{sl} q$ then $p=q$.
\end{theorem}

\begin{theorem}[Soundness of APTC with probabilistic static localities modulo probabilistic static location hp-bisimulation equivalence]\label{SAPTCHPBE22}
Let $x$ and $y$ be APTC with probabilistic static localities terms including encapsulation operator $\partial_H$. If $APPTC^{sl}\vdash x=y$, then $x\sim_{php}^{sl} y$.
\end{theorem}

\begin{theorem}[Completeness of APTC with probabilistic static localities modulo probabilistic static location hp-bisimulation equivalence]\label{CAPTCHPBE22}
Let $p$ and $q$ be closed APTC with probabilistic static localities terms including encapsulation operator $\partial_H$, if $p\sim_{php}^{sl} q$ then $p=q$.
\end{theorem}

\begin{theorem}[Soundness of APTC with probabilistic static localities modulo probabilistic static location hhp-bisimulation equivalence]\label{SAPTCHPBE22}
Let $x$ and $y$ be APTC with probabilistic static localities terms including encapsulation operator $\partial_H$. If $APPTC^{sl}\vdash x=y$, then $x\sim_{phhp}^{sl} y$.
\end{theorem}

\begin{theorem}[Completeness of APTC with probabilistic static localities modulo probabilistic static location hhp-bisimulation equivalence]\label{CAPTCHPBE22}
Let $p$ and $q$ be closed APTC with probabilistic static localities terms including encapsulation operator $\partial_H$, if $p\sim_{phhp}^{sl} q$ then $p=q$.
\end{theorem}

\subsubsection{Recursion}

In the following, $E,F,G$ are recursion specifications, $X,Y,Z$ are recursive variables.

%\subsubsection{Guarded Recursive Specifications}

\begin{definition}[Recursive specification]
A recursive specification is a finite set of recursive equations

$$X_1=t_1(X_1,\cdots,X_n)$$
$$\cdots$$
$$X_n=t_n(X_1,\cdots,X_n)$$

where the left-hand sides of $X_i$ are called recursion variables, and the right-hand sides $t_i(X_1,\cdots,X_n)$ are process terms in APTC with probabilistic static localities with possible
occurrences of the recursion variables $X_1,\cdots,X_n$.
\end{definition}

\begin{definition}[Solution]
Processes $p_1,\cdots,p_n$ are a solution for a recursive specification $\{X_i=t_i(X_1,\cdots,X_n)|i\in\{1,\cdots,n\}\}$ (with respect to probabilistic static location truly concurrent bisimulation equivalences
$\sim_{ps}^{sl}$($\sim_{pp}^{sl}$, $\sim_{php}^{sl}$, $\sim_{phhp}^{sl}$)) if $p_i\sim_{ps}^{sl} (\sim_{pp}^{sl}, \sim_{php}^{sl}, \sim_{phhp}^{sl})t_i(p_1,\cdots,p_n)$ for $i\in\{1,\cdots,n\}$.
\end{definition}

\begin{definition}[Guarded recursive specification]
A recursive specification

$$X_1=t_1(X_1,\cdots,X_n)$$
$$...$$
$$X_n=t_n(X_1,\cdots,X_n)$$

is guarded if the right-hand sides of its recursive equations can be adapted to the form by applications of the axioms in APTC with probabilistic static localities and replacing recursion variables
by the right-hand sides of their recursive equations,

$((u_{111}::a_{111}\leftmerge\cdots\leftmerge u_{11i_1}::a_{11i_1})\cdot s_1(X_1,\cdots,X_n)+\cdots+(u_{1k1}::a_{1k1}\leftmerge\cdots\leftmerge u_{1ki_k}::a_{1ki_k})\cdot s_k(X_1,\cdots,X_n)
+(v_{111}::b_{111}\leftmerge\cdots\leftmerge v_{11j_1}::b_{11j_1})+\cdots+(v_{11j_1}::b_{11j_1}\leftmerge\cdots\leftmerge v_{11j_l}::b_{1lj_l}))\boxplus_{\pi_1}\cdots\boxplus_{\pi_{m-1}}
((u_{m11}::a_{m11}\leftmerge\cdots\leftmerge u_{m1i_1}::a_{m1i_1})\cdot s_1(X_1,\cdots,X_n)+\cdots+(u_{mk1}::a_{mk1}\leftmerge\cdots\leftmerge u_{mki_k}::a_{mki_k})\cdot s_k(X_1,\cdots,X_n)
+(v_{m11}::b_{m11}\leftmerge\cdots\leftmerge v_{m1j_1}::b_{m1j_1})+\cdots+(v_{m1j_1}::b_{m1j_1}\leftmerge\cdots\leftmerge v_{m1j_l}::b_{mlj_l}))$

where $a_{111},\cdots,a_{11i_1},a_{1k1},\cdots,a_{1ki_k},b_{111},\cdots,b_{11j_1},b_{11j_1},\cdots,b_{1lj_l},\cdots, a_{m11},\cdots,a_{m1i_1},a_{mk1},\cdots,a_{mki_k},\\b_{m11},\cdots,
b_{m1j_1},b_{m1j_1},\cdots,b_{mlj_l}\in \mathbb{E}$, and the sum above is allowed to be empty, in which case it
represents the deadlock $\delta$.
\end{definition}

\begin{definition}[Linear recursive specification]\label{LRS22}
A recursive specification is linear if its recursive equations are of the form

$((u_{111}::a_{111}\leftmerge\cdots\leftmerge u_{11i_1}::a_{11i_1})X_1+\cdots+(u_{1k1}::a_{1k1}\leftmerge\cdots\leftmerge u_{1ki_k}::a_{1ki_k})X_k
+(v_{111}::b_{111}\leftmerge\cdots\leftmerge v_{11j_1}::b_{11j_1})+\cdots+(v_{11j_1}::b_{11j_1}\leftmerge\cdots\leftmerge v_{11j_l}::b_{1lj_l}))\boxplus_{\pi_1}\cdots\boxplus_{\pi_{m-1}}
((u_{m11}::a_{m11}\leftmerge\cdots\leftmerge u_{m1i_1}::a_{m1i_1})X_1+\cdots+(u_{mk1}::a_{mk1}\leftmerge\cdots\leftmerge u_{mki_k}::a_{mki_k})X_k
+(v_{m11}::b_{m11}\leftmerge\cdots\leftmerge v_{m1j_1}::b_{m1j_1})+\cdots+(v_{m1j_1}::b_{m1j_1}\leftmerge\cdots\leftmerge v_{m1j_l}::b_{mlj_l}))$

where $a_{111},\cdots,a_{11i_1},a_{1k1},\cdots,a_{1ki_k},b_{111},\cdots,b_{11j_1},b_{11j_1},\cdots,b_{1lj_l},\cdots,a_{m11},\cdots,a_{m1i_1},a_{mk1},\cdots,a_{mki_k},\\b_{m11},\cdots,
b_{m1j_1},b_{m1j_1},\cdots,b_{mlj_l}\in \mathbb{E}$, and the sum above is allowed to be empty, in which case it
represents the deadlock $\delta$.
\end{definition}

Firstly, we give the definition of PDFs in Table \ref{PDFGR22}.

\begin{center}
    \begin{table}
        $$\mu(\langle X|E\rangle,y)=\mu(\langle t_X|E\rangle,y)$$
        $$\mu(x,y)=0,\textrm{otherwise}$$
        \caption{PDF definitions of recursion}
        \label{PDFGR22}
    \end{table}
\end{center}

For a guarded recursive specifications $E$ with the form

$$X_1=t_1(X_1,\cdots,X_n)$$
$$\cdots$$
$$X_n=t_n(X_1,\cdots,X_n)$$

the behavior of the solution $\langle X_i|E\rangle$ for the recursion variable $X_i$ in $E$, where $i\in\{1,\cdots,n\}$, is exactly the behavior of their right-hand sides
$t_i(X_1,\cdots,X_n)$, which is captured by the two transition rules in Table \ref{TRForGR22}.

\begin{center}
    \begin{table}
        $$\frac{t_i(\langle X_1|E\rangle,\cdots,\langle X_n|E\rangle)\xrightarrow[u]{\{e_1,\cdots,e_k\}}\surd}{\langle X_i|E\rangle\xrightarrow[u]{\{e_1,\cdots,e_k\}}\surd}$$
        $$\frac{t_i(\langle X_1|E\rangle,\cdots,\langle X_n|E\rangle)\xrightarrow[u]{\{e_1,\cdots,e_k\}} y}{\langle X_i|E\rangle\xrightarrow[u]{\{e_1,\cdots,e_k\}} y}$$
        \caption{Transition rules of guarded recursion}
        \label{TRForGR22}
    \end{table}
\end{center}

\begin{theorem}[Conservitivity of APTC with probabilistic static localities and guarded recursion]
APTC with probabilistic static localities and guarded recursion is a conservative extension of APTC with probabilistic static localities.
\end{theorem}

\begin{theorem}[Congruence theorem of APTC with probabilistic static localities and guarded recursion]
Probabilistic static location truly concurrent bisimulation equivalences $\sim_{pp}^{sl}$, $\sim_{ps}^{sl}$, $\sim_{php}^{sl}$ and $\sim_{phhp}^{sl}$ are all congruences with respect to APTC with probabilistic static localities and guarded
recursion.
\end{theorem}

The $RDP$ (Recursive Definition Principle) and the $RSP$ (Recursive Specification Principle) are shown in Table \ref{RDPRSP22}.

\begin{center}
\begin{table}
  \begin{tabular}{@{}ll@{}}
\hline No. &Axiom\\
  $RDP$ & $\langle X_i|E\rangle = t_i(\langle X_1|E\rangle,\cdots,\langle X_n|E\rangle)\quad (i\in\{1,\cdots,n\})$\\
  $RSP$ & if $y_i=t_i(y_1,\cdots,y_n)$ for $i\in\{1,\cdots,n\}$, then $y_i=\langle X_i|E\rangle \quad(i\in\{1,\cdots,n\})$\\
\end{tabular}
\caption{Recursive definition and specification principle}
\label{RDPRSP22}
\end{table}
\end{center}

\begin{theorem}[Elimination theorem of APTC with probabilistic static localities and linear recursion]\label{ETRecursion22}
Each process term in APTC with probabilistic static localities and linear recursion is equal to a process term $\langle X_1|E\rangle$ with $E$ a linear recursive specification.
\end{theorem}

\begin{theorem}[Soundness of APTC with probabilistic static localities and guarded recursion]\label{SAPTCR22}
Let $x$ and $y$ be APTC with probabilistic static localities and guarded recursion terms. If $APTC\textrm{ with guarded recursion}\vdash x=y$, then
\begin{enumerate}
  \item $x\sim_{ps}^{sl} y$;
  \item $x\sim_{pp}^{sl} y$;
  \item $x\sim_{php}^{sl} y$;
  \item $x\sim_{phhp}^{sl} y$.
\end{enumerate}
\end{theorem}

\begin{theorem}[Completeness of APTC with probabilistic static localities and linear recursion]\label{CAPTCR22}
Let $p$ and $q$ be closed APTC with probabilistic static localities and linear recursion terms, then,
\begin{enumerate}
  \item if $p\sim_{ps}^{sl} q$ then $p=q$;
  \item if $p\sim_{pp}^{sl} q$ then $p=q$;
  \item if $p\sim_{php}^{sl} q$ then $p=q$;
  \item if $p\sim_{phhp}^{sl} q$ then $p=q$.
\end{enumerate}
\end{theorem}

\subsubsection{Abstraction}

In the following, let the atomic event $e$ range over $\mathbb{E}\cup\{\delta\}\cup\{\tau\}$, and let the communication function
$\gamma:\mathbb{E}\cup\{\tau\}\times \mathbb{E}\cup\{\tau\}\rightarrow \mathbb{E}\cup\{\delta\}$, with each communication involved $\tau$ resulting into $\delta$.

\begin{center}
    \begin{table}
        $$\frac{}{\tau\rightsquigarrow\breve{\tau}}$$
        $$\frac{}{\tau\xrightarrow{\tau}\surd}$$
        \caption{Transition rule of the silent step}
        \label{TRForTau22}
    \end{table}
\end{center}

%\subsubsection{Guarded Linear Recursion}

The silent step $\tau$ as an atomic event, is introduced into $E$. Considering the recursive specification $X=\tau X$, $\tau s$, $\tau\tau s$, and $\tau\cdots s$ are all its solutions,
that is, the solutions make the existence of $\tau$-loops which cause unfairness. To prevent $\tau$-loops, we extend the definition of linear recursive specification
(Definition \ref{LRS22}) to the guarded one.

\begin{definition}[Guarded linear recursive specification]\label{GLRS22}
A recursive specification is linear if its recursive equations are of the form

$((u_{111}::a_{111}\leftmerge\cdots\leftmerge u_{11i_1}::a_{11i_1})X_1+\cdots+(u_{1k1}::a_{1k1}\leftmerge\cdots\leftmerge u_{1ki_k}::a_{1ki_k})X_k
+(v_{111}::b_{111}\leftmerge\cdots\leftmerge v_{11j_1}::b_{11j_1})+\cdots+(v_{11j_1}::b_{11j_1}\leftmerge\cdots\leftmerge v_{11j_l}::b_{1lj_l}))\boxplus_{\pi_1}\cdots\boxplus_{\pi_{m-1}}
((u_{m11}::a_{m11}\leftmerge\cdots\leftmerge u_{m1i_1}::a_{m1i_1})X_1+\cdots+(u_{mk1}::a_{mk1}\leftmerge\cdots\leftmerge u_{mki_k}::a_{mki_k})X_k
+(v_{m11}::b_{m11}\leftmerge\cdots\leftmerge v_{m1j_1}::b_{m1j_1})+\cdots+(v_{m1j_1}::b_{m1j_1}\leftmerge\cdots\leftmerge v_{m1j_l}::b_{mlj_l}))$

where $a_{111},\cdots,a_{11i_1},a_{1k1},\cdots,a_{1ki_k},b_{111},\cdots,b_{11j_1},b_{11j_1},\cdots,b_{1lj_l}\cdots\\
a_{m11},\cdots,a_{m1i_1},a_{mk1},\cdots,a_{mki_k},b_{m11},\cdots,b_{m1j_1},b_{m1j_1},\cdots,b_{mlj_l}\in \mathbb{E}\cup\{\tau\}$,
and the sum above is allowed to be empty, in which case it represents the deadlock $\delta$.

A linear recursive specification $E$ is guarded if there does not exist an infinite sequence of $\tau$-transitions
$\langle X|E\rangle\xrightarrow{\tau}\langle X'|E\rangle\xrightarrow{\tau}\langle X''|E\rangle\xrightarrow{\tau}\cdots$.
\end{definition}

\begin{theorem}[Conservitivity of APTC with probabilistic static localities and silent step and guarded linear recursion]
APTC with probabilistic static localities and silent step and guarded linear recursion is a conservative extension of APTC with probabilistic static localities and linear recursion.
\end{theorem}

\begin{theorem}[Congruence theorem of APTC with probabilistic static localities and silent step and guarded linear recursion]
Rooted branching probabilistic static location truly concurrent bisimulation equivalences $\approx_{prbp}^{sl}$, $\approx_{prbs}^{sl}$ and $\approx_{prbhp}^{sl}$ are all congruences with respect to
APTC with probabilistic static localities and silent step and guarded linear recursion.
\end{theorem}

We design the axioms for the silent step $\tau$ in Table \ref{AxiomsForTau22}.

\begin{center}
\begin{table}
  \begin{tabular}{@{}ll@{}}
\hline No. &Axiom\\
  $B1$ & $(y=y+y,z=z+z)\quad x\cdot((y+\tau\cdot(y+z))\boxplus_{\pi}w)=x\cdot((y+z)\boxplus_{\pi}w)$\\
  $B2$ & $(y=y+y,z=z+z)\quad x\leftmerge((y+\tau\leftmerge(y+z))\boxplus_{\pi}w)=x\leftmerge((y+z)\boxplus_{\pi}w)$\\
  $L13$ & $u::\tau=\tau$\
\end{tabular}
\caption{Axioms of silent step}
\label{AxiomsForTau22}
\end{table}
\end{center}

\begin{theorem}[Elimination theorem of APTC with probabilistic static localities and silent step and guarded linear recursion]\label{ETTau22}
Each process term in APTC with probabilistic static localities and silent step and guarded linear recursion is equal to a process term $\langle X_1|E\rangle$ with $E$ a guarded linear recursive
specification.
\end{theorem}

\begin{theorem}[Soundness of APTC with probabilistic static localities and silent step and guarded linear recursion]\label{SAPTCTAU22}
Let $x$ and $y$ be APTC with probabilistic static localities and silent step and guarded linear recursion terms. If APTC with probabilistic static localities and silent step and guarded linear recursion
$\vdash x=y$, then
\begin{enumerate}
  \item $x\approx_{prbs}^{sl} y$;
  \item $x\approx_{prbp}^{sl} y$;
  \item $x\approx_{prbhp}^{sl} y$;
  \item $x\approx_{prbhhp}^{sl} y$.
\end{enumerate}
\end{theorem}

\begin{theorem}[Completeness of APTC with probabilistic static localities and silent step and guarded linear recursion]\label{CAPTCTAU22}
Let $p$ and $q$ be closed APTC with probabilistic static localities and silent step and guarded linear recursion terms, then,
\begin{enumerate}
  \item if $p\approx_{prbs}^{sl} q$ then $p=q$;
  \item if $p\approx_{prbp}^{sl} q$ then $p=q$;
  \item if $p\approx_{prbhp}^{sl} q$ then $p=q$;
  \item if $p\approx_{prbhhp}^{sl} q$ then $p=q$.
\end{enumerate}
\end{theorem}

The unary abstraction operator $\tau_I$ ($I\subseteq \mathbb{E}$) renames all atomic events in $I$ into $\tau$. APTC with probabilistic static localities and silent step and abstraction operator is
called $APPTC_{\tau}$ with probabilistic static localities. The transition rules of operator $\tau_I$ are shown in Table \ref{TRForAbstraction22}.

\begin{center}
    \begin{table}
        $$\frac{x\rightsquigarrow x'}{\tau_I(x)\rightsquigarrow\tau_I(x')}$$
        $$\frac{x\xrightarrow[u]{e}\surd}{\tau_I(x)\xrightarrow[u]{e}\surd}\quad e\notin I
        \quad\quad\frac{x\xrightarrow[u]{e}x'}{\tau_I(x)\xrightarrow[u]{e}\tau_I(x')}\quad e\notin I$$

        $$\frac{x\xrightarrow[u]{e}\surd}{\tau_I(x)\xrightarrow{\tau}\surd}\quad e\in I
        \quad\quad\frac{x\xrightarrow[u]{e}x'}{\tau_I(x)\xrightarrow{\tau}\tau_I(x')}\quad e\in I$$
        \caption{Transition rule of the abstraction operator}
        \label{TRForAbstraction22}
    \end{table}
\end{center}

\begin{theorem}[Conservitivity of $APPTC_{\tau}$ with probabilistic static localities and guarded linear recursion]
$APPTC_{\tau}$ with probabilistic static localities and guarded linear recursion is a conservative extension of APTC with probabilistic static localities and silent step and guarded linear recursion.
\end{theorem}

\begin{theorem}[Congruence theorem of $APPTC_{\tau}$ with probabilistic static localities and guarded linear recursion]
Rooted branching probabilistic static location truly concurrent bisimulation equivalences $\approx_{prbp}^{sl}$, $\approx_{prbs}^{sl}$, $\approx_{prbhp}^{sl}$ and $\approx_{prbhhp}^{sl}$ are all
congruences with respect to $APPTC_{\tau}$ with probabilistic static localities and guarded linear recursion.
\end{theorem}

We design the axioms for the abstraction operator $\tau_I$ in Table \ref{AxiomsForAbstraction22}.

\begin{center}
\begin{table}
  \begin{tabular}{@{}ll@{}}
\hline No. &Axiom\\
  $TI1$ & $e\notin I\quad \tau_I(e)=e$\\
  $TI2$ & $e\in I\quad \tau_I(e)=\tau$\\
  $TI3$ & $\tau_I(\delta)=\delta$\\
  $TI4$ & $\tau_I(x+y)=\tau_I(x)+\tau_I(y)$\\
  $TI5$ & $\tau_I(x\cdot y)=\tau_I(x)\cdot\tau_I(y)$\\
  $TI6$ & $\tau_I(x\leftmerge y)=\tau_I(x)\leftmerge\tau_I(y)$\\
  $L14$ & $u::\tau_I(x)=\tau_I(u::x)$\\
  $L15$ & $e\notin I\quad \tau_I(u::e)=u::e$\\
  $L16$ & $e\in I\quad \tau_I(u::e)=\tau$\\
  $PTI1$ & $\tau_I(x\boxplus_{\pi}y)=\tau_I(x)\boxplus_{\pi}\tau_I(y)$\\
\end{tabular}
\caption{Axioms of abstraction operator}
\label{AxiomsForAbstraction22}
\end{table}
\end{center}

\begin{theorem}[Soundness of $APPTC_{\tau}$ with probabilistic static localities and guarded linear recursion]\label{SAPTCABS22}
Let $x$ and $y$ be $APPTC_{\tau}$ with probabilistic static localities and guarded linear recursion terms. If $APPTC_{\tau}$ with probabilistic static localities and guarded linear recursion $\vdash x=y$, then
\begin{enumerate}
  \item $x\approx_{prbs}^{sl} y$;
  \item $x\approx_{prbp}^{sl} y$;
  \item $x\approx_{prbhp}^{sl} y$;
  \item $x\approx_{prbhhp}^{sl} y$.
\end{enumerate}
\end{theorem}

Though $\tau$-loops are prohibited in guarded linear recursive specifications in a specifiable way, they can be constructed using the abstraction operator, for example, there exist
$\tau$-loops in the process term $\tau_{\{a\}}(\langle X|X=aX\rangle)$. To avoid $\tau$-loops caused by $\tau_I$ and ensure fairness, we introduce the following recursive verification
rules as Table \ref{RVR22} shows, note that $i_1,\cdots, i_m,j_1,\cdots,j_n\in I\subseteq \mathbb{E}\setminus\{\tau\}$.

\begin{center}
\begin{table}
    $$VR_1\quad \frac{x=y+(u_1::i_1\leftmerge\cdots\leftmerge u_m::i_m)\cdot x, y=y+y}{\tau\cdot\tau_I(x)=\tau\cdot \tau_I(y)}$$
    $$VR_2\quad \frac{x=z\boxplus_{\pi}(u+(u_1::i_1\leftmerge\cdots\leftmerge u_m::i_m)\cdot x),z=z+u,z=z+z}{\tau\cdot\tau_I(x)=\tau\cdot\tau_I(z)}$$
    $$VR_3\quad \frac{x=z+(u_1::i_1\leftmerge\cdots\leftmerge u_m::i_m)\cdot y,y=z\boxplus_{\pi}(u+(v_1::j_1\leftmerge\cdots\leftmerge v_n::j_n)\cdot x), z=z+u,z=z+z}{\tau\cdot\tau_I(x)=\tau\cdot\tau_I(y')\textrm{ for }y'=z\boxplus_{\pi}(u+(u_1::i_1\leftmerge\cdots\leftmerge u_m::i_m)\cdot y')}$$
\caption{Recursive verification rules}
\label{RVR22}
\end{table}
\end{center}

\begin{theorem}[Soundness of $VR_1,VR_2,VR_3$]
$VR_1$, $VR_2$ and $VR_3$ are sound modulo probabilistic rooted branching probabilistic static location truly concurrent bisimulation equivalences $\approx_{prbp}^{sl}$, $\approx_{prbs}^{sl}$,
$\approx_{prbhp}^{sl}$ and $\approx_{prbhhp}^{sl}$.
\end{theorem}

\subsection{Operational Semantics for Quantum Computing}

in quantum processes, to avoid the abuse of quantum information which may violate the no-cloning theorem, a quantum configuration $\langle C,\varrho \rangle$
\cite{PSQP} \cite{QPA} \cite{QPA2} \cite{CQP} \cite{CQP2} \cite{qCCS} \cite{BQP} \cite{PSQP} \cite{SBQP} is usually consisted of a traditional configuration $C$ and state information $\varrho$ of
all (public) quantum information variables. Though quantum information variables are not explicitly defined and are hidden behind quantum operations or unitary operators, more importantly, the
state information $\varrho$ is the effects of execution of a series of quantum operations or unitary operators on involved quantum systems, the execution of a series of quantum operations
or unitary operators should not only obey the restrictions of the structure of the process terms, but also those of quantum mechanics principles. Through the state information
$\varrho$, we can check and observe the functions of quantum mechanics principles, such as quantum entanglement, quantum measurement, etc.

So, the operational semantics of quantum processes should be defined based on quantum process configuration $\langle C,\varrho\rangle$, in which $\varrho=\varsigma$ of two state
information $\varrho$ and $\varsigma$ means equality under the framework of quantum information and quantum computing, that is, these two quantum processes are in the same quantum
states.

Let $Loc$ be the set of locations, and $u,v\in Loc^*$. Let $\ll$ be the sequential ordering on $Loc^*$, we call $v$ is an extension or a sublocation of $u$ in $u\ll v$; and if $u\nll v$
$v\nll u$, then $u$ and $v$ are independent and denoted $u\diamond v$.

\begin{definition}[Consistent location association]
A relation $\varphi\subseteq (Loc^*\times Loc^*)$ is a consistent location association (cla), if $(u,v)\in \varphi \&(u',v')\in\varphi$, then $u\diamond u'\Leftrightarrow v\diamond v'$.
\end{definition}

\begin{definition}[Pomset transitions and step]
Let $\mathcal{E}$ be a PES and let $C\in\mathcal{C}(\mathcal{E})$, and $\emptyset\neq X\subseteq \mathbb{E}$, if $C\cap X=\emptyset$ and $C'=C\cup X\in\mathcal{C}(\mathcal{E})$, then
$\langle C,s\rangle\xrightarrow[u]{X} \langle C',s'\rangle$ is called a pomset transition from $\langle C,s\rangle$ to $\langle C',s'\rangle$ at location $u$. When the events in $X$ are pairwise
concurrent, we say that $\langle C,s\rangle\xrightarrow[u]{X}\langle C',s'\rangle$ is a step. It is obvious that $\rightarrow^*\xrightarrow[u]{X}\rightarrow^*=\xrightarrow[u]{X}$ and
$\rightarrow^*\xrightarrow[u]{e}\rightarrow^*=\xrightarrow[u]{e}$ for any $e\in\mathbb{E}$ and $X\subseteq\mathbb{E}$.
\end{definition}

\begin{definition}[Weak pomset transitions and weak step]
Let $\mathcal{E}$ be a PES and let $C\in\mathcal{C}(\mathcal{E})$, and $\emptyset\neq X\subseteq \hat{\mathbb{E}}$, if $C\cap X=\emptyset$ and
$\hat{C'}=\hat{C}\cup X\in\mathcal{C}(\mathcal{E})$, then $\langle C,\varrho\rangle\xRightarrow[u]{X} \langle C',\varrho'\rangle$ is called a weak pomset transition from $\langle C,\varrho\rangle$ to
$\langle C',\varrho'\rangle$ at location $u$, where we define $\xRightarrow[u]{e}\triangleq\xrightarrow{\tau^*}\xrightarrow[u]{e}\xrightarrow{\tau^*}$. And
$\xRightarrow[u]{X}\triangleq\xrightarrow{\tau^*}\xrightarrow[u]{e}\xrightarrow{\tau^*}$, for every $e\in X$. When the events in $X$ are pairwise concurrent, we say that
$\langle C,\varrho\rangle\xRightarrow[u]{X}\langle C',\varrho'\rangle$ is a weak step.
\end{definition}

\begin{definition}[Probabilistic transitions]
Let $\mathcal{E}$ be a PES and let $C\in\mathcal{C}(\mathcal{E})$, the transition $\langle C,\varrho\rangle\xrsquigarrow{\pi} \langle C^{\pi},\varrho\rangle$ is called a probabilistic
transition
from $\langle C,\varrho\rangle$ to $\langle C^{\pi},\varrho\rangle$.
\end{definition}

We will also suppose that all the PESs in this chapter are image finite, that is, for any PES $\mathcal{E}$ and $C\in \mathcal{C}(\mathcal{E})$ and $a\in \Lambda$,
$\{\langle C,\varrho\rangle\xrsquigarrow{\pi} \langle C^{\pi},\varrho\rangle\}$,
$\{e\in \mathbb{E}|\langle C,\varrho\rangle\xrightarrow[u]{e} \langle C',\varrho'\rangle\wedge \lambda(e)=a\}$ and
$\{e\in\hat{\mathbb{E}}|\langle C,\varrho\rangle\xRightarrow[u]{e} \langle C',\varrho'\rangle\wedge \lambda(e)=a\}$ is finite.

\begin{definition}[Static location pomset, step bisimulation]
Let $\mathcal{E}_1$, $\mathcal{E}_2$ be PESs. A static location pomset bisimulation is a relation $R_{\varphi}\subseteq\langle\mathcal{C}(\mathcal{E}_1),S\rangle\times\langle\mathcal{C}(\mathcal{E}_2),S\rangle$,
such that if $(\langle C_1,\varrho\rangle,\langle C_2,\varrho\rangle)\in R_{\varphi}$, and $\langle C_1,\varrho\rangle\xrightarrow[u]{X_1}\langle C_1',\varrho'\rangle$ then
$\langle C_2,\varrho\rangle\xrightarrow[v]{X_2}\langle C_2',\varrho'\rangle$, with $X_1\subseteq \mathbb{E}_1$, $X_2\subseteq \mathbb{E}_2$, $X_1\sim X_2$ and
$(\langle C_1',\varrho'\rangle,\langle C_2',\varrho'\rangle)\in R_{\varphi\cup\{(u,v)\}}$ for all $\varrho,\varrho'\in S$, and vice-versa. We say that $\mathcal{E}_1$, $\mathcal{E}_2$ are static location pomset bisimilar, written
$\mathcal{E}_1\sim_p^{sl}\mathcal{E}_2$, if there exists a static location pomset bisimulation $R_{\varphi}$, such that $(\langle\emptyset,\emptyset\rangle,\langle\emptyset,\emptyset\rangle)\in R_{\varphi}$. By replacing
pomset transitions with steps, we can get the definition of static location step bisimulation. When PESs $\mathcal{E}_1$ and $\mathcal{E}_2$ are static location step bisimilar, we write
$\mathcal{E}_1\sim_s^{sl}\mathcal{E}_2$.
\end{definition}

\begin{definition}[Weak static location pomset, step bisimulation]
Let $\mathcal{E}_1$, $\mathcal{E}_2$ be PESs. A weak static location pomset bisimulation is a relation
$R_{\varphi}\subseteq\langle\mathcal{C}(\mathcal{E}_1),S\rangle\times\langle\mathcal{C}(\mathcal{E}_2),S\rangle$, such that if $(\langle C_1,\varrho\rangle,\langle C_2,\varrho\rangle)\in R_{\varphi}$, and
$\langle C_1,\varrho\rangle\xRightarrow[u]{X_1}\langle C_1',\varrho'\rangle$ then $\langle C_2,\varrho\rangle\xRightarrow[v]{X_2}\langle C_2',\varrho'\rangle$, with $X_1\subseteq \hat{\mathbb{E}_1}$,
$X_2\subseteq \hat{\mathbb{E}_2}$, $X_1\sim X_2$ and $(\langle C_1',\varrho'\rangle,\langle C_2',\varrho'\rangle)\in R_{\varphi\cup\{(u,v)\}}$ for all $\varrho,\varrho'\in S$, and vice-versa. We say that $\mathcal{E}_1$,
$\mathcal{E}_2$ are weak static location pomset bisimilar, written $\mathcal{E}_1\approx_p^{sl}\mathcal{E}_2$, if there exists a weak static location pomset bisimulation $R_{\varphi}$, such that
$(\langle\emptyset,\emptyset\rangle,\langle\emptyset,\emptyset\rangle)\in R_{\varphi}$. By replacing weak pomset transitions with weak steps, we can get the definition of weak static location step bisimulation.
When PESs $\mathcal{E}_1$ and $\mathcal{E}_2$ are weak static location step bisimilar, we write $\mathcal{E}_1\approx_s^{sl}\mathcal{E}_2$.
\end{definition}

\begin{definition}[Posetal product]
Given two PESs $\mathcal{E}_1$, $\mathcal{E}_2$, the posetal product of their configurations, denoted
$\langle\mathcal{C}(\mathcal{E}_1),S\rangle\overline{\times}\langle\mathcal{C}(\mathcal{E}_2),S\rangle$, is defined as

$$\{(\langle C_1,\varrho\rangle,f,\langle C_2,\varrho\rangle)|C_1\in\mathcal{C}(\mathcal{E}_1),C_2\in\mathcal{C}(\mathcal{E}_2),f:C_1\rightarrow C_2 \textrm{ isomorphism}\}.$$

A subset $R\subseteq\langle\mathcal{C}(\mathcal{E}_1),S\rangle\overline{\times}\langle\mathcal{C}(\mathcal{E}_2),S\rangle$ is called a posetal relation. We say that $R$ is downward
closed when for any
$(\langle C_1,\varrho\rangle,f,\langle C_2,\varrho\rangle),(\langle C_1',\varrho'\rangle,f',\langle C_2',\varrho'\rangle)\in \langle\mathcal{C}(\mathcal{E}_1),S\rangle\overline{\times}\langle\mathcal{C}(\mathcal{E}_2),S\rangle$,
if $(\langle C_1,\varrho\rangle,f,\langle C_2,\varrho\rangle)\subseteq (\langle C_1',\varrho'\rangle,f',\langle C_2',\varrho'\rangle)$ pointwise and $(\langle C_1',\varrho'\rangle,f',\langle C_2',\varrho'\rangle)\in R$,
then $(\langle C_1,\varrho\rangle,f,\langle C_2,\varrho\rangle)\in R$.

For $f:X_1\rightarrow X_2$, we define $f[x_1\mapsto x_2]:X_1\cup\{x_1\}\rightarrow X_2\cup\{x_2\}$, $z\in X_1\cup\{x_1\}$,(1)$f[x_1\mapsto x_2](z)=
x_2$,if $z=x_1$;(2)$f[x_1\mapsto x_2](z)=f(z)$, otherwise. Where $X_1\subseteq \mathbb{E}_1$, $X_2\subseteq \mathbb{E}_2$, $x_1\in \mathbb{E}_1$, $x_2\in \mathbb{E}_2$.
\end{definition}

\begin{definition}[Weakly posetal product]
Given two PESs $\mathcal{E}_1$, $\mathcal{E}_2$, the weakly posetal product of their configurations, denoted
$\langle\mathcal{C}(\mathcal{E}_1),S\rangle\overline{\times}\langle\mathcal{C}(\mathcal{E}_2),S\rangle$, is defined as

$$\{(\langle C_1,\varrho\rangle,f,\langle C_2,\varrho\rangle)|C_1\in\mathcal{C}(\mathcal{E}_1),C_2\in\mathcal{C}(\mathcal{E}_2),f:\hat{C_1}\rightarrow \hat{C_2} \textrm{ isomorphism}\}.$$

A subset $R\subseteq\langle\mathcal{C}(\mathcal{E}_1),S\rangle\overline{\times}\langle\mathcal{C}(\mathcal{E}_2),S\rangle$ is called a weakly posetal relation. We say that $R$ is
downward closed when for any
$(\langle C_1,\varrho\rangle,f,\langle C_2,\varrho\rangle),(\langle C_1',\varrho'\rangle,f,\langle C_2',\varrho'\rangle)\in \langle\mathcal{C}(\mathcal{E}_1),S\rangle\overline{\times}\langle\mathcal{C}(\mathcal{E}_2),S\rangle$,
if $(\langle C_1,\varrho\rangle,f,\langle C_2,\varrho\rangle)\subseteq (\langle C_1',\varrho'\rangle,f',\langle C_2',\varrho'\rangle)$ pointwise and $(\langle C_1',\varrho'\rangle,f',\langle C_2',\varrho'\rangle)\in R$,
then $(\langle C_1,\varrho\rangle,f,\langle C_2,\varrho\rangle)\in R$.

For $f:X_1\rightarrow X_2$, we define $f[x_1\mapsto x_2]:X_1\cup\{x_1\}\rightarrow X_2\cup\{x_2\}$, $z\in X_1\cup\{x_1\}$,(1)$f[x_1\mapsto x_2](z)=
x_2$,if $z=x_1$;(2)$f[x_1\mapsto x_2](z)=f(z)$, otherwise. Where $X_1\subseteq \hat{\mathbb{E}_1}$, $X_2\subseteq \hat{\mathbb{E}_2}$, $x_1\in \hat{\mathbb{E}}_1$,
$x_2\in \hat{\mathbb{E}}_2$. Also, we define $f(\tau^*)=f(\tau^*)$.
\end{definition}

\begin{definition}[Static location (hereditary) history-preserving bisimulation]
A static location history-preserving (hp-) bisimulation is a posetal relation $R_{\varphi}\subseteq\langle\mathcal{C}(\mathcal{E}_1),S\rangle\overline{\times}\langle\mathcal{C}(\mathcal{E}_2),S\rangle$ such
that if $(\langle C_1,\varrho\rangle,f,\langle C_2,\varrho\rangle)\in R_{\varphi}$, and $\langle C_1,\varrho\rangle\xrightarrow[u]{e_1} \langle C_1',\varrho'\rangle$, then
$\langle C_2,\varrho\rangle\xrightarrow[v]{e_2} \langle C_2',\varrho'\rangle$, with $(\langle C_1',\varrho'\rangle,f[e_1\mapsto e_2],\langle C_2',\varrho'\rangle)\in R_{\varphi}$ for all $\varrho,\varrho'\in S$, and vice-versa.
$\mathcal{E}_1,\mathcal{E}_2$ are static location history-preserving (hp-)bisimilar and are written $\mathcal{E}_1\sim_{hp}^{sl}\mathcal{E}_2$ if there exists a static location hp-bisimulation $R_{\varphi}$ such that
$(\langle\emptyset,\emptyset\rangle,\emptyset,\langle\emptyset,\emptyset\rangle)\in R_{\varphi}$.

A static location hereditary history-preserving (hhp-)bisimulation is a downward closed static location hp-bisimulation. $\mathcal{E}_1,\mathcal{E}_2$ are static location hereditary history-preserving (hhp-)bisimilar and are written
$\mathcal{E}_1\sim_{hhp}^{sl}\mathcal{E}_2$.
\end{definition}

\begin{definition}[Weak static location (hereditary) history-preserving bisimulation]
A weak static location history-preserving (hp-) bisimulation is a weakly posetal relation
$R_{\varphi}\subseteq\langle\mathcal{C}(\mathcal{E}_1),S\rangle\overline{\times}\langle\mathcal{C}(\mathcal{E}_2),S\rangle$ such that if $(\langle C_1,\varrho\rangle,f,\langle C_2,\varrho\rangle)\in R_{\varphi}$, and
$\langle C_1,\varrho\rangle\xRightarrow[u]{e_1} \langle C_1',\varrho'\rangle$, then $\langle C_2,\varrho\rangle\xRightarrow[v]{e_2} \langle C_2',\varrho'\rangle$, with $(\langle C_1',\varrho'\rangle,f[e_1\mapsto e_2],\langle C_2',\varrho'\rangle)\in R_{\varphi}$
for all $\varrho,\varrho'\in S$, and vice-versa. $\mathcal{E}_1,\mathcal{E}_2$ are weak static location history-preserving (hp-)bisimilar and are written $\mathcal{E}_1\approx_{hp}^{sl}\mathcal{E}_2$ if there exists
a weak static location hp-bisimulation $R_{\varphi}$ such that $(\langle\emptyset,\emptyset\rangle,\emptyset,\langle\emptyset,\emptyset\rangle)\in R_{\varphi}$.

A weakly static location hereditary history-preserving (hhp-)bisimulation is a downward closed weak static location hp-bisimulation. $\mathcal{E}_1,\mathcal{E}_2$ are weak static location hereditary history-preserving
(hhp-)bisimilar and are written $\mathcal{E}_1\approx_{hhp}^{sl}\mathcal{E}_2$.
\end{definition}

\begin{definition}[Branching static location pomset, step bisimulation]
Assume a special termination predicate $\downarrow$, and let $\surd$ represent a state with $\surd\downarrow$. Let $\mathcal{E}_1$, $\mathcal{E}_2$ be PESs. A branching static location pomset
bisimulation is a relation $R_{\varphi}\subseteq\langle\mathcal{C}(\mathcal{E}_1),S\rangle\times\langle\mathcal{C}(\mathcal{E}_2),S\rangle$, such that:
 \begin{enumerate}
   \item if $(\langle C_1,\varrho\rangle,\langle C_2,\varrho\rangle)\in R_{\varphi}$, and $\langle C_1,\varrho\rangle\xrightarrow[u]{X}\langle C_1',\varrho'\rangle$ then
   \begin{itemize}
     \item either $X\equiv \tau^*$, and $(\langle C_1',\varrho'\rangle,\langle C_2,\varrho\rangle)\in R_{\varphi}$ with $\varrho'\in \tau(\varrho)$;
     \item or there is a sequence of (zero or more) $\tau$-transitions $\langle C_2,\varrho\rangle\xrightarrow{\tau^*} \langle C_2^0,\varrho^0\rangle$, such that
     $(\langle C_1,\varrho\rangle,\langle C_2^0,\varrho^0\rangle)\in R_{\varphi}$ and $\langle C_2^0,\varrho^0\rangle\xRightarrow[v]{X}\langle C_2',\varrho'\rangle$ with
     $(\langle C_1',\varrho'\rangle,\langle C_2',\varrho'\rangle)\in R_{\varphi\cup\{(u,v)\}}$;
   \end{itemize}
   \item if $(\langle C_1,\varrho\rangle,\langle C_2,\varrho\rangle)\in R_{\varphi}$, and $\langle C_2,\varrho\rangle\xrightarrow[v]{X}\langle C_2',\varrho'\rangle$ then
   \begin{itemize}
     \item either $X\equiv \tau^*$, and $(\langle C_1,\varrho\rangle,\langle C_2',\varrho'\rangle)\in R_{\varphi}$;
     \item or there is a sequence of (zero or more) $\tau$-transitions $\langle C_1,\varrho\rangle\xrightarrow{\tau^*} \langle C_1^0,\varrho^0\rangle$, such that $(\langle C_1^0,\varrho^0\rangle,\langle C_2,\varrho\rangle)\in R_{\varphi}$ and $\langle C_1^0,\varrho^0\rangle\xRightarrow[u]{X}\langle C_1',\varrho'\rangle$ with $(\langle C_1',\varrho'\rangle,\langle C_2',\varrho'\rangle)\in R_{\varphi\cup\{(u,v)\}}$;
   \end{itemize}
   \item if $(\langle C_1,\varrho\rangle,\langle C_2,\varrho\rangle)\in R_{\varphi}$ and $\langle C_1,\varrho\rangle\downarrow$, then there is a sequence of (zero or more) $\tau$-transitions
   $\langle C_2,\varrho\rangle\xrightarrow{\tau^*}\langle C_2^0,\varrho^0\rangle$ such that $(\langle C_1,\varrho\rangle,\langle C_2^0,\varrho^0\rangle)\in R_{\varphi}$ and
   $\langle C_2^0,\varrho^0\rangle\downarrow$;
   \item if $(\langle C_1,\varrho\rangle,\langle C_2,\varrho\rangle)\in R_{\varphi}$ and $\langle C_2,\varrho\rangle\downarrow$, then there is a sequence of (zero or more) $\tau$-transitions
   $\langle C_1,\varrho\rangle\xrightarrow{\tau^*}\langle C_1^0,\varrho^0\rangle$ such that $(\langle C_1^0,\varrho^0\rangle,\langle C_2,\varrho\rangle)\in R_{\varphi}$ and
   $\langle C_1^0,\varrho^0\rangle\downarrow$.
 \end{enumerate}

We say that $\mathcal{E}_1$, $\mathcal{E}_2$ are branching static location pomset bisimilar, written $\mathcal{E}_1\approx_{bp}^{sl}\mathcal{E}_2$, if there exists a branching static location pomset bisimulation $R_{\varphi}$, such
that $(\langle\emptyset,\emptyset\rangle,\langle\emptyset,\emptyset\rangle)\in R_{\varphi}$.

By replacing pomset transitions with steps, we can get the definition of branching static location step bisimulation. When PESs $\mathcal{E}_1$ and $\mathcal{E}_2$ are branching static location step bisimilar, we
write $\mathcal{E}_1\approx_{bs}^{sl}\mathcal{E}_2$.
\end{definition}

\begin{definition}[Rooted branching static location pomset, step bisimulation]
Assume a special termination predicate $\downarrow$, and let $\surd$ represent a state with $\surd\downarrow$. Let $\mathcal{E}_1$, $\mathcal{E}_2$ be PESs. A rooted branching static location pomset bisimulation is a relation $R_{\varphi}\subseteq\langle\mathcal{C}(\mathcal{E}_1),S\rangle\times\langle\mathcal{C}(\mathcal{E}_2),S\rangle$, such that:
 \begin{enumerate}
   \item if $(\langle C_1,\varrho\rangle,\langle C_2,\varrho\rangle)\in R_{\varphi}$, and $\langle C_1,\varrho\rangle\xrightarrow[u]{X}\langle C_1',\varrho'\rangle$ then
   $\langle C_2,\varrho\rangle\xrightarrow[v]{X}\langle C_2',\varrho'\rangle$ with $\langle C_1',\varrho'\rangle\approx_{bp}^{sl}\langle C_2',\varrho'\rangle$;
   \item if $(\langle C_1,\varrho\rangle,\langle C_2,\varrho\rangle)\in R_{\varphi}$, and $\langle C_2,\varrho\rangle\xrightarrow[v]{X}\langle C_2',\varrho'\rangle$ then
   $\langle C_1,\varrho\rangle\xrightarrow[u]{X}\langle C_1',\varrho'\rangle$ with $\langle C_1',\varrho'\rangle\approx_{bp}^{sl}\langle C_2',\varrho'\rangle$;
   \item if $(\langle C_1,\varrho\rangle,\langle C_2,\varrho\rangle)\in R_{\varphi}$ and $\langle C_1,\varrho\rangle\downarrow$, then $\langle C_2,\varrho\rangle\downarrow$;
   \item if $(\langle C_1,\varrho\rangle,\langle C_2,\varrho\rangle)\in R_{\varphi}$ and $\langle C_2,\varrho\rangle\downarrow$, then $\langle C_1,\varrho\rangle\downarrow$.
 \end{enumerate}

We say that $\mathcal{E}_1$, $\mathcal{E}_2$ are rooted branching static location pomset bisimilar, written $\mathcal{E}_1\approx_{rbp}^{sl}\mathcal{E}_2$, if there exists a rooted branching static location pomset
bisimulation $R_{\varphi}$, such that $(\langle\emptyset,\emptyset\rangle,\langle\emptyset,\emptyset\rangle)\in R_{\varphi}$.

By replacing pomset transitions with steps, we can get the definition of rooted branching static location step bisimulation. When PESs $\mathcal{E}_1$ and $\mathcal{E}_2$ are rooted branching static location step
bisimilar, we write $\mathcal{E}_1\approx_{rbs}^{sl}\mathcal{E}_2$.
\end{definition}

\begin{definition}[Branching static location (hereditary) history-preserving bisimulation]
Assume a special termination predicate $\downarrow$, and let $\surd$ represent a state with $\surd\downarrow$. A branching static location history-preserving (hp-) bisimulation is a weakly posetal
relation $R_{\varphi}\subseteq\langle\mathcal{C}(\mathcal{E}_1),S\rangle\overline{\times}\langle\mathcal{C}(\mathcal{E}_2),S\rangle$ such that:

 \begin{enumerate}
   \item if $(\langle C_1,\varrho\rangle,f,\langle C_2,\varrho\rangle)\in R_{\varphi}$, and $\langle C_1,\varrho\rangle\xrightarrow[u]{e_1}\langle C_1',\varrho'\rangle$ then
   \begin{itemize}
     \item either $e_1\equiv \tau$, and $(\langle C_1',\varrho'\rangle,f[e_1\mapsto \tau^{e_1}],\langle C_2,\varrho\rangle)\in R_{\varphi}$;
     \item or there is a sequence of (zero or more) $\tau$-transitions $\langle C_2,\varrho\rangle\xrightarrow{\tau^*} \langle C_2^0,\varrho^0\rangle$, such that
     $(\langle C_1,\varrho\rangle,f,\langle C_2^0,\varrho^0\rangle)\in R_{\varphi}$ and $\langle C_2^0,\varrho^0\rangle\xrightarrow[v]{e_2}\langle C_2',\varrho'\rangle$ with
     $(\langle C_1',\varrho'\rangle,f[e_1\mapsto e_2],\langle C_2',\varrho'\rangle)\in R_{\varphi\cup\{(u,v)\}}$;
   \end{itemize}
   \item if $(\langle C_1,\varrho\rangle,f,\langle C_2,\varrho\rangle)\in R_{\varphi}$, and $\langle C_2,\varrho\rangle\xrightarrow[v]{e_2}\langle C_2',\varrho'\rangle$ then
   \begin{itemize}
     \item either $e_2\equiv \tau$, and $(\langle C_1,\varrho\rangle,f[e_2\mapsto \tau^{e_2}],\langle C_2',\varrho'\rangle)\in R_{\varphi}$;
     \item or there is a sequence of (zero or more) $\tau$-transitions $\langle C_1,\varrho\rangle\xrightarrow{\tau^*} \langle C_1^0,\varrho^0\rangle$, such that
     $(\langle C_1^0,\varrho^0\rangle,f,\langle C_2,\varrho\rangle)\in R_{\varphi}$ and $\langle C_1^0,\varrho^0\rangle\xrightarrow[u]{e_1}\langle C_1',\varrho'\rangle$ with
     $(\langle C_1',\varrho'\rangle,f[e_2\mapsto e_1],\langle C_2',\varrho'\rangle)\in R_{\varphi\cup\{(u,v)\}}$;
   \end{itemize}
   \item if $(\langle C_1,\varrho\rangle,f,\langle C_2,\varrho\rangle)\in R_{\varphi}$ and $\langle C_1,\varrho\rangle\downarrow$, then there is a sequence of (zero or more)
   $\tau$-transitions $\langle C_2,\varrho\rangle\xrightarrow{\tau^*}\langle C_2^0,\varrho^0\rangle$ such that $(\langle C_1,\varrho\rangle,f,\langle C_2^0,\varrho^0\rangle)\in R_{\varphi}$
   and $\langle C_2^0,\varrho^0\rangle\downarrow$;
   \item if $(\langle C_1,\varrho\rangle,f,\langle C_2,\varrho\rangle)\in R_{\varphi}$ and $\langle C_2,\varrho\rangle\downarrow$, then there is a sequence of (zero or more) $\tau$-transitions
   $\langle C_1,\varrho\rangle\xrightarrow{\tau^*}\langle C_1^0,\varrho^0\rangle$ such that $(\langle C_1^0,\varrho^0\rangle,f,\langle C_2,\varrho\rangle)\in R_{\varphi}$ and
   $\langle C_1^0,\varrho^0\rangle\downarrow$.
 \end{enumerate}

$\mathcal{E}_1,\mathcal{E}_2$ are branching static location history-preserving (hp-)bisimilar and are written $\mathcal{E}_1\approx_{bhp}^{sl}\mathcal{E}_2$ if there exists a branching static location hp-bisimulation $R$
such that $(\langle\emptyset,\emptyset\rangle,\emptyset,\langle\emptyset,\emptyset\rangle)\in R_{\varphi}$.

A branching static location hereditary history-preserving (hhp-)bisimulation is a downward closed branching static location hp-bisimulation. $\mathcal{E}_1,\mathcal{E}_2$ are branching static location hereditary history-preserving
(hhp-)bisimilar and are written $\mathcal{E}_1\approx_{bhhp}^{sl}\mathcal{E}_2$.
\end{definition}

\begin{definition}[Rooted branching static location (hereditary) history-preserving bisimulation]
Assume a special termination predicate $\downarrow$, and let $\surd$ represent a state with $\surd\downarrow$. A rooted branching static location history-preserving (hp-) bisimulation is a weakly
posetal relation $R_{\varphi}\subseteq\langle\mathcal{C}(\mathcal{E}_1),S\rangle\overline{\times}\langle\mathcal{C}(\mathcal{E}_2),S\rangle$ such that:

 \begin{enumerate}
   \item if $(\langle C_1,\varrho\rangle,f,\langle C_2,\varrho\rangle)\in R_{\varphi}$, and $\langle C_1,\varrho\rangle\xrightarrow[u]{e_1}\langle C_1',\varrho'\rangle$, then
   $\langle C_2,\varrho\rangle\xrightarrow[v]{e_2}\langle C_2',\varrho'\rangle$ with $\langle C_1',\varrho'\rangle\approx_{bhp}^{sl}\langle C_2',\varrho'\rangle$;
   \item if $(\langle C_1,\varrho\rangle,f,\langle C_2,\varrho\rangle)\in R_{\varphi}$, and $\langle C_2,\varrho\rangle\xrightarrow[v]{e_2}\langle C_2',\varrho'\rangle$, then
   $\langle C_1,\varrho\rangle\xrightarrow[u]{e_1}\langle C_1',\varrho'\rangle$ with $\langle C_1',\varrho'\rangle\approx_{bhp}^{sl}\langle C_2',\varrho'\rangle$;
   \item if $(\langle C_1,\varrho\rangle,f,\langle C_2,\varrho\rangle)\in R_{\varphi}$ and $\langle C_1,\varrho\rangle\downarrow$, then $\langle C_2,\varrho\rangle\downarrow$;
   \item if $(\langle C_1,\varrho\rangle,f,\langle C_2,\varrho\rangle)\in R_{\varphi}$ and $\langle C_2,\varrho\rangle\downarrow$, then $\langle C_1,\varrho\rangle\downarrow$.
 \end{enumerate}

$\mathcal{E}_1,\mathcal{E}_2$ are rooted branching static location history-preserving (hp-)bisimilar and are written $\mathcal{E}_1\approx_{rbhp}^{sl}\mathcal{E}_2$ if there exists a rooted branching
static location hp-bisimulation $R_{\varphi}$ such that $(\langle\emptyset,\emptyset\rangle,\emptyset,\langle\emptyset,\emptyset\rangle)\in R_{\varphi}$.

A rooted branching static location hereditary history-preserving (hhp-)bisimulation is a downward closed rooted branching static location hp-bisimulation. $\mathcal{E}_1,\mathcal{E}_2$ are rooted branching static location hereditary
history-preserving (hhp-)bisimilar and are written $\mathcal{E}_1\approx_{rbhhp}^{sl}\mathcal{E}_2$.
\end{definition}

\begin{definition}[Probabilistic static location pomset, step bisimulation]
Let $\mathcal{E}_1$, $\mathcal{E}_2$ be PESs. A probabilistic static location pomset bisimulation is a relation $R_{\varphi}\subseteq\langle\mathcal{C}(\mathcal{E}_1),S\rangle\times\langle\mathcal{C}(\mathcal{E}_2),S\rangle$,
such that (1) if $(\langle C_1,\varrho\rangle,\langle C_2,\varrho\rangle)\in R_{varphi}$, and $\langle C_1,\varrho\rangle\xrightarrow[u]{X_1}\langle C_1',\varrho'\rangle$ then
$\langle C_2,\varrho\rangle\xrightarrow[v]{X_2}\langle C_2',\varrho'\rangle$, with $X_1\subseteq \mathbb{E}_1$, $X_2\subseteq \mathbb{E}_2$, $X_1\sim X_2$ and
$(\langle C_1',\varrho'\rangle,\langle C_2',\varrho'\rangle)\in R_{\varphi\cup\{(u,v)\}}$ for all $\varrho,\varrho'\in S$, and vice-versa; (2) if $(\langle C_1,\varrho\rangle,\langle C_2,\varrho\rangle)\in R_{\varphi}$, and $\langle C_1,\varrho\rangle\xrsquigarrow{\pi}\langle C_1^{\pi},\varrho\rangle$
then $\langle C_2,\varrho\rangle\xrsquigarrow{\pi}\langle C_2^{\pi},\varrho\rangle$ and $(\langle C_1^{\pi},\varrho\rangle,\langle C_2^{\pi},\varrho\rangle)\in R_{\varphi}$, and vice-versa; (3) if $(\langle C_1,\varrho\rangle,\langle C_2,\varrho\rangle)\in R_{\varphi}$,
then $\mu(C_1,C)=\mu(C_2,C)$ for each $C\in\mathcal{C}(\mathcal{E})/R_{\varphi}$; (4) $[\surd]_{R_{\varphi}}=\{\surd\}$. We say that $\mathcal{E}_1$, $\mathcal{E}_2$ are probabilistic static location pomset bisimilar, written
$\mathcal{E}_1\sim_{pp}^{sl}\mathcal{E}_2$, if there exists a probabilistic static location pomset bisimulation $R_{\varphi}$, such that $(\langle\emptyset,\emptyset\rangle,\langle\emptyset,\emptyset\rangle)\in R_{\varphi}$.
By replacing probabilistic pomset transitions with probabilistic steps, we can get the definition of probabilistic static location step bisimulation. When PESs $\mathcal{E}_1$ and $\mathcal{E}_2$ are
probabilistic static location step bisimilar, we write $\mathcal{E}_1\sim_{ps}^{sl}\mathcal{E}_2$.
\end{definition}

\begin{definition}[Weakly probabilistic static location pomset, step bisimulation]
Let $\mathcal{E}_1$, $\mathcal{E}_2$ be PESs. A weakly probabilistic static location pomset bisimulation is a relation $R_{\varphi}\subseteq\langle\mathcal{C}(\mathcal{E}_1),S\rangle\times\langle\mathcal{C}(\mathcal{E}_2),S\rangle$,
such that (1) if $(\langle C_1,\varrho\rangle,\langle C_2,\varrho\rangle)\in R_{\varphi}$, and $\langle C_1,\varrho\rangle\xRightarrow[u]{X_1}\langle C_1',\varrho'\rangle$ then
$\langle C_2,\varrho\rangle\xRightarrow[v]{X_2}\langle C_2',\varrho'\rangle$, with $X_1\subseteq \hat{\mathbb{E}_1}$, $X_2\subseteq \hat{\mathbb{E}_2}$, $X_1\sim X_2$ and
$(\langle C_1',\varrho'\rangle,\langle C_2',\varrho'\rangle)\in R_{\varphi\cup\{(u,v)\}}$ for all $\varrho,\varrho'\in S$, and vice-versa; (2) if $(\langle C_1,\varrho\rangle,\langle C_2,\varrho\rangle)\in R_{\varphi}$, and $\langle C_1,\varrho\rangle\xrsquigarrow{\pi}\langle C_1^{\pi},\varrho\rangle$
then $\langle C_2,\varrho\rangle\xrsquigarrow{\pi}\langle C_2^{\pi},\varrho\rangle$ and $(\langle C_1^{\pi},\varrho\rangle,\langle C_2^{\pi},\varrho\rangle)\in R_{\varphi}$, and vice-versa; (3) if $(\langle C_1,\varrho\rangle,\langle C_2,\varrho\rangle)\in R_{\varphi}$,
then $\mu(C_1,C)=\mu(C_2,C)$ for each $C\in\mathcal{C}(\mathcal{E})/R_{\varphi}$; (4) $[\surd]_{R_{\varphi}}=\{\surd\}$. We say that $\mathcal{E}_1$, $\mathcal{E}_2$ are weakly probabilistic static location pomset bisimilar,
written $\mathcal{E}_1\approx_{pp}^{sl}\mathcal{E}_2$, if there exists a weakly probabilistic static location pomset bisimulation $R_{\varphi}$, such that
$(\langle\emptyset,\emptyset\rangle,\langle\emptyset,\emptyset\rangle)\in R_{\varphi}$. By replacing weakly probabilistic static location pomset transitions with weakly probabilistic static location steps, we can get the
definition of weakly probabilistic static location step bisimulation. When PESs $\mathcal{E}_1$ and $\mathcal{E}_2$ are weakly probabilistic static location step bisimilar, we write
$\mathcal{E}_1\approx_{ps}^{sl}\mathcal{E}_2$.
\end{definition}

\begin{definition}[Posetal product]
Given two PESs $\mathcal{E}_1$, $\mathcal{E}_2$, the posetal product of their configurations, denoted
$\langle\mathcal{C}(\mathcal{E}_1),S\rangle\overline{\times}\langle\mathcal{C}(\mathcal{E}_2),S\rangle$, is defined as

$$\{(\langle C_1,\varrho\rangle,f,\langle C_2,\varrho\rangle)|C_1\in\mathcal{C}(\mathcal{E}_1),C_2\in\mathcal{C}(\mathcal{E}_2),f:C_1\rightarrow C_2 \textrm{ isomorphism}\}.$$

A subset $R_{\varphi}\subseteq\langle\mathcal{C}(\mathcal{E}_1),S\rangle\overline{\times}\langle\mathcal{C}(\mathcal{E}_2),S\rangle$ is called a posetal relation. We say that $R_{\varphi}$ is downward
closed when for any $(\langle C_1,\varrho\rangle,f,\langle C_2,\varrho\rangle),(\langle C_1',\varrho'\rangle,f',\langle C_2',\varrho'\rangle)\in \langle\mathcal{C}(\mathcal{E}_1),S\rangle\overline{\times}\langle\mathcal{C}(\mathcal{E}_2),S\rangle$,
if $(\langle C_1,\varrho\rangle,f,\langle C_2,\varrho\rangle)\subseteq (\langle C_1',\varrho'\rangle,f',\langle C_2',\varrho'\rangle)$ pointwise and
$(\langle C_1',\varrho'\rangle,f',\langle C_2',\varrho'\rangle)\in R_{\varphi}$, then $(\langle C_1,\varrho\rangle,f,\langle C_2,\varrho\rangle)\in R_{\varphi}$.

For $f:X_1\rightarrow X_2$, we define $f[x_1\mapsto x_2]:X_1\cup\{x_1\}\rightarrow X_2\cup\{x_2\}$, $z\in X_1\cup\{x_1\}$,(1)$f[x_1\mapsto x_2](z)=
x_2$,if $z=x_1$;(2)$f[x_1\mapsto x_2](z)=f(z)$, otherwise. Where $X_1\subseteq \mathbb{E}_1$, $X_2\subseteq \mathbb{E}_2$, $x_1\in \mathbb{E}_1$, $x_2\in \mathbb{E}_2$.
\end{definition}

\begin{definition}[Weakly posetal product]
Given two PESs $\mathcal{E}_1$, $\mathcal{E}_2$, the weakly posetal product of their configurations, denoted
$\langle\mathcal{C}(\mathcal{E}_1),S\rangle\overline{\times}\langle\mathcal{C}(\mathcal{E}_2),S\rangle$, is defined as

$$\{(\langle C_1,\varrho\rangle,f,\langle C_2,\varrho\rangle)|C_1\in\mathcal{C}(\mathcal{E}_1),C_2\in\mathcal{C}(\mathcal{E}_2),f:\hat{C_1}\rightarrow \hat{C_2} \textrm{ isomorphism}\}.$$

A subset $R_{\varphi}\subseteq\langle\mathcal{C}(\mathcal{E}_1),S\rangle\overline{\times}\langle\mathcal{C}(\mathcal{E}_2),S\rangle$ is called a weakly posetal relation. We say that $R_{\varphi}$ is
downward closed when for any $(\langle C_1,\varrho\rangle,f,\langle C_2,\varrho\rangle),(\langle C_1',\varrho'\rangle,f,\langle C_2',\varrho'\rangle)\in \langle\mathcal{C}(\mathcal{E}_1),S\rangle\overline{\times}\langle\mathcal{C}(\mathcal{E}_2),S\rangle$,
if $(\langle C_1,\varrho\rangle,f,\langle C_2,\varrho\rangle)\subseteq (\langle C_1',\varrho'\rangle,f',\langle C_2',\varrho'\rangle)$ pointwise and
$(\langle C_1',\varrho'\rangle,f',\langle C_2',\varrho'\rangle)\in R_{\varphi}$, then $(\langle C_1,\varrho\rangle,f,\langle C_2,\varrho\rangle)\in R_{\varphi}$.

For $f:X_1\rightarrow X_2$, we define $f[x_1\mapsto x_2]:X_1\cup\{x_1\}\rightarrow X_2\cup\{x_2\}$, $z\in X_1\cup\{x_1\}$,(1)$f[x_1\mapsto x_2](z)=
x_2$,if $z=x_1$;(2)$f[x_1\mapsto x_2](z)=f(z)$, otherwise. Where $X_1\subseteq \hat{\mathbb{E}_1}$, $X_2\subseteq \hat{\mathbb{E}_2}$, $x_1\in \hat{\mathbb{E}}_1$,
$x_2\in \hat{\mathbb{E}}_2$. Also, we define $f(\tau^*)=f(\tau^*)$.
\end{definition}

\begin{definition}[Probabilistic static location (hereditary) history-preserving bisimulation]
A probabilistic static location history-preserving (hp-) bisimulation is a posetal relation
$R_{\varphi}\subseteq\langle\mathcal{C}(\mathcal{E}_1),S\rangle\overline{\times}\langle\mathcal{C}(\mathcal{E}_2),S\rangle$ such that (1) if $(\langle C_1,\varrho\rangle,f,\langle C_2,\varrho\rangle)\in R_{\varphi}$,
and $\langle C_1,\varrho\rangle\xrightarrow[u]{e_1} \langle C_1',\varrho'\rangle$, then $\langle C_2,\varrho\rangle\xrightarrow[v]{e_2} \langle C_2',\varrho'\rangle$, with
$(\langle C_1',\varrho'\rangle,f[e_1\mapsto e_2],\langle C_2',\varrho'\rangle)\in R_{\varphi\cup\{(u,v)\}}$ for all $\varrho,\varrho'\in S$, and vice-versa; (2) if $(\langle C_1,\varrho\rangle,f,\langle C_2,\varrho\rangle)\in R_{\varphi}$, and
$\langle C_1,\varrho\rangle\xrsquigarrow{\pi}\langle C_1^{\pi},\varrho\rangle$ then $\langle C_2,\varrho\rangle\xrsquigarrow{\pi}\langle C_2^{\pi},\varrho\rangle$ and $(\langle C_1^{\pi},\varrho\rangle,f,\langle C_2^{\pi},\varrho\rangle)\in R_{\varphi}$,
and vice-versa; (3) if $(C_1,f,C_2)\in R_{\varphi}$, then $\mu(C_1,C)=\mu(C_2,C)$ for each $C\in\mathcal{C}(\mathcal{E})/R_{\varphi}$; (4) $[\surd]_{R_{\varphi}}=\{\surd\}$. $\mathcal{E}_1,\mathcal{E}_2$ are
probabilistic static location history-preserving (hp-)bisimilar and are written $\mathcal{E}_1\sim_{php}^{sl}\mathcal{E}_2$ if there exists a probabilistic static location hp-bisimulation $R_{\varphi}$ such that
$(\langle\emptyset,\emptyset\rangle,\emptyset,\langle\emptyset,\emptyset\rangle)\in R_{\varphi}$.

A probabilistic static location hereditary history-preserving (hhp-)bisimulation is a downward closed probabilistic static location hp-bisimulation. $\mathcal{E}_1,\mathcal{E}_2$ are probabilistic static location hereditary
history-preserving (hhp-)bisimilar and are written $\mathcal{E}_1\sim_{phhp}^{sl}\mathcal{E}_2$.
\end{definition}

\begin{definition}[Weakly probabilistic static location (hereditary) history-preserving bisimulation]
A weakly probabilistic static location history-preserving (hp-) bisimulation is a weakly posetal relation
$R_{\varphi}\subseteq\langle\mathcal{C}(\mathcal{E}_1),S\rangle\overline{\times}\langle\mathcal{C}(\mathcal{E}_2),S\rangle$ such that (1) if $(\langle C_1,\varrho\rangle,f,\langle C_2,\varrho\rangle)\in R_{\varphi}$,
and $\langle C_1,\varrho\rangle\xRightarrow[u]{e_1} \langle C_1',\varrho'\rangle$, then $\langle C_2,\varrho\rangle\xRightarrow[v]{e_2} \langle C_2',\varrho'\rangle$, with
$(\langle C_1',\varrho'\rangle,f[e_1\mapsto e_2],\langle C_2',\varrho'\rangle)\in R_{\varphi\cup\{(u,v)\}}$ for all $\varrho,\varrho'\in S$, and vice-versa; (2) if $(\langle C_1,\varrho\rangle,f,\langle C_2,\varrho\rangle)\in R_{\varphi}$, and
$\langle C_1,\varrho\rangle\xrsquigarrow{\pi}\langle C_1^{\pi},\varrho\rangle$ then $\langle C_2,\varrho\rangle\xrsquigarrow{\pi}\langle C_2^{\pi},\varrho\rangle$ and
$(\langle C_1^{\pi},\varrho\rangle,f,\langle C_2^{\pi},\varrho\rangle)\in R_{\varphi}$, and vice-versa; (3) if $(C_1,f,C_2)\in R_{\varphi}$, then $\mu(C_1,C)=\mu(C_2,C)$ for each $C\in\mathcal{C}(\mathcal{E})/R_{\varphi}$;
(4) $[\surd]_{R_{\varphi}}=\{\surd\}$. $\mathcal{E}_1,\mathcal{E}_2$ are weakly probabilistic static location history-preserving (hp-)bisimilar and are written $\mathcal{E}_1\approx_{php}^{sl}\mathcal{E}_2$ if there
exists a weakly probabilistic static location hp-bisimulation $R_{\varphi}$ such that $(\langle\emptyset,\emptyset\rangle,\emptyset,\langle\emptyset,\emptyset\rangle)\in R_{\varphi}$.

A weakly probabilistic static location hereditary history-preserving (hhp-)bisimulation is a downward closed weakly probabilistic static location hp-bisimulation. $\mathcal{E}_1,\mathcal{E}_2$ are weakly
probabilistic static location hereditary history-preserving (hhp-)bisimilar and are written $\mathcal{E}_1\approx_{phhp}^{sl}\mathcal{E}_2$.
\end{definition}

\begin{definition}[Probabilistic static location branching pomset, step bisimulation]
Assume a special termination predicate $\downarrow$, and let $\surd$ represent a state with $\surd\downarrow$. Let $\mathcal{E}_1$, $\mathcal{E}_2$ be PESs. A probabilistic static location branching
pomset bisimulation is a relation $R_{\varphi}\subseteq\langle\mathcal{C}(\mathcal{E}_1),S\rangle\times\langle\mathcal{C}(\mathcal{E}_2),S\rangle$, such that:

 \begin{enumerate}
   \item if $(\langle C_1,\varrho\rangle,\langle C_2,\varrho\rangle)\in R_{\varphi}$, and $\langle C_1,\varrho\rangle\xrightarrow[u]{X}\langle C_1',\varrho'\rangle$ then
   \begin{itemize}
     \item either $X\equiv \tau^*$, and $(\langle C_1',\varrho'\rangle,\langle C_2,\varrho\rangle)\in R_{\varphi}$ with $\varrho'\in \tau(\varrho)$;
     \item or there is a sequence of (zero or more) probabilistic transitions and $\tau$-transitions $\langle C_2,\varrho\rangle\rightsquigarrow^*\xrightarrow{\tau^*} \langle C_2^0,\varrho^0\rangle$, such that
     $(\langle C_1,\varrho\rangle,\langle C_2^0,\varrho^0\rangle)\in R_{\varphi}$ and $\langle C_2^0,\varrho^0\rangle\xRightarrow[v]{X}\langle C_2',\varrho'\rangle$ with
     $(\langle C_1',\varrho'\rangle,\langle C_2',\varrho'\rangle)\in R_{\varphi\cup\{(u,v)\}}$;
   \end{itemize}
   \item if $(\langle C_1,\varrho\rangle,\langle C_2,\varrho\rangle)\in R_{\varphi}$, and $\langle C_2,\varrho\rangle\xrightarrow[v]{X}\langle C_2',\varrho'\rangle$ then
   \begin{itemize}
     \item either $X\equiv \tau^*$, and $(\langle C_1,\varrho\rangle,\langle C_2',\varrho'\rangle)\in R_{\varphi}$;
     \item or there is a sequence of (zero or more) probabilistic transitions and $\tau$-transitions $\langle C_1,\varrho\rangle\rightsquigarrow^*\xrightarrow{\tau^*} \langle C_1^0,\varrho^0\rangle$, such that
     $(\langle C_1^0,\varrho^0\rangle,\langle C_2,\varrho\rangle)\in R_{\varphi}$ and $\langle C_1^0,\varrho^0\rangle\xRightarrow[u]{X}\langle C_1',\varrho'\rangle$ with
     $(\langle C_1',\varrho'\rangle,\langle C_2',\varrho'\rangle)\in R_{\varphi\cup\{(u,v)\}}$;
   \end{itemize}
   \item if $(\langle C_1,\varrho\rangle,\langle C_2,\varrho\rangle)\in R_{\varphi}$ and $\langle C_1,\varrho\rangle\downarrow$, then there is a sequence of (zero or more) probabilistic transitions and $\tau$-transitions
   $\langle C_2,\varrho\rangle\rightsquigarrow^*\xrightarrow{\tau^*}\langle C_2^0,\varrho^0\rangle$ such that $(\langle C_1,\varrho\rangle,\langle C_2^0,\varrho^0\rangle)\in R_{\varphi}$ and
   $\langle C_2^0,\varrho^0\rangle\downarrow$;
   \item if $(\langle C_1,\varrho\rangle,\langle C_2,\varrho\rangle)\in R_{\varphi}$ and $\langle C_2,\varrho\rangle\downarrow$, then there is a sequence of (zero or more) probabilistic transitions and $\tau$-transitions
   $\langle C_1,\varrho\rangle\rightsquigarrow^*\xrightarrow{\tau^*}\langle C_1^0,\varrho^0\rangle$ such that $(\langle C_1^0,\varrho^0\rangle,\langle C_2,\varrho\rangle)\in R_{\varphi}$ and
   $\langle C_1^0,\varrho^0\rangle\downarrow$;
   \item if $(C_1,C_2)\in R_{\varphi}$,then $\mu(C_1,C)=\mu(C_2,C)$ for each $C\in\mathcal{C}(\mathcal{E})/R_{\varphi}$;
   \item $[\surd]_{R_{\varphi}}=\{\surd\}$.
 \end{enumerate}

We say that $\mathcal{E}_1$, $\mathcal{E}_2$ are probabilistic static location branching pomset bisimilar, written $\mathcal{E}_1\approx_{pbp}^{sl}\mathcal{E}_2$, if there exists a probabilistic static location branching
pomset bisimulation $R_{\varphi}$, such that $(\langle\emptyset,\emptyset\rangle,\langle\emptyset,\emptyset\rangle)\in R_{\varphi}$.

By replacing probabilistic pomset transitions with steps, we can get the definition of probabilistic branching step bisimulation. When PESs $\mathcal{E}_1$ and $\mathcal{E}_2$ are
probabilistic branching step bisimilar, we write $\mathcal{E}_1\approx_{pbs}^{sl}\mathcal{E}_2$.
\end{definition}

\begin{definition}[Probabilistic static location rooted branching pomset, step bisimulation]
Assume a special termination predicate $\downarrow$, and let $\surd$ represent a state with $\surd\downarrow$. Let $\mathcal{E}_1$, $\mathcal{E}_2$ be PESs. A probabilistic static location rooted
branching pomset bisimulation is a relation $R_{\varphi}\subseteq\langle\mathcal{C}(\mathcal{E}_1),S\rangle\times\langle\mathcal{C}(\mathcal{E}_2),S\rangle$, such that:

 \begin{enumerate}
   \item if $(\langle C_1,\varrho\rangle,\langle C_2,\varrho\rangle)\in R_{\varphi}$, and $\langle C_1,\varrho\rangle\rightsquigarrow\xrightarrow[u]{X}\langle C_1',\varrho'\rangle$ then
   $\langle C_2,\varrho\rangle\rightsquigarrow\xrightarrow[v]{X}\langle C_2',\varrho'\rangle$ with $\langle C_1',\varrho'\rangle\approx_{pbp}^{sl}\langle C_2',\varrho'\rangle$;
   \item if $(\langle C_1,\varrho\rangle,\langle C_2,\varrho\rangle)\in R_{\varphi}$, and $\langle C_2,\varrho\rangle\rightsquigarrow\xrightarrow[v]{X}\langle C_2',\varrho'\rangle$ then
   $\langle C_1,\varrho\rangle\rightsquigarrow\xrightarrow[u]{X}\langle C_1',\varrho'\rangle$ with $\langle C_1',\varrho'\rangle\approx_{pbp}^{sl}\langle C_2',\varrho'\rangle$;
   \item if $(\langle C_1,\varrho\rangle,\langle C_2,\varrho\rangle)\in R_{\varphi}$ and $\langle C_1,\varrho\rangle\downarrow$, then $\langle C_2,\varrho\rangle\downarrow$;
   \item if $(\langle C_1,\varrho\rangle,\langle C_2,\varrho\rangle)\in R_{\varphi}$ and $\langle C_2,\varrho\rangle\downarrow$, then $\langle C_1,\varrho\rangle\downarrow$.
 \end{enumerate}

We say that $\mathcal{E}_1$, $\mathcal{E}_2$ are probabilistic static location rooted branching pomset bisimilar, written $\mathcal{E}_1\approx_{prbp}^{sl}\mathcal{E}_2$, if there exists a probabilistic
static location rooted branching pomset bisimulation $R_{\varphi}$, such that $(\langle\emptyset,\emptyset\rangle,\langle\emptyset,\emptyset\rangle)\in R_{\varphi}$.

By replacing pomset transitions with steps, we can get the definition of probabilistic static location rooted branching step bisimulation. When PESs $\mathcal{E}_1$ and $\mathcal{E}_2$ are probabilistic static location
rooted branching step bisimilar, we write $\mathcal{E}_1\approx_{prbs}^{sl}\mathcal{E}_2$.
\end{definition}

\begin{definition}[Probabilistic static location branching (hereditary) history-preserving bisimulation]
Assume a special termination predicate $\downarrow$, and let $\surd$ represent a state with $\surd\downarrow$. A probabilistic static location branching history-preserving (hp-) bisimulation is a
weakly posetal relation $R_{\varphi}\subseteq\langle\mathcal{C}(\mathcal{E}_1),S\rangle\overline{\times}\langle\mathcal{C}(\mathcal{E}_2),S\rangle$ such that:

 \begin{enumerate}
   \item if $(\langle C_1,\varrho\rangle,f,\langle C_2,\varrho\rangle)\in R_{\varphi}$, and $\langle C_1,\varrho\rangle\xrightarrow[u]{e_1}\langle C_1',\varrho'\rangle$ then
   \begin{itemize}
     \item either $e_1\equiv \tau$, and $(\langle C_1',\varrho'\rangle,f[e_1\mapsto \tau],\langle C_2,\varrho\rangle)\in R_{\varphi}$;
     \item or there is a sequence of (zero or more) probabilistic transitions and $\tau$-transitions $\langle C_2,\varrho\rangle\rightsquigarrow^*\xrightarrow{\tau^*} \langle C_2^0,\varrho^0\rangle$, such that
     $(\langle C_1,\varrho\rangle,f,\langle C_2^0,\varrho^0\rangle)\in R_{\varphi}$ and $\langle C_2^0,\varrho^0\rangle\xrightarrow[v]{e_2}\langle C_2',\varrho'\rangle$ with
     $(\langle C_1',\varrho'\rangle,f[e_1\mapsto e_2],\langle C_2',\varrho'\rangle)\in R_{\varphi\cup\{(u,v)\}}$;
   \end{itemize}
   \item if $(\langle C_1,\varrho\rangle,f,\langle C_2,\varrho\rangle)\in R_{\varphi}$, and $\langle C_2,\varrho\rangle\xrightarrow[v]{e_2}\langle C_2',\varrho'\rangle$ then
   \begin{itemize}
     \item either $e_2\equiv \tau$, and $(\langle C_1,\varrho\rangle,f[e_2\mapsto \tau],\langle C_2',\varrho'\rangle)\in R_{\varphi}$;
     \item or there is a sequence of (zero or more) probabilistic transitions and $\tau$-transitions $\langle C_1,\varrho\rangle\rightsquigarrow^*\xrightarrow{\tau^*} \langle C_1^0,\varrho^0\rangle$, such that
     $(\langle C_1^0,\varrho^0\rangle,f,\langle C_2,\varrho\rangle)\in R_{\varphi}$ and $\langle C_1^0,\varrho^0\rangle\xrightarrow[u]{e_1}\langle C_1',\varrho'\rangle$ with
     $(\langle C_1',\varrho'\rangle,f[e_2\mapsto e_1],\langle C_2',\varrho'\rangle)\in R_{\varphi\cup\{(u,v)\}}$;
   \end{itemize}
   \item if $(\langle C_1,\varrho\rangle,f,\langle C_2,\varrho\rangle)\in R_{\varphi}$ and $\langle C_1,\varrho\rangle\downarrow$, then there is a sequence of (zero or more) probabilistic transitions and $\tau$-transitions
   $\langle C_2,\varrho\rangle\rightsquigarrow^*\xrightarrow{\tau^*}\langle C_2^0,\varrho^0\rangle$ such that $(\langle C_1,\varrho\rangle,f,\langle C_2^0,\varrho^0\rangle)\in R_{\varphi}$ and
   $\langle C_2^0,\varrho^0\rangle\downarrow$;
   \item if $(\langle C_1,\varrho\rangle,f,\langle C_2,\varrho\rangle)\in R_{\varphi}$ and $\langle C_2,\varrho\rangle\downarrow$, then there is a sequence of (zero or more) probabilistic transitions and $\tau$-transitions
   $\langle C_1,\varrho\rangle\rightsquigarrow^*\xrightarrow{\tau^*}\langle C_1^0,\varrho^0\rangle$ such that $(\langle C_1^0,\varrho^0\rangle,f,\langle C_2,\varrho\rangle)\in R_{\varphi}$ and
   $\langle C_1^0,\varrho^0\rangle\downarrow$;
   \item if $(C_1,C_2)\in R_{\varphi}$,then $\mu(C_1,C)=\mu(C_2,C)$ for each $C\in\mathcal{C}(\mathcal{E})/R_{\varphi}$;
   \item $[\surd]_{R_{\varphi}}=\{\surd\}$.
 \end{enumerate}

$\mathcal{E}_1,\mathcal{E}_2$ are probabilistic static location branching history-preserving (hp-)bisimilar and are written $\mathcal{E}_1\approx_{pbhp}^{sl}\mathcal{E}_2$ if there exists a probabilistic static location
branching hp-bisimulation $R_{\varphi}$ such that $(\langle\emptyset,\emptyset\rangle,\emptyset,\langle\emptyset,\emptyset\rangle)\in R_{\varphi}$.

A probabilistic branching hereditary history-preserving (hhp-)bisimulation is a downward closed probabilistic branching hp-bisimulation. $\mathcal{E}_1,\mathcal{E}_2$ are probabilistic
branching hereditary history-preserving (hhp-)bisimilar and are written $\mathcal{E}_1\approx_{pbhhp}^{sl}\mathcal{E}_2$.
\end{definition}

\begin{definition}[Probabilistic static location rooted branching (hereditary) history-preserving bisimulation]
Assume a special termination predicate $\downarrow$, and let $\surd$ represent a state with $\surd\downarrow$. A probabilistic static location rooted branching history-preserving (hp-) bisimulation is
a weakly posetal relation $R_{\varphi}\subseteq\langle\mathcal{C}(\mathcal{E}_1),S\rangle\overline{\times}\langle\mathcal{C}(\mathcal{E}_2),S\rangle$ such that:

 \begin{enumerate}
   \item if $(\langle C_1,\varrho\rangle,f,\langle C_2,\varrho\rangle)\in R_{\varphi}$, and $\langle C_1,\varrho\rangle\rightsquigarrow\xrightarrow[u]{e_1}\langle C_1',\varrho'\rangle$, then
   $\langle C_2,\varrho\rangle\rightsquigarrow\xrightarrow[v]{e_2}\langle C_2',\varrho'\rangle$ with $\langle C_1',\varrho'\rangle\approx_{pbhp}^{sl}\langle C_2',\varrho'\rangle$;
   \item if $(\langle C_1,\varrho\rangle,f,\langle C_2,\varrho\rangle)\in R_{\varphi}$, and $\langle C_2,\varrho\rangle\rightsquigarrow\xrightarrow[v]{e_2}\langle C_2',\varrho'\rangle$, then
   $\langle C_1,\varrho\rangle\rightsquigarrow\xrightarrow[u]{e_1}\langle C_1',\varrho'\rangle$ with $\langle C_1',\varrho'\rangle\approx_{pbhp}^{sl}\langle C_2',\varrho'\rangle$;
   \item if $(\langle C_1,\varrho\rangle,f,\langle C_2,\varrho\rangle)\in R_{\varphi}$ and $\langle C_1,\varrho\rangle\downarrow$, then $\langle C_2,\varrho\rangle\downarrow$;
   \item if $(\langle C_1,\varrho\rangle,f,\langle C_2,\varrho\rangle)\in R_{\varphi}$ and $\langle C_2,\varrho\rangle\downarrow$, then $\langle C_1,\varrho\rangle\downarrow$.
 \end{enumerate}

$\mathcal{E}_1,\mathcal{E}_2$ are probabilistic static location rooted branching history-preserving (hp-)bisimilar and are written $\mathcal{E}_1\approx_{prbhp}^{sl}\mathcal{E}_2$ if there exists a probabilistic static location
rooted branching hp-bisimulation $R_{\varphi}$ such that $(\langle\emptyset,\emptyset\rangle,\emptyset,\langle\emptyset,\emptyset\rangle)\in R_{\varphi}$.

A probabilistic static location rooted branching hereditary history-preserving (hhp-)bisimulation is a downward closed probabilistic static location rooted branching hp-bisimulation. $\mathcal{E}_1,\mathcal{E}_2$ are
probabilistic static location rooted branching hereditary history-preserving (hhp-)bisimilar and are written $\mathcal{E}_1\approx_{prbhhp}^{sl}\mathcal{E}_2$.
\end{definition}

\newpage\section{APTC with Localities for Open Quantum Systems}\label{qaptcl}

In this chapter, we introduce APTC with localities for open quantum systems, including BATC with localities for open quantum systems abbreviated $qBATC^{sl}$ in section \ref{qbatcl}, APTC
with localities for open quantum systems
abbreviated $qAPTC^{sl}$ in section \ref{qaptcl2}, recursion in section \ref{qorec}, abstraction in section \ref{qoabs}, quantum entanglement in section \ref{qe1} and unification of quantum
and classical computing for open quantum systems in section \ref{uni1}.

Note that, in open quantum systems, quantum operations denoted $\mathbb{E}$ are the atomic actions (events), and a quantum operation $e\in\mathbb{E}$.

\subsection{BATC with Localities for Open Quantum Systems}\label{qbatcl}

Let $Loc$ be the set of locations, and $loc\in Loc$, $u,v\in Loc^*$, $\epsilon$ is the empty location. A distribution allocates a location $u\in Loc*$ to an action $e$ denoted
$u::e$ or a process $x$ denoted $u::x$.

In the following, $x,y,z$ range over the set of terms for true concurrency, $p,q,s$ range over the set of closed terms.
The set of axioms of qBATC with localities consists of the laws given in Table \ref{AxiomsForqBATC}.

\begin{center}
    \begin{table}
        \begin{tabular}{@{}ll@{}}
            \hline No. &Axiom\\
            $A1$ & $x+ y = y+ x$\\
            $A2$ & $(x+ y)+ z = x+ (y+ z)$\\
            $A3$ & $x+ x = x$\\
            $A4$ & $(x+ y)\cdot z = x\cdot z + y\cdot z$\\
            $A5$ & $(x\cdot y)\cdot z = x\cdot(y\cdot z)$\\
            $L1$ & $\epsilon::x=x$\\
            $L2$ & $u::(x\cdot y)=u::x\cdot u::y$\\
            $L3$ & $u::(x+ y)=u::x+ u::y$\\
            $L4$ & $u::(v::x)=uv::x$\\
        \end{tabular}
        \caption{Axioms of qBATC with localities}
        \label{AxiomsForqBATC}
    \end{table}
\end{center}

\begin{definition}[Basic terms of $qBATC$ with localities]
The set of basic terms of $qBATC$ with localities, $\mathcal{B}(qBATC^{sl})$, is inductively defined as follows:
\begin{enumerate}
  \item $\mathbb{E}\subset\mathcal{B}(qBATC^{sl})$;
  \item if $u\in Loc^*, t\in\mathcal{B}(qBATC^{sl})$ then $u::t\in\mathcal{B}(qBATC^{sl})$;
  \item if $e\in \mathbb{E}, t\in\mathcal{B}(qBATC^{sl})$ then $e\cdot t\in\mathcal{B}(qBATC^{sl})$;
  \item if $t,s\in\mathcal{B}(qBATC^{sl})$ then $t+ s\in\mathcal{B}(qBATC^{sl})$.
\end{enumerate}
\end{definition}

\begin{theorem}[Elimination theorem of $qBATC$ with localities]
Let $p$ be a closed $qBATC$ with localities term. Then there is a basic $qBATC$ with localities term $q$ such that $qBATC^{sl}\vdash p=q$.
\end{theorem}

\begin{proof}
The same as that of $BATC^{sl}$, we omit the proof, please refer to \cite{LOC1} for details.
\end{proof}

We give the operational transition rules of operators $::$, $\cdot$ and $+$ as Table \ref{TRForqBATC} shows.

\begin{center}
    \begin{table}
        $$\frac{}{\langle e,\varrho\rangle\xrightarrow[\epsilon]{e}\langle\surd,\varrho'\rangle}\quad \frac{}{\langle loc::e,\varrho\rangle\xrightarrow[loc]{e}\langle\surd,\varrho'\rangle}$$
        $$\frac{\langle x,\varrho\rangle\xrightarrow[u]{e}\langle x',\varrho'\rangle}{\langle loc::x,\varrho\rangle\xrightarrow[loc\ll u]{e}\langle loc::x',\varrho'\rangle}$$
        $$\frac{\langle x,\varrho\rangle\xrightarrow[u]{e}\langle\surd,\varrho'\rangle}{\langle x+ y,\varrho\rangle\xrightarrow[u]{e}\langle\surd,\varrho'\rangle}
        \quad\frac{\langle x,\varrho\rangle\xrightarrow[u]{e}\langle x',\varrho'\rangle}{\langle x+ y,\varrho\rangle\xrightarrow[u]{e}\langle x',\varrho'\rangle}$$
        $$\frac{\langle y,\varrho\rangle\xrightarrow[u]{e}\langle \surd,\varrho'\rangle}{\langle x+ y,\varrho\rangle\xrightarrow[u]{e}\langle\surd,\varrho'\rangle}
        \quad\frac{\langle y,\varrho\rangle\xrightarrow[u]{e}\langle y',\varrho'\rangle}{\langle x+ y,\varrho\rangle\xrightarrow[u]{e}\langle y',\varrho'\rangle}$$
        $$\frac{\langle x,\varrho\rangle\xrightarrow[u]{e}\langle\surd,\varrho'\rangle}{\langle x\cdot y,\varrho\rangle\xrightarrow[u]{e}\langle y,\varrho'\rangle}
        \quad\frac{\langle x,\varrho\rangle\xrightarrow[u]{e}\langle x',\varrho'\rangle}{\langle x\cdot y,\varrho\rangle\xrightarrow[u]{e}\langle x'\cdot y,\varrho'\rangle}$$
        \caption{Transition rules of qBATC with localities}
        \label{TRForqBATC}
    \end{table}
\end{center}

\begin{theorem}[Congruence of $qBATC$ with localities with respect to static location truly concurrent bisimulations]
Static location truly concurrent bisimulations $\sim_p^{sl}$, $\sim_s^{sl}$, $\sim_{hp}^{sl}$ and $\sim_{hhp}^{sl}$ are all congruences with respect to $qBATC$ with localities.
\end{theorem}

\begin{proof}
It is obvious that static location truly concurrent bisimulations $\sim_p^{sl}$, $\sim_s^{sl}$, $\sim_{hp}^{sl}$ and $\sim_{hhp}^{sl}$ are all equivalent relations with respect to $qBATC$ with localities. So, it is sufficient to prove
that static location truly concurrent bisimulations $\sim_p^{sl}$, $\sim_s^{sl}$, $\sim_{hp}^{sl}$ and $\sim_{hhp}^{sl}$ are preserved for $::$, $\cdot$ and $+$ according to the transition rules in Table \ref{TRForqBATC},
that is, if $x\sim_p^{sl}x'$ and $y\sim_p^{sl}y'$, then $loc::x\sim_p^{sl}loc::y$, $x+ y\sim_p^{sl}x'+ y'$ and $x\cdot y\sim_p^{sl}x'\cdot y'$; if $x\sim_s^{sl}x'$ and $y\sim_s^{sl}y'$, then $loc::x\sim_s^{sl}loc::y$, $x+ y\sim_s^{sl}x'+ y'$ and $x\cdot y\sim_s^{sl}x'\cdot y'$;
if $x\sim_{hp}^{sl}x'$ and $y\sim_{hp}^{sl}y'$, then $loc::x\sim_{hp}^{sl}loc::y$, $x+ y\sim_{hp}^{sl}x'+ y'$ and  $x\cdot y\sim_{hp}^{sl}x'\cdot y'$; and if $x\sim_{hhp}^{sl}x'$ and $y\sim_{hhp}^{sl}y'$, then $loc::x\sim_{hhp}^{sl}loc::y$, $x+ y\sim_{hhp}^{sl}x'+ y'$ and $x\cdot y\sim_{hhp}^{sl}x'\cdot y'$.
The proof is quit trivial, and we leave the proof as an exercise for the readers.
\end{proof}

\begin{theorem}[Soundness of qBATC with localities modulo static location truly concurrent bisimulation equivalences]\label{SBATC}
The axiomatization of qBATC with localities is sound modulo static location truly concurrent bisimulation equivalences $\sim_p^{sl}$, $\sim_s^{sl}$, $\sim_{hp}^{sl}$ and $\sim_{hhp}^{sl}$. That is,

\begin{enumerate}
  \item let $x$ and $y$ be qBATC with localities terms. If qBATC with localities $\vdash x=y$, then $x\sim_p^{sl} y$;
  \item let $x$ and $y$ be qBATC with localities terms. If qBATC with localities $\vdash x=y$, then $x\sim_s^{sl} y$;
  \item let $x$ and $y$ be qBATC with localities terms. If qBATC with localities $\vdash x=y$, then $x\sim_{hp}^{sl} y$;
  \item let $x$ and $y$ be qBATC with localities terms. If qBATC with localities $\vdash x=y$, then $x\sim_{hhp}^{sl} y$.
\end{enumerate}
\end{theorem}

\begin{proof}
(1) Since static location pomset bisimulation $\sim_p^{sl}$ is both an equivalent and a congruent relation, we only need to check if each axiom in Table \ref{AxiomsForqBATC} is sound
modulo static location pomset bisimulation equivalence. We leave the proof as an exercise for the readers.

(2) Since static location step bisimulation $\sim_s^{sl}$ is both an equivalent and a congruent relation, we only need to check if each axiom in Table \ref{AxiomsForqBATC} is sound modulo
static location step bisimulation equivalence. We leave the proof as an exercise for the readers.

(3) Since static location hp-bisimulation $\sim_{hp}^{sl}$ is both an equivalent and a congruent relation, we only need to check if each axiom in Table \ref{AxiomsForqBATC} is sound modulo
static location hp-bisimulation equivalence. We leave the proof as an exercise for the readers.

(4) Since static location hhp-bisimulation $\sim_{hhp}^{sl}$ is both an equivalent and a congruent relation, we only need to check if each axiom in Table \ref{AxiomsForqBATC} is sound modulo
static location hhp-bisimulation equivalence. We leave the proof as an exercise for the readers.
\end{proof}

\begin{theorem}[Completeness of qBATC with localities modulo static location truly concurrent bisimulation equivalences]\label{CBATC}
The axiomatization of qBATC with localities is complete modulo static location truly concurrent bisimulation equivalences $\sim_p^{sl}$, $\sim_s^{sl}$, $\sim_{hp}^{sl}$ and $\sim_{hhp}^{sl}$. That is,

\begin{enumerate}
  \item let $p$ and $q$ be closed qBATC with localities terms, if $p\sim_p^{sl} q$ then $p=q$;
  \item let $p$ and $q$ be closed qBATC with localities terms, if $p\sim_s^{sl} q$ then $p=q$;
  \item let $p$ and $q$ be closed qBATC with localities terms, if $p\sim_{hp}^{sl} q$ then $p=q$;
  \item let $p$ and $q$ be closed qBATC with localities terms, if $p\sim_{hhp}^{sl} q$ then $p=q$.
\end{enumerate}
\end{theorem}

\begin{proof}
According to the definition of static location truly concurrent bisimulation equivalences $\sim_p^{sl}$, $\sim_s^{sl}$, $\sim_{hp}^{sl}$ and $\sim_{hhp}^{sl}$, $p\sim_p^{sl}q$, $p\sim_s^{sl}q$, $p\sim_{hp}^{sl}q$ and $p\sim_{hhp}^{sl}q$ implies
both the bisimilarities between $p$ and $q$, and also the in the same quantum states. According to the completeness of $BATC^{sl}$ (please refer to \cite{LOC1} for details), we can get the
completeness of qBATC with localities.
\end{proof}

\subsection{APTC with Localities for Open Quantum Systems}\label{qaptcl2}

We give the transition rules of qAPTC with localities in Table \ref{TRForqAPTC}, it is suitable for all static location truly concurrent behavioral equivalence, including static location pomset bisimulation, static location step bisimulation,
static location hp-bisimulation and static location hhp-bisimulation.

\begin{center}
    \begin{table}
        $$\frac{\langle x,\varrho\rangle\xrightarrow[u]{e_1}\langle\surd,\varrho'\rangle\quad \langle y,\varrho\rangle\xrightarrow[v]{e_2}\langle\surd,\varrho''\rangle}{\langle x\parallel y,\varrho\rangle\xrightarrow[u\diamond v]{\{e_1,e_2\}}\langle\surd,\varrho'\cup \varrho''\rangle} \quad\frac{\langle x,\varrho\rangle\xrightarrow[u]{e_1}\langle x',\varrho'\rangle\quad \langle y,\varrho\rangle\xrightarrow[v]{e_2}\langle\surd,\varrho''\rangle}{\langle x\parallel y,\varrho\rangle\xrightarrow[u\diamond v]{\{e_1,e_2\}}\langle x',\varrho'\cup \varrho''\rangle}$$

        $$\frac{\langle x,\varrho\rangle\xrightarrow[u]{e_1}\langle\surd,\varrho'\rangle\quad \langle y,\varrho\rangle\xrightarrow[v]{e_2}\langle y',\varrho''\rangle}{\langle x\parallel y,\varrho\rangle\xrightarrow[u\diamond v]{\{e_1,e_2\}}\langle y',\varrho'\cup \varrho''\rangle} \quad\frac{\langle x,\varrho\rangle\xrightarrow[u]{e_1}\langle x',\varrho'\rangle\quad \langle y,\varrho\rangle\xrightarrow[v]{e_2}\langle y',\varrho''\rangle}{\langle x\parallel y,\varrho\rangle\xrightarrow[u\diamond v]{\{e_1,e_2\}}\langle x'\between y',\varrho'\cup \varrho''\rangle}$$

        $$\frac{\langle x,\varrho\rangle\xrightarrow[u]{e_1}\langle\surd,\varrho'\rangle\quad \langle y,\varrho\rangle\xnrightarrow[v]{e_2}\quad(e_1\%e_2)}{\langle x\parallel y,\varrho\rangle\xrightarrow[u]{e_1}\langle y,\varrho'\rangle} \quad\frac{\langle x,\varrho\rangle\xrightarrow[u]{e_1}\langle x',\varrho'\rangle\quad \langle y,\varrho\rangle\xnrightarrow[v]{e_2}\quad(e_1\%e_2)}{\langle x\parallel y,\varrho\rangle\xrightarrow[u]{e_1}\langle x'\between y,\varrho'\rangle}$$

        $$\frac{\langle x,\varrho\rangle\xnrightarrow[u]{e_1}\quad \langle y,\varrho\rangle\xrightarrow[v]{e_2}\langle\surd,\varrho''\rangle\quad(e_1\%e_2)}{\langle x\parallel y,\varrho\rangle\xrightarrow[v]{e_2}\langle x,\varrho''\rangle} \quad\frac{\langle x,\varrho\rangle\xnrightarrow[u]{e_1}\quad \langle y,\varrho\rangle\xrightarrow[v]{e_2}\langle y',\varrho''\rangle\quad(e_1\%e_2)}{\langle x\parallel y,\varrho\rangle\xrightarrow[v]{e_2}\langle x\between y',\varrho''\rangle}$$

        $$\frac{\langle x,\varrho\rangle\xrightarrow[u]{e_1}\langle\surd,\varrho'\rangle\quad \langle y,\varrho\rangle\xrightarrow[v]{e_2}\langle\surd,\varrho''\rangle \quad(e_1\leq e_2)}{\langle x\leftmerge y,\varrho\rangle\xrightarrow[u\diamond v]{\{e_1,e_2\}}\langle \surd,\varrho'\cup \varrho''\rangle} \quad\frac{\langle x,\varrho\rangle\xrightarrow[u]{e_1}\langle x',\varrho'\rangle\quad \langle y,\varrho\rangle\xrightarrow[v]{e_2}\langle\surd,\varrho''\rangle \quad(e_1\leq e_2)}{\langle x\leftmerge y,\varrho\rangle\xrightarrow[u\diamond v]{\{e_1,e_2\}}\langle x',\varrho'\cup \varrho''\rangle}$$

        $$\frac{\langle x,\varrho\rangle\xrightarrow[u]{e_1}\langle\surd,\varrho'\rangle\quad \langle y,\varrho\rangle\xrightarrow[v]{e_2}\langle y',\varrho''\rangle \quad(e_1\leq e_2)}{\langle x\leftmerge y,\varrho\rangle\xrightarrow[u\diamond v]{\{e_1,e_2\}}\langle y',\varrho'\cup \varrho''\rangle} \quad\frac{\langle x,\varrho\rangle\xrightarrow[u]{e_1}\langle x',\varrho'\rangle\quad \langle y,\varrho\rangle\xrightarrow[v]{e_2}\langle y',\varrho''\rangle \quad(e_1\leq e_2)}{\langle x\leftmerge y,\varrho\rangle\xrightarrow[u\diamond v]{\{e_1,e_2\}}\langle x'\between y',\varrho'\cup \varrho''\rangle}$$

        $$\frac{\langle x,\varrho\rangle\xrightarrow[u]{e_1}\langle\surd,\varrho'\rangle\quad \langle y,\varrho\rangle\xrightarrow[v]{e_2}\langle\surd,\varrho''\rangle}{\langle x\mid y,\varrho\rangle\xrightarrow[u\diamond v]{\gamma(e_1,e_2)}\langle\surd,effect(\gamma(e_1,e_2),\varrho)\rangle} \quad\frac{\langle x,\varrho\rangle\xrightarrow[u]{e_1}\langle x',\varrho'\rangle\quad \langle y,\varrho\rangle\xrightarrow[v]{e_2}\langle\surd,\varrho''\rangle}{\langle x\mid y,\varrho\rangle\xrightarrow[u\diamond v]{\gamma(e_1,e_2)}\langle x',effect(\gamma(e_1,e_2),\varrho)\rangle}$$

        $$\frac{\langle x,\varrho\rangle\xrightarrow[u]{e_1}\langle\surd,\varrho'\rangle\quad \langle y,\varrho\rangle\xrightarrow[v]{e_2}\langle y',\varrho''\rangle}{\langle x\mid y,\varrho\rangle\xrightarrow[u\diamond v]{\gamma(e_1,e_2)}\langle y',effect(\gamma(e_1,e_2),\varrho)\rangle} \quad\frac{\langle x,\varrho\rangle\xrightarrow[u]{e_1}\langle x',\varrho'\rangle\quad \langle y,\varrho\rangle\xrightarrow[v]{e_2}\langle y',\varrho''\rangle}{\langle x\mid y,\varrho\rangle\xrightarrow[u\diamond v]{\gamma(e_1,e_2)}\langle x'\between y',effect(\gamma(e_1,e_2),\varrho)\rangle}$$

        \caption{Transition rules of qAPTC with localities}
        \label{TRForqAPTC}
    \end{table}
\end{center}

\begin{center}
    \begin{table}

        $$\frac{\langle x,\varrho\rangle\xrightarrow[u]{e_1}\langle\surd,\varrho'\rangle\quad (\sharp(e_1,e_2))}{\langle \Theta(x),\varrho\rangle\xrightarrow[u]{e_1}\langle\surd,\varrho'\rangle} \quad\frac{\langle x,\varrho\rangle\xrightarrow[u]{e_2}\langle\surd,\varrho''\rangle\quad (\sharp(e_1,e_2))}{\langle\Theta(x),\varrho\rangle\xrightarrow[u]{e_2}\langle\surd,\varrho''\rangle}$$

        $$\frac{\langle x,\varrho\rangle\xrightarrow[u]{e_1}\langle x',\varrho'\rangle\quad (\sharp(e_1,e_2))}{\langle\Theta(x),\varrho\rangle\xrightarrow[u]{e_1}\langle\Theta(x'),\varrho'\rangle} \quad\frac{\langle x,\varrho\rangle\xrightarrow[u]{e_2}\langle x'',\varrho''\rangle\quad (\sharp(e_1,e_2))}{\langle\Theta(x),\varrho\rangle\xrightarrow[u]{e_2}\langle\Theta(x''),\varrho''\rangle}$$

        $$\frac{\langle x,\varrho\rangle\xrightarrow[u]{e_1}\langle\surd,\varrho'\rangle \quad \langle y,\varrho\rangle\nrightarrow^{e_2}\quad (\sharp(e_1,e_2))}{\langle x\triangleleft y,\varrho\rangle\xrightarrow[u]{\tau}\langle\surd,\varrho'\rangle}
        \quad\frac{\langle x,\varrho\rangle\xrightarrow[u]{e_1}\langle x',\varrho'\rangle \quad \langle y,\varrho\rangle\nrightarrow^{e_2}\quad (\sharp(e_1,e_2))}{\langle x\triangleleft y,\varrho\rangle\xrightarrow[u]{\tau}\langle x',\varrho'\rangle}$$

        $$\frac{\langle x,\varrho\rangle\xrightarrow[u]{e_1}\langle\surd,\varrho\rangle \quad \langle y,\varrho\rangle\nrightarrow^{e_3}\quad (\sharp(e_1,e_2),e_2\leq e_3)}{\langle x\triangleleft y,\varrho\rangle\xrightarrow[u]{e_1}\langle\surd,\varrho'\rangle}
        \quad\frac{\langle x,\varrho\rangle\xrightarrow[u]{e_1}\langle x',\varrho'\rangle \quad \langle y,\varrho\rangle\nrightarrow^{e_3}\quad (\sharp(e_1,e_2),e_2\leq e_3)}{\langle x\triangleleft y,\varrho\rangle\xrightarrow[u]{e_1}\langle x',\varrho'\rangle}$$

        $$\frac{\langle x,\varrho\rangle\xrightarrow[u]{e_3}\langle\surd,\varrho'\rangle \quad \langle y,\varrho\rangle\nrightarrow^{e_2}\quad (\sharp(e_1,e_2),e_1\leq e_3)}{\langle x\triangleleft y,\varrho\rangle\xrightarrow[u]{\tau}\langle\surd,\varrho'\rangle}
        \quad\frac{\langle x,\varrho\rangle\xrightarrow[u]{e_3}\langle x',\varrho'\rangle \quad \langle y,\varrho\rangle\nrightarrow^{e_2}\quad (\sharp(e_1,e_2),e_1\leq e_3)}{\langle x\triangleleft y,\varrho\rangle\xrightarrow[u]{\tau}\langle x',\varrho'\rangle}$$

        $$\frac{\langle x,\varrho\rangle\xrightarrow[u]{e}\langle\surd,\varrho'\rangle}{\langle\partial_H(x),\varrho\rangle\xrightarrow[u]{e}\langle\surd,\varrho'\rangle}\quad (e\notin H)\quad\frac{\langle x,\varrho\rangle\xrightarrow[u]{e}\langle x',\varrho'\rangle}{\langle\partial_H(x),\varrho\rangle\xrightarrow[u]{e}\langle\partial_H(x'),\varrho'\rangle}\quad(e\notin H)$$
        \caption{Transition rules of qAPTC with localities}
        \label{TRForqAPTC}
    \end{table}
\end{center}

The axioms for qAPTC with localities are listed in Table \ref{AxiomsForqLeftParallelism}.

\begin{center}
    \begin{table}
        \begin{tabular}{@{}ll@{}}
            \hline No. &Axiom\\
            $A6$ & $x+ \delta = x$\\
            $A7$ & $\delta\cdot x =\delta$\\
            $P1$ & $x\between y = x\parallel y + x\mid y$\\
            $P2$ & $x\parallel y = y \parallel x$\\
            $P3$ & $(x\parallel y)\parallel z = x\parallel (y\parallel z)$\\
            $P4$ & $x\parallel y = x\leftmerge y + y\leftmerge x$\\
            $P5$ & $(e_1\leq e_2)\quad e_1\leftmerge (e_2\cdot y) = (e_1\leftmerge e_2)\cdot y$\\
            $P6$ & $(e_1\leq e_2)\quad (e_1\cdot x)\leftmerge e_2 = (e_1\leftmerge e_2)\cdot x$\\
            $P7$ & $(e_1\leq e_2)\quad (e_1\cdot x)\leftmerge (e_2\cdot y) = (e_1\leftmerge e_2)\cdot (x\between y)$\\
            $P8$ & $(x+ y)\leftmerge z = (x\leftmerge z)+ (y\leftmerge z)$\\
            $P9$ & $\delta\leftmerge x = \delta$\\
            $C1$ & $e_1\mid e_2 = \gamma(e_1,e_2)$\\
            $C2$ & $e_1\mid (e_2\cdot y) = \gamma(e_1,e_2)\cdot y$\\
            $C3$ & $(e_1\cdot x)\mid e_2 = \gamma(e_1,e_2)\cdot x$\\
            $C4$ & $(e_1\cdot x)\mid (e_2\cdot y) = \gamma(e_1,e_2)\cdot (x\between y)$\\
            $C5$ & $(x+ y)\mid z = (x\mid z) + (y\mid z)$\\
            $C6$ & $x\mid (y+ z) = (x\mid y)+ (x\mid z)$\\
            $C7$ & $\delta\mid x = \delta$\\
            $C8$ & $x\mid\delta = \delta$\\
            $CE1$ & $\Theta(e) = e$\\
            $CE2$ & $\Theta(\delta) = \delta$\\
            $CE3$ & $\Theta(x+ y) = \Theta(x)\triangleleft y + \Theta(y)\triangleleft x$\\
            $CE4$ & $\Theta(x\cdot y)=\Theta(x)\cdot\Theta(y)$\\
            $CE5$ & $\Theta(x\parallel y) = ((\Theta(x)\triangleleft y)\parallel y)+ ((\Theta(y)\triangleleft x)\parallel x)$\\
            $CE6$ & $\Theta(x\mid y) = ((\Theta(x)\triangleleft y)\mid y)+ ((\Theta(y)\triangleleft x)\mid x)$\\
            $U1$ & $(\sharp(e_1,e_2))\quad e_1\triangleleft e_2 = \tau$\\
            $U2$ & $(\sharp(e_1,e_2),e_2\leq e_3)\quad e_1\triangleleft e_3 = e_1$\\
            $U3$ & $(\sharp(e_1,e_2),e_2\leq e_3)\quad e3\triangleleft e_1 = \tau$\\
            $U4$ & $e\triangleleft \delta = e$\\
            $U5$ & $\delta \triangleleft e = \delta$\\
            $U6$ & $(x+ y)\triangleleft z = (x\triangleleft z)+ (y\triangleleft z)$\\
            $U7$ & $(x\cdot y)\triangleleft z = (x\triangleleft z)\cdot (y\triangleleft z)$\\
            $U8$ & $(x\leftmerge y)\triangleleft z = (x\triangleleft z)\leftmerge (y\triangleleft z)$\\
            $U9$ & $(x\mid y)\triangleleft z = (x\triangleleft z)\mid (y\triangleleft z)$\\
            $U10$ & $x\triangleleft (y+ z) = (x\triangleleft y)\triangleleft z$\\
            $U11$ & $x\triangleleft (y\cdot z)=(x\triangleleft y)\triangleleft z$\\
            $U12$ & $x\triangleleft (y\leftmerge z) = (x\triangleleft y)\triangleleft z$\\
            $U13$ & $x\triangleleft (y\mid z) = (x\triangleleft y)\triangleleft z$\\
            $L5$ & $u::(x\between y) = u::x\between u:: y$\\
            $L6$ & $u::(x\parallel y) = u::x\parallel u:: y$\\
            $L7$ & $u::(x\mid y) = u::x\mid u:: y$\\
            $L8$ & $u::(\Theta(x)) = \Theta(u::x)$\\
            $L9$ & $u::(x\triangleleft y) = u::x\triangleleft u:: y$\\
            $L10$ & $u::\delta=\delta$\\
            $D1$ & $e\notin H\quad\partial_H(e) = e$\\
            $D2$ & $e\in H\quad \partial_H(e) = \delta$\\
            $D3$ & $\partial_H(\delta) = \delta$\\
            $D4$ & $\partial_H(x+ y) = \partial_H(x)+\partial_H(y)$\\
            $D5$ & $\partial_H(x\cdot y) = \partial_H(x)\cdot\partial_H(y)$\\
            $D6$ & $\partial_H(x\leftmerge y) = \partial_H(x)\leftmerge\partial_H(y)$\\
            $L11$ & $u::\partial_H(x) = \partial_H(u::x)$\\
        \end{tabular}
        \caption{Axioms of parallelism with left parallel composition}
        \label{AxiomsForqLeftParallelism}
    \end{table}
\end{center}

\begin{definition}[Basic terms of $qAPTC$ with localities]
The set of basic terms of $qAPTC$ with localities, $\mathcal{B}(qAPTC^{sl})$, is inductively defined as follows:
\begin{enumerate}
  \item $\mathbb{E}\subset\mathcal{B}(qAPTC^{sl})$;
  \item if $u\in Loc^*, t\in\mathcal{B}(qAPTC^{sl})$ then $u::t\in\mathcal{B}(qAPTC^{sl})$;
  \item if $e\in \mathbb{E}, t\in\mathcal{B}(qAPTC^{sl})$ then $e\cdot t\in\mathcal{B}(qAPTC^{sl})$;
  \item if $t,s\in\mathcal{B}(qAPTC^{sl})$ then $t+ s\in\mathcal{B}(qAPTC^{sl})$;
  \item if $t,s\in\mathcal{B}(qAPTC^{sl})$ then $t\leftmerge s\in\mathcal{B}(qAPTC^{sl})$.
\end{enumerate}
\end{definition}

\begin{theorem}[Elimination theorem of $qAPTC$ with localities]
Let $p$ be a closed $qAPTC$ with localities term. Then there is a basic $qAPTC$ with localities term $q$ such that $qAPTC^{sl}\vdash p=q$.
\end{theorem}

\begin{proof}
The same as that of $APTC^{sl}$, we omit the proof, please refer to \cite{LOC1} for details.
\end{proof}

\begin{theorem}[Generalization of $qAPTC^{sl}$ with localities]
$qAPTC^{sl}$ is a generalization of $qBATC$ with localities.
\end{theorem}

\begin{proof}
It follows from the following three facts.

\begin{enumerate}
  \item The transition rules of $qBATC$ with localities in are all source-dependent;
  \item The sources of the transition rules $qAPTC$ with localities contain an occurrence of $\between$, or $\parallel$, or $\leftmerge$, or $\mid$, or $\Theta$, or $\triangleleft$, or $\partial_H$;
  \item The transition rules of $qAPTC$ with localities are all source-dependent.
\end{enumerate}

So, $qAPTC$ with localities is a generalization of $qBATC$ with localities, that is, $qBATC$ with localities is an embedding of $qAPTC$ with localities, as desired.
\end{proof}

\begin{theorem}[Congruence theorem of $qAPTC$ with localities]
Static location truly concurrent bisimulation equivalences $\sim_p^{sl}$, $\sim_s^{sl}$, $\sim_{hp}^{sl}$ and $\sim_{hhp}^{sl}$ are all congruences with respect to $qAPTC$ with localities.
\end{theorem}

\begin{proof}
It is obvious that static location truly concurrent bisimulations $\sim_p^{sl}$, $\sim_s^{sl}$, $\sim_{hp}^{sl}$ and $\sim_{hhp}^{sl}$ are all equivalent relations with respect to $qAPTC$ with localities. So, it is sufficient to prove
that static location truly concurrent bisimulations $\sim_p^{sl}$, $\sim_s^{sl}$, $\sim_{hp}^{sl}$ and $\sim_{hhp}^{sl}$ are preserved for $\between$, $\parallel$, $\leftmerge$, $\mid$, $\Theta$, $\triangleleft$ and $\partial_H$
according to the transition rules in Table \ref{TRForqAPTC}, that is, if $x\sim_p^{sl}x'$ and $y\sim_p^{sl}y'$, then $x\between y\sim_p^{sl}x'\between y'$, $x\parallel y\sim_p^{sl}x'\parallel y'$,
$x\leftmerge y\sim_p^{sl}x'\leftmerge y'$, $x\mid y\sim_p^{sl}x'\mid y'$, $\Theta(x)\sim_p^{sl}\Theta(x')$, $x\triangleleft y\sim_p^{sl}x'\triangleleft y'$, and $\partial_H(x)\sim_p^{sl}\partial_H(x')$; if $x\sim_s^{sl}x'$ and $y\sim_s^{sl}y'$,
then $x\between y\sim_s^{sl}x'\between y'$, $x\parallel y\sim_s^{sl}x'\parallel y'$,
$x\leftmerge y\sim_s^{sl}x'\leftmerge y'$, $x\mid y\sim_s^{sl}x'\mid y'$, $\Theta(x)\sim_s^{sl}\Theta(x')$, $x\triangleleft y\sim_s^{sl}x'\triangleleft y'$, and $\partial_H(x)\sim_s^{sl}\partial_H(x')$;
if $x\sim_{hp}^{sl}x'$ and $y\sim_{hp}^{sl}y'$, then $x\between y\sim_{hp}^{sl}x'\between y'$, $x\parallel y\sim_{hp}^{sl}x'\parallel y'$,
$x\leftmerge y\sim_{hp}^{sl}x'\leftmerge y'$, $x\mid y\sim_{hp}^{sl}x'\mid y'$, $\Theta(x)\sim_{hp}^{sl}\Theta(x')$, $x\triangleleft y\sim_{hp}^{sl}x'\triangleleft y'$, and $\partial_H(x)\sim_{hp}^{sl}\partial_H(x')$; and if $x\sim_{hhp}^{sl}x'$ and $y\sim_{hhp}^{sl}y'$,
then $x\between y\sim_{hhp}^{sl}x'\between y'$, $x\parallel y\sim_{hhp}^{sl}x'\parallel y'$,
$x\leftmerge y\sim_{hhp}^{sl}x'\leftmerge y'$, $x\mid y\sim_{hhp}^{sl}x'\mid y'$, $\Theta(x)\sim_{hhp}^{sl}\Theta(x')$, $x\triangleleft y\sim_{hhp}^{sl}x'\triangleleft y'$ and $\partial_H(x)\sim_{hhp}^{sl}\partial_H(x')$.
The proof is quit trivial, and we leave the proof as an exercise for the readers.
\end{proof}

\begin{theorem}[Soundness of $qAPTC$ with localities modulo static location truly concurrent bisimulation equivalences]
Let $x$ and $y$ be $qAPTC$ with localities terms. If $qAPTC^{sl}\vdash x=y$, then

\begin{enumerate}
  \item $x\sim_s^{sl} y$;
  \item $x\sim_p^{sl} y$;
  \item $x\sim_{hp}^{sl} y$;
  \item $x\sim_{hhp}^{sl} y$.
\end{enumerate}
\end{theorem}

\begin{proof}
(1) Since static location pomset bisimulation $\sim_p^{sl}$ is both an equivalent and a congruent relation, we only need to check if each axiom in Table \ref{AxiomsForqLeftParallelism} is sound
modulo static location pomset bisimulation equivalence. We leave the proof as an exercise for the readers.

(2) Since  static location step bisimulation $\sim_s^{sl}$ is both an equivalent and a congruent relation, we only need to check if each axiom in Table \ref{AxiomsForqLeftParallelism} is sound modulo
static location step bisimulation equivalence. We leave the proof as an exercise for the readers.

(3) Since static location hp-bisimulation $\sim_{hp}^{sl}$ is both an equivalent and a congruent relation, we only need to check if each axiom in Table \ref{AxiomsForqLeftParallelism} is sound modulo
static location hp-bisimulation equivalence. We leave the proof as an exercise for the readers.

(4) Since static location hhp-bisimulation $\sim_{hhp}^{sl}$ is both an equivalent and a congruent relation, we only need to check if each axiom in Table \ref{AxiomsForqLeftParallelism} is sound modulo
static location hhp-bisimulation equivalence. We leave the proof as an exercise for the readers.
\end{proof}

\begin{theorem}[Completeness of $qAPTC$ with localities modulo static location truly concurrent bisimulation equivalences]
Let $x$ and $y$ be $qAPTC$ with localities terms.

\begin{enumerate}
  \item If $x\sim_s^{sl} y$, then $qAPTC^{sl}\vdash x=y$;
  \item if $x\sim_p^{sl} y$, then $qAPTC^{sl}\vdash x=y$;
  \item if $x\sim_{hp}^{sl} y$, then $qAPTC^{sl}\vdash x=y$;
  \item if $x\sim_{hhp}^{sl} y$, then $qAPTC^{sl}\vdash x=y$.
\end{enumerate}
\end{theorem}

\begin{proof}
According to the definition of static location truly concurrent bisimulation equivalences $\sim_p^{sl}$, $\sim_s^{sl}$, $\sim_{hp}^{sl}$ and $\sim_{hhp}^{sl}$, $p\sim_p^{sl}q$, $p\sim_s^{sl}q$, $p\sim_{hp}^{sl}q$ and $p\sim_{hhp}^{sl}q$ implies
both the bisimilarities between $p$ and $q$, and also the in the same quantum states. According to the completeness of APTC (please refer to \cite{LOC1} for details), we can get the
completeness of qAPTC with localities.
\end{proof}

\subsection{Recursion}\label{qorec}

\begin{definition}[Recursive specification]
A recursive specification is a finite set of recursive equations

$$X_1=t_1(X_1,\cdots,X_n)$$
$$\cdots$$
$$X_n=t_n(X_1,\cdots,X_n)$$

where the left-hand sides of $X_i$ are called recursion variables, and the right-hand sides $t_i(X_1,\cdots,X_n)$ are process terms in $qAPTC$ with localities and possible occurrences of the recursion
variables $X_1,\cdots,X_n$.
\end{definition}

\begin{definition}[Solution]
Processes $p_1,\cdots,p_n$ are a solution for a recursive specification $\{X_i=t_i(X_1,\cdots,X_n)|i\in\{1,\cdots,n\}\}$ (with respect to static location truly concurrent bisimulation equivalences
$\sim_s^{sl}$($\sim_p^{sl}$, $\sim_{hp}^{sl}$, $\sim_{hhp}^{sl}$)) if $p_i\sim_s^{sl} (\sim_p^{sl}, \sim_{hp}^{sl},\sim{hhp})t_i(p_1,\cdots,p_n)$ for $i\in\{1,\cdots,n\}$.
\end{definition}

\begin{definition}[Guarded recursive specification]
A recursive specification

$$X_1=t_1(X_1,\cdots,X_n)$$
$$...$$
$$X_n=t_n(X_1,\cdots,X_n)$$

is guarded if the right-hand sides of its recursive equations can be adapted to the form by applications of the axioms in $qAPTC$ with localities and replacing recursion variables by the right-hand
sides of their recursive equations,

$(u_{11}::a_{11}\leftmerge\cdots\leftmerge u_{1i_1}::a_{1i_1})\cdot s_1(X_1,\cdots,X_n)+\cdots+(u_{k1}::a_{k1}\leftmerge\cdots\leftmerge u_{ki_k}::a_{ki_k})\cdot s_k(X_1,\cdots,X_n)\\
+(v_{11}::b_{11}\leftmerge\cdots\leftmerge v_{1j_1}::b_{1j_1})+\cdots+(v_{1j_1}::b_{1j_1}\leftmerge\cdots\leftmerge v_{1j_l}::b_{lj_l})$

where $a_{11},\cdots,a_{1i_1},a_{k1},\cdots,a_{ki_k},b_{11},\cdots,b_{1j_1},b_{1j_1},\cdots,b_{lj_l}\in \mathbb{E}$, and the sum above is allowed to be empty, in which case it
represents the deadlock $\delta$.
\end{definition}

\begin{definition}[Linear recursive specification]
A recursive specification is linear if its recursive equations are of the form

$(u_{11}::a_{11}\leftmerge\cdots\leftmerge u_{1i_1}::a_{1i_1})X_1+\cdots+(u_{k1}::a_{k1}\leftmerge\cdots\leftmerge u_{ki_k}::a_{ki_k})X_k\\
+(v_{11}::b_{11}\leftmerge\cdots\leftmerge v_{1j_1}::b_{1j_1})+\cdots+(v_{1j_1}::b_{1j_1}\leftmerge\cdots\leftmerge v_{1j_l}::b_{lj_l})$

where $a_{11},\cdots,a_{1i_1},a_{k1},\cdots,a_{ki_k},b_{11},\cdots,b_{1j_1},b_{1j_1},\cdots,b_{lj_l}\in \mathbb{E}$, and the sum above is allowed to be empty, in which case it
represents the deadlock $\delta$.
\end{definition}

\begin{center}
    \begin{table}
        $$\frac{t_i(\langle X_1|E\rangle,\cdots,\langle X_n|E\rangle)\xrightarrow[u]{\{e_1,\cdots,e_k\}}\surd}{\langle X_i|E\rangle\xrightarrow[u]{\{e_1,\cdots,e_k\}}\surd}$$
        $$\frac{t_i(\langle X_1|E\rangle,\cdots,\langle X_n|E\rangle)\xrightarrow[u]{\{e_1,\cdots,e_k\}} y}{\langle X_i|E\rangle\xrightarrow[u]{\{e_1,\cdots,e_k\}} y}$$
        \caption{Transition rules of guarded recursion}
        \label{TRForGR}
    \end{table}
\end{center}

\begin{theorem}[Conservitivity of $qAPTC$ with localities and guarded recursion]
$qAPTC$ with localities and guarded recursion is a conservative extension of $qAPTC$ with localities.
\end{theorem}

\begin{proof}
It follows from the following three facts.

\begin{enumerate}
  \item The transition rules of $qAPTC$ with localities in are all source-dependent;
  \item The sources of the transition rules $qAPTC$ with localities and guarded recursion contain only one constant;
  \item The transition rules of $qAPTC$ with localities and guarded recursion are all source-dependent.
\end{enumerate}

So, $qAPTC$ with localities and guarded recursion is a conservative extension of $qAPTC$ with localities, as desired.
\end{proof}

\begin{theorem}[Congruence theorem of $qAPTC$ with localities and guarded recursion]
Static location truly concurrent bisimulation equivalences $\sim_p^{sl}$, $\sim_s^{sl}$, $\sim_{hp}^{sl}$, $\sim_{hhp}^{sl}$ are all congruences with respect to $qAPTC$ with localities and guarded recursion.
\end{theorem}

\begin{proof}
It follows the following two facts:
\begin{enumerate}
  \item in a guarded recursive specification, right-hand sides of its recursive equations can be adapted to the form by applications of the axioms in $qAPTC$ with localities and replacing recursion
  variables by the right-hand sides of their recursive equations;
  \item static location truly concurrent bisimulation equivalences $\sim_p^{sl}$, $\sim_s^{sl}$, $\sim_{hp}^{sl}$ and $\sim_{hhp}^{sl}$ are all congruences with respect to all operators of $qAPTC$ with localities.
\end{enumerate}
\end{proof}

The $RDP$ (Recursive Definition Principle) and the $RSP$ (Recursive Specification Principle) are shown in Table \ref{RDPRSP}.

\begin{center}
\begin{table}
  \begin{tabular}{@{}ll@{}}
\hline No. &Axiom\\
  $RDP$ & $\langle X_i|E\rangle = t_i(\langle X_1|E\rangle,\cdots,\langle X_n|E\rangle)\quad (i\in\{1,\cdots,n\})$\\
  $RSP$ & if $y_i=t_i(y_1,\cdots,y_n)$ for $i\in\{1,\cdots,n\}$, then $y_i=\langle X_i|E\rangle \quad(i\in\{1,\cdots,n\})$\\
\end{tabular}
\caption{Recursive definition and specification principle}
\label{RDPRSP}
\end{table}
\end{center}

\begin{theorem}[Elimination theorem of $qAPTC$ with localities and linear recursion]
Each process term in $qAPTC$ with localities and linear recursion is equal to a process term $\langle X_1|E\rangle$ with $E$ a linear recursive specification.
\end{theorem}

\begin{proof}
The same as that of $APTC^{sl}$ with linear recursion, we omit the proof, please refer to \cite{LOC1} for details.
\end{proof}

\begin{theorem}[Soundness of $qAPTC$ with localities and guarded recursion]
Let $x$ and $y$ be $qAPTC$ with localities and guarded recursion terms. If $qAPTC^{sl}\textrm{ with guarded recursion}\vdash x=y$, then
\begin{enumerate}
  \item $x\sim_s^{sl} y$;
  \item $x\sim_p^{sl} y$;
  \item $x\sim_{hp}^{sl} y$;
  \item $x\sim_{hhp}^{sl} y$.
\end{enumerate}
\end{theorem}

\begin{proof}
(1) Since static location pomset bisimulation $\sim_p^{sl}$ is both an equivalent and a congruent relation, we only need to check if each axiom in Table \ref{RDPRSP} is sound
modulo static location pomset bisimulation equivalence. We leave the proof as an exercise for the readers.

(2) Since static location step bisimulation $\sim_s^{sl}$ is both an equivalent and a congruent relation, we only need to check if each axiom in Table \ref{RDPRSP} is sound modulo
static location step bisimulation equivalence. We leave the proof as an exercise for the readers.

(3) Since static location hp-bisimulation $\sim_{hp}^{sl}$ is both an equivalent and a congruent relation, we only need to check if each axiom in Table \ref{RDPRSP} is sound modulo
static location hp-bisimulation equivalence. We leave the proof as an exercise for the readers.

(4) Since static location hhp-bisimulation $\sim_{hhp}^{sl}$ is both an equivalent and a congruent relation, we only need to check if each axiom in Table \ref{RDPRSP} is sound modulo
static location hhp-bisimulation equivalence. We leave the proof as an exercise for the readers.
\end{proof}

\begin{theorem}[Completeness of $qAPTC$ with localities and linear recursion]
Let $p$ and $q$ be closed $qAPTC$ with localities and linear recursion terms, then,
\begin{enumerate}
  \item if $p\sim_s^{sl} q$ then $p=q$;
  \item if $p\sim_p^{sl} q$ then $p=q$;
  \item if $p\sim_{hp}^{sl} q$ then $p=q$;
  \item if $p\sim_{hhp}^{sl} q$ then $p=q$.
\end{enumerate}
\end{theorem}

\begin{proof}
According to the definition of static location truly concurrent bisimulation equivalences $\sim_p^{sl}$, $\sim_s^{sl}$, $\sim_{hp}^{sl}$ and $\sim_{hhp}^{sl}$, $p\sim_p^{sl}q$, $p\sim_s^{sl}q$, $p\sim_{hp}^{sl}q$ and $p\sim_{hhp}^{sl}q$ implies
both the bisimilarities between $p$ and $q$, and also the in the same quantum states. According to the completeness of $APTC^{sl}$ with linear recursion (please refer to \cite{LOC1} for details), we can get the
completeness of qAPTC with localities and linear recursion.
\end{proof}

\subsection{Abstraction}\label{qoabs}

\begin{definition}[Guarded linear recursive specification]
A recursive specification is linear if its recursive equations are of the form

$(u_{11}::a_{11}\leftmerge\cdots\leftmerge u_{1i_1}::a_{1i_1})X_1+\cdots+(u_{k1}::a_{k1}\leftmerge\cdots\leftmerge u_{ki_k}::a_{ki_k})X_k\\
+(v_{11}::b_{11}\leftmerge\cdots\leftmerge v_{1j_1}::b_{1j_1})+\cdots+(v_{1j_1}::b_{1j_1}\leftmerge\cdots\leftmerge v_{1j_l}::b_{lj_l})$

where $a_{11},\cdots,a_{1i_1},a_{k1},\cdots,a_{ki_k},b_{11},\cdots,b_{1j_1},b_{1j_1},\cdots,b_{lj_l}\in \mathbb{E}\cup\{\tau\}$, and the sum above is allowed to be empty, in which case
it represents the deadlock $\delta$.

A linear recursive specification $E$ is guarded if there does not exist an infinite sequence of $\tau$-transitions
$\langle X|E\rangle\xrightarrow{\tau}\langle X'|E\rangle\xrightarrow{\tau}\langle X''|E\rangle\xrightarrow{\tau}\cdots$.
\end{definition}

The transition rules of $\tau$ are shown in Table \ref{TRForqAbstraction}, and axioms of $\tau$ are as Table \ref{AxiomsForqTauLeft} shows.

\begin{center}
    \begin{table}
        $$\frac{}{\langle\tau,\varrho\rangle\xrightarrow{\tau}\langle\surd,\tau(\varrho)\rangle}$$
        $$\frac{\langle x,\varrho\rangle\xrightarrow[u]{e}\langle\surd,\varrho'\rangle}{\langle\tau_I(x),\varrho\rangle\xrightarrow[u]{e}\langle\surd,\varrho'\rangle}\quad e\notin I
        \quad\quad\frac{\langle x,\varrho\rangle\xrightarrow[u]{e}\langle x',\varrho'\rangle}{\langle\tau_I(x),\varrho\rangle\xrightarrow[u]{e}\langle\tau_I(x'),\varrho'\rangle}\quad e\notin I$$

        $$\frac{\langle x,\varrho\rangle\xrightarrow[u]{e}\langle\surd,\varrho'\rangle}{\langle\tau_I(x),\varrho\rangle\xrightarrow{\tau}\langle\surd,\tau(\varrho)\rangle}\quad e\in I
        \quad\quad\frac{\langle x,\varrho\rangle\xrightarrow[u]{e}\langle x',\varrho'\rangle}{\langle\tau_I(x),\varrho\rangle\xrightarrow{\tau}\langle\tau_I(x'),\tau(\varrho)\rangle}\quad e\in I$$
        \caption{Transition rule of $qAPTC^{sl}_{\tau}$}
        \label{TRForqAbstraction}
    \end{table}
\end{center}

\begin{theorem}[Conservitivity of $qAPTC$ with localities and silent step and guarded linear recursion]
$qAPTC$ with localities and silent step and guarded linear recursion is a conservative extension of $qAPTC$ with localities and linear recursion.
\end{theorem}

\begin{proof}
Since the transition rules of $qAPTC$ with localities and silent step and guarded linear recursion are source-dependent, and the transition rules for $\tau$ in Table
\ref{TRForqAbstraction} contain only a fresh constant $\tau$ in their source, so the transition rules of $qAPTC$ with localities and silent step and guarded linear recursion is a conservative extension
of those of $qAPTC$ with localities and guarded linear recursion.
\end{proof}

\begin{theorem}[Congruence theorem of $qAPTC$ with localities and silent step and guarded linear recursion]
Rooted branching static location truly concurrent bisimulation equivalences $\approx_{rbp}^{sl}$, $\approx_{rbs}^{sl}$, $\approx_{rbhp}^{sl}$, and $\approx_{rbhhp}^{sl}$ are all congruences with respect to $qAPTC$ with localities and
silent step and guarded linear recursion.
\end{theorem}

\begin{proof}
It follows the following three facts:
\begin{enumerate}
  \item in a guarded linear recursive specification, right-hand sides of its recursive equations can be adapted to the form by applications of the axioms in $qAPTC$ with localities and replacing
  recursion variables by the right-hand sides of their recursive equations;
  \item static location truly concurrent bisimulation equivalences $\sim_p^{sl}$, $\sim_s^{sl}$, $\sim_{hp}^{sl}$ and $\sim_{hhp}^{sl}$ are all congruences with respect to all operators of
  $qAPTC$ with localities, while static location truly concurrent bisimulation equivalences $\sim_p^{sl}$, $\sim_s^{sl}$, $\sim_{hp}^{sl}$ and $\sim_{hhp}^{sl}$ imply the corresponding rooted
  branching static location truly concurrent bisimulations $\approx_{rbp}^{sl}$, $\approx_{rbs}^{sl}$, $\approx_{rbhp}^{sl}$ and $\approx_{rbhhp}^{sl}$, so rooted branching static location truly concurrent
  bisimulations $\approx_{rbp}^{sl}$, $\approx_{rbs}^{sl}$, $\approx_{rbhp}^{sl}$ and $\approx_{rbhhp}^{sl}$ are all congruences with respect to all operators of $qAPTC$ with localities;
  \item While $\mathbb{E}$ is extended to $\mathbb{E}\cup\{\tau\}$, it can be proved that rooted branching static location truly concurrent
  bisimulations $\approx_{rbp}^{sl}$, $\approx_{rbs}^{sl}$, $\approx_{rbhp}^{sl}$ and $\approx_{rbhhp}^{sl}$ are all congruences with respect to all operators of $qAPTC$ with localities, we omit it.
\end{enumerate}
\end{proof}

\begin{center}
\begin{table}
  \begin{tabular}{@{}ll@{}}
\hline No. &Axiom\\
  $B1$ & $e\cdot\tau=e$\\
  $B2$ & $e\cdot(\tau\cdot(x+y)+x)=e\cdot(x+y)$\\
  $B3$ & $x\leftmerge\tau=x$\\
  $L13$ & $u::\tau=\tau$\
\end{tabular}
\caption{Axioms of silent step}
\label{AxiomsForqTauLeft}
\end{table}
\end{center}

\begin{theorem}[Elimination theorem of $qAPTC$ with localities and silent step and guarded linear recursion]
Each process term in $qAPTC$ with localities and silent step and guarded linear recursion is equal to a process term $\langle X_1|E\rangle$ with $E$ a guarded linear recursive specification.
\end{theorem}

\begin{proof}
The same as that of $APTC^{sl}$ with silent step and guarded linear recursion, we omit the proof, please refer to \cite{LOC1} for details.
\end{proof}

\begin{theorem}[Soundness of $qAPTC$ with localities and silent step and guarded linear recursion]
Let $x$ and $y$ be $qAPTC$ with localities and silent step and guarded linear recursion terms. If $qAPTC$ with localities and silent step and guarded linear recursion $\vdash x=y$, then
\begin{enumerate}
  \item $x\approx_{rbs}^{sl} y$;
  \item $x\approx_{rbp}^{sl} y$;
  \item $x\approx_{rbhp}^{sl} y$;
  \item $x\approx_{rbhhp}^{sl} y$.
\end{enumerate}
\end{theorem}

\begin{proof}
(1) Since rooted branching static location pomset bisimulation $\approx_{rbp}^{sl}$ is both an equivalent and a congruent relation with respect to $qAPTC$ with localities and silent step and guarded
linear recursion, we only need to check if each axiom in Table \ref{AxiomsForqTauLeft} is sound modulo rooted branching static location pomset bisimulation $\approx_{rbp}^{sl}$. We leave them as
exercises to the readers.

(2) Since rooted branching static location step bisimulation $\approx_{rbs}^{sl}$ is both an equivalent and a congruent relation with respect to $qAPTC$ with localities and silent step and guarded
linear recursion, we only need to check if each axiom in Table \ref{AxiomsForqTauLeft} is sound modulo rooted branching static location step bisimulation $\approx_{rbs}^{sl}$. We leave them
as exercises to the readers.

(3) Since rooted branching static location hp-bisimulation $\approx_{rbhp}^{sl}$ is both an equivalent and a congruent relation with respect to $qAPTC$ with localities and silent step and guarded linear
recursion, we only need to check if each axiom in Table \ref{AxiomsForqTauLeft} is sound modulo rooted branching static location hp-bisimulation $\approx_{rbhp}^{sl}$. We leave them as exercises
to the readers.

(4) Since rooted branching static location hhp-bisimulation $\approx_{rbhhp}^{sl}$ is both an equivalent and a congruent relation with respect to $qAPTC$ with localities and silent step and guarded linear
recursion, we only need to check if each axiom in Table \ref{AxiomsForqTauLeft} is sound modulo rooted branching static location hhp-bisimulation $\approx_{rbhhp}^{sl}$. We leave them as exercises
to the readers.
\end{proof}

\begin{theorem}[Completeness of $qAPTC$ with localities and silent step and guarded linear recursion]
Let $p$ and $q$ be closed $qAPTC$ with localities and silent step and guarded linear recursion terms, then,
\begin{enumerate}
  \item if $p\approx_{rbs}^{sl} q$ then $p=q$;
  \item if $p\approx_{rbp}^{sl} q$ then $p=q$;
  \item if $p\approx_{rbhp}^{sl} q$ then $p=q$;
  \item if $p\approx_{rbhhp}^{sl} q$ then $p=q$.
\end{enumerate}
\end{theorem}

\begin{proof}
According to the definition of static location truly concurrent rooted branching bisimulation equivalences $\approx_{rbp}^{sl}$, $\approx_{rbs}^{sl}$, $\approx_{rbhp}^{sl}$ and $\approx_{rbhhp}^{sl}$,
$p\approx_{rbp}^{sl}q$, $p\approx_{rbs}^{sl}q$, $p\approx_{rbhp}^{sl}q$ and $p\approx_{rbhhp}^{sl}q$ implies
both the rooted branching bisimilarities between $p$ and $q$, and also the in the same quantum states. According to the completeness of $APTC^{sl}$ with silent step and guarded linear recursion (please refer to \cite{LOC1} for details), we can get the
completeness of qAPTC with localities and silent step and guarded linear recursion.
\end{proof}

The transition rules of $\tau_I$ are shown in Table \ref{TRForqAbstraction}, and the axioms are shown in Table \ref{AxiomsForqAbstractionLeft}.

\begin{theorem}[Conservitivity of $qAPTC^{sl}_{\tau}$ with guarded linear recursion]
$qAPTC^{sl}_{\tau}$ with guarded linear recursion is a conservative extension of $qAPTC$ with localities and silent step and guarded linear recursion.
\end{theorem}

\begin{proof}
Since the transition rules of $qAPTC$ with localities and silent step and guarded linear recursion are source-dependent, and the transition rules for abstraction operator in Table
\ref{TRForqAbstraction} contain only a fresh operator $\tau_I$ in their source, so the transition rules of $qAPTC^{sl}_{\tau}$ with guarded linear recursion is a conservative extension
of those of $qAPTC$ with localities and silent step and guarded linear recursion.
\end{proof}

\begin{theorem}[Congruence theorem of $qAPTC^{sl}_{\tau}$ with guarded linear recursion]
Rooted branching static location truly concurrent bisimulation equivalences $\approx_{rbp}^{sl}$, $\approx_{rbs}^{sl}$, $\approx_{rbhp}^{sl}$ and $\approx_{rbhhp}^{sl}$ are all congruences with respect to $qAPTC^{sl}_{\tau}$
with guarded linear recursion.
\end{theorem}

\begin{proof}
(1) It is easy to see that rooted branching static location pomset bisimulation is an equivalent relation on $qAPTC^{sl}_{\tau}$ with guarded linear recursion terms, we only need to
prove that $\approx_{rbp}^{sl}$ is preserved by the operator $\tau_I$. It is trivial and we leave the proof as an exercise for the readers.

(2) It is easy to see that rooted branching static location step bisimulation is an equivalent relation on $qAPTC^{sl}_{\tau}$ with guarded linear recursion terms, we only need to
prove that $\approx_{rbs}^{sl}$ is preserved by the operator $\tau_I$. It is trivial and we leave the proof as an exercise for the readers.

(3) It is easy to see that rooted branching static location hp-bisimulation is an equivalent relation on $qAPTC^{sl}_{\tau}$ with guarded linear recursion terms, we only need to
prove that $\approx_{rbhp}^{sl}$ is preserved by the operator $\tau_I$. It is trivial and we leave the proof as an exercise for the readers.

(4) It is easy to see that rooted branching static location hhp-bisimulation is an equivalent relation on $qAPTC^{sl}_{\tau}$ with guarded linear recursion terms, we only need to
prove that $\approx_{rbhhp}^{sl}$ is preserved by the operator $\tau_I$. It is trivial and we leave the proof as an exercise for the readers.
\end{proof}

\begin{center}
\begin{table}
  \begin{tabular}{@{}ll@{}}
\hline No. &Axiom\\
  $TI1$ & $e\notin I\quad \tau_I(e)=e$\\
  $TI2$ & $e\in I\quad \tau_I(e)=\tau$\\
  $TI3$ & $\tau_I(\delta)=\delta$\\
  $TI4$ & $\tau_I(x+y)=\tau_I(x)+\tau_I(y)$\\
  $TI5$ & $\tau_I(x\cdot y)=\tau_I(x)\cdot\tau_I(y)$\\
  $TI6$ & $\tau_I(x\leftmerge y)=\tau_I(x)\leftmerge\tau_I(y)$\\
  $L14$ & $u::\tau_I(x)=\tau_I(u::x)$\\
  $L15$ & $e\notin I\quad \tau_I(u::e)=u::e$\\
  $L16$ & $e\in I\quad \tau_I(u::e)=\tau$\\
\end{tabular}
\caption{Axioms of abstraction operator}
\label{AxiomsForqAbstractionLeft}
\end{table}
\end{center}

\begin{theorem}[Soundness of $qAPTC^{sl}_{\tau}$ with guarded linear recursion]
Let $x$ and $y$ be $qAPTC^{sl}_{\tau}$ with guarded linear recursion terms. If $qAPTC^{sl}_{\tau}$ with guarded linear recursion$\vdash x=y$, then
\begin{enumerate}
  \item $x\approx_{rbs}^{sl} y$;
  \item $x\approx_{rbp}^{sl} y$;
  \item $x\approx_{rbhp}^{sl} y$;
  \item $x\approx_{rbhhp}^{sl} y$.
\end{enumerate}
\end{theorem}

\begin{proof}
(1) Since rooted branching static location step bisimulation $\approx_{rbs}^{sl}$ is both an equivalent and a congruent relation with respect to $APTC^{sl}_{\tau}$ with guarded linear
recursion, we only need to check if each axiom in Table \ref{AxiomsForqAbstractionLeft} is sound modulo rooted branching static location step bisimulation $\approx_{rbs}^{sl}$. We leave them as
exercises to the readers.

(2) Since rooted branching static location pomset bisimulation $\approx_{rbp}^{sl}$ is both an equivalent and a congruent relation with respect to $APTC^{sl}_{\tau}$ with guarded linear
recursion, we only need to check if each axiom in Table \ref{AxiomsForqAbstractionLeft} is sound modulo rooted branching static location pomset bisimulation $\approx_{rbp}^{sl}$. We leave them
as exercises to the readers.

(3) Since rooted branching static location hp-bisimulation $\approx_{rbhp}^{sl}$ is both an equivalent and a congruent relation with respect to $APTC^{sl}_{\tau}$ with guarded linear
recursion, we only need to check if each axiom in Table \ref{AxiomsForqAbstractionLeft} is sound modulo rooted branching static location hp-bisimulation $\approx_{rbhp}^{sl}$. We leave them as
exercises to the readers.

(4) Since rooted branching static location hhp-bisimulation $\approx_{rbhhp}^{sl}$ is both an equivalent and a congruent relation with respect to $APTC^{sl}_{\tau}$ with guarded linear
recursion, we only need to check if each axiom in Table \ref{AxiomsForqAbstractionLeft} is sound modulo rooted branching static location hhp-bisimulation $\approx_{rbhhp}^{sl}$. We leave them as
exercises to the readers.
\end{proof}

\begin{definition}[Cluster]
Let $E$ be a guarded linear recursive specification, and $I\subseteq \mathbb{E}$. Two recursion variable $X$ and $Y$ in $E$ are in the same cluster for $I$ iff there exist sequences of
transitions $\langle X|E\rangle\xrightarrow[u]{\{b_{11},\cdots, b_{1i}\}}\cdots[u]\xrightarrow{\{b_{m1},\cdots, b_{mi}\}}\langle Y|E\rangle$ and
$\langle Y|E\rangle\xrightarrow[v]{\{c_{11},\cdots, c_{1j}\}}\cdots\xrightarrow[v]{\{c_{n1},\cdots, c_{nj}\}}\langle X|E\rangle$, where
$b_{11},\cdots,b_{mi},c_{11},\cdots,c_{nj}\in I\cup\{\tau\}$.

$u_1::a_1\leftmerge\cdots\leftmerge u_k::a_k$ or $(u_1::a_1\leftmerge\cdots\leftmerge u_k::a_k) X$ is an exit for the cluster $C$ iff: (1) $u_1::a_1\leftmerge\cdots\leftmerge u_k::a_k$
or $(u_1::a_1\leftmerge\cdots\leftmerge u_k::a_k) X$ is a summand at the right-hand side of the recursive equation for a recursion variable in $C$, and (2) in the case of
$(u_1::a_1\leftmerge\cdots\leftmerge u_k::a_k) X$, either $a_l\notin I\cup\{\tau\}(l\in\{1,2,\cdots,k\})$ or $X\notin C$.
\end{definition}

\begin{center}
\begin{table}
  \begin{tabular}{@{}ll@{}}
\hline No. &Axiom\\
  $CFAR$ & If $X$ is in a cluster for $I$ with exits \\
           & $\{(u_{11}::a_{11}\leftmerge\cdots\leftmerge u_{1i}::a_{1i})Y_1,\cdots,(u_{m1}::a_{m1}\leftmerge\cdots\leftmerge u_{mi}::a_{mi})Y_m,$ \\
           & $v_{11}::b_{11}\leftmerge\cdots\leftmerge v_{1j}::b_{1j},\cdots,v_{n1}::b_{n1}\leftmerge\cdots\leftmerge v_{nj}::b_{nj}\}$, \\
           & then $\tau\cdot\tau_I(\langle X|E\rangle)=$\\
           & $\tau\cdot\tau_I((u_{11}::a_{11}\leftmerge\cdots\leftmerge u_{1i}::a_{1i})\langle Y_1|E\rangle+\cdots+(u_{m1}::a_{m1}\leftmerge\cdots\leftmerge u_{mi}::a_{mi})\langle Y_m|E\rangle$\\
           & $+v_{11}::b_{11}\leftmerge\cdots\leftmerge v_{1j}::b_{1j}+\cdots+v_{n1}::b_{n1}\leftmerge\cdots\leftmerge v_{nj}::b_{nj})$\\
  \end{tabular}
\caption{Cluster fair abstraction rule}
\label{qCFARLeft}
\end{table}
\end{center}

\begin{theorem}[Soundness of $CFAR$]
$CFAR$ is sound modulo rooted branching static location truly concurrent bisimulation equivalences $\approx_{rbs}^{sl}$, $\approx_{rbp}^{sl}$, $\approx_{rbhp}^{sl}$ and $\approx_{rbhhp}^{sl}$.
\end{theorem}

\begin{proof}
(1) Since rooted branching static location step bisimulation $\approx_{rbs}^{sl}$ is both an equivalent and a congruent relation with respect to $APTC^{sl}_{\tau}$ with guarded linear
recursion, we only need to check if each axiom in Table \ref{qCFARLeft} is sound modulo rooted branching static location step bisimulation $\approx_{rbs}^{sl}$. We leave them as
exercises to the readers.

(2) Since rooted branching static location pomset bisimulation $\approx_{rbp}^{sl}$ is both an equivalent and a congruent relation with respect to $APTC^{sl}_{\tau}$ with guarded linear
recursion, we only need to check if each axiom in Table \ref{qCFARLeft} is sound modulo rooted branching static location pomset bisimulation $\approx_{rbp}^{sl}$. We leave them
as exercises to the readers.

(3) Since rooted branching static location hp-bisimulation $\approx_{rbhp}^{sl}$ is both an equivalent and a congruent relation with respect to $APTC^{sl}_{\tau}$ with guarded linear
recursion, we only need to check if each axiom in Table \ref{qCFARLeft} is sound modulo rooted branching static location hp-bisimulation $\approx_{rbhp}^{sl}$. We leave them as
exercises to the readers.

(4) Since rooted branching static location hhp-bisimulation $\approx_{rbhhp}^{sl}$ is both an equivalent and a congruent relation with respect to $APTC^{sl}_{\tau}$ with guarded linear
recursion, we only need to check if each axiom in Table \ref{qCFARLeft} is sound modulo rooted branching static location hhp-bisimulation $\approx_{rbhhp}^{sl}$. We leave them as
exercises to the readers.
\end{proof}

\begin{theorem}[Completeness of $qAPTC^{sl}_{\tau}$ with guarded linear recursion and $CFAR$]
Let $p$ and $q$ be closed $qAPTC^{sl}_{\tau}$ with guarded linear recursion and $CFAR$ terms, then,
\begin{enumerate}
  \item if $p\approx_{rbs}^{sl} q$ then $p=q$;
  \item if $p\approx_{rbp}^{sl} q$ then $p=q$;
  \item if $p\approx_{rbhp}^{sl} q$ then $p=q$;
  \item if $p\approx_{rbhhp}^{sl} q$ then $p=q$.
\end{enumerate}
\end{theorem}

\begin{proof}
According to the definition of static location truly concurrent rooted branching bisimulation equivalences $\approx_{rbp}^{sl}$, $\approx_{rbs}^{sl}$, $\approx_{rbhp}^{sl}$ and $\approx_{rbhhp}^{sl}$,
$p\approx_{rbp}^{sl}q$, $p\approx_{rbs}^{sl}q$, $p\approx_{rbhp}^{sl}q$ and $p\approx_{rbhhp}^{sl}q$ implies
both the rooted branching bisimilarities between $p$ and $q$, and also the in the same quantum states. According to the completeness of $APTC^{sl}_{\tau}$ with guarded linear recursion (please refer to \cite{LOC1} for details), we can get the
completeness of $qAPTC^{sl}_{\tau}$ with guarded linear recursion.
\end{proof}

\subsection{Quantum Entanglement}\label{qe1}

If two quantum variables are entangled, then a quantum operation performed on one variable, then state of the other quantum variable is also changed. So, the entangled states must be
all the inner variables or all the public variables. We will introduced a mechanism to explicitly define quantum entanglement in open quantum systems.
A new constant called shadow constant denoted $\circledS^e_i$ corresponding to a specific quantum operation.
If there are $n$ quantum variables entangled, they maybe be distributed in different quantum systems, with a quantum operation performed on one variable, there should be one
$\circledS^e_i$ ($1\leq i\leq n-1$) executed on each variable in the other $n-1$ variables. Thus, distributed variables are all hidden behind actions.
In the following, we let $\circledS\in \mathbb{E}$.

The axiom system of the shadow constant $\circledS$ is shown in Table \ref{AxiomsForQE1}.

\begin{center}
\begin{table}
  \begin{tabular}{@{}ll@{}}
  \hline No. &Axiom\\
  $SC1$ & $\circledS\cdot x = x$ \\
  $SC2$ & $x\cdot\circledS = x$\\
  $SC3$ & $e\leftmerge\circledS^e=e$\\
  $SC4$ & $\circledS^e\leftmerge e=e$\\
  $SC5$ & $e\leftmerge(\circledS^e\cdot y) = e\cdot y$\\
  $SC6$ & $\circledS^e\leftmerge(e\cdot y) = e\cdot y$\\
  $SC7$ & $(e\cdot x)\leftmerge\circledS^e = e\cdot x$\\
  $SC8$ & $(\circledS^e\cdot x)\leftmerge e = e\cdot x$\\
  $SC9$ & $(e\cdot x)\leftmerge(\circledS^e\cdot y) = e\cdot (x\between y)$\\
  $SC10$ & $(\circledS^e\cdot x)\leftmerge(e\cdot y) = e\cdot (x\between y)$\\
  $L17$ & $loc::\circledS = \circledS$\\
\end{tabular}
\caption{Axioms of quantum entanglement}
\label{AxiomsForQE1}
\end{table}
\end{center}

The transition rules of constant $\circledS$ are as Table \ref{TRForENT1} shows.

\begin{center}
    \begin{table}
        $$\frac{}{\langle\circledS,\varrho\rangle\rightarrow\langle\surd,\varrho\rangle}$$
        $$\frac{\langle x, \varrho\rangle\xrightarrow[u]{e}\langle x',\varrho'\rangle\quad \langle y, \varrho'\rangle\xrightarrow{\circledS^e}\langle y',\varrho'\rangle}{\langle x\leftmerge y,\varrho\rangle\xrightarrow[u]{e}\langle x'\between y', \varrho'\rangle}$$
        $$\frac{\langle x, \varrho\rangle\xrightarrow[u]{e}\langle\surd,\varrho'\rangle\quad \langle y, \varrho'\rangle\xrightarrow{\circledS^e}\langle y',\varrho'\rangle}{\langle x\leftmerge y,\varrho\rangle\xrightarrow[u]{e}\langle y', \varrho'\rangle}$$
        $$\frac{\langle x, \varrho'\rangle\xrightarrow{\circledS^e}\langle\surd,\varrho'\rangle\quad \langle y, \varrho\rangle\xrightarrow[u]{e}\langle y',\varrho'\rangle}{\langle x\leftmerge y,\varrho\rangle\xrightarrow[u]{e}\langle y', \varrho'\rangle}$$
        $$\frac{\langle x, \varrho\rangle\xrightarrow[u]{e}\langle\surd,\varrho'\rangle\quad \langle y, \varrho'\rangle\xrightarrow{\circledS^e}\langle\surd,\varrho'\rangle}{\langle x\leftmerge y,\varrho\rangle\xrightarrow[u]{e}\langle \surd, \varrho'\rangle}$$
        \caption{Transition rules of constant $\circledS$}
        \label{TRForENT1}
    \end{table}
\end{center}

\begin{theorem}[Elimination theorem of $qAPTC^{sl}_{\tau}$ with guarded linear recursion and shadow constant]
Let $p$ be a closed $qAPTC^{sl}_{\tau}$ with guarded linear recursion and shadow constant term. Then there is a closed $qAPTC$ with localities term such that $qAPTC^{sl}_{\tau}$ with guarded linear recursion and shadow constant$\vdash p=q$.
\end{theorem}

\begin{proof}
We leave the proof to the readers as an excise.
\end{proof}

\begin{theorem}[Conservitivity of $qAPTC^{sl}_{\tau}$ with guarded linear recursion and shadow constant]
$qAPTC^{sl}_{\tau}$ with guarded linear recursion and shadow constant is a conservative extension of $qAPTC^{sl}_{\tau}$ with guarded linear recursion.
\end{theorem}

\begin{proof}
We leave the proof to the readers as an excise.
\end{proof}

\begin{theorem}[Congruence theorem of $qAPTC^{sl}_{\tau}$ with guarded linear recursion and shadow constant]
Rooted branching static location truly concurrent bisimulation equivalences $\approx_{rbp}^{sl}$, $\approx_{rbs}^{sl}$, $\approx_{rbhp}^{sl}$ and $\approx_{rbhhp}^{sl}$ are all congruences with respect to $qAPTC^{sl}_{\tau}$
with guarded linear recursion and shadow constant.
\end{theorem}

\begin{proof}
We leave the proof to the readers as an excise.
\end{proof}

\begin{theorem}[Soundness of $qAPTC^{sl}_{\tau}$ with guarded linear recursion and shadow constant]
Let $x$ and $y$ be closed $qAPTC^{sl}_{\tau}$ with guarded linear recursion and shadow constant terms. If $qAPTC^{sl}_{\tau}$ with guarded linear recursion and shadow constant$\vdash x=y$, then

\begin{enumerate}
  \item $x\approx_{rbs}^{sl} y$;
  \item $x\approx_{rbp}^{sl} y$;
  \item $x\approx_{rbhp}^{sl} y$;
  \item $x\approx_{rbhhp}^{sl} y$.
\end{enumerate}
\end{theorem}

\begin{proof}
We leave the proof to the readers as an excise.
\end{proof}

\begin{theorem}[Completeness of $qAPTC^{sl}_{\tau}$ with guarded linear recursion and shadow constant]
Let $p$ and $q$ are closed $qAPTC^{sl}_{\tau}$ with guarded linear recursion and shadow constant terms, then,

\begin{enumerate}
  \item if $p\approx_{rbs}^{sl} q$ then $p=q$;
  \item if $p\approx_{rbp}^{sl} q$ then $p=q$;
  \item if $p\approx_{rbhp}^{sl} q$ then $p=q$;
  \item if $p\approx_{rbhhp}^{sl} q$ then $p=q$.
\end{enumerate}
\end{theorem}

\begin{proof}
We leave the proof to the readers as an excise.
\end{proof}

\subsection{Unification of Quantum and Classical Computing for Open Quantum Systems}\label{uni1}

We give the transition rules under quantum configuration for traditional atomic actions (events) $e'\in\mathbb{E}$ as Table \ref{TRForBPA3} shows.

\begin{center}
    \begin{table}
        $$\frac{}{\langle e',\varrho\rangle\xrightarrow[\epsilon]{e'}\langle\surd,\varrho\rangle}\quad \frac{}{\langle loc::e',\varrho\rangle\xrightarrow[loc]{e'}\langle\surd,\varrho'\rangle}$$
        $$\frac{\langle x,\varrho\rangle\xrightarrow[u]{e'}\langle x',\varrho'\rangle}{\langle loc::x,\varrho\rangle\xrightarrow[loc\ll u]{e'}\langle loc::x',\varrho'\rangle}$$
        $$\frac{\langle x,\varrho\rangle\xrightarrow[u]{e'}\langle\surd,\varrho\rangle}{\langle x+y,\varrho\rangle\xrightarrow[u]{e'}\langle\surd,\varrho\rangle}$$
        $$\frac{\langle x,\varrho\rangle\xrightarrow[u]{e'}\langle x',\varrho\rangle}{\langle x+y,\varrho\rangle\xrightarrow[u]{e'}\langle x',\varrho\rangle}$$
        $$\frac{\langle y,\varrho\rangle\xrightarrow[u]{e'}\langle\surd,\varrho\rangle}{\langle x+y,\varrho\rangle\xrightarrow[u]{e'}\langle\surd,\varrho\rangle}$$
        $$\frac{\langle y,\varrho\rangle\xrightarrow[u]{e'}\langle y',\varrho\rangle}{\langle x+y,\varrho\rangle\xrightarrow[u]{e'}\langle y',\varrho\rangle}$$
        $$\frac{\langle x,\varrho\rangle\xrightarrow[u]{e'}\langle\surd,\varrho\rangle}{\langle x\cdot y,\varrho\rangle\xrightarrow[u]{e'}\langle y,\varrho\rangle}$$
        $$\frac{\langle x,\varrho\rangle\xrightarrow[u]{e'}\langle x',\varrho\rangle}{\langle x\cdot y,\varrho\rangle\xrightarrow[u]{e'}\langle x'\cdot y,\varrho\rangle}$$
        \caption{Transition rules of BATC under quantum configuration}
        \label{TRForBPA3}
    \end{table}
\end{center}

And the axioms for traditional actions are the same as those of qBATC with localities. And it is natural can be extended to qAPTC with localities, recursion and abstraction. So, quantum and classical computing
are unified under the framework of qAPTC with localities for open quantum systems.

\newpage\section{Applications of qAPTC with Localities}\label{aqaptcl}

Quantum and classical computing in open systems are unified with qAPTC with localities, which have the same equational logic and the same quantum configuration based operational semantics.
The unification can be used widely in verification for the behaviors of quantum and classical computing mixed systems with distributed characteristics. In this chapter, we show its usage in verification of the
distributed quantum communication protocols.

\subsection{Verification of BB84 Protocol}\label{VBB844}

The BB84 protocol is used to create a private key between two parities, Alice and Bob. Firstly, we introduce the basic BB84 protocol briefly, which is illustrated in Figure \ref{BB844}.

\begin{enumerate}
  \item Alice create two string of bits with size $n$ randomly, denoted as $B_a$ and $K_a$.
  \item Alice generates a string of qubits $q$ with size $n$, and the $i$th qubit in $q$ is $|x_y\rangle$, where $x$ is the $i$th bit of $B_a$ and $y$ is the $i$th bit of $K_a$.
  \item Alice sends $q$ to Bob through a quantum channel $Q$ between Alice and Bob.
  \item Bob receives $q$ and randomly generates a string of bits $B_b$ with size $n$.
  \item Bob measures each qubit of $q$ according to a basis by bits of $B_b$. And the measurement results would be $K_b$, which is also with size $n$.
  \item Bob sends his measurement bases $B_b$ to Alice through a public channel $P$.
  \item Once receiving $B_b$, Alice sends her bases $B_a$ to Bob through channel $P$, and Bob receives $B_a$.
  \item Alice and Bob determine that at which position the bit strings $B_a$ and $B_b$ are equal, and they discard the mismatched bits of $B_a$ and $B_b$. Then the remaining bits of $K_a$ and $K_b$, denoted as $K_a'$ and $K_b'$ with $K_{a,b}=K_a'=K_b'$.
\end{enumerate}

\begin{figure}
  \centering
  %\vspace{5cm}
  \includegraphics{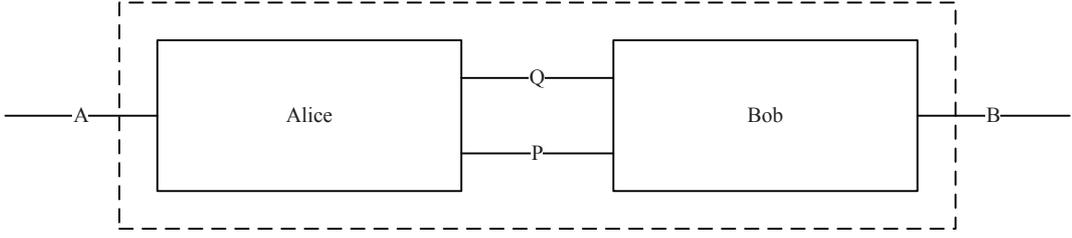}
  \caption{The BB84 protocol.}
  \label{BB844}
\end{figure}

We re-introduce the basic BB84 protocol in an abstract way with more technical details as Figure \ref{BB844} illustrates.

Now, we assume a special measurement operation $Rand[q;B_a]$ which create a string of $n$ random bits $B_a$ from the $q$ quantum system, and the same as $Rand[q;K_a]$, $Rand[q';B_b]$. $M[q;K_b]$ denotes the Bob's measurement operation of $q$. The generation of $n$ qubits $q$ through two quantum operations $Set_{K_a}[q]$ and $H_{B_a}[q]$. Alice sends $q$ to Bob through the quantum channel $Q$ by quantum communicating action $send_{Q}(q)$ and Bob receives $q$ through $Q$ by quantum communicating action $receive_{Q}(q)$. Bob sends $B_b$ to Alice through the public channel $P$ by classical communicating action $send_{P}(B_b)$ and Alice receives $B_b$ through channel $P$ by classical communicating action $receive_{P}(B_b)$, and the same as $send_{P}(B_a)$ and $receive_{P}(B_a)$. Alice and Bob generate the private key $K_{a,b}$ by a classical comparison action $cmp(K_{a,b},K_a,K_b,B_a,B_b)$. Let Alice and Bob be a system $AB$ and let interactions between Alice and Bob be internal actions. $AB$ receives external input $D_i$ through channel $A$ by communicating action $receive_A(D_i)$ and sends results $D_o$ through channel $B$ by communicating action $send_B(D_o)$.

Then the state transition of Alice can be described as follows.

\begin{eqnarray}
&&A=loc_A::(\sum_{D_i\in \Delta_i}receive_A(D_i)\cdot A_1)\nonumber\\
&&A_1=Rand[q;B_a]\cdot A_2\nonumber\\
&&A_2=Rand[q;K_a]\cdot A_3\nonumber\\
&&A_3=Set_{K_a}[q]\cdot A_4\nonumber\\
&&A_4=H_{B_a}[q]\cdot A_5\nonumber\\
&&A_5=send_Q(q)\cdot A_6\nonumber\\
&&A_6=receive_P(B_b)\cdot A_7\nonumber\\
&&A_7=send_P(B_a)\cdot A_8\nonumber\\
&&A_8=cmp(K_{a,b},K_a,K_b,B_a,B_b)\cdot A\nonumber
\end{eqnarray}

where $\Delta_i$ is the collection of the input data.

And the state transition of Bob can be described as follows.

\begin{eqnarray}
&&B=loc_B::(receive_Q(q)\cdot B_1)\nonumber\\
&&B_1=Rand[q';B_b]\cdot B_2\nonumber\\
&&B_2=M[q;K_b]\cdot B_3\nonumber\\
&&B_3=send_P(B_b)\cdot B_4\nonumber\\
&&B_4=receive_P(B_a)\cdot B_5\nonumber\\
&&B_5=cmp(K_{a,b},K_a,K_b,B_a,B_b)\cdot B_6\nonumber\\
&&B_6=\sum_{D_o\in\Delta_o}send_B(D_o)\cdot B\nonumber
\end{eqnarray}

where $\Delta_o$ is the collection of the output data.

The send action and receive action of the same data through the same channel can communicate each other, otherwise, a deadlock $\delta$ will be caused. We define the following communication functions.

\begin{eqnarray}
&&\gamma(send_Q(q),receive_Q(q))\triangleq c_Q(q)\nonumber\\
&&\gamma(send_P(B_b),receive_P(B_b))\triangleq c_P(B_b)\nonumber\\
&&\gamma(send_P(B_a),receive_P(B_a))\triangleq c_P(B_a)\nonumber
\end{eqnarray}

Let $A$ and $B$ in parallel, then the system $AB$ can be represented by the following process term.

$$\tau_I(\partial_H(\Theta(A\between B)))$$

where $H=\{send_Q(q),receive_Q(q),send_P(B_b),receive_P(B_b),send_P(B_a),receive_P(B_a)\}$ and $I=\{Rand[q;B_a], Rand[q;K_a], Set_{K_a}[q], H_{B_a}[q], Rand[q';B_b], M[q;K_b], c_Q(q), c_P(B_b),\\ c_P(B_a), cmp(K_{a,b},K_a,K_b,B_a,B_b)\}$.

Then we get the following conclusion.

\begin{theorem}
The basic BB84 protocol $\tau_I(\partial_H(\Theta(A\between B)))$ can exhibit desired external behaviors.
\end{theorem}

\begin{proof}
We can get $\tau_I(\partial_H(\Theta(A\between B)))=\sum_{D_i\in \Delta_i}\sum_{D_o\in\Delta_o}loc_A::receive_A(D_i)\leftmerge loc_B::send_B(D_o)\leftmerge \tau_I(\partial_H(\Theta(A\between B)))$. So, the basic
BB84 protocol $\tau_I(\partial_H(\Theta(A\between B)))$ can exhibit desired external behaviors.
\end{proof}

\subsection{Verification of E91 Protocol}\label{VE914}

The E91 protocol\cite{E91} is the first quantum protocol which utilizes entanglement and mixes quantum and classical information. In this section, we take an example of verification for the E91 protocol.

The E91 protocol is used to create a private key between two parities, Alice and Bob. Firstly, we introduce the basic E91 protocol briefly, which is illustrated in Figure \ref{E914}.

\begin{enumerate}
  \item Alice generates a string of EPR pairs $q$ with size $n$, i.e., $2n$ particles, and sends a string of qubits $q_b$ from each EPR pair with $n$ to Bob through a quantum channel $Q$, remains the other string of qubits $q_a$ from each pair with size $n$.
  \item Alice create two string of bits with size $n$ randomly, denoted as $B_a$ and $K_a$.
%  \item Alice generates a string of qubits $q$ with size $n$, and the $i$th qubit is $q$ is $|x_y\rangle$, where $x$ is the $i$th bit of $B_a$ and $y$ is the $i$th bit of $K_a$.
%  \item Alice sends $q$ to Bob through a quantum channel $Q$ between Alice and Bob.
  \item Bob receives $q_b$ and randomly generates a string of bits $B_b$ with size $n$.
  \item Alice measures each qubit of $q_a$ according to a basis by bits of $B_a$. And the measurement results would be $K_a$, which is also with size $n$.
  \item Bob measures each qubit of $q_b$ according to a basis by bits of $B_b$. And the measurement results would be $K_b$, which is also with size $n$.
  \item Bob sends his measurement bases $B_b$ to Alice through a public channel $P$.
  \item Once receiving $B_b$, Alice sends her bases $B_a$ to Bob through channel $P$, and Bob receives $B_a$.
  \item Alice and Bob determine that at which position the bit strings $B_a$ and $B_b$ are equal, and they discard the mismatched bits of $B_a$ and $B_b$. Then the remaining bits of $K_a$ and $K_b$, denoted as $K_a'$ and $K_b'$ with $K_{a,b}=K_a'=K_b'$.
\end{enumerate}

\begin{figure}
  \centering
  %\vspace{5cm}
  \includegraphics{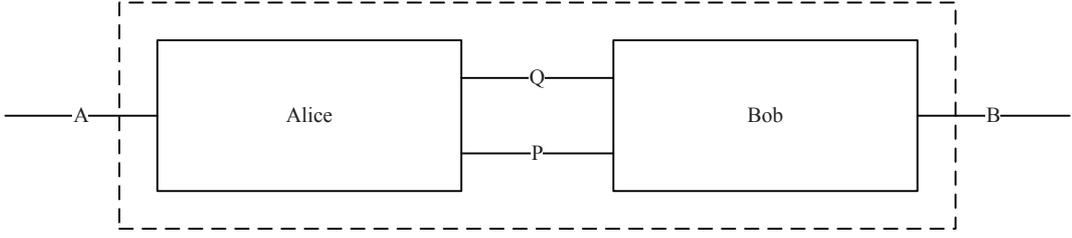}
  \caption{The E91 protocol.}
  \label{E914}
\end{figure}

We re-introduce the basic E91 protocol in an abstract way with more technical details as Figure \ref{E914} illustrates.

Now, $M[q_a;K_a]$ denotes the Alice's measurement operation of $q_a$, and $\circledS_{M[q_a;K_a]}$ denotes the responding shadow constant; $M[q_b;K_b]$ denotes the Bob's measurement operation of $q_b$, and $\circledS_{M[q_b;K_b]}$ denotes the responding shadow constant. Alice sends $q_b$ to Bob through the quantum channel $Q$ by quantum communicating action $send_{Q}(q_b)$ and Bob receives $q_b$ through $Q$ by quantum communicating action $receive_{Q}(q_b)$. Bob sends $B_b$ to Alice through the public channel $P$ by classical communicating action $send_{P}(B_b)$ and Alice receives $B_b$ through channel $P$ by classical communicating action $receive_{P}(B_b)$, and the same as $send_{P}(B_a)$ and $receive_{P}(B_a)$. Alice and Bob generate the private key $K_{a,b}$ by a classical comparison action $cmp(K_{a,b},K_a,K_b,B_a,B_b)$. Let Alice and Bob be a system $AB$ and let interactions between Alice and Bob be internal actions. $AB$ receives external input $D_i$ through channel $A$ by communicating action $receive_A(D_i)$ and sends results $D_o$ through channel $B$ by communicating action $send_B(D_o)$.

Then the state transition of Alice can be described as follows.

\begin{eqnarray}
&&A=loc_A::(\sum_{D_i\in \Delta_i}receive_A(D_i)\cdot A_1)\nonumber\\
%&&A_1=Rand[q;B_a]\cdot A_2\nonumber\\
%&&A_2=Rand[q;K_a]\cdot A_3\nonumber\\
%&&A_3=Set_{K_a}[q]\cdot A_4\nonumber\\
%&&A_4=H_{B_a}[q]\cdot A_5\nonumber\\
&&A_1=send_Q(q_b)\cdot A_2\nonumber\\
&&A_2=M[q_a;K_a]\cdot A_3\nonumber\\
&&A_3=\circledS_{M[q_b;K_b]}\cdot A_4\nonumber\\
&&A_4=receive_P(B_b)\cdot A_5\nonumber\\
&&A_5=send_P(B_a)\cdot A_6\nonumber\\
&&A_6=cmp(K_{a,b},K_a,K_b,B_a,B_b)\cdot A\nonumber
\end{eqnarray}

where $\Delta_i$ is the collection of the input data.

And the state transition of Bob can be described as follows.

\begin{eqnarray}
&&B=loc_B::(receive_Q(q_b)\cdot B_1)\nonumber\\
&&B_1=\circledS_{M[q_a;K_a]}\cdot B_2\nonumber\\
&&B_2=M[q_b;K_b]\cdot B_3\nonumber\\
&&B_3=send_P(B_b)\cdot B_4\nonumber\\
&&B_4=receive_P(B_a)\cdot B_5\nonumber\\
&&B_5=cmp(K_{a,b},K_a,K_b,B_a,B_b)\cdot B_6\nonumber\\
&&B_6=\sum_{D_o\in\Delta_o}send_B(D_o)\cdot B\nonumber
\end{eqnarray}

where $\Delta_o$ is the collection of the output data.

The send action and receive action of the same data through the same channel can communicate each other, otherwise, a deadlock $\delta$ will be caused. The quantum operation and its shadow constant pair will lead entanglement occur, otherwise, a deadlock $\delta$ will occur. We define the following communication functions.

\begin{eqnarray}
&&\gamma(send_Q(q_b),receive_Q(q_b))\triangleq c_Q(q_b)\nonumber\\
&&\gamma(send_P(B_b),receive_P(B_b))\triangleq c_P(B_b)\nonumber\\
&&\gamma(send_P(B_a),receive_P(B_a))\triangleq c_P(B_a)\nonumber
\end{eqnarray}

Let $A$ and $B$ in parallel, then the system $AB$ can be represented by the following process term.

$$\tau_I(\partial_H(\Theta(A\between B)))$$

where $H=\{send_Q(q_b),receive_Q(q_b),send_P(B_b),receive_P(B_b),send_P(B_a),receive_P(B_a),\\ M[q_a;K_a], \circledS_{M[q_a;K_a]}, M[q_b;K_b], \circledS_{M[q_b;K_b]}\}$ and $I=\{c_Q(q_b), c_P(B_b), c_P(B_a), M[q_a;K_a], M[q_b;K_b],\\ cmp(K_{a,b},K_a,K_b,B_a,B_b)\}$.

Then we get the following conclusion.

\begin{theorem}
The basic E91 protocol $\tau_I(\partial_H(A\parallel B))$ can exhibit desired external behaviors.
\end{theorem}

\begin{proof}
We can get $\tau_I(\partial_H(\Theta(A\between B)))=\sum_{D_i\in \Delta_i}\sum_{D_o\in\Delta_o}loc_A::receive_A(D_i)\leftmerge loc_B::send_B(D_o)\leftmerge \tau_I(\partial_H(\Theta(A\between B)))$.
So, the basic E91 protocol $\tau_I(\partial_H(\Theta(A\between B)))$ can exhibit desired external behaviors.
\end{proof}

\subsection{Verification of B92 Protocol}\label{VB924}

The famous B92 protocol\cite{B92} is a quantum key distribution protocol, in which quantum information and classical information are mixed.

The B92 protocol is used to create a private key between two parities, Alice and Bob. B92 is a protocol of quantum key distribution (QKD) which uses polarized photons as information carriers. Firstly, we introduce the basic B92 protocol briefly, which is illustrated in Figure \ref{B924}.

\begin{enumerate}
  \item Alice create a string of bits with size $n$ randomly, denoted as $A$.
  \item Alice generates a string of qubits $q$ with size $n$, carried by polarized photons. If $A_i=0$, the ith qubit is $|0\rangle$; else if $A_i=1$, the ith qubit is $|+\rangle$.
  \item Alice sends $q$ to Bob through a quantum channel $Q$ between Alice and Bob.
  \item Bob receives $q$ and randomly generates a string of bits $B$ with size $n$.
  \item If $B_i=0$, Bob chooses the basis $\oplus$; else if $B_i=1$, Bob chooses the basis $\otimes$. Bob measures each qubit of $q$ according to the above basses. And Bob builds a String of bits $T$, if the measurement produces $|0\rangle$ or $|+\rangle$, then $T_i=0$; else if the measurement produces $|1\rangle$ or $|-\rangle$, then $T_i=1$, which is also with size $n$.
  \item Bob sends $T$ to Alice through a public channel $P$.
  \item Alice and Bob determine that at which position the bit strings $A$ and $B$ are remained for which $T_i=1$. In absence of Eve, $A_i=1-B_i$, a shared raw key $K_{a,b}$ is formed by $A_i$.
\end{enumerate}

\begin{figure}
  \centering
  %\vspace{5cm}
  \includegraphics{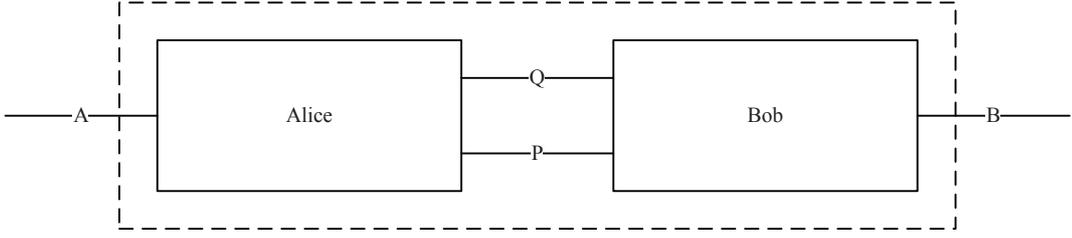}
  \caption{The B92 protocol.}
  \label{B924}
\end{figure}

We re-introduce the basic B92 protocol in an abstract way with more technical details as Figure \ref{B924} illustrates.

Now, we assume a special measurement operation $Rand[q;A]$ which create a string of $n$ random bits $A$ from the $q$ quantum system, and the same as $Rand[q';B]$. $M[q;T]$ denotes the Bob's measurement operation of $q$. The generation of $n$ qubits $q$ through a quantum operation $Set_{A}[q]$. Alice sends $q$ to Bob through the quantum channel $Q$ by quantum communicating action $send_{Q}(q)$ and Bob receives $q$ through $Q$ by quantum communicating action $receive_{Q}(q)$. Bob sends $T$ to Alice through the public channel $P$ by classical communicating action $send_{P}(T)$ and Alice receives $T$ through channel $P$ by classical communicating action $receive_{P}(T)$. Alice and Bob generate the private key $K_{a,b}$ by a classical comparison action $cmp(K_{a,b},T,A,B)$. Let Alice and Bob be a system $AB$ and let interactions between Alice and Bob be internal actions. $AB$ receives external input $D_i$ through channel $A$ by communicating action $receive_A(D_i)$ and sends results $D_o$ through channel $B$ by communicating action $send_B(D_o)$.

Then the state transition of Alice can be described as follows.

\begin{eqnarray}
&&A=loc_A::(\sum_{D_i\in \Delta_i}receive_A(D_i)\cdot A_1)\nonumber\\
&&A_1=Rand[q;A]\cdot A_2\nonumber\\
&&A_2=Set_{A}[q]\cdot A_3\nonumber\\
&&A_3=send_Q(q)\cdot A_4\nonumber\\
&&A_4=receive_P(T)\cdot A_5\nonumber\\
&&A_5=cmp(K_{a,b},T,A,B)\cdot A\nonumber
\end{eqnarray}

where $\Delta_i$ is the collection of the input data.

And the state transition of Bob can be described as follows.

\begin{eqnarray}
&&B=loc_B::(receive_Q(q)\cdot B_1)\nonumber\\
&&B_1=Rand[q';B]\cdot B_2\nonumber\\
&&B_2=M[q;T]\cdot B_3\nonumber\\
&&B_3=send_P(T)\cdot B_4\nonumber\\
&&B_4=cmp(K_{a,b},T,A,B)\cdot B_5\nonumber\\
&&B_5=\sum_{D_o\in\Delta_o}send_B(D_o)\cdot B\nonumber
\end{eqnarray}

where $\Delta_o$ is the collection of the output data.

The send action and receive action of the same data through the same channel can communicate each other, otherwise, a deadlock $\delta$ will be caused. We define the following communication functions.

\begin{eqnarray}
&&\gamma(send_Q(q),receive_Q(q))\triangleq c_Q(q)\nonumber\\
&&\gamma(send_P(T),receive_P(T))\triangleq c_P(T)\nonumber
\end{eqnarray}

Let $A$ and $B$ in parallel, then the system $AB$ can be represented by the following process term.

$$\tau_I(\partial_H(\Theta(A\between B)))$$

where $H=\{send_Q(q),receive_Q(q),send_P(T),receive_P(T)\}$ and $I=\{Rand[q;A], Set_{A}[q], Rand[q';B], \\ M[q;T], c_Q(q), c_P(T), cmp(K_{a,b},T,A,B)\}$.

Then we get the following conclusion.

\begin{theorem}
The basic B92 protocol $\tau_I(\partial_H(A\parallel B))$ can exhibit desired external behaviors.
\end{theorem}

\begin{proof}
We can get $\tau_I(\partial_H(\Theta(A\between B)))=\sum_{D_i\in \Delta_i}\sum_{D_o\in\Delta_o}loc_A::receive_A(D_i)\leftmerge loc_B::send_B(D_o)\leftmerge \tau_I(\partial_H(\Theta(A\between B)))$.
So, the basic B92 protocol $\tau_I(\partial_H(\Theta(A\between B)))$ can exhibit desired external behaviors.
\end{proof}

\subsection{Verification of DPS Protocol}\label{VDPS4}

The famous DPS protocol\cite{DPS} is a quantum key distribution protocol, in which quantum information and classical information are mixed.

The DPS protocol is used to create a private key between two parities, Alice and Bob. DPS is a protocol of quantum key distribution (QKD) which uses pulses of a photon which has nonorthogonal four states. Firstly, we introduce the basic DPS protocol briefly, which is illustrated in Figure \ref{DPS4}.

\begin{enumerate}
  \item Alice generates a string of qubits $q$ with size $n$, carried by a series of single photons possily at four time instances.
  \item Alice sends $q$ to Bob through a quantum channel $Q$ between Alice and Bob.
  \item Bob receives $q$ by detectors clicking at the second or third time instance, and records the time into $T$ with size $n$ and which detector clicks into $D$ with size $n$.
  \item Bob sends $T$ to Alice through a public channel $P$.
  \item Alice receives $T$. From $T$ and her modulation data, Alice knows which detector clicked in Bob's site, i.e. $D$.
  \item Alice and Bob have an identical bit string, provided that the first detector click represents "0" and the other detector represents "1", then a shared raw key $K_{a,b}$ is formed.
\end{enumerate}

\begin{figure}
  \centering
  %\vspace{5cm}
  \includegraphics{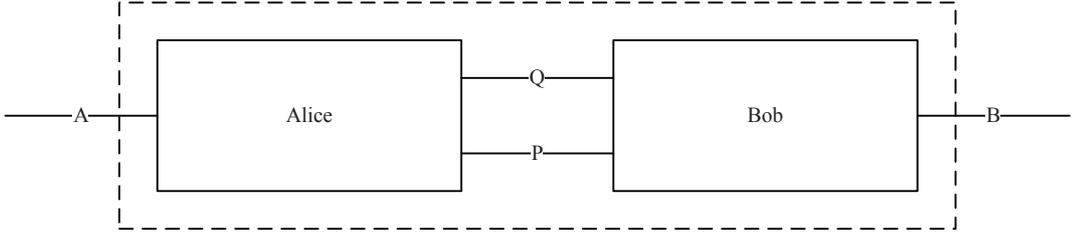}
  \caption{The DPS protocol.}
  \label{DPS4}
\end{figure}

We re-introduce the basic DPS protocol in an abstract way with more technical details as Figure \ref{DPS4} illustrates.

Now, we assume $M[q;T]$ denotes the Bob's measurement operation of $q$. The generation of $n$ qubits $q$ through a quantum operation $Set_{A}[q]$. Alice sends $q$ to Bob through the quantum channel $Q$ by quantum communicating action $send_{Q}(q)$ and Bob receives $q$ through $Q$ by quantum communicating action $receive_{Q}(q)$. Bob sends $T$ to Alice through the public channel $P$ by classical communicating action $send_{P}(T)$ and Alice receives $T$ through channel $P$ by classical communicating action $receive_{P}(T)$. Alice and Bob generate the private key $K_{a,b}$ by a classical comparison action $cmp(K_{a,b},D)$. Let Alice and Bob be a system $AB$ and let interactions between Alice and Bob be internal actions. $AB$ receives external input $D_i$ through channel $A$ by communicating action $receive_A(D_i)$ and sends results $D_o$ through channel $B$ by communicating action $send_B(D_o)$.

Then the state transition of Alice can be described as follows.

\begin{eqnarray}
&&A=loc_A::(\sum_{D_i\in \Delta_i}receive_A(D_i)\cdot A_1)\nonumber\\
&&A_1=Set_{A}[q]\cdot A_2\nonumber\\
&&A_2=send_Q(q)\cdot A_3\nonumber\\
&&A_3=receive_P(T)\cdot A_4\nonumber\\
&&A_4=cmp(K_{a,b},D)\cdot A\nonumber
\end{eqnarray}

where $\Delta_i$ is the collection of the input data.

And the state transition of Bob can be described as follows.

\begin{eqnarray}
&&B=loc_B::(receive_Q(q)\cdot B_1)\nonumber\\
&&B_1=M[q;T]\cdot B_2\nonumber\\
&&B_2=send_P(T)\cdot B_3\nonumber\\
&&B_3=cmp(K_{a,b},D)\cdot B_4\nonumber\\
&&B_4=\sum_{D_o\in\Delta_o}send_B(D_o)\cdot B\nonumber
\end{eqnarray}

where $\Delta_o$ is the collection of the output data.

The send action and receive action of the same data through the same channel can communicate each other, otherwise, a deadlock $\delta$ will be caused. We define the following communication functions.

\begin{eqnarray}
&&\gamma(send_Q(q),receive_Q(q))\triangleq c_Q(q)\nonumber\\
&&\gamma(send_P(T),receive_P(T))\triangleq c_P(T)\nonumber\\
\end{eqnarray}

Let $A$ and $B$ in parallel, then the system $AB$ can be represented by the following process term.

$$\tau_I(\partial_H(\Theta(A\between B)))$$

where $H=\{send_Q(q),receive_Q(q),send_P(T),receive_P(T)\}$

and $I=\{Set_{A}[q], M[q;T], c_Q(q), c_P(T), cmp(K_{a,b},D)\}$.

Then we get the following conclusion.

\begin{theorem}
The basic DPS protocol $\tau_I(\partial_H(\Theta(A\between B)))$ can exhibit desired external behaviors.
\end{theorem}

\begin{proof}
We can get $\tau_I(\partial_H(\Theta(A\between B)))=\sum_{D_i\in \Delta_i}\sum_{D_o\in\Delta_o}loc_A::receive_A(D_i)\leftmerge loc_B::send_B(D_o)\leftmerge \tau_I(\partial_H(\Theta(A\between B)))$.
So, the basic DPS protocol $\tau_I(\partial_H(\Theta(A\between B)))$ can exhibit desired external behaviors.
\end{proof}

\subsection{Verification of BBM92 Protocol}\label{VBBM924}

The famous BBM92 protocol\cite{BBM92} is a quantum key distribution protocol, in which quantum information and classical information are mixed.

The BBM92 protocol is used to create a private key between two parities, Alice and Bob. BBM92 is a protocol of quantum key distribution (QKD) which uses EPR pairs as information carriers. Firstly, we introduce the basic BBM92 protocol briefly, which is illustrated in Figure \ref{BBM924}.

\begin{enumerate}
  \item Alice generates a string of EPR pairs $q$ with size $n$, i.e., $2n$ particles, and sends a string of qubits $q_b$ from each EPR pair with $n$ to Bob through a quantum channel $Q$, remains the other string of qubits $q_a$ from each pair with size $n$.
  \item Alice create a string of bits with size $n$ randomly, denoted as $B_a$.
  \item Bob receives $q_b$ and randomly generates a string of bits $B_b$ with size $n$.
  \item Alice measures each qubit of $q_a$ according to bits of $B_a$, if $B_{a_i}=0$, then uses $x$ axis ($\rightarrow$); else if $B_{a_i}=1$, then uses $z$ axis ($\uparrow$).
  \item Bob measures each qubit of $q_b$ according to bits of $B_b$, if $B_{b_i}=0$, then uses $x$ axis ($\rightarrow$); else if $B_{b_i}=1$, then uses $z$ axis ($\uparrow$).
  \item Bob sends his measurement axis choices $B_b$ to Alice through a public channel $P$.
  \item Once receiving $B_b$, Alice sends her axis choices $B_a$ to Bob through channel $P$, and Bob receives $B_a$.
  \item Alice and Bob agree to discard all instances in which they happened to measure along different axes, as well as instances in which measurements fails because of imperfect quantum efficiency of the detectors. Then the remaining instances can be used to generate a private key $K_{a,b}$.
\end{enumerate}

\begin{figure}
  \centering
  %\vspace{5cm}
  \includegraphics{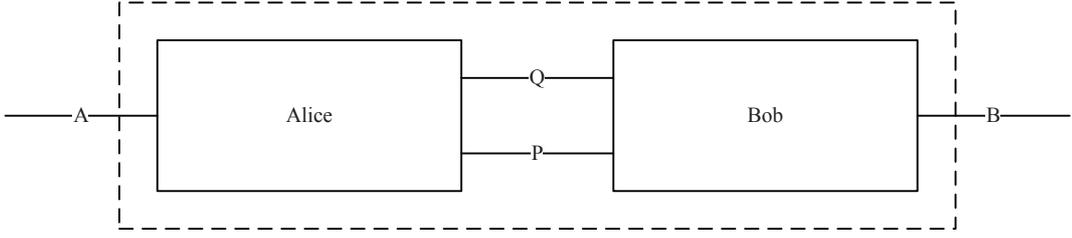}
  \caption{The BBM92 protocol.}
  \label{BBM924}
\end{figure}

We re-introduce the basic BBM92 protocol in an abstract way with more technical details as Figure \ref{BBM924} illustrates.

Now, $M[q_a;B_a]$ denotes the Alice's measurement operation of $q_a$, and $\circledS_{M[q_a;B_a]}$ denotes the responding shadow constant; $M[q_b;B_b]$ denotes the Bob's measurement operation of $q_b$, and $\circledS_{M[q_b;B_b]}$ denotes the responding shadow constant. Alice sends $q_b$ to Bob through the quantum channel $Q$ by quantum communicating action $send_{Q}(q_b)$ and Bob receives $q_b$ through $Q$ by quantum communicating action $receive_{Q}(q_b)$. Bob sends $B_b$ to Alice through the public channel $P$ by classical communicating action $send_{P}(B_b)$ and Alice receives $B_b$ through channel $P$ by classical communicating action $receive_{P}(B_b)$, and the same as $send_{P}(B_a)$ and $receive_{P}(B_a)$. Alice and Bob generate the private key $K_{a,b}$ by a classical comparison action $cmp(K_{a,b},B_a,B_b)$. Let Alice and Bob be a system $AB$ and let interactions between Alice and Bob be internal actions. $AB$ receives external input $D_i$ through channel $A$ by communicating action $receive_A(D_i)$ and sends results $D_o$ through channel $B$ by communicating action $send_B(D_o)$.

Then the state transition of Alice can be described as follows.

\begin{eqnarray}
&&A=loc_A::(\sum_{D_i\in \Delta_i}receive_A(D_i)\cdot A_1)\nonumber\\
&&A_1=send_Q(q_b)\cdot A_2\nonumber\\
&&A_2=M[q_a;B_a]\cdot A_3\nonumber\\
&&A_3=\circledS_{M[q_b;B_b]}\cdot A_4\nonumber\\
&&A_4=receive_P(B_b)\cdot A_5\nonumber\\
&&A_5=send_P(B_a)\cdot A_6\nonumber\\
&&A_6=cmp(K_{a,b},B_a,B_b)\cdot A\nonumber
\end{eqnarray}

where $\Delta_i$ is the collection of the input data.

And the state transition of Bob can be described as follows.

\begin{eqnarray}
&&B=loc_B::(receive_Q(q_b)\cdot B_1)\nonumber\\
&&B_1=\circledS_{M[q_a;B_a]}\cdot B_2\nonumber\\
&&B_2=M[q_b;B_b]\cdot B_3\nonumber\\
&&B_3=send_P(B_b)\cdot B_4\nonumber\\
&&B_4=receive_P(B_a)\cdot B_5\nonumber\\
&&B_5=cmp(K_{a,b},B_a,B_b)\cdot B_6\nonumber\\
&&B_6=\sum_{D_o\in\Delta_o}send_B(D_o)\cdot B\nonumber
\end{eqnarray}

where $\Delta_o$ is the collection of the output data.

The send action and receive action of the same data through the same channel can communicate each other, otherwise, a deadlock $\delta$ will be caused. The quantum operation and its shadow constant pair will lead entanglement occur, otherwise, a deadlock $\delta$ will occur. We define the following communication functions.

\begin{eqnarray}
&&\gamma(send_Q(q_b),receive_Q(q_b))\triangleq c_Q(q_b)\nonumber\\
&&\gamma(send_P(B_b),receive_P(B_b))\triangleq c_P(B_b)\nonumber\\
&&\gamma(send_P(B_a),receive_P(B_a))\triangleq c_P(B_a)\nonumber
\end{eqnarray}

Let $A$ and $B$ in parallel, then the system $AB$ can be represented by the following process term.

$$\tau_I(\partial_H(\Theta(A\between B)))$$

where $H=\{send_Q(q_b),receive_Q(q_b),send_P(B_b),receive_P(B_b),send_P(B_a),receive_P(B_a),\\ M[q_a;B_a], \circledS_{M[q_a;B_a]}, M[q_b;B_b], \circledS_{M[q_b;B_b]}\}$ and $I=\{c_Q(q_b), c_P(B_b), c_P(B_a), M[q_a;B_a], M[q_b;B_b],\\ cmp(K_{a,b},B_a,B_b)\}$.

Then we get the following conclusion.

\begin{theorem}
The basic BBM92 protocol $\tau_I(\partial_H(\Theta(A\between B)))$ can exhibit desired external behaviors.
\end{theorem}

\begin{proof}
We can get $\tau_I(\partial_H(\Theta(A\between B)))=\sum_{D_i\in \Delta_i}\sum_{D_o\in\Delta_o}loc_A::receive_A(D_i)\leftmerge loc_B::loc_send_B(D_o)\leftmerge \tau_I(\partial_H(\Theta(A\between B)))$.
So, the basic BBM92 protocol $\tau_I(\partial_H(\Theta(A\between B)))$ can exhibit desired external behaviors.
\end{proof}

\subsection{Verification of SARG04 Protocol}\label{VSARG044}

The famous SARG04 protocol\cite{SARG04} is a quantum key distribution protocol, in which quantum information and classical information are mixed.

The SARG04 protocol is used to create a private key between two parities, Alice and Bob. SARG04 is a protocol of quantum key distribution (QKD) which refines the BB84 protocol against PNS (Photon Number Splitting) attacks. The main innovations are encoding bits in nonorthogonal states and the classical sifting procedure. Firstly, we introduce the basic SARG04 protocol briefly, which is illustrated in Figure \ref{SARG044}.

\begin{enumerate}
  \item Alice create a string of bits with size $n$ randomly, denoted as $K_a$.
  \item Alice generates a string of qubits $q$ with size $n$, and the $i$th qubit of $q$ has four nonorthogonal states, it is $|\pm x\rangle$ if $K_a=0$; it is $|\pm z\rangle$ if $K_a=1$. And she records the corresponding one of the four pairs of nonorthogonal states into $B_a$ with size $2n$.
  \item Alice sends $q$ to Bob through a quantum channel $Q$ between Alice and Bob.
  \item Alice sends $B_a$ through a public channel $P$.
  \item Bob measures each qubit of $q$ $\sigma_x$ or $\sigma_z$. And he records the unambiguous discriminations into $K_b$ with a raw size $n/4$, and the unambiguous discrimination information into $B_b$ with size $n$.
  \item Bob sends $B_b$ to Alice through the public channel $P$.
  \item Alice and Bob determine that at which position the bit should be remained. Then the remaining bits of $K_a$ and $K_b$ is the private key $K_{a,b}$.
\end{enumerate}

\begin{figure}
  \centering
  %\vspace{5cm}
  \includegraphics{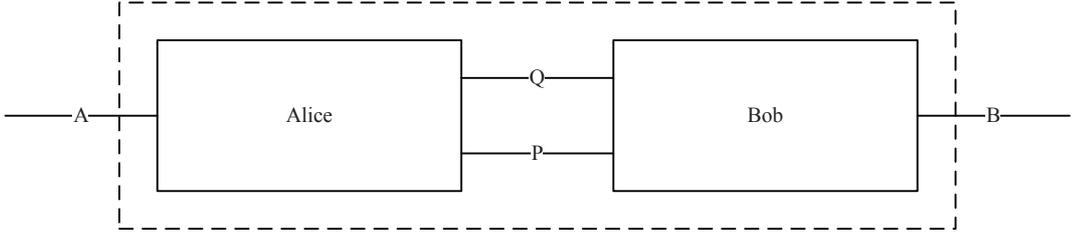}
  \caption{The SARG04 protocol.}
  \label{SARG044}
\end{figure}

We re-introduce the basic SARG04 protocol in an abstract way with more technical details as Figure \ref{SARG044} illustrates.

Now, we assume a special measurement operation $Rand[q;K_a]$ which create a string of $n$ random bits $K_a$ from the $q$ quantum system. $M[q;K_b]$ denotes the Bob's measurement operation of $q$. The generation of $n$ qubits $q$ through a quantum operation $Set_{K_a}[q]$. Alice sends $q$ to Bob through the quantum channel $Q$ by quantum communicating action $send_{Q}(q)$ and Bob receives $q$ through $Q$ by quantum communicating action $receive_{Q}(q)$. Bob sends $B_b$ to Alice through the public channel $P$ by classical communicating action $send_{P}(B_b)$ and Alice receives $B_b$ through channel $P$ by classical communicating action $receive_{P}(B_b)$, and the same as $send_{P}(B_a)$ and $receive_{P}(B_a)$. Alice and Bob generate the private key $K_{a,b}$ by a classical comparison action $cmp(K_{a,b},K_a,K_b,B_a,B_b)$. Let Alice and Bob be a system $AB$ and let interactions between Alice and Bob be internal actions. $AB$ receives external input $D_i$ through channel $A$ by communicating action $receive_A(D_i)$ and sends results $D_o$ through channel $B$ by communicating action $send_B(D_o)$.

Then the state transition of Alice can be described as follows.

\begin{eqnarray}
&&A=loc_A::(\sum_{D_i\in \Delta_i}receive_A(D_i)\cdot A_1)\nonumber\\
&&A_1=Rand[q;K_a]\cdot A_2\nonumber\\
&&A_2=Set_{K_a}[q]\cdot A_3\nonumber\\
&&A_3=send_Q(q)\cdot A_4\nonumber\\
&&A_4=send_P(B_a)\cdot A_5\nonumber\\
&&A_5=receive_P(B_b)\cdot A_6\nonumber\\
&&A_6=cmp(K_{a,b},K_a,K_b,B_a,B_b)\cdot A\nonumber
\end{eqnarray}

where $\Delta_i$ is the collection of the input data.

And the state transition of Bob can be described as follows.

\begin{eqnarray}
&&B=loc_B::(receive_Q(q)\cdot B_1)\nonumber\\
&&B_1=receive_P(B_a)\cdot B_2\nonumber\\
&&B_2=M[q;K_b]\cdot B_3\nonumber\\
&&B_3=send_P(B_b)\cdot B_4\nonumber\\
&&B_4=cmp(K_{a,b},K_a,K_b,B_a,B_b)\cdot B_5\nonumber\\
&&B_5=\sum_{D_o\in\Delta_o}send_B(D_o)\cdot B\nonumber
\end{eqnarray}

where $\Delta_o$ is the collection of the output data.

The send action and receive action of the same data through the same channel can communicate each other, otherwise, a deadlock $\delta$ will be caused. We define the following communication functions.

\begin{eqnarray}
&&\gamma(send_Q(q),receive_Q(q))\triangleq c_Q(q)\nonumber\\
&&\gamma(send_P(B_b),receive_P(B_b))\triangleq c_P(B_b)\nonumber\\
&&\gamma(send_P(B_a),receive_P(B_a))\triangleq c_P(B_a)\nonumber
\end{eqnarray}

Let $A$ and $B$ in parallel, then the system $AB$ can be represented by the following process term.

$$\tau_I(\partial_H(\Theta(A\between B)))$$

where $H=\{send_Q(q),receive_Q(q),send_P(B_b),receive_P(B_b),send_P(B_a),receive_P(B_a)\}$ and $I=\{Rand[q;K_a], Set_{K_a}[q], M[q;K_b], c_Q(q), c_P(B_b),\\ c_P(B_a), cmp(K_{a,b},K_a,K_b,B_a,B_b)\}$.

Then we get the following conclusion.

\begin{theorem}
The basic SARG04 protocol $\tau_I(\partial_H(\Theta(A\between B)))$ can exhibit desired external behaviors.
\end{theorem}

\begin{proof}
We can get $\tau_I(\partial_H(\Theta(A\between B)))=\sum_{D_i\in \Delta_i}\sum_{D_o\in\Delta_o}loc_A::receive_A(D_i)\leftmerge loc_B::send_B(D_o)\leftmerge \tau_I(\partial_H(\Theta(A\between B)))$.
So, the basic SARG04 protocol $\tau_I(\partial_H(\Theta(A\between B)))$ can exhibit desired external behaviors.
\end{proof}

\subsection{Verification of COW Protocol}\label{VCOW4}

The famous COW protocol\cite{COW} is a quantum key distribution protocol, in which quantum information and classical information are mixed.

The COW protocol is used to create a private key between two parities, Alice and Bob. COW is a protocol of quantum key distribution (QKD) which is practical. Firstly, we introduce the basic COW protocol briefly, which is illustrated in Figure \ref{COW4}.

\begin{enumerate}
  \item Alice generates a string of qubits $q$ with size $n$, and the $i$th qubit of $q$ is "0" with probability $\frac{1-f}{2}$, "1" with probability $\frac{1-f}{2}$ and the decoy sequence with probability $f$.
  \item Alice sends $q$ to Bob through a quantum channel $Q$ between Alice and Bob.
  \item Alice sends $A$ of the items corresponding to a decoy sequence through a public channel $P$.
  \item Bob removes all the detections at times $2A-1$ and $2A$ from his raw key and looks whether detector $D_{2M}$ has ever fired at time $2A$.
  \item Bob sends $B$ of the times $2A+1$ in which he had a detector in $D_{2M}$ to Alice through the public channel $P$.
  \item Alice receives $B$ and verifies if some of these items corresponding to a bit sequence "1,0".
  \item Bob sends $C$ of the items that he has detected through the public channel $P$.
  \item Alice and Bob run error correction and privacy amplification on these bits, and the private key $K_{a,b}$ is established.
\end{enumerate}

\begin{figure}
  \centering
  %\vspace{5cm}
  \includegraphics{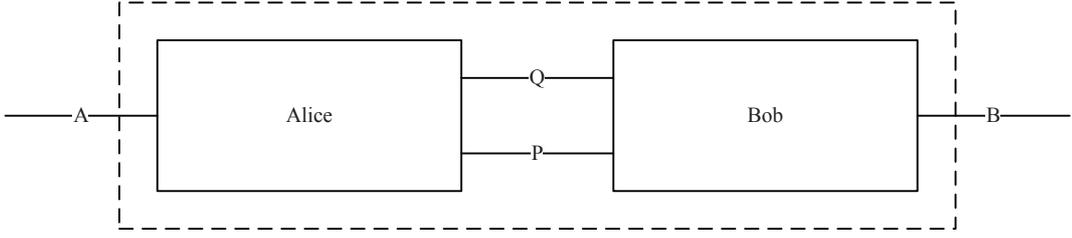}
  \caption{The COW protocol.}
  \label{COW4}
\end{figure}

We re-introduce the basic COW protocol in an abstract way with more technical details as Figure \ref{COW4} illustrates.

Now, we assume The generation of $n$ qubits $q$ through a quantum operation $Set[q]$. $M[q]$ denotes the Bob's measurement operation of $q$.  Alice sends $q$ to Bob through the quantum channel $Q$ by quantum communicating action $send_{Q}(q)$ and Bob receives $q$ through $Q$ by quantum communicating action $receive_{Q}(q)$. Alice sends $A$ to Alice through the public channel $P$ by classical communicating action $send_{P}(A)$ and Alice receives $A$ through channel $P$ by classical communicating action $receive_{P}(A)$, and the same as $send_{P}(B)$ and $receive_{P}(B)$, and $send_{P}(C)$ and $receive_{P}(C)$. Alice and Bob generate the private key $K_{a,b}$ by a classical comparison action $cmp(K_{a,b})$. Let Alice and Bob be a system $AB$ and let interactions between Alice and Bob be internal actions. $AB$ receives external input $D_i$ through channel $A$ by communicating action $receive_A(D_i)$ and sends results $D_o$ through channel $B$ by communicating action $send_B(D_o)$.

Then the state transition of Alice can be described as follows.

\begin{eqnarray}
&&A=loc_A::(\sum_{D_i\in \Delta_i}receive_A(D_i)\cdot A_1)\nonumber\\
&&A_1=Set[q]\cdot A_2\nonumber\\
&&A_2=send_Q(q)\cdot A_3\nonumber\\
&&A_3=send_P(A)\cdot A_4\nonumber\\
&&A_4=receive_P(B)\cdot A_5\nonumber\\
&&A_5=receive_P(C)\cdot A_6\nonumber\\
&&A_6=cmp(K_{a,b})\cdot A\nonumber
\end{eqnarray}

where $\Delta_i$ is the collection of the input data.

And the state transition of Bob can be described as follows.

\begin{eqnarray}
&&B=loc_B::(receive_Q(q)\cdot B_1)\nonumber\\
&&B_1=receive_P(A)\cdot B_2\nonumber\\
&&B_2=M[q]\cdot B_3\nonumber\\
&&B_3=send_P(B)\cdot B_4\nonumber\\
&&B_4=send_P(C)\cdot B_5\nonumber\\
&&B_5=cmp(K_{a,b})\cdot B_6\nonumber\\
&&B_6=\sum_{D_o\in\Delta_o}send_B(D_o)\cdot B\nonumber
\end{eqnarray}

where $\Delta_o$ is the collection of the output data.

The send action and receive action of the same data through the same channel can communicate each other, otherwise, a deadlock $\delta$ will be caused. We define the following communication functions.

\begin{eqnarray}
&&\gamma(send_Q(q),receive_Q(q))\triangleq c_Q(q)\nonumber\\
&&\gamma(send_P(A),receive_P(A))\triangleq c_P(A)\nonumber\\
&&\gamma(send_P(B),receive_P(B))\triangleq c_P(B)\nonumber\\
&&\gamma(send_P(C),receive_P(C))\triangleq c_P(C)\nonumber
\end{eqnarray}

Let $A$ and $B$ in parallel, then the system $AB$ can be represented by the following process term.

$$\tau_I(\partial_H(\Theta(A\between B)))$$

where $H=\{send_Q(q),receive_Q(q),send_P(A),receive_P(A),send_P(B),receive_P(B),\\send_P(C),receive_P(C)\}$ and $I=\{Set[q], M[q], c_Q(q), c_P(A),\\ c_P(B),c_P(C), cmp(K_{a,b})\}$.

Then we get the following conclusion.

\begin{theorem}
The basic COW protocol $\tau_I(\partial_H(\Theta(A\between B)))$ can exhibit desired external behaviors.
\end{theorem}

\begin{proof}
We can get $\tau_I(\partial_H(\Theta(A\between B)))=\sum_{D_i\in \Delta_i}\sum_{D_o\in\Delta_o}loc_A::receive_A(D_i)\leftmerge loc_B::send_B(D_o)\leftmerge \tau_I(\partial_H(\Theta(A\between B)))$.
So, the basic COW protocol $\tau_I(\partial_H(\Theta(A\between B)))$ can exhibit desired external behaviors.
\end{proof}

\subsection{Verification of SSP Protocol}\label{VSSP4}

The famous SSP protocol\cite{SSP} is a quantum key distribution protocol, in which quantum information and classical information are mixed.

The SSP protocol is used to create a private key between two parities, Alice and Bob. SSP is a protocol of quantum key distribution (QKD) which uses six states. Firstly, we introduce the basic SSP protocol briefly, which is illustrated in Figure \ref{SSP4}.

\begin{enumerate}
  \item Alice create two string of bits with size $n$ randomly, denoted as $B_a$ and $K_a$.
  \item Alice generates a string of qubits $q$ with size $n$, and the $i$th qubit in $q$ is one of the six states $\pm x$, $\pm y$ and $\pm z$.
  \item Alice sends $q$ to Bob through a quantum channel $Q$ between Alice and Bob.
  \item Bob receives $q$ and randomly generates a string of bits $B_b$ with size $n$.
  \item Bob measures each qubit of $q$ according to a basis by bits of $B_b$, i.e., $x$, $y$ or $z$ basis. And the measurement results would be $K_b$, which is also with size $n$.
  \item Bob sends his measurement bases $B_b$ to Alice through a public channel $P$.
  \item Once receiving $B_b$, Alice sends her bases $B_a$ to Bob through channel $P$, and Bob receives $B_a$.
  \item Alice and Bob determine that at which position the bit strings $B_a$ and $B_b$ are equal, and they discard the mismatched bits of $B_a$ and $B_b$. Then the remaining bits of $K_a$ and $K_b$, denoted as $K_a'$ and $K_b'$ with $K_{a,b}=K_a'=K_b'$.
\end{enumerate}

\begin{figure}
  \centering
  %\vspace{5cm}
  \includegraphics{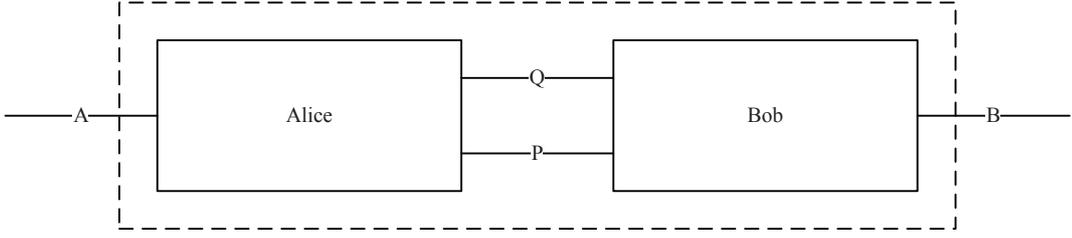}
  \caption{The SSP protocol.}
  \label{SSP4}
\end{figure}

We re-introduce the basic SSP protocol in an abstract way with more technical details as Figure \ref{SSP4} illustrates.

Now, we assume a special measurement operation $Rand[q;B_a]$ which create a string of $n$ random bits $B_a$ from the $q$ quantum system, and the same as $Rand[q;K_a]$, $Rand[q';B_b]$. $M[q;K_b]$ denotes the Bob's measurement operation of $q$. The generation of $n$ qubits $q$ through two quantum operations $Set_{K_a}[q]$ and $H_{B_a}[q]$. Alice sends $q$ to Bob through the quantum channel $Q$ by quantum communicating action $send_{Q}(q)$ and Bob receives $q$ through $Q$ by quantum communicating action $receive_{Q}(q)$. Bob sends $B_b$ to Alice through the public channel $P$ by classical communicating action $send_{P}(B_b)$ and Alice receives $B_b$ through channel $P$ by classical communicating action $receive_{P}(B_b)$, and the same as $send_{P}(B_a)$ and $receive_{P}(B_a)$. Alice and Bob generate the private key $K_{a,b}$ by a classical comparison action $cmp(K_{a,b},K_a,K_b,B_a,B_b)$. Let Alice and Bob be a system $AB$ and let interactions between Alice and Bob be internal actions. $AB$ receives external input $D_i$ through channel $A$ by communicating action $receive_A(D_i)$ and sends results $D_o$ through channel $B$ by communicating action $send_B(D_o)$.

Then the state transition of Alice can be described as follows.

\begin{eqnarray}
&&A=loc_A::(\sum_{D_i\in \Delta_i}receive_A(D_i)\cdot A_1)\nonumber\\
&&A_1=Rand[q;B_a]\cdot A_2\nonumber\\
&&A_2=Rand[q;K_a]\cdot A_3\nonumber\\
&&A_3=Set_{K_a}[q]\cdot A_4\nonumber\\
&&A_4=H_{B_a}[q]\cdot A_5\nonumber\\
&&A_5=send_Q(q)\cdot A_6\nonumber\\
&&A_6=receive_P(B_b)\cdot A_7\nonumber\\
&&A_7=send_P(B_a)\cdot A_8\nonumber\\
&&A_8=cmp(K_{a,b},K_a,K_b,B_a,B_b)\cdot A\nonumber
\end{eqnarray}

where $\Delta_i$ is the collection of the input data.

And the state transition of Bob can be described as follows.

\begin{eqnarray}
&&B=loc_B::(receive_Q(q)\cdot B_1)\nonumber\\
&&B_1=Rand[q';B_b]\cdot B_2\nonumber\\
&&B_2=M[q;K_b]\cdot B_3\nonumber\\
&&B_3=send_P(B_b)\cdot B_4\nonumber\\
&&B_4=receive_P(B_a)\cdot B_5\nonumber\\
&&B_5=cmp(K_{a,b},K_a,K_b,B_a,B_b)\cdot B_6\nonumber\\
&&B_6=\sum_{D_o\in\Delta_o}send_B(D_o)\cdot B\nonumber
\end{eqnarray}

where $\Delta_o$ is the collection of the output data.

The send action and receive action of the same data through the same channel can communicate each other, otherwise, a deadlock $\delta$ will be caused. We define the following communication functions.

\begin{eqnarray}
&&\gamma(send_Q(q),receive_Q(q))\triangleq c_Q(q)\nonumber\\
&&\gamma(send_P(B_b),receive_P(B_b))\triangleq c_P(B_b)\nonumber\\
&&\gamma(send_P(B_a),receive_P(B_a))\triangleq c_P(B_a)\nonumber
\end{eqnarray}

Let $A$ and $B$ in parallel, then the system $AB$ can be represented by the following process term.

$$\tau_I(\partial_H(\Theta(A\between B)))$$

where $H=\{send_Q(q),receive_Q(q),send_P(B_b),receive_P(B_b),send_P(B_a),receive_P(B_a)\}$ and $I=\{Rand[q;B_a], Rand[q;K_a], Set_{K_a}[q], H_{B_a}[q], Rand[q';B_b], M[q;K_b], c_Q(q), c_P(B_b),\\ c_P(B_a), cmp(K_{a,b},K_a,K_b,B_a,B_b)\}$.

Then we get the following conclusion.

\begin{theorem}
The basic SSP protocol $\tau_I(\partial_H(\Theta(A\between B)))$ can exhibit desired external behaviors.
\end{theorem}

\begin{proof}
We can get $\tau_I(\partial_H(\Theta(A\between B)))=\sum_{D_i\in \Delta_i}\sum_{D_o\in\Delta_o}loc_A::receive_A(D_i)\leftmerge loc_B::send_B(D_o)\leftmerge \tau_I(\partial_H(\Theta(A\between B)))$.
So, the basic SSP protocol $\tau_I(\partial_H(\Theta(A\between B)))$ can exhibit desired external behaviors.
\end{proof}

\subsection{Verification of S09 Protocol}\label{VS094}

The famous S09 protocol\cite{S09} is a quantum key distribution protocol, in which quantum information and classical information are mixed.

The S09 protocol is used to create a private key between two parities, Alice and Bob, by use of pure quantum information. Firstly, we introduce the basic S09 protocol briefly, which is illustrated in Figure \ref{S094}.

\begin{enumerate}
  \item Alice create two string of bits with size $n$ randomly, denoted as $B_a$ and $K_a$.
  \item Alice generates a string of qubits $q$ with size $n$, and the $i$th qubit in $q$ is $|x_y\rangle$, where $x$ is the $i$th bit of $B_a$ and $y$ is the $i$th bit of $K_a$.
  \item Alice sends $q$ to Bob through a quantum channel $Q$ between Alice and Bob.
  \item Bob receives $q$ and randomly generates a string of bits $B_b$ with size $n$.
  \item Bob measures each qubit of $q$ according to a basis by bits of $B_b$. After the measurement, the state of $q$ evolves into $q'$.
  \item Bob sends $q'$ to Alice through the quantum channel $Q$.
  \item Alice measures each qubit of $q'$ to generate a string $C$.
  \item Alice sums $C_i\oplus B_{a_i}$ to get the private key $K_{a,b}=B_b$.
\end{enumerate}

\begin{figure}
  \centering
  %\vspace{5cm}
  \includegraphics{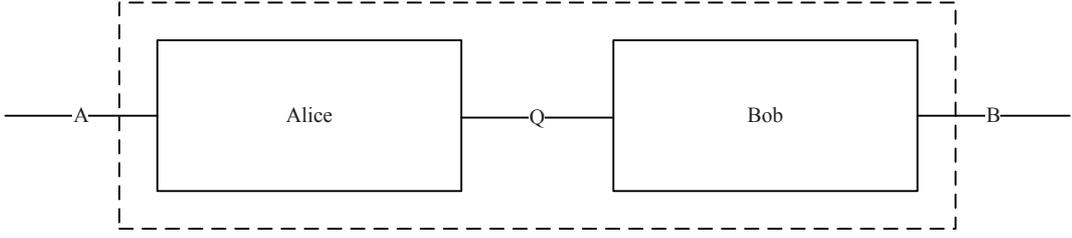}
  \caption{The S09 protocol.}
  \label{S094}
\end{figure}

We re-introduce the basic S09 protocol in an abstract way with more technical details as Figure \ref{S094} illustrates.

Now, we assume a special measurement operation $Rand[q;B_a]$ which create a string of $n$ random bits $B_a$ from the $q$ quantum system, and the same as $Rand[q;K_a]$, $Rand[q';B_b]$. $M[q;B_b]$ denotes the Bob's measurement operation of $q$, and the same as $M[q';C]$. The generation of $n$ qubits $q$ through two quantum operations $Set_{K_a}[q]$ and $H_{B_a}[q]$. Alice sends $q$ to Bob through the quantum channel $Q$ by quantum communicating action $send_{Q}(q)$ and Bob receives $q$ through $Q$ by quantum communicating action $receive_{Q}(q)$, and the same as $send_{Q}(q')$ and $receive_{Q}(q')$. Alice and Bob generate the private key $K_{a,b}$ by a classical comparison action $cmp(K_{a,b},B_b)$. We omit the sum classical $\oplus$ actions without of loss of generality. Let Alice and Bob be a system $AB$ and let interactions between Alice and Bob be internal actions. $AB$ receives external input $D_i$ through channel $A$ by communicating action $receive_A(D_i)$ and sends results $D_o$ through channel $B$ by communicating action $send_B(D_o)$.

Then the state transition of Alice can be described as follows.

\begin{eqnarray}
&&A=loc_A::(\sum_{D_i\in \Delta_i}receive_A(D_i)\cdot A_1)\nonumber\\
&&A_1=Rand[q;B_a]\cdot A_2\nonumber\\
&&A_2=Rand[q;K_a]\cdot A_3\nonumber\\
&&A_3=Set_{K_a}[q]\cdot A_4\nonumber\\
&&A_4=H_{B_a}[q]\cdot A_5\nonumber\\
&&A_5=send_Q(q)\cdot A_6\nonumber\\
&&A_6=receive_Q(q')\cdot A_{7}\nonumber\\
&&A_7=M[q';C]\cdot A_8\nonumber\\
&&A_{8}=cmp(K_{a,b},B_b)\cdot A\nonumber
\end{eqnarray}

where $\Delta_i$ is the collection of the input data.

And the state transition of Bob can be described as follows.

\begin{eqnarray}
&&B=loc_B::(receive_Q(q)\cdot B_1)\nonumber\\
&&B_1=Rand[q';B_b]\cdot B_2\nonumber\\
&&B_2=M[q;B_b]\cdot B_3\nonumber\\
&&B_3=send_Q(q')\cdot B_4\nonumber\\
&&B_4=cmp(K_{a,b},B_b)\cdot B_{5}\nonumber\\
&&B_{5}=\sum_{D_o\in\Delta_o}send_B(D_o)\cdot B\nonumber
\end{eqnarray}

where $\Delta_o$ is the collection of the output data.

The send action and receive action of the same data through the same channel can communicate each other, otherwise, a deadlock $\delta$ will be caused. We define the following communication functions.

\begin{eqnarray}
&&\gamma(send_Q(q),receive_Q(q))\triangleq c_Q(q)\nonumber\\
&&\gamma(send_Q(q'),receive_Q(q'))\triangleq c_Q(q')\nonumber
\end{eqnarray}

Let $A$ and $B$ in parallel, then the system $AB$ can be represented by the following process term.

$$\tau_I(\partial_H(\Theta(A\between B)))$$

where $H=\{send_Q(q),receive_Q(q),send_Q(q'),receive_Q(q')\}$ and $I=\{Rand[q;B_a], Rand[q;K_a], Set_{K_a}[q], \\ H_{B_a}[q], Rand[q';B_b], M[q;K_b], M[q';C], c_Q(q), c_Q(q'), cmp(K_{a,b},B_b)\}$.

Then we get the following conclusion.

\begin{theorem}
The basic S09 protocol $\tau_I(\partial_H(\Theta(A\between B)))$ can exhibit desired external behaviors.
\end{theorem}

\begin{proof}
We can get $\tau_I(\partial_H(\Theta(A\between B)))=\sum_{D_i\in \Delta_i}\sum_{D_o\in\Delta_o}loc_A::receive_A(D_i)\leftmerge loc_B::send_B(D_o)\leftmerge \tau_I(\partial_H(\Theta(A\between B)))$.
So, the basic S09 protocol $\tau_I(\partial_H(\Theta(A\between B)))$ can exhibit desired external behaviors.
\end{proof}

\subsection{Verification of KMB09 Protocol}\label{VKMB094}

The famous KMB09 protocol\cite{KMB09} is a quantum key distribution protocol, in which quantum information and classical information are mixed.

The KMB09 protocol is used to create a private key between two parities, Alice and Bob. KMB09 is a protocol of quantum key distribution (QKD) which refines the BB84 protocol against PNS (Photon Number Splitting) attacks. The main innovations are encoding bits in nonorthogonal states and the classical sifting procedure. Firstly, we introduce the basic KMB09 protocol briefly, which is illustrated in Figure \ref{KMB094}.

\begin{enumerate}
  \item Alice create a string of bits with size $n$ randomly, denoted as $K_a$, and randomly assigns each bit value a random index $i=1,2,...,N$ into $B_a$.
  \item Alice generates a string of qubits $q$ with size $n$, accordingly either in $|e_i\rangle$ or $|f_i\rangle$.
  \item Alice sends $q$ to Bob through a quantum channel $Q$ between Alice and Bob.
  \item Alice sends $B_a$ through a public channel $P$.
  \item Bob measures each qubit of $q$ by randomly switching the measurement basis between $e$ and $f$. And he records the unambiguous discriminations into $K_b$, and the unambiguous discrimination information into $B_b$.
  \item Bob sends $B_b$ to Alice through the public channel $P$.
  \item Alice and Bob determine that at which position the bit should be remained. Then the remaining bits of $K_a$ and $K_b$ is the private key $K_{a,b}$.
\end{enumerate}

\begin{figure}
  \centering
  %\vspace{5cm}
  \includegraphics{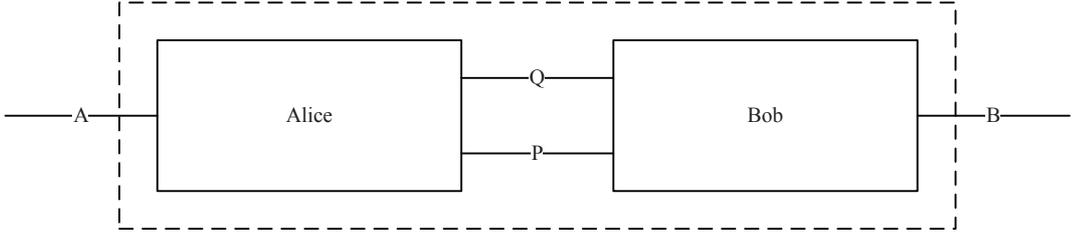}
  \caption{The KMB09 protocol.}
  \label{KMB094}
\end{figure}

We re-introduce the basic KMB09 protocol in an abstract way with more technical details as Figure \ref{KMB094} illustrates.

Now, we assume a special measurement operation $Rand[q;K_a]$ which create a string of $n$ random bits $K_a$ from the $q$ quantum system. $M[q;K_b]$ denotes the Bob's measurement operation of $q$. The generation of $n$ qubits $q$ through a quantum operation $Set_{K_a}[q]$. Alice sends $q$ to Bob through the quantum channel $Q$ by quantum communicating action $send_{Q}(q)$ and Bob receives $q$ through $Q$ by quantum communicating action $receive_{Q}(q)$. Bob sends $B_b$ to Alice through the public channel $P$ by classical communicating action $send_{P}(B_b)$ and Alice receives $B_b$ through channel $P$ by classical communicating action $receive_{P}(B_b)$, and the same as $send_{P}(B_a)$ and $receive_{P}(B_a)$. Alice and Bob generate the private key $K_{a,b}$ by a classical comparison action $cmp(K_{a,b},K_a,K_b,B_a,B_b)$. Let Alice and Bob be a system $AB$ and let interactions between Alice and Bob be internal actions. $AB$ receives external input $D_i$ through channel $A$ by communicating action $receive_A(D_i)$ and sends results $D_o$ through channel $B$ by communicating action $send_B(D_o)$.

Then the state transition of Alice can be described as follows.

\begin{eqnarray}
&&A=loc_A::(\sum_{D_i\in \Delta_i}receive_A(D_i)\cdot A_1)\nonumber\\
&&A_1=Rand[q;K_a]\cdot A_2\nonumber\\
&&A_2=Set_{K_a}[q]\cdot A_3\nonumber\\
&&A_3=send_Q(q)\cdot A_4\nonumber\\
&&A_4=send_P(B_a)\cdot A_5\nonumber\\
&&A_5=receive_P(B_b)\cdot A_6\nonumber\\
&&A_6=cmp(K_{a,b},K_a,K_b,B_a,B_b)\cdot A\nonumber
\end{eqnarray}

where $\Delta_i$ is the collection of the input data.

And the state transition of Bob can be described as follows.

\begin{eqnarray}
&&B=loc_B::(receive_Q(q)\cdot B_1)\nonumber\\
&&B_1=receive_P(B_a)\cdot B_2\nonumber\\
&&B_2=M[q;K_b]\cdot B_3\nonumber\\
&&B_3=send_P(B_b)\cdot B_4\nonumber\\
&&B_4=cmp(K_{a,b},K_a,K_b,B_a,B_b)\cdot B_5\nonumber\\
&&B_5=\sum_{D_o\in\Delta_o}send_B(D_o)\cdot B\nonumber
\end{eqnarray}

where $\Delta_o$ is the collection of the output data.

The send action and receive action of the same data through the same channel can communicate each other, otherwise, a deadlock $\delta$ will be caused. We define the following communication functions.

\begin{eqnarray}
&&\gamma(send_Q(q),receive_Q(q))\triangleq c_Q(q)\nonumber\\
&&\gamma(send_P(B_b),receive_P(B_b))\triangleq c_P(B_b)\nonumber\\
&&\gamma(send_P(B_a),receive_P(B_a))\triangleq c_P(B_a)\nonumber
\end{eqnarray}

Let $A$ and $B$ in parallel, then the system $AB$ can be represented by the following process term.

$$\tau_I(\partial_H(\Theta(A\between B)))$$

where $H=\{send_Q(q),receive_Q(q),send_P(B_b),receive_P(B_b),send_P(B_a),receive_P(B_a)\}$ and $I=\{Rand[q;K_a], Set_{K_a}[q], M[q;K_b], c_Q(q), c_P(B_b),\\ c_P(B_a), cmp(K_{a,b},K_a,K_b,B_a,B_b)\}$.

Then we get the following conclusion.

\begin{theorem}
The basic KMB09 protocol $\tau_I(\partial_H(\Theta(A\between B)))$ can exhibit desired external behaviors.
\end{theorem}

\begin{proof}
We can get $\tau_I(\partial_H(\Theta(A\between B)))=\sum_{D_i\in \Delta_i}\sum_{D_o\in\Delta_o}loc_A::receive_A(D_i)\leftmerge loc_B::send_B(D_o)\leftmerge \tau_I(\partial_H(\Theta(A\between B)))$.
So, the basic KMB09 protocol $\tau_I(\partial_H(\Theta(A\between B)))$ can exhibit desired external behaviors.
\end{proof}

\subsection{Verification of S13 Protocol}\label{VS134}

The famous S13 protocol\cite{S13} is a quantum key distribution protocol, in which quantum information and classical information are mixed.

The S13 protocol is used to create a private key between two parities, Alice and Bob. Firstly, we introduce the basic S13 protocol briefly, which is illustrated in Figure \ref{S134}.

\begin{enumerate}
  \item Alice create two string of bits with size $n$ randomly, denoted as $B_a$ and $K_a$.
  \item Alice generates a string of qubits $q$ with size $n$, and the $i$th qubit in $q$ is $|x_y\rangle$, where $x$ is the $i$th bit of $B_a$ and $y$ is the $i$th bit of $K_a$.
  \item Alice sends $q$ to Bob through a quantum channel $Q$ between Alice and Bob.
  \item Bob receives $q$ and randomly generates a string of bits $B_b$ with size $n$.
  \item Bob measures each qubit of $q$ according to a basis by bits of $B_b$. And the measurement results would be $K_b$, which is also with size $n$.
  \item Alice sends a random binary string $C$ to Bob through the public channel $P$.
  \item Alice sums $B_{a_i}\oplus C_i$ to obtain $T$ and generates other random string of binary values $J$. From the elements occupying a concrete position, $i$, of the preceding strings, Alice get the new states of $q'$, and sends it to Bob through the quantum channel $Q$.
  \item Bob sums $1\oplus B_{b_i}$ to obtain the string of binary basis $N$ and measures $q'$ according to these bases, and generating $D$.
  \item Alice sums $K_{a_i}\oplus J_i$ to obtain the binary string $Y$ and sends it to Bob through the public channel $P$.
  \item Bob encrypts $B_b$ to obtain $U$ and sends to Alice through the public channel $P$.
  \item Alice decrypts $U$ to obtain $B_b$. She sums $B_{a_i}\oplus B_{b_i}$ to obtain $L$ and sends $L$ to Bob through the public channel $P$.
  \item Bob sums $B_{b_i}\oplus L_i$ to get the private key $K_{a,b}$.
\end{enumerate}

\begin{figure}
  \centering
  %\vspace{5cm}
  \includegraphics{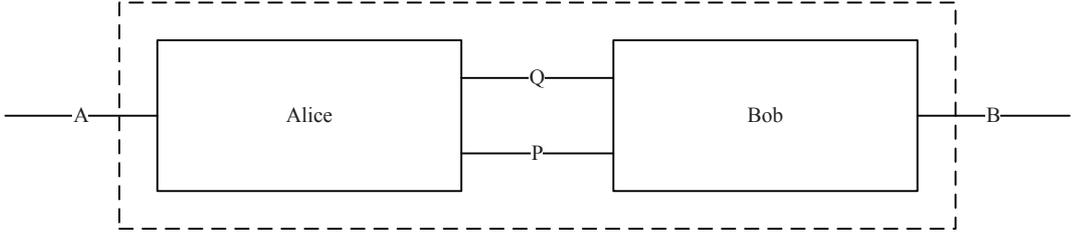}
  \caption{The S13 protocol.}
  \label{S134}
\end{figure}

We re-introduce the basic S13 protocol in an abstract way with more technical details as Figure \ref{S134} illustrates.

Now, we assume a special measurement operation $Rand[q;B_a]$ which create a string of $n$ random bits $B_a$ from the $q$ quantum system, and the same as $Rand[q;K_a]$, $Rand[q';B_b]$. $M[q;K_b]$ denotes the Bob's measurement operation of $q$, and the same as $M[q';D]$. The generation of $n$ qubits $q$ through two quantum operations $Set_{K_a}[q]$ and $H_{B_a}[q]$, and the same as $Set_{T}[q']$. Alice sends $q$ to Bob through the quantum channel $Q$ by quantum communicating action $send_{Q}(q)$ and Bob receives $q$ through $Q$ by quantum communicating action $receive_{Q}(q)$, and the same as $send_{Q}(q')$ and $receive_{Q}(q')$. Bob sends $B_b$ to Alice through the public channel $P$ by classical communicating action $send_{P}(B_b)$ and Alice receives $B_b$ through channel $P$ by classical communicating action $receive_{P}(B_b)$, and the same as $send_{P}(B_a)$ and $receive_{P}(B_a)$, $send_{P}(C)$ and $receive_{P}(C)$, $send_{P}(Y)$ and $receive_{P}(Y)$, $send_{P}(U)$ and $receive_{P}(U)$, $send_{P}(L)$ and $receive_{P}(L)$. Alice and Bob generate the private key $K_{a,b}$ by a classical comparison action $cmp(K_{a,b},K_a,K_b,B_a,B_b)$. We omit the sum classical $\oplus$ actions without of loss of generality. Let Alice and Bob be a system $AB$ and let interactions between Alice and Bob be internal actions. $AB$ receives external input $D_i$ through channel $A$ by communicating action $receive_A(D_i)$ and sends results $D_o$ through channel $B$ by communicating action $send_B(D_o)$.

Then the state transition of Alice can be described as follows.

\begin{eqnarray}
&&A=loc_A::(\sum_{D_i\in \Delta_i}receive_A(D_i)\cdot A_1)\nonumber\\
&&A_1=Rand[q;B_a]\cdot A_2\nonumber\\
&&A_2=Rand[q;K_a]\cdot A_3\nonumber\\
&&A_3=Set_{K_a}[q]\cdot A_4\nonumber\\
&&A_4=H_{B_a}[q]\cdot A_5\nonumber\\
&&A_5=send_Q(q)\cdot A_6\nonumber\\
&&A_6=send_P(C)\cdot A_7\nonumber\\
&&A_7=send_Q(q')\cdot A_8\nonumber\\
&&A_8=send_P(Y)\cdot A_9\nonumber\\
&&A_9=receive_P(U)\cdot A_{10}\nonumber\\
&&A_{10}=send_P(L)\cdot A_{11}\nonumber\\
&&A_{11}=cmp(K_{a,b},K_a,K_b,B_a,B_b)\cdot A\nonumber
\end{eqnarray}

where $\Delta_i$ is the collection of the input data.

And the state transition of Bob can be described as follows.

\begin{eqnarray}
&&B=loc_B::(receive_Q(q)\cdot B_1)\nonumber\\
&&B_1=Rand[q';B_b]\cdot B_2\nonumber\\
&&B_2=M[q;K_b]\cdot B_3\nonumber\\
&&B_3=receive_P(C)\cdot B_4\nonumber\\
&&B_4=receive_Q(q')\cdot B_5\nonumber\\
&&B_5=M[q';D]\cdot B_6\nonumber\\
&&B_6=receive_P(Y)\cdot B_7\nonumber\\
&&B_7=send_P(U)\cdot B_8\nonumber\\
&&B_8=receive_P(L)\cdot B_9\nonumber\\
&&B_9=cmp(K_{a,b},K_a,K_b,B_a,B_b)\cdot B_{10}\nonumber\\
&&B_{10}=\sum_{D_o\in\Delta_o}send_B(D_o)\cdot B\nonumber
\end{eqnarray}

where $\Delta_o$ is the collection of the output data.

The send action and receive action of the same data through the same channel can communicate each other, otherwise, a deadlock $\delta$ will be caused. We define the following communication functions.

\begin{eqnarray}
&&\gamma(send_Q(q),receive_Q(q))\triangleq c_Q(q)\nonumber\\
&&\gamma(send_Q(q'),receive_Q(q'))\triangleq c_Q(q')\nonumber\\
&&\gamma(send_P(C),receive_P(C))\triangleq c_P(C)\nonumber\\
&&\gamma(send_P(Y),receive_P(Y))\triangleq c_P(Y)\nonumber\\
&&\gamma(send_P(U),receive_P(U))\triangleq c_P(U)\nonumber\\
&&\gamma(send_P(L),receive_P(L))\triangleq c_P(L)\nonumber
\end{eqnarray}

Let $A$ and $B$ in parallel, then the system $AB$ can be represented by the following process term.

$$\tau_I(\partial_H(\Theta(A\between B)))$$

where $H=\{send_Q(q),receive_Q(q),send_Q(q'),receive_Q(q'),send_P(C),\\receive_P(C),send_P(Y),receive_P(Y),send_P(U),receive_P(U),send_P(L),receive_P(L)\}$ and $I=\{Rand[q;B_a], Rand[q;K_a],\\ Set_{K_a}[q], H_{B_a}[q], Rand[q';B_b], M[q;K_b], M[q';D], c_Q(q), c_P(C),\\c_Q(q'), c_P(Y), c_P(U), c_P(L), cmp(K_{a,b},K_a,K_b,B_a,B_b)\}$.

Then we get the following conclusion.

\begin{theorem}
The basic S13 protocol $\tau_I(\partial_H(\Theta(A\between B)))$ can exhibit desired external behaviors.
\end{theorem}

\begin{proof}
We can get $\tau_I(\partial_H(\Theta(A\between B)))=\sum_{D_i\in \Delta_i}\sum_{D_o\in\Delta_o}loc_A::receive_A(D_i)\leftmerge loc_B::send_B(D_o)\leftmerge \tau_I(\partial_H(\Theta(A\between B)))$.
So, the basic S13 protocol $\tau_I(\partial_H(\Theta(A\between B)))$ can exhibit desired external behaviors.
\end{proof}

\newpage\section{APPTC with Localities for Closed Quantum Systems}\label{qapptcl2}

The theory $APPTC$ with localities for closed quantum systems abbreviated $qAPPTC^{sl}$ has four modules: $qBAPTC^{sl}$ , $qAPPTC^{sl}$, recursion and abstraction.

This chapter is organized as follows. We introduce $qBAPTC^{sl}$ in section \ref{qbaptcl}, $APPTC$ in section \ref{qapptcl}, recursion in section \ref{qcrec}, and abstraction in section
\ref{qcabs}. And we introduce quantum measurement in section \ref{qm}, quantum entanglement in section \ref{qe2}, and unification of quantum and classical computing in section \ref{uni2}.

Note that, for a closed quantum system, the unitary operators are the atomic actions (events) and let unitary operators into $\mathbb{E}$. And for the existence of quantum measurement,
the probabilism is unavoidable.

\subsection{$BAPTC$ with Localities for Closed Quantum Systems}{\label{qbaptcl}}

In this subsection, we will discuss $qBAPTC$ with localities. Let $\mathbb{E}$ be the set of atomic events (actions, unitary operators).

Let $Loc$ be the set of locations, and $loc\in Loc$, $u,v\in Loc^*$, $\epsilon$ is the empty location.

In the following, let $e_1, e_2, e_1', e_2'\in \mathbb{E}$, and let variables $x,y,z$ range over the set of terms for true concurrency, $p,q$ range over the set of
closed terms.

The set of axioms of $qBAPTC$ with localities consists of the laws given in Table \ref{AxiomsForqBAPTC}.

\begin{center}
    \begin{table}
        \begin{tabular}{@{}ll@{}}
            \hline No. &Axiom\\
            $A1$ & $x+ y = y+ x$\\
            $A2$ & $(x+ y)+ z = x+ (y+ z)$\\
            $A3$ & $e+ e = e$\\
            $A4$ & $(x+ y)\cdot z = x\cdot z + y\cdot z$\\
            $A5$ & $(x\cdot y)\cdot z = x\cdot(y\cdot z)$\\
            $A6$ & $x+\delta = x$\\
            $A7$ & $\delta\cdot x = \delta$\\
            $PA1$ & $x\boxplus_{\pi} y=y\boxplus_{1-\pi} x$\\
            $PA2$ & $x\boxplus_{\pi}(y\boxplus_{\rho} z)=(x\boxplus_{\frac{\pi}{\pi+\rho-\pi\rho}}y)\boxplus_{\pi+\rho-\pi\rho} z$\\
            $PA3$ & $x\boxplus_{\pi}x=x$\\
            $PA4$ & $(x\boxplus_{\pi}y)\cdot z=x\cdot z\boxplus_{\pi}y\cdot z$\\
            $PA5$ & $(x\boxplus_{\pi}y)+z=(x+z)\boxplus_{\pi}(y+z)$\\
            $L1$ & $\epsilon::x=x$\\
            $L2$ & $u::(x\cdot y)=u::x\cdot u::y$\\
            $L3$ & $u::(x+ y)=u::x+ u::y$\\
            $L4$ & $u::(v::x)=uv::x$\\
        \end{tabular}
        \caption{Axioms of $qBAPTC$ with localities}
        \label{AxiomsForqBAPTC}
    \end{table}
\end{center}

\begin{definition}[Basic terms of $qBAPTC$ with localities]\label{BTBATCG}
The set of basic terms of $qBAPTC$ with localities, $\mathcal{B}(qBAPTC^{sl})$, is inductively defined as follows:

\begin{enumerate}
  \item $\mathbb{E}\subset\mathcal{B}(qBAPTC^{sl})$;
  \item if $u\in Loc^*, t\in\mathcal{B}(qBAPTC^{sl})$ then $u::t\in\mathcal{B}(qBAPTC^{sl})$;
  \item if $e\in \mathbb{E}, t\in\mathcal{B}(qBAPTC^{sl})$ then $e\cdot t\in\mathcal{B}(qBAPTC^{sl})$;
  \item if $t,t'\in\mathcal{B}(qBAPTC^{sl})$ then $t+ t'\in\mathcal{B}(qBAPTC^{sl})$;
  \item if $t,t'\in\mathcal{B}(qBAPTC^{sl})$ then $t\boxplus_{\pi} t'\in\mathcal{B}(qBAPTC^{sl})$.
\end{enumerate}
\end{definition}

\begin{theorem}[Elimination theorem of $qBAPTC$ with localities]\label{ETBATCG}
Let $p$ be a closed $qBAPTC$ with localities term. Then there is a basic $qBAPTC$ with localities term $q$ such that $qBAPTC\vdash p=q$.
\end{theorem}

\begin{proof}
The same as that of $BAPTC^{sl}$, we omit the proof, please refer to \cite{LOC2} for details.
\end{proof}

We will define a term-deduction system which gives the operational semantics of $qBAPTC$ with localities. We give the operational transition rules for
atomic event $e\in\mathbb{E}$, operators $::$, $\cdot$ and $+$ as Table \ref{SETRForqBAPTC} shows. And the predicate $\xrightarrow{e}\surd$ represents successful termination after execution
of the event $e$.

\begin{center}
    \begin{table}
        $$\frac{}{\langle e,\varrho\rangle\rightsquigarrow\langle\breve{e},\varrho\rangle}$$
        $$\frac{\langle x,\varrho\rangle\rightsquigarrow \langle x',\varrho\rangle}{\langle x\cdot y,\varrho\rangle\rightsquigarrow \langle x'\cdot y,\varrho\rangle}$$
        $$\frac{\langle x,\varrho\rangle\rightsquigarrow \langle x',\varrho\rangle\quad \langle y,\varrho\rangle\rightsquigarrow \langle y',\varrho\rangle}{\langle x+y,\varrho\rangle\rightsquigarrow \langle x'+y',\varrho\rangle}$$
        $$\frac{\langle x,\varrho\rangle\rightsquigarrow \langle x',\varrho\rangle}{\langle x\boxplus_{\pi}y,\varrho\rangle\rightsquigarrow \langle x',\varrho\rangle}\quad \frac{\langle y,\varrho\rangle\rightsquigarrow \langle y',\varrho\rangle}{\langle x\boxplus_{\pi}y,\varrho\rangle\rightsquigarrow \langle y',\varrho\rangle}$$

        $$\frac{}{\langle \breve{e},\varrho\rangle\xrightarrow[\epsilon]{e}\langle\surd,\varrho'\rangle}\textrm{ if }\varrho'\in effect(e,\varrho)$$
        $$\frac{}{\langle loc::\breve{e},\varrho\rangle\xrightarrow[loc]{e}\langle\surd,\varrho'\rangle}$$
        $$\frac{\langle x,\varrho\rangle\xrightarrow[u]{e}x'}{\langle loc::x,\varrho\rangle\xrightarrow[loc\ll u]{e}\langle loc::x',\varrho'\rangle}$$
        $$\frac{\langle x,\varrho\rangle\xrightarrow[u]{e}\langle\surd,\varrho'\rangle}{\langle x+ y,\varrho\rangle\xrightarrow[u]{e}\langle\surd,\varrho'\rangle} \quad\frac{\langle x,\varrho\rangle\xrightarrow[u]{e}\langle x',\varrho'\rangle}{\langle x+ y,\varrho\rangle\xrightarrow[u]{e}\langle x',\varrho'\rangle}$$
        $$\frac{\langle y,\varrho\rangle\xrightarrow[u]{e}\langle\surd,\varrho'\rangle}{\langle x+ y,\varrho\rangle\xrightarrow[u]{e}\langle\surd,\varrho'\rangle} \quad\frac{\langle y,\varrho\rangle\xrightarrow[u]{e}\langle y',\varrho'\rangle}{\langle x+ y,\varrho\rangle\xrightarrow[u]{e}\langle y',\varrho'\rangle}$$
        $$\frac{\langle x,\varrho\rangle\xrightarrow[u]{e}\langle\surd,\varrho'\rangle}{\langle x\cdot y,\varrho\rangle\xrightarrow[u]{e} \langle y,\varrho'\rangle} \quad\frac{\langle x,\varrho\rangle\xrightarrow[u]{e}\langle x',\varrho'\rangle}{\langle x\cdot y,\varrho\rangle\xrightarrow[u]{e}\langle x'\cdot y,\varrho'\rangle}$$
        \caption{Single event transition rules of $qBAPTC$ with localities}
        \label{SETRForqBAPTC}
    \end{table}
\end{center}

Note that, we replace the single atomic event $e\in\mathbb{E}$ by $X\subseteq\mathbb{E}$, we can obtain the static location pomset transition rules of $qBAPTC$ with localities, and omit them.

\begin{theorem}[Congruence of $qBAPTC$ with localities with respect to probabilistic static location truly concurrent bisimulation equivalences]
(1) Probabilistic static location pomset bisimulation equivalence $\sim_{pp}^{sl}$ is a congruence with respect to $qBAPTC$ with localities.

(2) Probabilistic static location step bisimulation equivalence $\sim_{ps}^{sl}$ is a congruence with respect to $qBAPTC$ with localities.

(3) Probabilistic static location hp-bisimulation equivalence $\sim_{php}^{sl}$ is a congruence with respect to $qBAPTC$ with localities.

(4) Probabilistic static location hhp-bisimulation equivalence $\sim_{phhp}^{sl}$ is a congruence with respect to $qBAPTC$ with localities.
\end{theorem}

\begin{proof}
(1) It is easy to see that probabilistic static location pomset bisimulation is an equivalent relation on $qBAPTC$ with localities terms, we only need to prove that $\sim_{pp}^{sl}$ is preserved by the operators $::$, $\cdot$
, $+$ and $\boxplus_{\pi}$. It is trivial and we leave the proof as an exercise for the readers.

(2) It is easy to see that probabilistic static location step bisimulation is an equivalent relation on $qBAPTC$ with localities terms, we only need to prove that $\sim_{ps}^{sl}$ is preserved by the operators $::$, $\cdot$,
$+$ and $\boxplus_{\pi}$. It is trivial and we leave the proof as an exercise for the readers.

(3) It is easy to see that probabilistic static location hp-bisimulation is an equivalent relation on $qBAPTC$ with localities terms, we only need to prove that $\sim_{php}^{sl}$ is preserved by the operators $::$, $\cdot$,
$+$, and $\boxplus_{\pi}$. It is trivial and we leave the proof as an exercise for the readers.

(4) It is easy to see that probabilistic static location hhp-bisimulation is an equivalent relation on $qBAPTC$ with localities terms, we only need to prove that $\sim_{phhp}^{sl}$ is preserved by the operators $::$, $\cdot$,
$+$, and $\boxplus_{\pi}$. It is trivial and we leave the proof as an exercise for the readers.
\end{proof}

\begin{theorem}[Soundness of $qBAPTC$ with localities modulo probabilistic static location truly concurrent bisimulation equivalences]
(1) Let $x$ and $y$ be $qBAPTC$ with localities terms. If $qBAPTC^{sl}\vdash x=y$, then $x\sim_{pp}^{sl} y$.

(2) Let $x$ and $y$ be $qBAPTC$ with localities terms. If $qBAPTC^{sl}\vdash x=y$, then $x\sim_{ps}^{sl} y$.

(3) Let $x$ and $y$ be $qBAPTC$ with localities terms. If $qBAPTC^{sl}\vdash x=y$, then $x\sim_{php}^{sl} y$.

(4) Let $x$ and $y$ be $qBAPTC$ with localities terms. If $qBAPTC^{sl}\vdash x=y$, then $x\sim_{phhp}^{sl} y$.
\end{theorem}

\begin{proof}
(1) Since probabilistic static location pomset bisimulation $\sim_{pp}^{sl}$ is both an equivalent and a congruent relation, we only need to check if each axiom in Table \ref{AxiomsForqBAPTC} is sound
modulo probabilistic static location pomset bisimulation equivalence. We leave the proof as an exercise for the readers.

(2) Since probabilistic static location step bisimulation $\sim_{ps}^{sl}$ is both an equivalent and a congruent relation, we only need to check if each axiom in Table \ref{AxiomsForqBAPTC} is sound modulo
probabilistic static location step bisimulation equivalence. We leave the proof as an exercise for the readers.

(3) Since probabilistic static location hp-bisimulation $\sim_{php}^{sl}$ is both an equivalent and a congruent relation, we only need to check if each axiom in Table \ref{AxiomsForqBAPTC} is sound modulo
probabilistic static location hp-bisimulation equivalence. We leave the proof as an exercise for the readers.

(4) Since probabilistic static location hhp-bisimulation $\sim_{phhp}^{sl}$ is both an equivalent and a congruent relation, we only need to check if each axiom in Table \ref{AxiomsForqBAPTC} is sound modulo
probabilistic static location hhp-bisimulation equivalence. We leave the proof as an exercise for the readers.
\end{proof}

\begin{theorem}[Completeness of $qBAPTC$ with localities modulo probabilistic static location truly concurrent bisimulation equivalences]\label{CBATCG}
(1) Let $p$ and $q$ be closed $qBAPTC$ with localities terms, if $p\sim_{pp}^{sl} q$ then $p=q$.

(2) Let $p$ and $q$ be closed $qBAPTC$ with localities terms, if $p\sim_{ps}^{sl} q$ then $p=q$.

(3) Let $p$ and $q$ be closed $qBAPTC$ with localities terms, if $p\sim_{php}^{sl} q$ then $p=q$.

(4) Let $p$ and $q$ be closed $qBAPTC$ with localities terms, if $p\sim_{phhp}^{sl} q$ then $p=q$.
\end{theorem}

\begin{proof}
According to the definition of probabilistic static location truly concurrent bisimulation equivalences $\sim_{pp}^{sl}$, $\sim_{ps}^{sl}$, $\sim_{php}^{sl}$ and $\sim_{phhp}^{sl}$, $p\sim_{pp}^{sl}q$, $p\sim_{ps}^{sl}q$, $p\sim_{php}^{sl}q$ and $p\sim_{phhp}^{sl}q$ implies
both the bisimilarities between $p$ and $q$, and also the in the same quantum states. According to the completeness of $BAPTC^{sl}$ (please refer to \cite{LOC2} for details), we can get the
completeness of $qBAPTC^{sl}$.
\end{proof}

\subsection{$APPTC$ with Localites for Closed Quantum Systems}{\label{qapptcl}}

In this subsection, we will extend $APPTC^{sl}$ for closed quantum systems, which is abbreviated $qAPPTC$ with localities.

The set of axioms of $qAPPTC$ with localities including axioms of $qBAPTC$ with localities in Table \ref{AxiomsForqBAPTC} and the axioms are shown in Table \ref{AxiomsForqAPTC}.

\begin{center}
    \begin{table}
        \begin{tabular}{@{}ll@{}}
            \hline No. &Axiom\\
            $P1$ & $(x+x=x,y+y=y)\quad x\between y = x\parallel y + x\mid y$\\
            $P2$ & $x\parallel y = y \parallel x$\\
            $P3$ & $(x\parallel y)\parallel z = x\parallel (y\parallel z)$\\
            $P4$ & $(x+x=x,y+y=y)\quad x\parallel y = x\leftmerge y + y\leftmerge x$\\
            $P5$ & $(e_1\leq e_2)\quad e_1\leftmerge (e_2\cdot y) = (e_1\leftmerge e_2)\cdot y$\\
            $P6$ & $(e_1\leq e_2)\quad (e_1\cdot x)\leftmerge e_2 = (e_1\leftmerge e_2)\cdot x$\\
            $P7$ & $(e_1\leq e_2)\quad (e_1\cdot x)\leftmerge (e_2\cdot y) = (e_1\leftmerge e_2)\cdot (x\between y)$\\
            $P8$ & $(x+ y)\leftmerge z = (x\leftmerge z)+ (y\leftmerge z)$\\
            $P9$ & $\delta\leftmerge x = \delta$\\
            $C1$ & $e_1\mid e_2 = \gamma(e_1,e_2)$\\
            $C2$ & $e_1\mid (e_2\cdot y) = \gamma(e_1,e_2)\cdot y$\\
            $C3$ & $(e_1\cdot x)\mid e_2 = \gamma(e_1,e_2)\cdot x$\\
            $C4$ & $(e_1\cdot x)\mid (e_2\cdot y) = \gamma(e_1,e_2)\cdot (x\between y)$\\
            $C5$ & $(x+ y)\mid z = (x\mid z) + (y\mid z)$\\
            $C6$ & $x\mid (y+ z) = (x\mid y)+ (x\mid z)$\\
            $C7$ & $\delta\mid x = \delta$\\
            $C8$ & $x\mid\delta = \delta$\\
            $PM1$ & $x\parallel (y\boxplus_{\pi} z)=(x\parallel y)\boxplus_{\pi}(x\parallel z)$\\
            $PM2$ & $(x\boxplus_{\pi} y)\parallel z=(x\parallel z)\boxplus_{\pi}(y\parallel z)$\\
            $PM3$ & $x\mid (y\boxplus_{\pi} z)=(x\mid y)\boxplus_{\pi}(x\mid z)$\\
            $PM4$ & $(x\boxplus_{\pi} y)\mid z=(x\mid z)\boxplus_{\pi}(y\mid z)$\\
            $CE1$ & $\Theta(e) = e$\\
            $CE2$ & $\Theta(\delta) = \delta$\\
            $CE3$ & $\Theta(x+ y) = \Theta(x)\triangleleft y + \Theta(y)\triangleleft x$\\
            $CE4$ & $\Theta(x\cdot y)=\Theta(x)\cdot\Theta(y)$\\
            $CE5$ & $\Theta(x\parallel y) = ((\Theta(x)\triangleleft y)\parallel y)+ ((\Theta(y)\triangleleft x)\parallel x)$\\
            $CE6$ & $\Theta(x\mid y) = ((\Theta(x)\triangleleft y)\mid y)+ ((\Theta(y)\triangleleft x)\mid x)$\\
            $PCE1$ & $\Theta(x\boxplus_{\pi} y) = \Theta(x)\triangleleft y \boxplus_{\pi} \Theta(y)\triangleleft x$\\
        \end{tabular}
        \caption{Axioms of $qAPPTC$ with localities}
        \label{AxiomsForqAPTC}
    \end{table}
\end{center}

\begin{center}
    \begin{table}
        \begin{tabular}{@{}ll@{}}
            \hline No. &Axiom\\
            $U1$ & $(\sharp(e_1,e_2))\quad e_1\triangleleft e_2 = \tau$\\
            $U2$ & $(\sharp(e_1,e_2),e_2\leq e_3)\quad e_1\triangleleft e_3 = e_1$\\
            $U3$ & $(\sharp(e_1,e_2),e_2\leq e_3)\quad e3\triangleleft e_1 = \tau$\\
            $U4$ & $e\triangleleft \delta = e$\\
            $U5$ & $\delta \triangleleft e = \delta$\\
            $U6$ & $(x+ y)\triangleleft z = (x\triangleleft z)+ (y\triangleleft z)$\\
            $U7$ & $(x\cdot y)\triangleleft z = (x\triangleleft z)\cdot (y\triangleleft z)$\\
            $U8$ & $(x\leftmerge y)\triangleleft z = (x\triangleleft z)\leftmerge (y\triangleleft z)$\\
            $U9$ & $(x\mid y)\triangleleft z = (x\triangleleft z)\mid (y\triangleleft z)$\\
            $U10$ & $x\triangleleft (y+ z) = (x\triangleleft y)\triangleleft z$\\
            $U11$ & $x\triangleleft (y\cdot z)=(x\triangleleft y)\triangleleft z$\\
            $U12$ & $x\triangleleft (y\leftmerge z) = (x\triangleleft y)\triangleleft z$\\
            $U13$ & $x\triangleleft (y\mid z) = (x\triangleleft y)\triangleleft z$\\
            $PU1$ & $(\sharp_{\pi}(e_1,e_2))\quad e_1\triangleleft e_2 = \tau$\\
            $PU2$ & $(\sharp_{\pi}(e_1,e_2),e_2\leq e_3)\quad e_1\triangleleft e_3 = e_1$\\
            $PU3$ & $(\sharp_{\pi}(e_1,e_2),e_2\leq e_3)\quad e_3\triangleleft e_1 = \tau$\\
            $PU4$ & $(x\boxplus_{\pi} y)\triangleleft z = (x\triangleleft z)\boxplus_{\pi} (y\triangleleft z)$\\
            $PU5$ & $x\triangleleft (y\boxplus_{\pi} z) = (x\triangleleft y)\triangleleft z$\\
            $L5$ & $u::(x\between y) = u::x\between u:: y$\\
            $L6$ & $u::(x\parallel y) = u::x\parallel u:: y$\\
            $L7$ & $u::(x\mid y) = u::x\mid u:: y$\\
            $L8$ & $u::(\Theta(x)) = \Theta(u::x)$\\
            $L9$ & $u::(x\triangleleft y) = u::x\triangleleft u:: y$\\
            $L10$ & $u::\delta=\delta$\\
            $PL1$ & $u::(x\boxplus_{\pi} y) = u::x\boxplus_{\pi} u:: y$\\
        \end{tabular}
        \caption{Axioms of $qAPPTC$ with localities (continuing)}
        \label{AxiomsForqAPTC2}
    \end{table}
\end{center}

\begin{definition}[Basic terms of $qAPPTC$ with localities]\label{BTAPTCG}
The set of basic terms of $qAPPTC$ with localities, $\mathcal{B}(qAPPTC^{sl})$, is inductively defined as follows:

\begin{enumerate}
    \item $\mathbb{E}\subset\mathcal{B}(qAPPTC^{sl})$;
    \item if $u\in Loc^*, t\in\mathcal{B}(qAPPTC^{sl})$ then $u::t\in\mathcal{B}(qAPPTC^{sl})$;
    \item if $e\in \mathbb{E}, t\in\mathcal{B}(qAPPTC^{sl})$ then $e\cdot t\in\mathcal{B}(qAPPTC^{sl})$;
    \item if $t,t'\in\mathcal{B}(qAPPTC^{sl})$ then $t+ t'\in\mathcal{B}(qAPPTC^{sl})$;
    \item if $t,t'\in\mathcal{B}(qAPPTC^{sl})$ then $t\boxplus_{\pi} t'\in\mathcal{B}(qAPPTC^{sl})$
    \item if $t,t'\in\mathcal{B}(qAPPTC^{sl})$ then $t\leftmerge t'\in\mathcal{B}(qAPPTC^{sl})$.
\end{enumerate}
\end{definition}

Based on the definition of basic terms for $qAPPTC$ with localities (see Definition \ref{BTAPTCG}) and axioms of $qAPPTC$ with localities, we can prove the elimination theorem of $qAPPTC$ with localities.

\begin{theorem}[Elimination theorem of $qAPPTC$ with localities]\label{ETAPTCG}
Let $p$ be a closed $qAPPTC$ with localities term. Then there is a basic $qAPPTC$ with localities term $q$ such that $qAPPTC\vdash p=q$.
\end{theorem}

\begin{proof}
The same as that of $APPTC^{sl}$, we omit the proof, please refer to \cite{LOC2} for details.
\end{proof}

We will define a term-deduction system which gives the operational semantics of $qAPPTC$ with localities. Two atomic events $e_1$ and $e_2$ are in race condition, which are denoted 
$e_1\% e_2$.

\begin{center}
    \begin{table}
        $$\frac{x\rightsquigarrow x'\quad y\rightsquigarrow y'}{x\between y\rightsquigarrow x'\parallel y'+x'\mid y'}$$
        $$\frac{x\rightsquigarrow x'\quad y\rightsquigarrow y'}{x\parallel y\rightsquigarrow x'\leftmerge y+y'\leftmerge x}$$
        $$\frac{x\rightsquigarrow x'}{x\leftmerge y\rightsquigarrow x'\leftmerge y}$$
        $$\frac{x\rightsquigarrow x'\quad y\rightsquigarrow y'}{x\mid y\rightsquigarrow x'\mid y'}$$
        $$\frac{x\rightsquigarrow x'}{\Theta(x)\rightsquigarrow \Theta(x')}$$
        $$\frac{x\rightsquigarrow x'}{x\triangleleft y\rightsquigarrow x'\triangleleft y}$$
        \caption{Probabilistic transition rules of $qAPPTC$ with localities}
        \label{TRForAPPTCG1}
    \end{table}
\end{center}

\begin{center}
    \begin{table}
        $$\frac{}{\langle \breve{e_1}\parallel\cdots \parallel \breve{e_n},\varrho\rangle\xrightarrow[\epsilon]{\{e_1,\cdots,e_n\}}\langle\surd,\varrho'\rangle}\textrm{ if }\varrho'\in effect(e_1,\varrho)\cup\cdots\cup effect(e_n,\varrho)$$

        $$\frac{\langle x,\varrho\rangle\xrightarrow[u]{e_1}\langle\surd,\varrho'\rangle\quad \langle y,\varrho\rangle\xrightarrow[v]{e_2}\langle\surd,\varrho''\rangle}{\langle x\parallel y,\varrho\rangle\xrightarrow[u\diamond v]{\{e_1,e_2\}}\langle\surd,\varrho'\cup \varrho''\rangle} \quad\frac{\langle x,\varrho\rangle\xrightarrow[u]{e_1}\langle x',\varrho'\rangle\quad \langle y,\varrho\rangle\xrightarrow[v]{e_2}\langle\surd,\varrho''\rangle}{\langle x\parallel y,\varrho\rangle\xrightarrow[u\diamond v]{\{e_1,e_2\}}\langle x',\varrho'\cup \varrho''\rangle}$$

        $$\frac{\langle x,\varrho\rangle\xrightarrow[u]{e_1}\langle\surd,\varrho'\rangle\quad \langle y,\varrho\rangle\xrightarrow[v]{e_2}\langle y',\varrho''\rangle}{\langle x\parallel y,\varrho\rangle\xrightarrow[u\diamond v]{\{e_1,e_2\}}\langle y',\varrho'\cup \varrho''\rangle} \quad\frac{\langle x,\varrho\rangle\xrightarrow[u]{e_1}\langle x',\varrho'\rangle\quad \langle y,\varrho\rangle\xrightarrow[v]{e_2}\langle y',\varrho''\rangle}{\langle x\parallel y,\varrho\rangle\xrightarrow[u\diamond v]{\{e_1,e_2\}}\langle x'\between y',\varrho'\cup \varrho''\rangle}$$

        $$\frac{\langle x,\varrho\rangle\xrightarrow[u]{e_1}\langle\surd,\varrho'\rangle\quad \langle y,\varrho\rangle\xnrightarrow[v]{e_2}\quad(e_1\%e_2)}{\langle x\parallel y,\varrho\rangle\xrightarrow[u]{e_1}\langle y,\varrho'\rangle} \quad\frac{\langle x,\varrho\rangle\xrightarrow[u]{e_1}\langle x',\varrho'\rangle\quad \langle y,\varrho\rangle\xnrightarrow{e_2}\quad(e_1\%e_2)}{\langle x\parallel y,\varrho\rangle\xrightarrow[u]{e_1}\langle x'\between y,\varrho'\rangle}$$

        $$\frac{\langle x,\varrho\rangle\xnrightarrow[u]{e_1}\quad \langle y,\varrho\rangle\xrightarrow[v]{e_2}\langle\surd,\varrho''\rangle\quad(e_1\%e_2)}{\langle x\parallel y,\varrho\rangle\xrightarrow[v]{e_2}\langle x,\varrho''\rangle} \quad\frac{\langle x,\varrho\rangle\xnrightarrow[u]{e_1}\quad \langle y,\varrho\rangle\xrightarrow[v]{e_2}\langle y',\varrho''\rangle\quad(e_1\%e_2)}{\langle x\parallel y,\varrho\rangle\xrightarrow[v]{e_2}\langle x\between y',\varrho''\rangle}$$

        $$\frac{\langle x,\varrho\rangle\xrightarrow[u]{e_1}\langle\surd,\varrho'\rangle\quad \langle y,\varrho\rangle\xrightarrow[v]{e_2}\langle\surd,\varrho''\rangle \quad(e_1\leq e_2)}{\langle x\leftmerge y,\varrho\rangle\xrightarrow[u\diamond v]{\{e_1,e_2\}}\langle \surd,\varrho'\cup \varrho''\rangle} \quad\frac{\langle x,\varrho\rangle\xrightarrow[u]{e_1}\langle x',\varrho'\rangle\quad \langle y,\varrho\rangle\xrightarrow[v]{e_2}\langle\surd,\varrho''\rangle \quad(e_1\leq e_2)}{\langle x\leftmerge y,\varrho\rangle\xrightarrow[u\diamond v]{\{e_1,e_2\}}\langle x',\varrho'\cup \varrho''\rangle}$$

        $$\frac{\langle x,\varrho\rangle\xrightarrow[u]{e_1}\langle\surd,\varrho'\rangle\quad \langle y,\varrho\rangle\xrightarrow[v]{e_2}\langle y',\varrho''\rangle \quad(e_1\leq e_2)}{\langle x\leftmerge y,\varrho\rangle\xrightarrow[u\diamond v]{\{e_1,e_2\}}\langle y',\varrho'\cup \varrho''\rangle} \quad\frac{\langle x,\varrho\rangle\xrightarrow[u]{e_1}\langle x',\varrho'\rangle\quad \langle y,\varrho\rangle\xrightarrow[v]{e_2}\langle y',\varrho''\rangle \quad(e_1\leq e_2)}{\langle x\leftmerge y,\varrho\rangle\xrightarrow[u\diamond v]{\{e_1,e_2\}}\langle x'\between y',\varrho'\cup \varrho''\rangle}$$

        $$\frac{\langle x,\varrho\rangle\xrightarrow[u]{e_1}\langle\surd,\varrho'\rangle\quad \langle y,\varrho\rangle\xrightarrow[v]{e_2}\langle\surd,\varrho''\rangle}{\langle x\mid y,\varrho\rangle\xrightarrow[u\diamond v]{\gamma(e_1,e_2)}\langle\surd,effect(\gamma(e_1,e_2),\varrho)\rangle} \quad\frac{\langle x,\varrho\rangle\xrightarrow[u]{e_1}\langle x',\varrho'\rangle\quad \langle y,\varrho\rangle\xrightarrow[v]{e_2}\langle\surd,\varrho''\rangle}{\langle x\mid y,\varrho\rangle\xrightarrow[u\diamond v]{\gamma(e_1,e_2)}\langle x',effect(\gamma(e_1,e_2),\varrho)\rangle}$$

        $$\frac{\langle x,\varrho\rangle\xrightarrow[u]{e_1}\langle\surd,\varrho'\rangle\quad \langle y,\varrho\rangle\xrightarrow[v]{e_2}\langle y',\varrho''\rangle}{\langle x\mid y,\varrho\rangle\xrightarrow[u\diamond v]{\gamma(e_1,e_2)}\langle y',effect(\gamma(e_1,e_2),\varrho)\rangle} \quad\frac{\langle x,\varrho\rangle\xrightarrow[u]{e_1}\langle x',\varrho'\rangle\quad \langle y,\varrho\rangle\xrightarrow[v]{e_2}\langle y',\varrho''\rangle}{\langle x\mid y,\varrho\rangle\xrightarrow[u\diamond v]{\gamma(e_1,e_2)}\langle x'\between y',effect(\gamma(e_1,e_2),\varrho)\rangle}$$

        \caption{Action transition rules of $qAPPTC$ with localities}
        \label{TRForAPTCG}
    \end{table}
\end{center}

\begin{center}
    \begin{table}
        $$\frac{\langle x,\varrho\rangle\xrightarrow[u]{e_1}\langle\surd,\varrho'\rangle\quad (\sharp(e_1,e_2))}{\langle \Theta(x),\varrho\rangle\xrightarrow[u]{e_1}\langle\surd,\varrho'\rangle} \quad\frac{\langle x,\varrho\rangle\xrightarrow[v]{e_2}\langle\surd,\varrho''\rangle\quad (\sharp(e_1,e_2))}{\langle\Theta(x),\varrho\rangle\xrightarrow[v]{e_2}\langle\surd,\varrho''\rangle}$$

        $$\frac{\langle x,\varrho\rangle\xrightarrow[u]{e_1}\langle x',\varrho'\rangle\quad (\sharp(e_1,e_2))}{\langle\Theta(x),\varrho\rangle\xrightarrow[u]{e_1}\langle\Theta(x'),\varrho'\rangle} \quad\frac{\langle x,\varrho\rangle\xrightarrow[v]{e_2}\langle x'',\varrho''\rangle\quad (\sharp(e_1,e_2))}{\langle\Theta(x),\varrho\rangle\xrightarrow[v]{e_2}\langle\Theta(x''),\varrho''\rangle}$$

        $$\frac{\langle x,\varrho\rangle\xrightarrow[u]{e_1}\langle\surd,\varrho'\rangle \quad \langle y,\varrho\rangle\nrightarrow^{e_2}\quad (\sharp(e_1,e_2))}{\langle x\triangleleft y,\varrho\rangle\xrightarrow[u]{\tau}\langle\surd,\varrho'\rangle}
        \quad\frac{\langle x,\varrho\rangle\xrightarrow{e_1}\langle x',\varrho'\rangle \quad \langle y,\varrho\rangle\nrightarrow^{e_2}\quad (\sharp(e_1,e_2))}{\langle x\triangleleft y,\varrho\rangle\xrightarrow[u]{\tau}\langle x',\varrho'\rangle}$$

        $$\frac{\langle x,\varrho\rangle\xrightarrow[u]{e_1}\langle\surd,\varrho\rangle \quad \langle y,\varrho\rangle\nrightarrow^{e_3}\quad (\sharp(e_1,e_2),e_2\leq e_3)}{\langle x\triangleleft y,\varrho\rangle\xrightarrow[u]{e_1}\langle\surd,\varrho'\rangle}
        \quad\frac{\langle x,\varrho\rangle\xrightarrow[u]{e_1}\langle x',\varrho'\rangle \quad \langle y,\varrho\rangle\nrightarrow^{e_3}\quad (\sharp(e_1,e_2),e_2\leq e_3)}{\langle x\triangleleft y,\varrho\rangle\xrightarrow[u]{e_1}\langle x',\varrho'\rangle}$$

        $$\frac{\langle x,\varrho\rangle\xrightarrow{[u]e_3}\langle\surd,\varrho'\rangle \quad \langle y,\varrho\rangle\nrightarrow^{e_2}\quad (\sharp(e_1,e_2),e_1\leq e_3)}{\langle x\triangleleft y,\varrho\rangle\xrightarrow[u]{\tau}\langle\surd,\varrho'\rangle}
        \quad\frac{\langle x,\varrho\rangle\xrightarrow[u]{e_3}\langle x',\varrho'\rangle \quad \langle y,\varrho\rangle\nrightarrow^{e_2}\quad (\sharp(e_1,e_2),e_1\leq e_3)}{\langle x\triangleleft y,\varrho\rangle\xrightarrow[u]{\tau}\langle x',\varrho'\rangle}$$

        $$\frac{\langle x,\varrho\rangle\xrightarrow[u]{e_1}\langle\surd,\varrho'\rangle\quad (\sharp_{\pi}(e_1,e_2))}{\langle \Theta(x),\varrho\rangle\xrightarrow[u]{e_1}\langle\surd,\varrho'\rangle} \quad\frac{\langle x,\varrho\rangle\xrightarrow[v]{e_2}\langle\surd,\varrho''\rangle\quad (\sharp_{\pi}(e_1,e_2))}{\langle\Theta(x),\varrho\rangle\xrightarrow[v]{e_2}\langle\surd,\varrho''\rangle}$$

        $$\frac{\langle x,\varrho\rangle\xrightarrow[u]{e_1}\langle x',\varrho'\rangle\quad (\sharp_{\pi}(e_1,e_2))}{\langle\Theta(x),\varrho\rangle\xrightarrow[u]{e_1}\langle\Theta(x'),\varrho'\rangle} \quad\frac{\langle x,\varrho\rangle\xrightarrow[v]{e_2}\langle x'',\varrho''\rangle\quad (\sharp_{\pi}(e_1,e_2))}{\langle\Theta(x),\varrho\rangle\xrightarrow[v]{e_2}\langle\Theta(x''),\varrho''\rangle}$$

        $$\frac{\langle x,\varrho\rangle\xrightarrow[u]{e_1}\langle\surd,\varrho'\rangle \quad \langle y,\varrho\rangle\nrightarrow^{e_2}\quad (\sharp_{\pi}(e_1,e_2))}{\langle x\triangleleft y,\varrho\rangle\xrightarrow[u]{\tau}\langle\surd,\varrho'\rangle}
        \quad\frac{\langle x,\varrho\rangle\xrightarrow[u]{e_1}\langle x',\varrho'\rangle \quad \langle y,\varrho\rangle\nrightarrow^{e_2}\quad (\sharp_{\pi}(e_1,e_2))}{\langle x\triangleleft y,\varrho\rangle\xrightarrow[u]{\tau}\langle x',\varrho'\rangle}$$

        $$\frac{\langle x,\varrho\rangle\xrightarrow[u]{e_1}\langle\surd,\varrho\rangle \quad \langle y,\varrho\rangle\nrightarrow^{e_3}\quad (\sharp_{\pi}(e_1,e_2),e_2\leq e_3)}{\langle x\triangleleft y,\varrho\rangle\xrightarrow[u]{e_1}\langle\surd,\varrho'\rangle}
        \quad\frac{\langle x,\varrho\rangle\xrightarrow[u]{e_1}\langle x',\varrho'\rangle \quad \langle y,\varrho\rangle\nrightarrow^{e_3}\quad (\sharp_{\pi}(e_1,e_2),e_2\leq e_3)}{\langle x\triangleleft y,\varrho\rangle\xrightarrow[u]{e_1}\langle x',\varrho'\rangle}$$

        $$\frac{\langle x,\varrho\rangle\xrightarrow[u]{e_3}\langle\surd,\varrho'\rangle \quad \langle y,\varrho\rangle\nrightarrow^{e_2}\quad (\sharp_{\pi}(e_1,e_2),e_1\leq e_3)}{\langle x\triangleleft y,\varrho\rangle\xrightarrow[u]{\tau}\langle\surd,\varrho'\rangle}
        \quad\frac{\langle x,\varrho\rangle\xrightarrow[u]{e_3}\langle x',\varrho'\rangle \quad \langle y,\varrho\rangle\nrightarrow^{e_2}\quad (\sharp_{\pi}(e_1,e_2),e_1\leq e_3)}{\langle x\triangleleft y,\varrho\rangle\xrightarrow[u]{\tau}\langle x',\varrho'\rangle}$$

        $$\frac{\langle x,\varrho\rangle\xrightarrow[u]{e}\langle\surd,\varrho'\rangle}{\langle\partial_H(x),\varrho\rangle\xrightarrow[u]{e}\langle\surd,\varrho'\rangle}\quad (e\notin H)\quad\frac{\langle x,\varrho\rangle\xrightarrow[u]{e}\langle x',\varrho'\rangle}{\langle\partial_H(x),\varrho\rangle\xrightarrow[u]{e}\langle\partial_H(x'),\varrho'\rangle}\quad(e\notin H)$$

        $$\frac{\langle x,\varrho\rangle\xrightarrow[u]{e}\langle\surd,\varrho'\rangle}{\langle\partial_H(x),\varrho\rangle\xrightarrow[u]{e}\langle\surd,\varrho'\rangle}\quad (e\notin H)\quad\frac{\langle x,\varrho\rangle\xrightarrow[u]{e}\langle x',\varrho'\rangle}{\langle\partial_H(x),\varrho\rangle\xrightarrow[u]{e}\langle\partial_H(x'),\varrho'\rangle}\quad(e\notin H)$$
        \caption{Action transition rules of $qAPPTC$ with localities (continuing)}
        \label{TRForAPTCG2}
    \end{table}
\end{center}

\begin{theorem}[Generalization of $qAPPTC$ with localities with respect to $qBAPTC$ with localities]
$qAPPTC$ with localities is a generalization of $qBAPTC$ with localities.
\end{theorem}

\begin{proof}
It follows from the following three facts.

\begin{enumerate}
  \item The transition rules of $qBAPTC$ with localities in section \ref{qbaptcl} are all source-dependent;
  \item The sources of the transition rules $qAPPTC$ with localities contain an occurrence of $\between$, or $\parallel$, or $\leftmerge$, or $\mid$, or $\Theta$, or $\triangleleft$;
  \item The transition rules of $qAPPTC$ with localities are all source-dependent.
\end{enumerate}

So, $qAPPTC$ with localities is a generalization of $qBAPTC$ with localities, that is, $qBAPTC$ with localities is an embedding of $qAPPTC$ with localities, as desired.
\end{proof}

\begin{theorem}[Congruence of $qAPPTC$ with localities with respect to probabilistic static location truly concurrent bisimulation equivalences]\label{CAPTCG}
(1) Probabilistic static location pomset bisimulation equivalence $\sim_{pp}^{sl}$ is a congruence with respect to $qAPPTC$ with localities.

(2) Probabilistic static location step bisimulation equivalence $\sim_{ps}^{sl}$ is a congruence with respect to $qAPPTC$ with localities.

(3) Probabilistic static location hp-bisimulation equivalence $\sim_{php}^{sl}$ is a congruence with respect to $qAPPTC$ with localities.

(4) Probabilistic static location hhp-bisimulation equivalence $\sim_{phhp}^{sl}$ is a congruence with respect to $qAPPTC$ with localities.
\end{theorem}

\begin{proof}
(1) It is easy to see that probabilistic static location pomset bisimulation is an equivalent relation on $qAPPTC$ with localities terms, we only need to prove that $\sim_{pp}^{sl}$ is preserved by the operators
$\parallel$, $\leftmerge$, $\mid$, $\Theta$, $\triangleleft$, $\partial_H$. It is trivial and we leave the proof as an exercise for the readers.

(2) It is easy to see that probabilistic static location step bisimulation is an equivalent relation on $qAPPTC$ with localities terms, we only need to prove that $\sim_{ps}^{sl}$ is preserved by the operators
$\parallel$, $\leftmerge$, $\mid$, $\Theta$, $\triangleleft$, $\partial_H$. It is trivial and we leave the proof as an exercise for the readers.

(3) It is easy to see that probabilistic static location hp-bisimulation is an equivalent relation on $qAPPTC$ with localities terms, we only need to prove that $\sim_{php}^{sl}$ is preserved by the operators
$\parallel$, $\leftmerge$, $\mid$, $\Theta$, $\triangleleft$, $\partial_H$. It is trivial and we leave the proof as an exercise for the readers.

(4) It is easy to see that probabilistic static location hhp-bisimulation is an equivalent relation on $qAPPTC$ with localities terms, we only need to prove that $\sim_{phhp}^{sl}$ is preserved by the operators
$\parallel$, $\leftmerge$, $\mid$, $\Theta$, $\triangleleft$, $\partial_H$. It is trivial and we leave the proof as an exercise for the readers.
\end{proof}

\begin{theorem}[Soundness of $qAPPTC$ with localities modulo probabilistic static location truly concurrent bisimulation equivalences]\label{SAPTCG}
(1) Let $x$ and $y$ be $qAPPTC$ with localities terms. If $qAPPTC^{sl}\vdash x=y$, then $x\sim_{pp}^{sl} y$.

(2) Let $x$ and $y$ be $qAPPTC$ with localities terms. If $qAPPTC^{sl}\vdash x=y$, then $x\sim_{ps}^{sl} y$.

(3) Let $x$ and $y$ be $qAPPTC$ with localities terms. If $qAPPTC^{sl}\vdash x=y$, then $x\sim_{php}^{sl} y$;

(3) Let $x$ and $y$ be $qAPPTC$ with localities terms. If $qAPPTC^{sl}\vdash x=y$, then $x\sim_{phhp}^{sl} y$.
\end{theorem}

\begin{proof}
(1) Since probabilistic static location pomset bisimulation $\sim_{pp}^{sl}$ is both an equivalent and a congruent relation, we only need to check if each axiom in Table \ref{AxiomsForqAPTC} is sound modulo
probabilistic static location pomset bisimulation equivalence. We leave the proof as an exercise for the readers.

(2) Since probabilistic static location step bisimulation $\sim_{ps}^{sl}$ is both an equivalent and a congruent relation, we only need to check if each axiom in Table \ref{AxiomsForqAPTC} is sound modulo
probabilistic static location step bisimulation equivalence. We leave the proof as an exercise for the readers.

(3) Since probabilistic static location hp-bisimulation $\sim_{php}^{sl}$ is both an equivalent and a congruent relation, we only need to check if each axiom in Table \ref{AxiomsForqAPTC} is sound modulo
probabilistic static location hp-bisimulation equivalence. We leave the proof as an exercise for the readers.

(4) Since probabilistic static location hhp-bisimulation $\sim_{phhp}^{sl}$ is both an equivalent and a congruent relation, we only need to check if each axiom in Table \ref{AxiomsForqAPTC} is sound modulo
probabilistic static location hhp-bisimulation equivalence. We leave the proof as an exercise for the readers.
\end{proof}

\begin{theorem}[Completeness of $qAPPTC$ with localities modulo probabilistic static location truly concurrent bisimulation equivalences]\label{CAPTCG}
(1) Let $p$ and $q$ be closed $qAPPTC$ with localities terms, if $p\sim_{pp}^{sl} q$ then $p=q$.

(2) Let $p$ and $q$ be closed $qAPPTC$ with localities terms, if $p\sim_{ps}^{sl} q$ then $p=q$.

(3) Let $p$ and $q$ be closed $qAPPTC$ with localities terms, if $p\sim_{php}^{sl} q$ then $p=q$.

(3) Let $p$ and $q$ be closed $qAPPTC$ with localities terms, if $p\sim_{phhp}^{sl} q$ then $p=q$.
\end{theorem}

\begin{proof}
According to the definition of probabilistic static location truly concurrent bisimulation equivalences $\sim_{pp}^{sl}$, $\sim_{ps}^{sl}$, $\sim_{php}^{sl}$ and $\sim_{phhp}^{sl}$, $p\sim_{pp}^{sl}q$, $p\sim_{ps}^{sl}q$, $p\sim_{php}^{sl}q$ and $p\sim_{phhp}^{sl}q$ implies
both the bisimilarities between $p$ and $q$, and also the in the same quantum states. According to the completeness of $APPTC^{sl}$ (please refer to \cite{LOC2} for details), we can get the
completeness of $qAPPTC^{sl}$.
\end{proof}

\subsection{Recursion}{\label{qcrec}}

In this subsection, we introduce recursion to capture infinite processes based on $qAPPTC$ with localities. In the following, $E,F,G$ are recursion specifications, $X,Y,Z$ are recursive variables.

\begin{definition}[Guarded recursive specification]
A recursive specification

$$X_1=t_1(X_1,\cdots,X_n)$$
$$...$$
$$X_n=t_n(X_1,\cdots,X_n)$$

is guarded if the right-hand sides of its recursive equations can be adapted to the form by applications of the axioms in $APTC$ and replacing recursion variables by the right-hand
sides of their recursive equations,

$((u_{111}::a_{111}\leftmerge\cdots\leftmerge u_{11i_1}::a_{11i_1})\cdot s_1(X_1,\cdots,X_n)+\cdots+(u_{1k1}::a_{1k1}\leftmerge\cdots\leftmerge u_{1ki_k}::a_{1ki_k})\cdot s_k(X_1,\cdots,X_n)
+(v_{111}::b_{111}\leftmerge\cdots\leftmerge v_{11j_1}::b_{11j_1})+\cdots+(v_{11j_1}::b_{11j_1}\leftmerge\cdots\leftmerge v_{11j_l}::b_{1lj_l}))\boxplus_{\pi_1}\cdots\boxplus_{\pi_{m-1}}
((u_{m11}::a_{m11}\leftmerge\cdots\leftmerge u_{m1i_1}::a_{m1i_1})\cdot s_1(X_1,\cdots,X_n)+\cdots+(u_{mk1}::a_{mk1}\leftmerge\cdots\leftmerge u_{mki_k}::a_{mki_k})\cdot s_k(X_1,\cdots,X_n)
+(v_{m11}::b_{m11}\leftmerge\cdots\leftmerge v_{m1j_1}::b_{m1j_1})+\cdots+(v_{m1j_1}::b_{m1j_1}\leftmerge\cdots\leftmerge v_{m1j_l}::b_{mlj_l}))$

where $a_{111},\cdots,a_{11i_1},a_{1k1},\cdots,a_{1ki_k},b_{111},\cdots,b_{11j_1},b_{11j_1},\cdots,b_{1lj_l},\cdots, a_{m11},\cdots,a_{m1i_1},a_{1k1},\cdots,a_{mki_k},\\b_{111},\cdots,
b_{m1j_1},b_{m1j_1},\cdots,b_{mlj_l}\in \mathbb{E}$, and the sum above is allowed to be empty, in which case it represents the deadlock $\delta$.
\end{definition}

\begin{definition}[Linear recursive specification]\label{LRS}
A recursive specification is linear if its recursive equations are of the form

$((u_{111}::a_{111}\leftmerge\cdots\leftmerge u_{11i_1}::a_{11i_1})X_1+\cdots+(u_{1k1}::a_{1k1}\leftmerge\cdots\leftmerge u_{1ki_k}::a_{1ki_k})X_k
+(v_{111}::b_{111}\leftmerge\cdots\leftmerge v_{11j_1}::b_{11j_1})+\cdots+(v_{11j_1}::b_{11j_1}\leftmerge\cdots\leftmerge v_{11j_l}::b_{1lj_l}))\boxplus_{\pi_1}\cdots\boxplus_{\pi_{m-1}}
((u_{m11}::a_{m11}\leftmerge\cdots\leftmerge u_{m1i_1}::a_{m1i_1})X_1+\cdots+(u_{mk1}::a_{mk1}\leftmerge\cdots\leftmerge u_{mki_k}::a_{mki_k})X_k
+(v_{m11}::b_{m11}\leftmerge\cdots\leftmerge v_{m1j_1}::b_{m1j_1})+\cdots+(v_{m1j_1}::b_{m1j_1}\leftmerge\cdots\leftmerge v_{m1j_l}::b_{mlj_l}))$

where $a_{111},\cdots,a_{11i_1},a_{1k1},\cdots,a_{1ki_k},b_{111},\cdots,b_{11j_1},b_{11j_1},\cdots,b_{1lj_l},\cdots,a_{m11},\cdots,a_{m1i_1},a_{mk1},\cdots,a_{mki_k},\\b_{m11},\cdots,
b_{m1j_1},b_{m1j_1},\cdots,b_{mlj_l}\in \mathbb{E}$, and the sum above is allowed to be empty, in which case it
represents the deadlock $\delta$.
\end{definition}

\begin{center}
    \begin{table}
        $$\frac{\langle t_i(\langle X_1|E\rangle,\cdots,\langle X_n|E\rangle),\varrho\rangle\rightsquigarrow \langle y,\varrho\rangle}{\langle\langle X_i|E\rangle,\varrho\rangle\rightsquigarrow \langle y,\varrho\rangle}$$
        $$\frac{\langle t_i(\langle X_1|E\rangle,\cdots,\langle X_n|E\rangle),\varrho\rangle\xrightarrow[u]{\{e_1,\cdots,e_k\}}\langle\surd,\varrho'\rangle}{\langle\langle X_i|E\rangle,\varrho\rangle\xrightarrow[u]{\{e_1,\cdots,e_k\}}\langle\surd,\varrho'\rangle}$$
        $$\frac{\langle t_i(\langle X_1|E\rangle,\cdots,\langle X_n|E\rangle),\varrho\rangle\xrightarrow[u]{\{e_1,\cdots,e_k\}} \langle y,\varrho'\rangle}{\langle\langle X_i|E\rangle,\varrho\rangle\xrightarrow[u]{\{e_1,\cdots,e_k\}} \langle y,\varrho'\rangle}$$
        \caption{Transition rules of guarded recursion}
        \label{TRForGRG}
    \end{table}
\end{center}

\begin{theorem}[Conservitivity of $qAPPTC$ with localities and guarded recursion]
$qAPPTC$ with localities and guarded recursion is a conservative extension of $qAPPTC$ with localities.
\end{theorem}

\begin{proof}
Since the transition rules of $qAPPTC$ with localities are source-dependent, and the transition rules for guarded recursion in Table \ref{TRForGRG} contain only a fresh constant in their source, so
the transition rules of $qAPPTC$ with localities and guarded recursion are a conservative extension of those of $qAPPTC$ with localities.
\end{proof}

\begin{theorem}[Congruence theorem of $qAPPTC$ with localities and guarded recursion]
Probabilistic static location truly concurrent bisimulation equivalences $\sim_{pp}^{sl}$, $\sim_{ps}^{sl}$, $\sim_{php}^{sl}$ and $\sim_{phhp}^{sl}$ are all congruences with respect to $qAPPTC$ with localities and guarded recursion.
\end{theorem}

\begin{proof}
It follows the following two facts:
\begin{enumerate}
  \item in a guarded recursive specification, right-hand sides of its recursive equations can be adapted to the form by applications of the axioms in $qAPPTC$ with localities and replacing recursion
  variables by the right-hand sides of their recursive equations;
  \item probabilistic static location truly concurrent bisimulation equivalences $\sim_{pp}^{sl}$, $\sim_{ps}^{sl}$, $\sim_{php}^{sl}$ and $\sim_{phhp}^{sl}$ are all congruences with respect to all operators of $qAPPTC$ with localities.
\end{enumerate}
\end{proof}

\begin{theorem}[Elimination theorem of $qAPPTC$ with localities and linear recursion]\label{ETRecursionG}
Each process term in $qAPPTC$ with localities and linear recursion is equal to a process term $\langle X_1|E\rangle$ with $E$ a linear recursive specification.
\end{theorem}

\begin{proof}
The same as that of $APPTC^{sl}$, we omit the proof, please refer to \cite{LOC2} for details.
\end{proof}

\begin{theorem}[Soundness of $qAPPTC$ with localities and guarded recursion]\label{SAPTC_GRG}
Let $x$ and $y$ be $qAPPTC$ with localities and guarded recursion terms. If $qAPPTC^{sl}\textrm{ with guarded recursion}\vdash x=y$, then

(1) $x\sim_{ps}^{sl} y$.

(2) $x\sim_{pp}^{sl} y$.

(3) $x\sim_{php}^{sl} y$.

(4) $x\sim_{phhp}^{sl} y$.
\end{theorem}

\begin{proof}
(1) Since probabilistic static location step bisimulation $\sim_{ps}^{sl}$ is both an equivalent and a congruent relation with respect to $qAPPTC$ with localities and guarded recursion, we only need to check if each
axiom in Table \ref{RDPRSP} is sound modulo probabilistic static location step bisimulation equivalence. We leave them as exercises to the readers.

(2) Since probabilistic static location pomset bisimulation $\sim_{pp}^{sl}$ is both an equivalent and a congruent relation with respect to the guarded recursion, we only need to check if each axiom in
Table \ref{RDPRSP} is sound modulo probabilistic static location pomset bisimulation equivalence. We leave them as exercises to the readers.

(3) Since probabilistic static location hp-bisimulation $\sim_{php}^{sl}$ is both an equivalent and a congruent relation with respect to guarded recursion, we only need to check if each axiom in Table
\ref{RDPRSP} is sound modulo probabilistic static location hp-bisimulation equivalence. We leave them as exercises to the readers.

(4) Since probabilistic static location hhp-bisimulation $\sim_{phhp}^{sl}$ is both an equivalent and a congruent relation with respect to guarded recursion, we only need to check if each axiom in Table
\ref{RDPRSP} is sound modulo probabilistic static location hhp-bisimulation equivalence. We leave them as exercises to the readers.
\end{proof}

\begin{theorem}[Completeness of $qAPPTC$ with localities and linear recursion]\label{CAPTC_GRG}
Let $p$ and $q$ be closed $qAPPTC$ with localities and linear recursion terms, then,

(1) if $p\sim_{ps}^{sl} q$ then $p=q$.

(2) if $p\sim_{pp}^{sl} q$ then $p=q$.

(3) if $p\sim_{php}^{sl} q$ then $p=q$.

(4) if $p\sim_{phhp}^{sl} q$ then $p=q$.
\end{theorem}

\begin{proof}
According to the definition of probabilistic static location truly concurrent bisimulation equivalences $\sim_{pp}^{sl}$, $\sim_{ps}^{sl}$, $\sim_{php}^{sl}$ and $\sim_{phhp}^{sl}$, $p\sim_{pp}^{sl}q$, $p\sim_{ps}^{sl}q$, $p\sim_{php}^{sl}q$ and $p\sim_{phhp}^{sl}q$ implies
both the bisimilarities between $p$ and $q$, and also the in the same quantum states. According to the completeness of $APPTC^{sl}$ with linear recursion (please refer to \cite{LOC2} for details), we can get the
completeness of $qAPPTC^{sl}$ with linear recursion.
\end{proof}

\subsection{Abstraction}{\label{qcabs}}

To abstract away from the internal implementations of a program, and verify that the program exhibits the desired external behaviors, the silent step $\tau$ and abstraction operator
$\tau_I$ are introduced, where $I\subseteq \mathbb{E}\cup G_{at}$ denotes the internal events or guards. The silent step $\tau$ represents the internal events or guards, when we
consider the external behaviors of a process, $\tau$ steps can be removed, that is, $\tau$ steps must keep silent. The transition rule of $\tau$ is shown in Table \ref{TRForqTau2}. In
the following, let the atomic event $e$ range over $\mathbb{E}\cup\{\epsilon\}\cup\{\delta\}\cup\{\tau\}$, and $\phi$ range over $G\cup \{\tau\}$, and let the communication function
$\gamma:\mathbb{E}\cup\{\tau\}\times \mathbb{E}\cup\{\tau\}\rightarrow \mathbb{E}\cup\{\delta\}$, with each communication involved $\tau$ resulting in $\delta$. We use $\tau(\varrho)$ to
denote $effect(\tau,\varrho)$, for the fact that $\tau$ only change the state of internal data environment, that is, for the external data environments, $\varrho=\tau(\varrho)$.

\begin{center}
    \begin{table}
        $$\frac{}{\tau\rightsquigarrow\breve{\tau}}$$
%        $$\frac{}{\langle\tau,\varrho\rangle\rightarrow\langle\surd,\varrho\rangle}\textrm{ if }test(\tau,\varrho)$$
        $$\frac{}{\langle\breve{\tau},\varrho\rangle\xrightarrow{\tau}\langle\surd,\tau(\varrho)\rangle}$$
        \caption{Transition rule of the silent step}
        \label{TRForqTau2}
    \end{table}
\end{center}

\begin{theorem}[Conservitivity of $qAPPTC$ with localities and silent step and guarded linear recursion]
$qAPPTC$ with localities and silent step and guarded linear recursion is a conservative extension of $qAPPTC$ with localities and linear recursion.
\end{theorem}

\begin{proof}
Since the transition rules of $qAPPTC$ with localities and linear recursion are source-dependent, and the transition rules for silent step in Table \ref{TRForqTau2} contain only a fresh constant
$\tau$ in their source, so the transition rules of $qAPPTC$ with localities and silent step and guarded linear recursion is a conservative extension of those of $qAPPTC$ with localities and linear recursion.
\end{proof}

\begin{theorem}[Congruence theorem of $qAPPTC$ with localities and silent step and guarded linear recursion]
Probabilistic static location rooted branching truly concurrent bisimulation equivalences $\approx_{prbp}^{sl}$, $\approx_{prbs}^{sl}$, $\approx_{prbhp}^{sl}$ and $\approx_{rbhhp}$ are all congruences with respect
to $qAPPTC$ with localities and silent step and guarded linear recursion.
\end{theorem}

\begin{proof}
It follows the following three facts:
\begin{enumerate}
  \item in a guarded linear recursive specification, right-hand sides of its recursive equations can be adapted to the form by applications of the axioms in $qAPPTC$ with localities and replacing
  recursion variables by the right-hand sides of their recursive equations;
  \item probabilistic static location truly concurrent bisimulation equivalences $\sim_{pp}^{sl}$, $\sim_{ps}^{sl}$, $\sim_{php}^{sl}$ and $\sim_{phhp}^{sl}$ are all congruences with respect to all operators of
  $qAPPTC$ with localities, while probabilistic static location truly concurrent bisimulation equivalences $\sim_{pp}^{sl}$, $\sim_{ps}^{sl}$, $\sim_{php}^{sl}$ and $\sim_{phhp}^{sl}$ imply the corresponding probabilistic rooted
  branching truly concurrent bisimulations $\approx_{prbp}^{sl}$, $\approx_{prbs}^{sl}$, $\approx_{prbhp}^{sl}$ and $\approx_{prbhhp}^{sl}$, so probabilistic static location rooted branching truly concurrent
  bisimulations $\approx_{prbp}^{sl}$, $\approx_{prbs}^{sl}$, $\approx_{prbhp}^{sl}$ and $\approx_{prbhhp}^{sl}$ are all congruences with respect to all operators of $qAPPTC$ with localities;
  \item While $\mathbb{E}$ is extended to $\mathbb{E}\cup\{\tau\}$, and $G$ is extended to $G\cup\{\tau\}$, it can be proved that probabilistic static location rooted branching truly concurrent
  bisimulations $\approx_{prbp}^{sl}$, $\approx_{prbs}^{sl}$, $\approx_{prbhp}^{sl}$ and $\approx_{prbhhp}^{sl}$ are all congruences with respect to all operators of $qAPPTC$ with localities, we omit it.
\end{enumerate}
\end{proof}

We design the axioms for the silent step $\tau$ in Table \ref{AxiomsForqTau2}.

\begin{center}
\begin{table}
  \begin{tabular}{@{}ll@{}}
  \hline No. &Axiom\\
  $B1$ & $(y=y+y,z=z+z)\quad x\cdot((y+\tau\cdot(y+z))\boxplus_{\pi}w)=x\cdot((y+z)\boxplus_{\pi}w)$\\
  $B2$ & $(y=y+y,z=z+z)\quad x\leftmerge((y+\tau\leftmerge(y+z))\boxplus_{\pi}w)=x\leftmerge((y+z)\boxplus_{\pi}w)$\\
  $L13$ & $u::\tau=\tau$\\
\end{tabular}
\caption{Axioms of silent step}
\label{AxiomsForqTau2}
\end{table}
\end{center}

\begin{theorem}[Elimination theorem of $qAPPTC$ with localities and silent step and guarded linear recursion]\label{ETTauG}
Each process term in $qAPPTC$ with localities and silent step and guarded linear recursion is equal to a process term $\langle X_1|E\rangle$ with $E$ a guarded linear recursive specification.
\end{theorem}

\begin{proof}
The same as that of $APPTC^{sl}$, we omit the proof, please refer to \cite{LOC2} for details.
\end{proof}

\begin{theorem}[Soundness of $qAPPTC$ with localities and silent step and guarded linear recursion]\label{SAPTC_GTAUG}
Let $x$ and $y$ be $qAPPTC$ with localities and silent step and guarded linear recursion terms. If $qAPPTC$ with localities and silent step and guarded linear recursion $\vdash x=y$, then

(1) $x\approx_{prbs}^{sl} y$.

(2) $x\approx_{prbp}^{sl} y$.

(3) $x\approx_{prbhp}^{sl} y$.

(4) $x\approx_{prbhhp}^{sl} y$.
\end{theorem}

\begin{proof}
(1) Since probabilistic rooted branching static location step bisimulation $\approx_{prbs}^{sl}$ is both an equivalent and a congruent relation with respect to $qAPPTC$ with localities and silent step and guarded
linear recursion, we only need to check if each axiom in Table \ref{AxiomsForqTau2} is sound modulo probabilistic rooted branching static location step bisimulation $\approx_{prbs}^{sl}$. We leave them as
exercises to the readers.

(2) Since probabilistic rooted branching static location pomset bisimulation $\approx_{prbp}^{sl}$ is both an equivalent and a congruent relation with respect to $qAPPTC$ with localities and silent step and guarded
linear recursion, we only need to check if each axiom in Table \ref{AxiomsForqTau2} is sound modulo probabilistic rooted branching static location pomset bisimulation $\approx_{prbp}^{sl}$. We leave them
as exercises to the readers.

(3) Since probabilistic rooted branching static location hp-bisimulation $\approx_{prbhp}^{sl}$ is both an equivalent and a congruent relation with respect to $qAPPTC$ with localities and silent step and guarded linear
recursion, we only need to check if each axiom in Table \ref{AxiomsForqTau2} is sound modulo probabilistic rooted branching static location hp-bisimulation $\approx_{prbhp}^{sl}$. We leave them as exercises
to the readers.

(4) Since probabilistic rooted branching static location hhp-bisimulation $\approx_{prbhhp}^{sl}$ is both an equivalent and a congruent relation with respect to $qAPPTC$ with localities and silent step and guarded linear
recursion, we only need to check if each axiom in Table \ref{AxiomsForqTau2} is sound modulo probabilistic rooted branching static location hhp-bisimulation $\approx_{prbhhp}^{sl}$. We leave them as exercises
to the readers.
\end{proof}

\begin{theorem}[Completeness of $qAPPTC$ with localities and silent step and guarded linear recursion]\label{CAPTC_GTAUG}
Let $p$ and $q$ be closed $qAPPTC$ with localities and silent step and guarded linear recursion terms, then,

(1) if $p\approx_{prbs}^{sl} q$ then $p=q$.

(2) if $p\approx_{prbp}^{sl} q$ then $p=q$.

(3) if $p\approx_{prbhp}^{sl} q$ then $p=q$.

(3) if $p\approx_{prbhhp}^{sl} q$ then $p=q$.
\end{theorem}

\begin{proof}
According to the definition of probabilistic static location rooted branching truly concurrent bisimulation equivalences $\approx_{prbp}^{sl}$, $\approx_{prbs}^{sl}$, $\approx_{prbhp}^{sl}$ and $\approx_{prbhhp}^{sl}$, and $\approx_{prbp}^{sl}$, $\approx_{prbs}^{sl}$, $\approx_{prbhp}^{sl}$ and $\approx_{prbhhp}^{sl}$ implies
both the bisimilarities between $p$ and $q$, and also the in the same quantum states. According to the completeness of $APPTC^{sl}$ with silent step and guarded linear recursion (please refer to \cite{LOC2} for details), we can get the
completeness of $qAPPTC^{sl}$ with silent step and guarded linear recursion.
\end{proof}

The unary abstraction operator $\tau_I$ ($I\subseteq \mathbb{E}\cup G_{at}$) renames all atomic events or atomic guards in $I$ into $\tau$. $qAPPTC$ with localities and silent step and abstraction
operator is called $qAPPTC^{sl}_{\tau}$. The transition rules of operator $\tau_I$ are shown in Table \ref{TRForqAbstraction2}.

\begin{center}
    \begin{table}
        $$\frac{\langle x,\varrho\rangle\rightsquigarrow \langle x',\varrho\rangle}{\langle \tau_I(x),\varrho\rangle\rightsquigarrow\langle\tau_I(x'),\varrho\rangle}$$
        $$\frac{\langle x,\varrho\rangle\xrightarrow[u]{e}\langle\surd,\varrho'\rangle}{\langle\tau_I(x),\varrho\rangle\xrightarrow[u]{e}\langle\surd,\varrho'\rangle}\quad e\notin I
        \quad\quad\frac{\langle x,\varrho\rangle\xrightarrow[u]{e}\langle x',\varrho'\rangle}{\langle\tau_I(x),\varrho\rangle\xrightarrow[u]{e}\langle\tau_I(x'),\varrho'\rangle}\quad e\notin I$$

        $$\frac{\langle x,\varrho\rangle\xrightarrow[u]{e}\langle\surd,\varrho'\rangle}{\langle\tau_I(x),\varrho\rangle\xrightarrow{\tau}\langle\surd,\tau(\varrho)\rangle}\quad e\in I
        \quad\quad\frac{\langle x,\varrho\rangle\xrightarrow[u]{e}\langle x',\varrho'\rangle}{\langle\tau_I(x),\varrho\rangle\xrightarrow{\tau}\langle\tau_I(x'),\tau(\varrho)\rangle}\quad e\in I$$
        \caption{Transition rule of the abstraction operator}
        \label{TRForqAbstraction2}
    \end{table}
\end{center}

\begin{theorem}[Conservitivity of $qAPPTC^{sl}_{\tau}$ with guarded linear recursion]
$qAPPTC^{sl}_{\tau}$ with guarded linear recursion is a conservative extension of $qAPPTC$ with localities and silent step and guarded linear recursion.
\end{theorem}

\begin{proof}
Since the transition rules of $qAPPTC$ with localities and silent step and guarded linear recursion are source-dependent, and the transition rules for abstraction operator in Table
\ref{TRForqAbstraction2} contain only a fresh operator $\tau_I$ in their source, so the transition rules of $qAPPTC^{sl}_{\tau}$ with guarded linear recursion is a conservative extension
of those of $qAPPTC$ with localities and silent step and guarded linear recursion.
\end{proof}

\begin{theorem}[Congruence theorem of $qAPPTC^{sl}_{\tau}$ with guarded linear recursion]
Probabilistic static location rooted branching truly concurrent bisimulation equivalences $\approx_{prbp}^{sl}$, $\approx_{prbs}^{sl}$, $\approx_{prbhp}^{sl}$ and $\approx_{prbhhp}^{sl}$ are all congruences with respect
to $qAPPTC^{sl}_{\tau}$ with guarded linear recursion.
\end{theorem}

\begin{proof}
(1) It is easy to see that probabilistic rooted branching static location pomset bisimulation is an equivalent relation on $qAPPTC^{sl}_{\tau}$ with guarded linear recursion terms, we only need to
prove that $\approx_{prbp}^{sl}$ is preserved by the operators $\tau_I$. It is trivial and we leave the proof as an exercise for the readers.

(2) It is easy to see that probabilistic rooted branching static location step bisimulation is an equivalent relation on $qAPPTC^{sl}_{\tau}$ with guarded linear recursion terms, we only need to
prove that $\approx_{prbs}^{sl}$ is preserved by the operators $\tau_I$. It is trivial and we leave the proof as an exercise for the readers.

(3) It is easy to see that probabilistic rooted branching static location hp-bisimulation is an equivalent relation on $qAPPTC^{sl}_{\tau}$ with guarded linear recursion terms, we only need to
prove that $\approx_{prbhp}^{sl}$ is preserved by the operators $\tau_I$. It is trivial and we leave the proof as an exercise for the readers.

(4) It is easy to see that probabilistic rooted branching static location hhp-bisimulation is an equivalent relation on $qAPPTC^{sl}_{\tau}$ with guarded linear recursion terms, we only need to
prove that $\approx_{prbhhp}^{sl}$ is preserved by the operators $\tau_I$. It is trivial and we leave the proof as an exercise for the readers.
\end{proof}

We design the axioms for the abstraction operator $\tau_I$ in Table \ref{AxiomsForqAbstraction2}.

\begin{center}
\begin{table}
  \begin{tabular}{@{}ll@{}}
\hline No. &Axiom\\
  $TI1$ & $e\notin I\quad \tau_I(e)=e$\\
  $TI2$ & $e\in I\quad \tau_I(e)=\tau$\\
  $TI3$ & $\tau_I(\delta)=\delta$\\
  $TI4$ & $\tau_I(x+y)=\tau_I(x)+\tau_I(y)$\\
  $PTI1$ & $\tau_I(x\boxplus_{\pi}y)=\tau_I(x)\boxplus_{\pi}\tau_I(y)$\\
  $TI5$ & $\tau_I(x\cdot y)=\tau_I(x)\cdot\tau_I(y)$\\
  $TI6$ & $\tau_I(x\leftmerge y)=\tau_I(x)\leftmerge\tau_I(y)$\\
  $L14$ & $u::\tau_I(x)=\tau_I(u::x)$\\
  $L15$ & $e\notin I\quad \tau_I(u::e)=u::e$\\
  $L16$ & $e\in I\quad \tau_I(u::e)=\tau$\\
  $PTI1$ & $\tau_I(x\boxplus_{\pi}y)=\tau_I(x)\boxplus_{\pi}\tau_I(y)$\\
\end{tabular}
\caption{Axioms of abstraction operator}
\label{AxiomsForqAbstraction2}
\end{table}
\end{center}

\begin{theorem}[Soundness of $qAPPTC^{sl}_{\tau}$ with guarded linear recursion]
Let $x$ and $y$ be $qAPPTC^{sl}_{\tau}$ with guarded linear recursion terms. If $qAPPTC^{sl}_{\tau}$ with guarded linear recursion $\vdash x=y$, then

(1) $x\approx_{prbs}^{sl} y$.

(2) $x\approx_{prbp}^{sl} y$.

(3) $x\approx_{prbhp}^{sl} y$.

(4) $x\approx_{prbhhp}^{sl} y$.
\end{theorem}

\begin{proof}
(1) Since probabilistic rooted branching static location step bisimulation $\approx_{prbs}^{sl}$ is both an equivalent and a congruent relation with respect to $qAPPTC^{sl}_{\tau}$ with guarded linear
recursion, we only need to check if each axiom in Table \ref{AxiomsForqAbstraction2} is sound modulo probabilistic rooted branching static location step bisimulation $\approx_{prbs}^{sl}$. We leave them as
exercises to the readers.

(2) Since probabilistic rooted branching static location pomset bisimulation $\approx_{prbp}^{sl}$ is both an equivalent and a congruent relation with respect to $qAPPTC^{sl}_{\tau}$ with guarded linear
recursion, we only need to check if each axiom in Table \ref{AxiomsForqAbstraction2} is sound modulo probabilistic rooted branching static location pomset bisimulation $\approx_{prbp}^{sl}$. We leave them
as exercises to the readers.

(3) Since probabilistic rooted branching static location hp-bisimulation $\approx_{prbhp}^{sl}$ is both an equivalent and a congruent relation with respect to $qAPPTC^{sl}_{\tau}$ with guarded linear
recursion, we only need to check if each axiom in Table \ref{AxiomsForqAbstraction2} is sound modulo probabilistic rooted branching static location hp-bisimulation $\approx_{prbhp}^{sl}$. We leave them as
exercises to the readers.

(4) Since probabilistic rooted branching static location hhp-bisimulation $\approx_{prbhhp}^{sl}$ is both an equivalent and a congruent relation with respect to $qAPPTC^{sl}_{\tau}$ with guarded linear
recursion, we only need to check if each axiom in Table \ref{AxiomsForqAbstraction2} is sound modulo probabilistic rooted branching static location hhp-bisimulation $\approx_{prbhhp}^{sl}$. We leave them as
exercises to the readers.
\end{proof}

Though $\tau$-loops are prohibited in guarded linear recursive specifications in a specifiable way, they can be constructed using the abstraction operator, for example, there exist
$\tau$-loops in the process term $\tau_{\{a\}}(\langle X|X=aX\rangle)$. To avoid $\tau$-loops caused by $\tau_I$ and ensure fairness, we introduce the following recursive verification
rules as Table \ref{RVR} shows, note that $i_1,\cdots, i_m,j_1,\cdots,j_n\in I\subseteq \mathbb{E}\setminus\{\tau\}$.

\begin{center}
\begin{table}
    $$VR_1\quad \frac{x=y+(u_1::i_1\leftmerge\cdots\leftmerge u_m::i_m)\cdot x, y=y+y}{\tau\cdot\tau_I(x)=\tau\cdot \tau_I(y)}$$
    $$VR_2\quad \frac{x=z\boxplus_{\pi}(u+(u_1::i_1\leftmerge\cdots\leftmerge u_m::i_m)\cdot x),z=z+u,z=z+z}{\tau\cdot\tau_I(x)=\tau\cdot\tau_I(z)}$$
    $$VR_3\quad \frac{x=z+(u_1::i_1\leftmerge\cdots\leftmerge u_m::i_m)\cdot y,y=z\boxplus_{\pi}(u+(v_1::j_1\leftmerge\cdots\leftmerge v_n::j_n)\cdot x), z=z+u,z=z+z}{\tau\cdot\tau_I(x)=\tau\cdot\tau_I(y')\textrm{ for }y'=z\boxplus_{\pi}(u+(u_1::i_1\leftmerge\cdots\leftmerge u_m::i_m)\cdot y')}$$
\caption{Recursive verification rules}
\label{RVR}
\end{table}
\end{center}

\begin{theorem}[Soundness of $VR_1,VR_2,VR_3$]
$VR_1$, $VR_2$ and $VR_3$ are sound modulo probabilistic static location rooted branching truly concurrent bisimulation equivalences $\approx_{prbp}^{sl}$, $\approx_{prbs}^{sl}$, $\approx_{prbhp}^{sl}$ and $\approx_{prbhhp}^{sl}$.
\end{theorem}

\subsection{Quantum Measurement}\label{qm}

In closed quantum systems, there is another basic quantum operation -- quantum measurement, besides the unitary operator. Quantum measurements have a probabilistic nature.

There is a concrete but non-trivial problem in modeling quantum measurement.

Let the following process term represent quantum measurement during modeling phase,

$$\beta_1\cdot t_1\boxplus_{\pi_1}\beta_2\cdot t_2\boxplus_{\pi_2}\cdots\boxplus_{\pi_{i-1}}\beta_i\cdot t_i$$

where $\sum_i \pi_i=1$, $t_i\in\mathcal{B}(qBAPTC)$, $\beta$ denotes a quantum measurement, and $\beta=\sum_i\lambda_i \beta_i$, $\beta_i$ denotes the projection performed on the quantum
system $\varrho$, $\pi_i=Tr(\beta_i\varrho)$, $\varrho_i=\beta_i\varrho \beta_i/\pi_i$.

The above term means that, firstly, we choose a projection $\beta_i$ in a quantum measurement $\beta=\sum_i\lambda_i\beta_i$ probabilistically, then, we execute (perform) the
projection $\beta_i$ on the closed quantum system. This also adheres to the intuition on quantum mechanics.

We define $B$ as the collection of all projections of all quantum measurements, and make the collection of atomic actions be $\mathbb{E}=\mathbb{E}\cup B$. We see that a
projection $\beta_i\in B$ has the almost same semantics as a unitary operator $\alpha\in A$. So, we add the following (probabilistic and action)
transition rules into those of $PQRA$.

$$\frac{}{\langle\beta_i,\varrho\rangle\rightsquigarrow\langle\breve{\beta_i},\varrho\rangle}$$

$$\frac{}{\langle\breve{\beta_i},\varrho\rangle\xrightarrow[u]{\beta_i}\langle\surd,\varrho'\rangle}$$

Until now, $qAPPTC$ with localities works again. The two main quantum operations in a closed quantum system -- the unitary operator and the quantum measurement, are fully modeled in probabilistic
process algebra.

\subsection{Quantum Entanglement}\label{qe2}

As in section \ref{qe1}, The axiom system of the shadow constant $\circledS$ is shown in Table \ref{AxiomsForQE2}.

\begin{center}
\begin{table}
  \begin{tabular}{@{}ll@{}}
\hline No. &Axiom\\
  $SC1$ & $\circledS\cdot x = x$ \\
  $SC2$ & $x\cdot\circledS = x$\\
  $SC3$ & $e\leftmerge\circledS^e=e$\\
  $SC4$ & $\circledS^e\leftmerge e=e$\\
  $SC5$ & $e\leftmerge(\circledS^e\cdot y) = e\cdot y$\\
  $SC6$ & $\circledS^e\leftmerge(e\cdot y) = e\cdot y$\\
  $SC7$ & $(e\cdot x)\leftmerge\circledS^e = e\cdot x$\\
  $SC8$ & $(\circledS^e\cdot x)\leftmerge e = e\cdot x$\\
  $SC9$ & $(e\cdot x)\leftmerge(\circledS^e\cdot y) = e\cdot (x\between y)$\\
  $SC10$ & $(\circledS^e\cdot x)\leftmerge(e\cdot y) = e\cdot (x\between y)$\\
  $L17$ & $loc::\circledS = \circledS$\\
\end{tabular}
\caption{Axioms of quantum entanglement}
\label{AxiomsForQE2}
\end{table}
\end{center}

The transition rules of constant $\circledS$ are as Table \ref{TRForENT2} shows.

\begin{center}
    \begin{table}
        $$\frac{}{\langle\circledS,\varrho\rangle\rightsquigarrow\langle\breve{\circledS},\varrho\rangle}$$
        $$\frac{}{\langle\circledS,\varrho\rangle\rightarrow\langle\surd,\varrho\rangle}$$
        $$\frac{\langle x, \varrho\rangle\xrightarrow[u]{e}\langle x',\varrho'\rangle\quad \langle y, \varrho'\rangle\xrightarrow{\circledS^e}\langle y',\varrho'\rangle}{\langle x\leftmerge y,\varrho\rangle\xrightarrow[u]{e}\langle x'\between y', \varrho'\rangle}$$
        $$\frac{\langle x, \varrho\rangle\xrightarrow[u]{e}\langle\surd,\varrho'\rangle\quad \langle y, \varrho'\rangle\xrightarrow{\circledS^e}\langle y',\varrho'\rangle}{\langle x\leftmerge y,\varrho\rangle\xrightarrow[u]{e}\langle y', \varrho'\rangle}$$
        $$\frac{\langle x, \varrho'\rangle\xrightarrow{\circledS^e}\langle\surd,\varrho'\rangle\quad \langle y, \varrho\rangle\xrightarrow[u]{e}\langle y',\varrho'\rangle}{\langle x\leftmerge y,\varrho\rangle\xrightarrow[u]{e}\langle y', \varrho'\rangle}$$
        $$\frac{\langle x, \varrho\rangle\xrightarrow[u]{e}\langle\surd,\varrho'\rangle\quad \langle y, \varrho'\rangle\xrightarrow{\circledS^e}\langle\surd,\varrho'\rangle}{\langle x\leftmerge y,\varrho\rangle\xrightarrow[u]{e}\langle \surd, \varrho'\rangle}$$
        \caption{Transition rules of constant $\circledS$}
        \label{TRForENT2}
    \end{table}
\end{center}

\begin{theorem}[Elimination theorem of $qAPPTC^{sl}_{\tau}$ with guarded linear recursion and shadow constant]
Let $p$ be a closed $qAPPTC^{sl}_{\tau}$ with guarded linear recursion and shadow constant term. Then there is a closed $qAPPTC$ with localities term such that $qAPPTC^{sl}_{\tau}$ with guarded linear recursion and shadow constant$\vdash p=q$.
\end{theorem}

\begin{proof}
We leave the proof to the readers as an excise.
\end{proof}

\begin{theorem}[Conservitivity of $qAPPTC^{sl}_{\tau}$ with guarded linear recursion and shadow constant]
$qAPPTC^{sl}_{\tau}$ with guarded linear recursion and shadow constant is a conservative extension of $qAPPTC^{sl}_{\tau}$ with guarded linear recursion.
\end{theorem}

\begin{proof}
We leave the proof to the readers as an excise.
\end{proof}

\begin{theorem}[Congruence theorem of $qAPPTC^{sl}_{\tau}$ with guarded linear recursion and shadow constant]
Probabilistic static location rooted branching truly concurrent bisimulation equivalences $\approx_{prbp}^{sl}$, $\approx_{prbs}^{sl}$, $\approx_{prbhp}^{sl}$ and $\approx_{prbhhp}^{sl}$ are all congruences with respect
to $qAPPTC^{sl}_{\tau}$ with guarded linear recursion and shadow constant.
\end{theorem}

\begin{proof}
We leave the proof to the readers as an excise.
\end{proof}

\begin{theorem}[Soundness of $qAPPTC^{sl}_{\tau}$ with guarded linear recursion and shadow constant]
Let $x$ and $y$ be closed $qAPPTC^{sl}_{\tau}$ with guarded linear recursion and shadow constant terms. If $qAPPTC^{sl}_{\tau}$ with guarded linear recursion and shadow constant$\vdash x=y$, then

\begin{enumerate}
  \item $x\approx_{prbs}^{sl} y$;
  \item $x\approx_{prbp}^{sl} y$;
  \item $x\approx_{prbhp}^{sl} y$;
  \item $x\approx_{prbhhp}^{sl} y$.
\end{enumerate}
\end{theorem}

\begin{proof}
We leave the proof to the readers as an excise.
\end{proof}

\begin{theorem}[Completeness of $qAPPTC^{sl}_{\tau}$ with guarded linear recursion and shadow constant]
Let $p$ and $q$ are closed $qAPPTC^{sl}_{\tau}$ with guarded linear recursion and shadow constant terms, then,

\begin{enumerate}
  \item if $p\approx_{prbs}^{sl} q$ then $p=q$;
  \item if $p\approx_{prbp}^{sl} q$ then $p=q$;
  \item if $p\approx_{prbhp}^{sl} q$ then $p=q$;
  \item if $p\approx_{prbhhp}^{sl} q$ then $p=q$.
\end{enumerate}
\end{theorem}

\begin{proof}
We leave the proof to the readers as an excise.
\end{proof}

\subsection{Unification of Quantum and Classical Computing for Closed Quantum Systems}\label{uni2}

We give the transition rules under quantum configuration for traditional atomic actions (events) $e'\in\mathbb{E}$ as Table \ref{TRForBPA5} shows.

\begin{center}
    \begin{table}
%        $$\frac{}{\langle\epsilon,\varrho\rangle\rightsquigarrow\langle\breve{\epsilon},\varrho\rangle}$$
        $$\frac{}{\langle e',\varrho\rangle\rightsquigarrow\langle\breve{e},\varrho\rangle}$$
%        $$\frac{}{\langle\phi,\varrho\rangle\rightsquigarrow\langle\breve{\phi},\varrho\rangle}$$
        $$\frac{\langle x,\varrho\rangle\rightsquigarrow \langle x',\varrho\rangle}{\langle x\cdot y,\varrho\rangle\rightsquigarrow \langle x'\cdot y,\varrho\rangle}$$
        $$\frac{\langle x,\varrho\rangle\rightsquigarrow \langle x',\varrho\rangle\quad \langle y,\varrho\rangle\rightsquigarrow \langle y',\varrho\rangle}{\langle x+y,\varrho\rangle\rightsquigarrow \langle x'+y',\varrho\rangle}$$
        $$\frac{\langle x,\varrho\rangle\rightsquigarrow \langle x',\varrho\rangle}{\langle x\boxplus_{\pi}y,\varrho\rangle\rightsquigarrow \langle x',\varrho\rangle}\quad \frac{\langle y,\varrho\rangle\rightsquigarrow \langle y',\varrho\rangle}{\langle x\boxplus_{\pi}y,\varrho\rangle\rightsquigarrow \langle y',\varrho\rangle}$$

        $$\frac{}{\langle e',\varrho\rangle\xrightarrow[\epsilon]{e'}\langle\surd,\varrho\rangle}\quad \frac{}{\langle loc::e',\varrho\rangle\xrightarrow[loc]{e'}\langle\surd,\varrho'\rangle}$$
        $$\frac{\langle x,\varrho\rangle\xrightarrow[u]{e'}\langle x',\varrho'\rangle}{\langle loc::x,\varrho\rangle\xrightarrow[loc\ll u]{e'}\langle loc::x',\varrho'\rangle}$$
        $$\frac{\langle x,\varrho\rangle\xrightarrow[u]{e'}\langle\surd,\varrho\rangle}{\langle x+y,\varrho\rangle\xrightarrow[u]{e'}\langle\surd,\varrho\rangle}$$
        $$\frac{\langle x,\varrho\rangle\xrightarrow[u]{e'}\langle x',\varrho\rangle}{\langle x+y,\varrho\rangle\xrightarrow[u]{e'}\langle x',\varrho\rangle}$$
        $$\frac{\langle y,\varrho\rangle\xrightarrow[u]{e'}\langle\surd,\varrho\rangle}{\langle x+y,\varrho\rangle\xrightarrow[u]{e'}\langle\surd,\varrho\rangle}$$
        $$\frac{\langle y,\varrho\rangle\xrightarrow[u]{e'}\langle y',\varrho\rangle}{\langle x+y,\varrho\rangle\xrightarrow[u]{e'}\langle y',\varrho\rangle}$$
        $$\frac{\langle x,\varrho\rangle\xrightarrow[u]{e'}\langle\surd,\varrho\rangle}{\langle x\cdot y,\varrho\rangle\xrightarrow[u]{e'}\langle y,\varrho\rangle}$$
        $$\frac{\langle x,\varrho\rangle\xrightarrow[u]{e'}\langle x',\varrho\rangle}{\langle x\cdot y,\varrho\rangle\xrightarrow[u]{e'}\langle x'\cdot y,\varrho\rangle}$$
        \caption{Transition rules of BAPTC under quantum configuration}
        \label{TRForBPA5}
    \end{table}
\end{center}

And the axioms for traditional actions are the same as those of $qBAPTC^{sl}$. And it is natural can be extended to $qAPPTC^{sl}$, recursion and abstraction. So, quantum and classical computing
are unified under the framework of $qAPPTC^{sl}$ for closed quantum systems.

\newpage\section{Applications of qAPPTC with Localities}\label{aqapptcl}

Quantum and classical computing in closed systems are unified with qAPPTC with localities, which have the same equational logic and the same quantum configuration based operational semantics.
The unification can be used widely in verification for the behaviors of quantum and classical computing mixed systems with distribution characteristics. In this chapter, we show its usage in verification of the
distributed quantum communication protocols.

\subsection{Verification of Quantum Teleportation Protocol}\label{VQT6}

Quantum teleportation \cite{QT} is a famous quantum protocol in quantum information theory to teleport an unknown quantum state by sending only classical information, provided that the sender and the receiver, Alice and Bob, shared an entangled state in advance. Firstly, we introduce the basic quantum teleportation protocol briefly, which is illustrated in Figure \ref{QT}. In this section, we show how to process quantum entanglement in an implicit way.

\begin{enumerate}
  \item EPR generates 2-qubits entangled EPR pair $q=q_1\otimes q_2$, and he sends $q_1$ to Alice through quantum channel $Q_A$ and $q_2$ to Bob through quantum channel $Q_B$;
  \item Alice receives $q_1$, after some preparations, she measures on $q_1$, and sends the measurement results $x$ to Bob through classical channel $P$;
  \item Bob receives $q_2$ from EPR, and also the classical information $x$ from Alice. According to $x$, he chooses specific Pauli transformation on $q_2$.
\end{enumerate}

\begin{figure}
  \centering
  %\vspace{5cm}
  \includegraphics{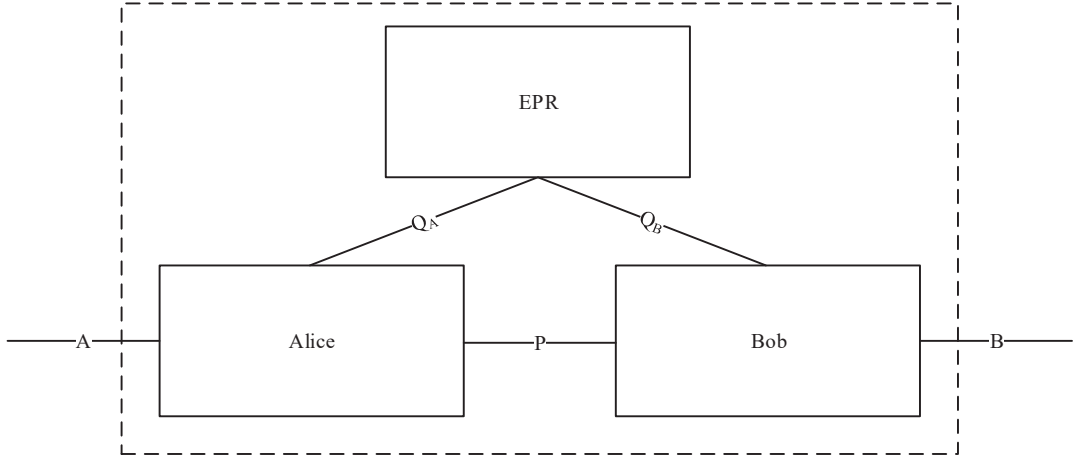}
  \caption{Quantum teleportation protocol.}
  \label{QT}
\end{figure}

We re-introduce the basic quantum teleportation protocol in an abstract way with more technical details as Figure \ref{QT} illustrates.

Now, we assume the generation of 2-qubits $q$ through two unitary operators $Set[q]$ and $H[q]$. EPR sends $q_1$ to Alice through the quantum channel $Q_A$ by quantum communicating action $send_{Q_A}(q_1)$ and Alice receives $q_1$ through $Q_A$ by quantum communicating action $receive_{Q_A}(q_1)$. Similarly, for Bob, those are $send_{Q_B}(q_2)$ and $receive_{Q_B}(q_2)$. After Alice receives $q_1$, she does some preparations, including a unitary transformation $CNOT$ and a Hadamard transformation $H$, then Alice do measurement $M=\sum^3_{i=0}M_i$, and sends measurement results $x$ to Bob through the public classical channel $P$ by classical communicating action $send_{P}(x)$, and Bob receives $x$ through channel $P$ by classical communicating action $receive_{P}(x)$. According to $x$, Bob performs specific Pauli transformations $\sigma_x$ on $q_2$. Let Alice, Bob and EPR be a system $ABE$ and let interactions between Alice, Bob and EPR be internal actions. $ABE$ receives external input $D_i$ through channel $A$ by communicating action $receive_A(D_i)$ and sends results $D_o$ through channel $B$ by communicating action $send_B(D_o)$. Note that the entangled EPR pair $q=q_1\otimes q_2$ is within $ABE$, so quantum entanglement can be processed implicitly.

Then the state transitions of EPR can be described as follows.

\begin{eqnarray}
&&E=loc_E::(Set[q]\cdot E_1)\nonumber\\
&&E_1=H[q]\cdot E_2\nonumber\\
&&E_2=send_{Q_A}(q_1)\cdot E_3\nonumber\\
&&E_3=send_{Q_B}(q_2)\cdot E\nonumber
\end{eqnarray}

And the state transitions of Alice can be described as follows.

\begin{eqnarray}
&&A=loc_A::(\sum_{D_i\in \Delta_i}receive_A(D_i)\cdot A_1)\nonumber\\
&&A_1=receive_{Q_A}(q_1)\cdot A_2\nonumber\\
&&A_2=CNOT\cdot A_3\nonumber\\
&&A_3=H\cdot A_4\nonumber\\
&&A_4=(M_0\cdot send_P(0)\boxplus_{\frac{1}{4}}M_1\cdot send_P(1)\boxplus_{\frac{1}{4}}M_2\cdot send_P(2)\boxplus_{\frac{1}{4}}M_3\cdot send_P(3))\cdot A\nonumber
\end{eqnarray}

where $\Delta_i$ is the collection of the input data.

And the state transitions of Bob can be described as follows.

\begin{eqnarray}
&&B=loc_B::(receive_{Q_B}(q_2)\cdot B_1)\nonumber\\
&&B_1=(receive_P(0)\cdot\sigma_0\boxplus_{\frac{1}{4}}receive_P(1)\cdot\sigma_1\boxplus_{\frac{1}{4}}receive_P(2) \cdot\sigma_2\boxplus_{\frac{1}{4}}receive_P(3)\cdot\sigma_3)\cdot B_2\nonumber\\
&&B_2=\sum_{D_o\in\Delta_o}send_B(D_o)\cdot B\nonumber
\end{eqnarray}

where $\Delta_o$ is the collection of the output data.

The send action and receive action of the same data through the same channel can communicate each other, otherwise, a deadlock $\delta$ will be caused. We define the following communication functions.

\begin{eqnarray}
&&\gamma(send_{Q_A}(q_1),receive_{Q_A}(q_1))\triangleq c_{Q_A}(q_1)\nonumber\\
&&\gamma(send_{Q_B}(q_2),receive_{Q_B}(q_2))\triangleq c_{Q_B}(q_2)\nonumber\\
&&\gamma(send_P(0),receive_P(0))\triangleq c_P(0)\nonumber\\
&&\gamma(send_P(1),receive_P(1))\triangleq c_P(1)\nonumber\\
&&\gamma(send_P(2),receive_P(2))\triangleq c_P(2)\nonumber\\
&&\gamma(send_P(3),receive_P(3))\triangleq c_P(3)\nonumber
\end{eqnarray}

Let $A$, $B$ and $E$ in parallel, then the system $ABE$ can be represented by the following process term.

$$\tau_I(\partial_H(\Theta(A\between B\between E)))$$

where $H=\{send_{Q_A}(q_1), receive_{Q_A}(q_1), send_{Q_B}(q_2), receive_{Q_B}(q_2),\\
send_P(0), receive_P(0), send_P(1), receive_P(1),\\
send_P(2), receive_P(2), send_P(3), receive_P(3)\}$ and $I=\{Set[q], H[q], CNOT, H, M_0, M_1,\\ M_2, M_3, \sigma_0, \sigma_1, \sigma_2, \sigma_3, \\ c_{Q_A}(q_1), c_{Q_B}(q_2), c_P(0), c_P(1), c_P(2), c_P(3)\}$.

Then we get the following conclusion.

\begin{theorem}
The basic quantum teleportation protocol $\tau_I(\partial_H(\Theta(A\between B\between E)))$ can exhibit desired external behaviors.
\end{theorem}

\begin{proof}
We can get $\tau_I(\partial_H(\Theta(A\between B\between E)))=\sum_{D_i\in \Delta_i}\sum_{D_o\in\Delta_o}loc_A::receive_A(D_i)\leftmerge loc_B::send_B(D_o)\leftmerge
\tau_I(\partial_H(\Theta(A\between B\between E)))$. So, the basic quantum teleportation protocol $\tau_I(\partial_H(\Theta(A\between B\between E)))$ can exhibit desired external behaviors.
\end{proof}

\subsection{Verification of BB84 Protocol}\label{VBB86}

The BB84 protocol \cite{BB84} is used to create a private key between two parities, Alice and Bob. Firstly, we introduce the basic BB84 protocol briefly, which is illustrated in Figure \ref{BB84}.

\begin{enumerate}
  \item Alice create two string of bits with size $n$ randomly, denoted as $B_a$ and $K_a$;
  \item Alice generates a string of qubits $q$ with size $n$, and the $i$th qubit in $q$ is $|x_y\rangle$, where $x$ is the $i$th bit of $B_a$ and $y$ is the $i$th bit of $K_a$;
  \item Alice sends $q$ to Bob through a quantum channel $Q$ between Alice and Bob;
  \item Bob receives $q$ and randomly generates a string of bits $B_b$ with size $n$;
  \item Bob measures each qubit of $q$ according to a basis by bits of $B_b$. And the measurement results would be $K_b$, which is also with size $n$;
  \item Bob sends his measurement bases $B_b$ to Alice through a public channel $P$;
  \item Once receiving $B_b$, Alice sends her bases $B_a$ to Bob through channel $P$, and Bob receives $B_a$;
  \item Alice and Bob determine that at which position the bit strings $B_a$ and $B_b$ are equal, and they discard the mismatched bits of $B_a$ and $B_b$. Then the remaining bits of $K_a$ and $K_b$, denoted as $K_a'$ and $K_b'$ with $K_{a,b}=K_a'=K_b'$.
\end{enumerate}

\begin{figure}
  \centering
  %\vspace{5cm}
  \includegraphics{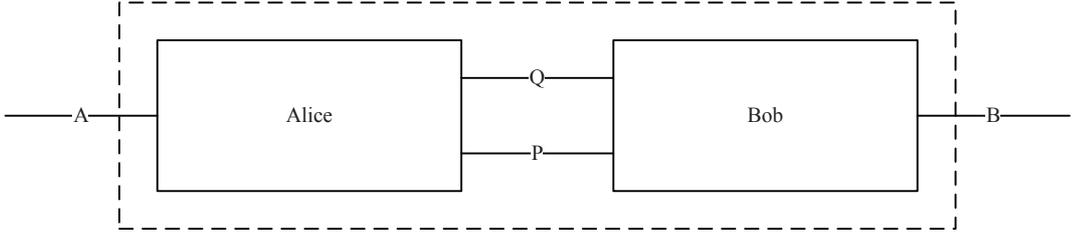}
  \caption{BB84 protocol.}
  \label{BB84}
\end{figure}

We re-introduce the basic BB84 protocol in an abstract way with more technical details as Figure \ref{BB84} illustrates.

Now, we assume a special measurement operation $Rand[q;B_a]=\sum^{2n-1}_{i=0}Rand[q;B_a]_i$ which create a string of $n$ random bits $B_a$ from the $q$ quantum system, and the same as $Rand[q;K_a]=\sum^{2n-1}_{i=0}Rand[q;K_a]_i$, $Rand[q';B_b]=\sum^{2n-1}_{i=0}Rand[q';B_b]_i$. $M[q;K_b]=\sum^{2n-1}_{i=0}M[q;K_b]_i$ denotes the Bob's measurement on $q$. The generation of $n$ qubits $q$ through two unitary operators $Set_{K_a}[q]$ and $H_{B_a}[q]$. Alice sends $q$ to Bob through the quantum channel $Q$ by quantum communicating action $send_{Q}(q)$ and Bob receives $q$ through $Q$ by quantum communicating action $receive_{Q}(q)$. Bob sends $B_b$ to Alice through the public classical channel $P$ by classical communicating action $send_{P}(B_b)$ and Alice receives $B_b$ through channel $P$ by classical communicating action $receive_{P}(B_b)$, and the same as $send_{P}(B_a)$ and $receive_{P}(B_a)$. Alice and Bob generate the private key $K_{a,b}$ by a classical comparison action $cmp(K_{a,b},K_a,K_b,B_a,B_b)$. Let Alice and Bob be a system $AB$ and let interactions between Alice and Bob be internal actions. $AB$ receives external input $D_i$ through channel $A$ by communicating action $receive_A(D_i)$ and sends results $D_o$ through channel $B$ by communicating action $send_B(D_o)$.

Then the state transitions of Alice can be described as follows.

\begin{eqnarray}
&&A=loc_A::(\sum_{D_i\in \Delta_i}receive_A(D_i)\cdot A_1)\nonumber\\
&&A_1=\boxplus_{\frac{1}{2n},i=0}^{2n-1}Rand[q;B_a]_i\cdot A_2\nonumber\\
&&A_2=\boxplus_{\frac{1}{2n},i=0}^{2n-1}Rand[q;K_a]_i\cdot A_3\nonumber\\
&&A_3=Set_{K_a}[q]\cdot A_4\nonumber\\
&&A_4=H_{B_a}[q]\cdot A_5\nonumber\\
&&A_5=send_Q(q)\cdot A_6\nonumber\\
&&A_6=receive_P(B_b)\cdot A_7\nonumber\\
&&A_7=send_P(B_a)\cdot A_8\nonumber\\
&&A_8=cmp(K_{a,b},K_a,K_b,B_a,B_b)\cdot A\nonumber
\end{eqnarray}

where $\Delta_i$ is the collection of the input data.

And the state transitions of Bob can be described as follows.

\begin{eqnarray}
&&B=loc_B::(receive_Q(q)\cdot B_1)\nonumber\\
&&B_1=\boxplus_{\frac{1}{2n},i=0}^{2n-1}Rand[q';B_b]_i\cdot B_2\nonumber\\
&&B_2=\boxplus_{\frac{1}{2n},i=0}^{2n-1}M[q;K_b]_i\cdot B_3\nonumber\\
&&B_3=send_P(B_b)\cdot B_4\nonumber\\
&&B_4=receive_P(B_a)\cdot B_5\nonumber\\
&&B_5=cmp(K_{a,b},K_a,K_b,B_a,B_b)\cdot B_6\nonumber\\
&&B_6=\sum_{D_o\in\Delta_o}send_B(D_o)\cdot B\nonumber
\end{eqnarray}

where $\Delta_o$ is the collection of the output data.

The send action and receive action of the same data through the same channel can communicate each other, otherwise, a deadlock $\delta$ will be caused. We define the following communication functions.

\begin{eqnarray}
&&\gamma(send_Q(q),receive_Q(q))\triangleq c_Q(q)\nonumber\\
&&\gamma(send_P(B_b),receive_P(B_b))\triangleq c_P(B_b)\nonumber\\
&&\gamma(send_P(B_a),receive_P(B_a))\triangleq c_P(B_a)\nonumber
\end{eqnarray}

Let $A$ and $B$ in parallel, then the system $AB$ can be represented by the following process term.

$$\tau_I(\partial_H(\Theta(A\between B)))$$

where $H=\{send_Q(q),receive_Q(q),send_P(B_b),receive_P(B_b),send_P(B_a),receive_P(B_a)\}$ and $I=\{Rand[q;B_a]_i, Rand[q;K_a]_i, Set_{K_a}[q], H_{B_a}[q], Rand[q';B_b]_i, M[q;K_b]_i, \\c_Q(q), c_P(B_b), c_P(B_a), cmp(K_{a,b},K_a,K_b,B_a,B_b)\}$.

Then we get the following conclusion.

\begin{theorem}
The basic BB84 protocol $\tau_I(\partial_H(\Theta(A\between B)))$ can exhibit desired external behaviors.
\end{theorem}

\begin{proof}
We can get $\tau_I(\partial_H(\Theta(A\between B)))=\sum_{D_i\in \Delta_i}\sum_{D_o\in\Delta_o}loc_A::receive_A(D_i)\leftmerge loc_B::send_B(D_o)\leftmerge \tau_I(\partial_H(\Theta(A\between B)))$.
So, the basic BB84 protocol $\tau_I(\partial_H(\Theta(A\between B)))$ can exhibit desired external behaviors.
\end{proof}

\subsection{Verification of E91 Protocol}\label{VE916}

E91 protocol\cite{E91} is the first quantum protocol which utilizes entanglement. E91 protocol is used to create a private key between two parities, Alice and Bob. Firstly, we introduce the basic E91 protocol briefly, which is illustrated in Figure \ref{E91}.

\begin{enumerate}
  \item Alice generates a string of EPR pairs $q$ with size $n$, i.e., $2n$ particles, and sends a string of qubits $q_b$ from each EPR pair with $n$ to Bob through a quantum channel $Q$, remains the other string of qubits $q_a$ from each pair with size $n$;
  \item Alice create two string of bits with size $n$ randomly, denoted as $B_a$ and $K_a$;
  \item Bob receives $q_b$ and randomly generates a string of bits $B_b$ with size $n$;
  \item Alice measures each qubit of $q_a$ according to a basis by bits of $B_a$. And the measurement results would be $K_a$, which is also with size $n$;
  \item Bob measures each qubit of $q_b$ according to a basis by bits of $B_b$. And the measurement results would be $K_b$, which is also with size $n$;
  \item Bob sends his measurement bases $B_b$ to Alice through a public channel $P$;
  \item Once receiving $B_b$, Alice sends her bases $B_a$ to Bob through channel $P$, and Bob receives $B_a$;
  \item Alice and Bob determine that at which position the bit strings $B_a$ and $B_b$ are equal, and they discard the mismatched bits of $B_a$ and $B_b$. Then the remaining bits of $K_a$ and $K_b$, denoted as $K_a'$ and $K_b'$ with $K_{a,b}=K_a'=K_b'$.
\end{enumerate}

\begin{figure}
  \centering
  %\vspace{5cm}
  \includegraphics{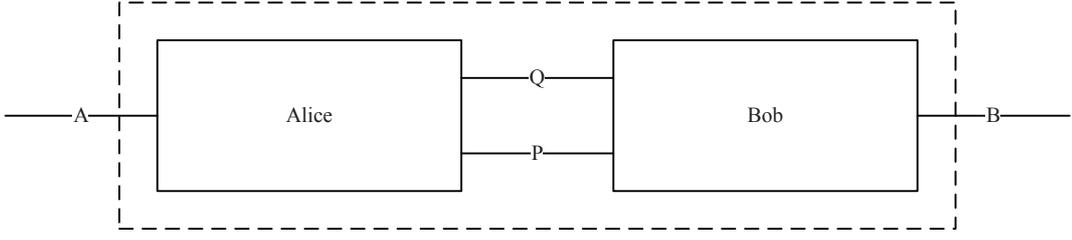}
  \caption{E91 protocol.}
  \label{E91}
\end{figure}

We re-introduce the basic E91 protocol in an abstract way with more technical details as Figure \ref{E91} illustrates.

Now, $M[q_a;K_a]=\sum_{i=0}^{2n-1}M[q_a;K_a]_i$ denotes the Alice's measurement operation of $q_a$, and $\circledS_{M[q_a;K_a]}=\sum_{i=0}^{2n-1}\circledS_{M[q_a;K_a]_i}$ denotes the responding shadow constant; $M[q_b;K_b]=\sum_{i=0}^{2n-1}M[q_b;K_b]_i$ denotes the Bob's measurement operation of $q_b$, and $\circledS_{M[q_b;K_b]}=\sum_{i=0}^{2n-1}\circledS_{M[q_b;K_b]_i}$ denotes the responding shadow constant. Alice sends $q_b$ to Bob through the quantum channel $Q$ by quantum communicating action $send_{Q}(q_b)$ and Bob receives $q_b$ through $Q$ by quantum communicating action $receive_{Q}(q_b)$. Bob sends $B_b$ to Alice through the public channel $P$ by classical communicating action $send_{P}(B_b)$ and Alice receives $B_b$ through channel $P$ by classical communicating action $receive_{P}(B_b)$, and the same as $send_{P}(B_a)$ and $receive_{P}(B_a)$. Alice and Bob generate the private key $K_{a,b}$ by a classical comparison action $cmp(K_{a,b},K_a,K_b,B_a,B_b)$. Let Alice and Bob be a system $AB$ and let interactions between Alice and Bob be internal actions. $AB$ receives external input $D_i$ through channel $A$ by communicating action $receive_A(D_i)$ and sends results $D_o$ through channel $B$ by communicating action $send_B(D_o)$.

Then the state transitions of Alice can be described as follows.

\begin{eqnarray}
&&A=loc_A::(\sum_{D_i\in \Delta_i}receive_A(D_i)\cdot A_1)\nonumber\\
&&A_1=send_Q(q_b)\cdot A_2\nonumber\\
&&A_2=\boxplus_{\frac{1}{2n},i=0}^{2n-1}M[q_a;K_a]_i\cdot A_3\nonumber\\
&&A_3=\boxplus_{\frac{1}{2n},i=0}^{2n-1}\circledS_{M[q_b;K_b]_i}\cdot A_4\nonumber\\
&&A_4=receive_P(B_b)\cdot A_5\nonumber\\
&&A_5=send_P(B_a)\cdot A_6\nonumber\\
&&A_6=cmp(K_{a,b},K_a,K_b,B_a,B_b)\cdot A\nonumber
\end{eqnarray}

where $\Delta_i$ is the collection of the input data.

And the state transitions of Bob can be described as follows.

\begin{eqnarray}
&&B=loc_B::(receive_Q(q_b)\cdot B_1)\nonumber\\
&&B_1=\boxplus_{\frac{1}{2n},i=0}^{2n-1}\circledS_{M[q_a;K_a]_i}\cdot B_2\nonumber\\
&&B_2=\boxplus_{\frac{1}{2n},i=0}^{2n-1}M[q_b;K_b]_i\cdot B_3\nonumber\\
&&B_3=send_P(B_b)\cdot B_4\nonumber\\
&&B_4=receive_P(B_a)\cdot B_5\nonumber\\
&&B_5=cmp(K_{a,b},K_a,K_b,B_a,B_b)\cdot B_6\nonumber\\
&&B_6=\sum_{D_o\in\Delta_o}send_B(D_o)\cdot B\nonumber
\end{eqnarray}

where $\Delta_o$ is the collection of the output data.

The send action and receive action of the same data through the same channel can communicate each other, otherwise, a deadlock $\delta$ will be caused. The quantum operation and its shadow constant pair will lead entanglement occur, otherwise, a deadlock $\delta$ will occur. We define the following communication functions.

\begin{eqnarray}
&&\gamma(send_Q(q_b),receive_Q(q_b))\triangleq c_Q(q_b)\nonumber\\
&&\gamma(send_P(B_b),receive_P(B_b))\triangleq c_P(B_b)\nonumber\\
&&\gamma(send_P(B_a),receive_P(B_a))\triangleq c_P(B_a)\nonumber
\end{eqnarray}

Let $A$ and $B$ in parallel, then the system $AB$ can be represented by the following process term.

$$\tau_I(\partial_H(\Theta(A\between B)))$$

where $H=\{send_Q(q_b),receive_Q(q_b),send_P(B_b),receive_P(B_b),send_P(B_a),receive_P(B_a),\\ M[q_a;K_a]_i, \circledS_{M[q_a;K_a]_i}, M[q_b;K_b]_i, \circledS_{M[q_b;K_b]_i}\}$ and $I=\{c_Q(q_b), c_P(B_b), c_P(B_a), M[q_a;K_a], M[q_b;K_b],\\ cmp(K_{a,b},K_a,K_b,B_a,B_b)\}$.

Then we get the following conclusion.

\begin{theorem}
The basic E91 protocol $\tau_I(\partial_H(\Theta(A\between B)))$ can exhibit desired external behaviors.
\end{theorem}

\begin{proof}
We can get $\tau_I(\partial_H(\Theta(A\between B)))=\sum_{D_i\in \Delta_i}\sum_{D_o\in\Delta_o}loc_A::receive_A(D_i)\leftmerge loc_B::send_B(D_o)\leftmerge \tau_I(\partial_H(\Theta(A\between B)))$.
So, the basic E91 protocol $\tau_I(\partial_H(\Theta(A\between B)))$ can exhibit desired external behaviors.
\end{proof}

\subsection{Verification of B92 Protocol}\label{VB926}

The famous B92 protocol\cite{B92} is a quantum key distribution protocol, in which quantum information and classical information are mixed.

The B92 protocol is used to create a private key between two parities, Alice and Bob. B92 is a protocol of quantum key distribution (QKD) which uses polarized photons as information carriers. Firstly, we introduce the basic B92 protocol briefly, which is illustrated in Figure \ref{B92}.

\begin{enumerate}
  \item Alice create a string of bits with size $n$ randomly, denoted as $A$.
  \item Alice generates a string of qubits $q$ with size $n$, carried by polarized photons. If $A_i=0$, the ith qubit is $|0\rangle$; else if $A_i=1$, the ith qubit is $|+\rangle$.
  \item Alice sends $q$ to Bob through a quantum channel $Q$ between Alice and Bob.
  \item Bob receives $q$ and randomly generates a string of bits $B$ with size $n$.
  \item If $B_i=0$, Bob chooses the basis $\oplus$; else if $B_i=1$, Bob chooses the basis $\otimes$. Bob measures each qubit of $q$ according to the above basses. And Bob builds a String of bits $T$, if the measurement produces $|0\rangle$ or $|+\rangle$, then $T_i=0$; else if the measurement produces $|1\rangle$ or $|-\rangle$, then $T_i=1$, which is also with size $n$.
  \item Bob sends $T$ to Alice through a public channel $P$.
  \item Alice and Bob determine that at which position the bit strings $A$ and $B$ are remained for which $T_i=1$. In absence of Eve, $A_i=1-B_i$, a shared raw key $K_{a,b}$ is formed by $A_i$.
\end{enumerate}

\begin{figure}
  \centering
  %\vspace{5cm}
  \includegraphics{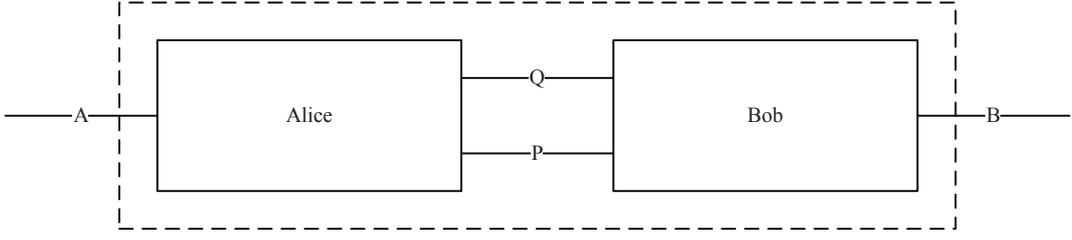}
  \caption{The B92 protocol.}
  \label{B92}
\end{figure}

We re-introduce the basic B92 protocol in an abstract way with more technical details as Figure \ref{B92} illustrates.

Now, we assume a special measurement operation $Rand[q;A]=\sum^{2n-1}_{i=0}Rand[q;A]_i$ which create a string of $n$ random bits $A$ from the $q$ quantum system, and the same as $Rand[q';B]=\sum^{2n-1}_{i=0}Rand[q';B]_i$. $M[q;T]=\sum^{2n-1}_{i=0}M[q;T]_i$ denotes the Bob's measurement operation of $q$. The generation of $n$ qubits $q$ through a unitary operator $Set_{A}[q]$. Alice sends $q$ to Bob through the quantum channel $Q$ by quantum communicating action $send_{Q}(q)$ and Bob receives $q$ through $Q$ by quantum communicating action $receive_{Q}(q)$. Bob sends $T$ to Alice through the public channel $P$ by classical communicating action $send_{P}(T)$ and Alice receives $T$ through channel $P$ by classical communicating action $receive_{P}(T)$. Alice and Bob generate the private key $K_{a,b}$ by a classical comparison action $cmp(K_{a,b},T,A,B)$. Let Alice and Bob be a system $AB$ and let interactions between Alice and Bob be internal actions. $AB$ receives external input $D_i$ through channel $A$ by communicating action $receive_A(D_i)$ and sends results $D_o$ through channel $B$ by communicating action $send_B(D_o)$.

Then the state transition of Alice can be described as follows.

\begin{eqnarray}
&&A=loc_A::(\sum_{D_i\in \Delta_i}receive_A(D_i)\cdot A_1)\nonumber\\
&&A_1=\boxplus_{\frac{1}{2n},i=0}^{2n-1}Rand[q;A]_i\cdot A_2\nonumber\\
&&A_2=Set_{A}[q]\cdot A_3\nonumber\\
&&A_3=send_Q(q)\cdot A_4\nonumber\\
&&A_4=receive_P(T)\cdot A_5\nonumber\\
&&A_5=cmp(K_{a,b},T,A,B)\cdot A\nonumber
\end{eqnarray}

where $\Delta_i$ is the collection of the input data.

And the state transition of Bob can be described as follows.

\begin{eqnarray}
&&B=loc_B::(receive_Q(q)\cdot B_1)\nonumber\\
&&B_1=\boxplus_{\frac{1}{2n},i=0}^{2n-1}Rand[q';B]_i\cdot B_2\nonumber\\
&&B_2=\boxplus_{\frac{1}{2n},i=0}^{2n-1}M[q;T]_i\cdot B_3\nonumber\\
&&B_3=send_P(T)\cdot B_4\nonumber\\
&&B_4=cmp(K_{a,b},T,A,B)\cdot B_5\nonumber\\
&&B_5=\sum_{D_o\in\Delta_o}send_B(D_o)\cdot B\nonumber
\end{eqnarray}

where $\Delta_o$ is the collection of the output data.

The send action and receive action of the same data through the same channel can communicate each other, otherwise, a deadlock $\delta$ will be caused. We define the following communication functions.

\begin{eqnarray}
&&\gamma(send_Q(q),receive_Q(q))\triangleq c_Q(q)\nonumber\\
&&\gamma(send_P(T),receive_P(T))\triangleq c_P(T)\nonumber\\
\end{eqnarray}

Let $A$ and $B$ in parallel, then the system $AB$ can be represented by the following process term.

$$\tau_I(\partial_H(\Theta(A\between B)))$$

where $H=\{send_Q(q),receive_Q(q),send_P(T),receive_P(T)\}$ and $I=\{\boxplus_{\frac{1}{2n},i=0}^{2n-1}Rand[q;A]_i, \\Set_{A}[q], \boxplus_{\frac{1}{2n},i=0}^{2n-1}Rand[q';B]_i, \boxplus_{\frac{1}{2n},i=0}^{2n-1}M[q;T]_i, c_Q(q), c_P(T), cmp(K_{a,b},T,A,B)\}$.

Then we get the following conclusion.

\begin{theorem}
The basic B92 protocol $\tau_I(\partial_H(\Theta(A\between B)))$ can exhibit desired external behaviors.
\end{theorem}

\begin{proof}
We can get $\tau_I(\partial_H(\Theta(A\between B)))=\sum_{D_i\in \Delta_i}\sum_{D_o\in\Delta_o}loc_A::receive_A(D_i)\leftmerge loc_B::send_B(D_o)\leftmerge \tau_I(\partial_H(\Theta(A\between B)))$.
So, the basic B92 protocol $\tau_I(\partial_H(\Theta(A\between B)))$ can exhibit desired external behaviors.
\end{proof}

\subsection{Verification of DPS Protocol}\label{VDPS6}

The famous DPS protocol\cite{DPS} is a quantum key distribution protocol, in which quantum information and classical information are mixed.

The DPS protocol is used to create a private key between two parities, Alice and Bob. DPS is a protocol of quantum key distribution (QKD) which uses pulses of a photon which has nonorthogonal four states. Firstly, we introduce the basic DPS protocol briefly, which is illustrated in Figure \ref{DPS}.

\begin{enumerate}
  \item Alice generates a string of qubits $q$ with size $n$, carried by a series of single photons possily at four time instances.
  \item Alice sends $q$ to Bob through a quantum channel $Q$ between Alice and Bob.
  \item Bob receives $q$ by detectors clicking at the second or third time instance, and records the time into $T$ with size $n$ and which detector clicks into $D$ with size $n$.
  \item Bob sends $T$ to Alice through a public channel $P$.
  \item Alice receives $T$. From $T$ and her modulation data, Alice knows which detector clicked in Bob's site, i.e. $D$.
  \item Alice and Bob have an identical bit string, provided that the first detector click represents "0" and the other detector represents "1", then a shared raw key $K_{a,b}$ is formed.
\end{enumerate}

\begin{figure}
  \centering
  %\vspace{5cm}
  \includegraphics{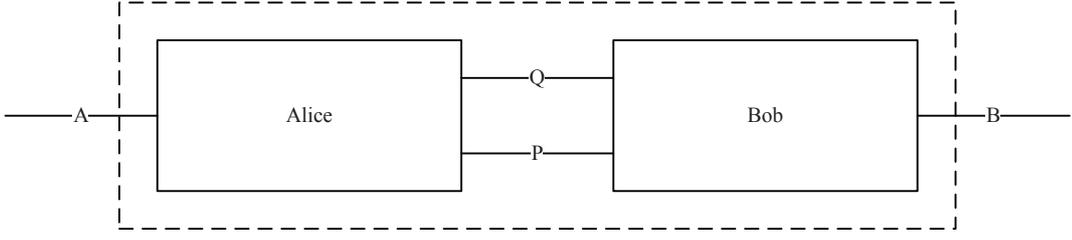}
  \caption{The DPS protocol.}
  \label{DPS}
\end{figure}

We re-introduce the basic DPS protocol in an abstract way with more technical details as Figure \ref{DPS} illustrates.

Now, we assume $M[q;T]=\sum^{2n-1}_{i=0}M[q;T]_i$ denotes the Bob's measurement operation of $q$. The generation of $n$ qubits $q$ through a unitary operator $Set_{A}[q]$. Alice sends $q$ to Bob through the quantum channel $Q$ by quantum communicating action $send_{Q}(q)$ and Bob receives $q$ through $Q$ by quantum communicating action $receive_{Q}(q)$. Bob sends $T$ to Alice through the public channel $P$ by classical communicating action $send_{P}(T)$ and Alice receives $T$ through channel $P$ by classical communicating action $receive_{P}(T)$. Alice and Bob generate the private key $K_{a,b}$ by a classical comparison action $cmp(K_{a,b},D)$. Let Alice and Bob be a system $AB$ and let interactions between Alice and Bob be internal actions. $AB$ receives external input $D_i$ through channel $A$ by communicating action $receive_A(D_i)$ and sends results $D_o$ through channel $B$ by communicating action $send_B(D_o)$.

Then the state transition of Alice can be described as follows.

\begin{eqnarray}
&&A=loc_A::(\sum_{D_i\in \Delta_i}receive_A(D_i)\cdot A_1)\nonumber\\
&&A_1=Set_{A}[q]\cdot A_2\nonumber\\
&&A_2=send_Q(q)\cdot A_3\nonumber\\
&&A_3=receive_P(T)\cdot A_4\nonumber\\
&&A_4=cmp(K_{a,b},D)\cdot A\nonumber
\end{eqnarray}

where $\Delta_i$ is the collection of the input data.

And the state transition of Bob can be described as follows.

\begin{eqnarray}
&&B=loc_B::(receive_Q(q)\cdot B_1)\nonumber\\
&&B_1=\boxplus_{\frac{1}{2n},i=0}^{2n-1}M[q;T]_i\cdot B_2\nonumber\\
&&B_2=send_P(T)\cdot B_3\nonumber\\
&&B_3=cmp(K_{a,b},D)\cdot B_4\nonumber\\
&&B_4=\sum_{D_o\in\Delta_o}send_B(D_o)\cdot B\nonumber
\end{eqnarray}

where $\Delta_o$ is the collection of the output data.

The send action and receive action of the same data through the same channel can communicate each other, otherwise, a deadlock $\delta$ will be caused. We define the following communication functions.

\begin{eqnarray}
&&\gamma(send_Q(q),receive_Q(q))\triangleq c_Q(q)\nonumber\\
&&\gamma(send_P(T),receive_P(T))\triangleq c_P(T)\nonumber\\
\end{eqnarray}

Let $A$ and $B$ in parallel, then the system $AB$ can be represented by the following process term.

$$\tau_I(\partial_H(\Theta(A\between B)))$$

where $H=\{send_Q(q),receive_Q(q),send_P(T),receive_P(T)\}$ and $I=\{Set_{A}[q], \\ \boxplus_{\frac{1}{2n},i=0}^{2n-1}M[q;T]_i, c_Q(q), c_P(T), cmp(K_{a,b},D)\}$.

Then we get the following conclusion.

\begin{theorem}
The basic DPS protocol $\tau_I(\partial_H(\Theta(A\between B)))$ can exhibit desired external behaviors.
\end{theorem}

\begin{proof}
We can get $\tau_I(\partial_H(\Theta(A\between B)))=\sum_{D_i\in \Delta_i}\sum_{D_o\in\Delta_o}loc_A::receive_A(D_i)\leftmerge loc_B::send_B(D_o)\leftmerge \tau_I(\partial_H(\Theta(A\between B)))$.
So, the basic DPS protocol $\tau_I(\partial_H(\Theta(A\between B)))$ can exhibit desired external behaviors.
\end{proof}

\subsection{Verification of BBM92 Protocol}\label{VBBM926}

The famous BBM92 protocol\cite{BBM92} is a quantum key distribution protocol, in which quantum information and classical information are mixed.

The BBM92 protocol is used to create a private key between two parities, Alice and Bob. BBM92 is a protocol of quantum key distribution (QKD) which uses EPR pairs as information carriers. Firstly, we introduce the basic BBM92 protocol briefly, which is illustrated in Figure \ref{BBM92}.

\begin{enumerate}
  \item Alice generates a string of EPR pairs $q$ with size $n$, i.e., $2n$ particles, and sends a string of qubits $q_b$ from each EPR pair with $n$ to Bob through a quantum channel $Q$, remains the other string of qubits $q_a$ from each pair with size $n$.
  \item Alice create a string of bits with size $n$ randomly, denoted as $B_a$.
  \item Bob receives $q_b$ and randomly generates a string of bits $B_b$ with size $n$.
  \item Alice measures each qubit of $q_a$ according to bits of $B_a$, if $B_{a_i}=0$, then uses $x$ axis ($\rightarrow$); else if $B_{a_i}=1$, then uses $z$ axis ($\uparrow$).
  \item Bob measures each qubit of $q_b$ according to bits of $B_b$, if $B_{b_i}=0$, then uses $x$ axis ($\rightarrow$); else if $B_{b_i}=1$, then uses $z$ axis ($\uparrow$).
  \item Bob sends his measurement axis choices $B_b$ to Alice through a public channel $P$.
  \item Once receiving $B_b$, Alice sends her axis choices $B_a$ to Bob through channel $P$, and Bob receives $B_a$.
  \item Alice and Bob agree to discard all instances in which they happened to measure along different axes, as well as instances in which measurements fails because of imperfect quantum efficiency of the detectors. Then the remaining instances can be used to generate a private key $K_{a,b}$.
\end{enumerate}

\begin{figure}
  \centering
  %\vspace{5cm}
  \includegraphics{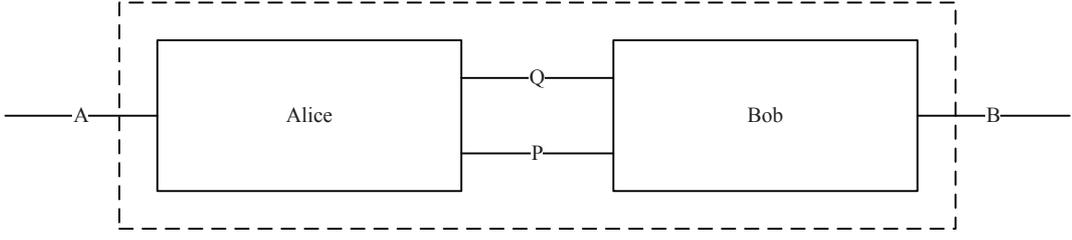}
  \caption{The BBM92 protocol.}
  \label{BBM92}
\end{figure}

We re-introduce the basic BBM92 protocol in an abstract way with more technical details as Figure \ref{BBM92} illustrates.

Now, $M[q_a;B_a]=\sum_{i=0}^{2n-1}M[q_a;K_a]_i$ denotes the Alice's measurement operation of $q_a$, and $\circledS_{M[q_a;B_a]}=\sum_{i=0}^{2n-1}\circledS_{M[q_a;B_a]_i}$ denotes the responding shadow constant; $M[q_b;B_b]=\sum_{i=0}^{2n-1}M[q_b;B_b]_i$ denotes the Bob's measurement operation of $q_b$, and $\circledS_{M[q_b;B_b]}=\sum_{i=0}^{2n-1}\circledS_{M[q_b;B_n]_i}$ denotes the responding shadow constant. Alice sends $q_b$ to Bob through the quantum channel $Q$ by quantum communicating action $send_{Q}(q_b)$ and Bob receives $q_b$ through $Q$ by quantum communicating action $receive_{Q}(q_b)$. Bob sends $B_b$ to Alice through the public channel $P$ by classical communicating action $send_{P}(B_b)$ and Alice receives $B_b$ through channel $P$ by classical communicating action $receive_{P}(B_b)$, and the same as $send_{P}(B_a)$ and $receive_{P}(B_a)$. Alice and Bob generate the private key $K_{a,b}$ by a classical comparison action $cmp(K_{a,b},B_a,B_b)$. Let Alice and Bob be a system $AB$ and let interactions between Alice and Bob be internal actions. $AB$ receives external input $D_i$ through channel $A$ by communicating action $receive_A(D_i)$ and sends results $D_o$ through channel $B$ by communicating action $send_B(D_o)$.

Then the state transition of Alice can be described as follows.

\begin{eqnarray}
&&A=loc_A::(\sum_{D_i\in \Delta_i}receive_A(D_i)\cdot A_1)\nonumber\\
&&A_1=send_Q(q_b)\cdot A_2\nonumber\\
&&A_2=\boxplus_{\frac{1}{2n},i=0}^{2n-1}M[q_a;B_a]_i\cdot A_3\nonumber\\
&&A_3=\boxplus_{\frac{1}{2n},i=0}^{2n-1}\circledS_{M[q_b;B_b]_i}\cdot A_4\nonumber\\
&&A_4=receive_P(B_b)\cdot A_5\nonumber\\
&&A_5=send_P(B_a)\cdot A_6\nonumber\\
&&A_6=cmp(K_{a,b},B_a,B_b)\cdot A\nonumber
\end{eqnarray}

where $\Delta_i$ is the collection of the input data.

And the state transition of Bob can be described as follows.

\begin{eqnarray}
&&B=loc_B::(receive_Q(q_b)\cdot B_1)\nonumber\\
&&B_1=\boxplus_{\frac{1}{2n},i=0}^{2n-1}\circledS_{M[q_a;B_a]_i}\cdot B_2\nonumber\\
&&B_2=\boxplus_{\frac{1}{2n},i=0}^{2n-1}M[q_b;B_b]_i\cdot B_3\nonumber\\
&&B_3=send_P(B_b)\cdot B_4\nonumber\\
&&B_4=receive_P(B_a)\cdot B_5\nonumber\\
&&B_5=cmp(K_{a,b},B_a,B_b)\cdot B_6\nonumber\\
&&B_6=\sum_{D_o\in\Delta_o}send_B(D_o)\cdot B\nonumber
\end{eqnarray}

where $\Delta_o$ is the collection of the output data.

The send action and receive action of the same data through the same channel can communicate each other, otherwise, a deadlock $\delta$ will be caused. The quantum measurement and its shadow constant pair will lead entanglement occur, otherwise, a deadlock $\delta$ will occur. We define the following communication functions.

\begin{eqnarray}
&&\gamma(send_Q(q_b),receive_Q(q_b))\triangleq c_Q(q_b)\nonumber\\
&&\gamma(send_P(B_b),receive_P(B_b))\triangleq c_P(B_b)\nonumber\\
&&\gamma(send_P(B_a),receive_P(B_a))\triangleq c_P(B_a)\nonumber
\end{eqnarray}

Let $A$ and $B$ in parallel, then the system $AB$ can be represented by the following process term.

$$\tau_I(\partial_H(\Theta(A\between B)))$$

where $H=\{send_Q(q_b),receive_Q(q_b),send_P(B_b),receive_P(B_b),send_P(B_a),receive_P(B_a),\\ \boxplus_{\frac{1}{2n},i=0}^{2n-1}M[q_a;B_a]_i, \boxplus_{\frac{1}{2n},i=0}^{2n-1}\circledS_{M[q_a;B_a]_i}, \boxplus_{\frac{1}{2n},i=0}^{2n-1}M[q_b;B_b]_i, \boxplus_{\frac{1}{2n},i=0}^{2n-1}\circledS_{M[q_b;B_b]_i}\}$

and $I=\{c_Q(q_b), c_P(B_b), c_P(B_a), M[q_a;B_a], M[q_b;B_b],\\ cmp(K_{a,b},B_a,B_b)\}$.

Then we get the following conclusion.

\begin{theorem}
The basic BBM92 protocol $\tau_I(\partial_H(\Theta(A\between B)))$ can exhibit desired external behaviors.
\end{theorem}

\begin{proof}
We can get $\tau_I(\partial_H(\Theta(A\between B)))=\sum_{D_i\in \Delta_i}\sum_{D_o\in\Delta_o}loc_A::receive_A(D_i)\leftmerge loc_B::send_B(D_o)\leftmerge \tau_I(\partial_H(\Theta(A\between B)))$.
So, the basic BBM92 protocol $\tau_I(\partial_H(\Theta(A\between B)))$ can exhibit desired external behaviors.
\end{proof}

\subsection{Verification of SARG04 Protocol}\label{VSARG046}

The famous SARG04 protocol\cite{SARG04} is a quantum key distribution protocol, in which quantum information and classical information are mixed.

The SARG04 protocol is used to create a private key between two parities, Alice and Bob. SARG04 is a protocol of quantum key distribution (QKD) which refines the BB84 protocol against PNS (Photon Number Splitting) attacks. The main innovations are encoding bits in nonorthogonal states and the classical sifting procedure. Firstly, we introduce the basic SARG04 protocol briefly, which is illustrated in Figure \ref{SARG04}.

\begin{enumerate}
  \item Alice create a string of bits with size $n$ randomly, denoted as $K_a$.
  \item Alice generates a string of qubits $q$ with size $n$, and the $i$th qubit of $q$ has four nonorthogonal states, it is $|\pm x\rangle$ if $K_a=0$; it is $|\pm z\rangle$ if $K_a=1$. And she records the corresponding one of the four pairs of nonorthogonal states into $B_a$ with size $2n$.
  \item Alice sends $q$ to Bob through a quantum channel $Q$ between Alice and Bob.
  \item Alice sends $B_a$ through a public channel $P$.
  \item Bob measures each qubit of $q$ $\sigma_x$ or $\sigma_z$. And he records the unambiguous discriminations into $K_b$ with a raw size $n/4$, and the unambiguous discrimination information into $B_b$ with size $n$.
  \item Bob sends $B_b$ to Alice through the public channel $P$.
  \item Alice and Bob determine that at which position the bit should be remained. Then the remaining bits of $K_a$ and $K_b$ is the private key $K_{a,b}$.
\end{enumerate}

\begin{figure}
  \centering
  %\vspace{5cm}
  \includegraphics{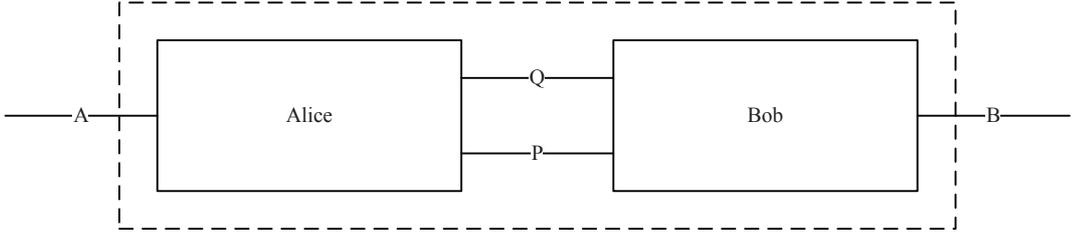}
  \caption{The SARG04 protocol.}
  \label{SARG04}
\end{figure}

We re-introduce the basic SARG04 protocol in an abstract way with more technical details as Figure \ref{SARG04} illustrates.

Now, we assume a special measurement operation $Rand[q;K_a]=\sum^{2n-1}_{i=0}Rand[q;K_a]_i$ which create a string of $n$ random bits $K_a$ from the $q$ quantum system. $M[q;K_b]=\sum^{2n-1}_{i=0}M[q;K_b]_i$ denotes the Bob's measurement operation of $q$. The generation of $n$ qubits $q$ through a unitary operator $Set_{K_a}[q]$. Alice sends $q$ to Bob through the quantum channel $Q$ by quantum communicating action $send_{Q}(q)$ and Bob receives $q$ through $Q$ by quantum communicating action $receive_{Q}(q)$. Bob sends $B_b$ to Alice through the public channel $P$ by classical communicating action $send_{P}(B_b)$ and Alice receives $B_b$ through channel $P$ by classical communicating action $receive_{P}(B_b)$, and the same as $send_{P}(B_a)$ and $receive_{P}(B_a)$. Alice and Bob generate the private key $K_{a,b}$ by a classical comparison action $cmp(K_{a,b},K_a,K_b,B_a,B_b)$. Let Alice and Bob be a system $AB$ and let interactions between Alice and Bob be internal actions. $AB$ receives external input $D_i$ through channel $A$ by communicating action $receive_A(D_i)$ and sends results $D_o$ through channel $B$ by communicating action $send_B(D_o)$.

Then the state transition of Alice can be described as follows.

\begin{eqnarray}
&&A=loc_A::(\sum_{D_i\in \Delta_i}receive_A(D_i)\cdot A_1)\nonumber\\
&&A_1=\boxplus_{\frac{1}{2n},i=0}^{2n-1}Rand[q;K_a]_i\cdot A_2\nonumber\\
&&A_2=Set_{K_a}[q]\cdot A_3\nonumber\\
&&A_3=send_Q(q)\cdot A_4\nonumber\\
&&A_4=send_P(B_a)\cdot A_5\nonumber\\
&&A_5=receive_P(B_b)\cdot A_6\nonumber\\
&&A_6=cmp(K_{a,b},K_a,K_b,B_a,B_b)\cdot A\nonumber
\end{eqnarray}

where $\Delta_i$ is the collection of the input data.

And the state transition of Bob can be described as follows.

\begin{eqnarray}
&&B=loc_B::(receive_Q(q)\cdot B_1)\nonumber\\
&&B_1=receive_P(B_a)\cdot B_2\nonumber\\
&&B_2=\boxplus_{\frac{1}{2n},i=0}^{2n-1}M[q;K_b]_i\cdot B_3\nonumber\\
&&B_3=send_P(B_b)\cdot B_4\nonumber\\
&&B_4=cmp(K_{a,b},K_a,K_b,B_a,B_b)\cdot B_5\nonumber\\
&&B_5=\sum_{D_o\in\Delta_o}send_B(D_o)\cdot B\nonumber
\end{eqnarray}

where $\Delta_o$ is the collection of the output data.

The send action and receive action of the same data through the same channel can communicate each other, otherwise, a deadlock $\delta$ will be caused. We define the following communication functions.

\begin{eqnarray}
&&\gamma(send_Q(q),receive_Q(q))\triangleq c_Q(q)\nonumber\\
&&\gamma(send_P(B_b),receive_P(B_b))\triangleq c_P(B_b)\nonumber\\
&&\gamma(send_P(B_a),receive_P(B_a))\triangleq c_P(B_a)\nonumber
\end{eqnarray}

Let $A$ and $B$ in parallel, then the system $AB$ can be represented by the following process term.

$$\tau_I(\partial_H(\Theta(A\between B)))$$

where $H=\{send_Q(q),receive_Q(q),send_P(B_b),receive_P(B_b),send_P(B_a),receive_P(B_a)\}$ and $I=\{\boxplus_{\frac{1}{2n},i=0}^{2n-1}Rand[q;K_a]_i, Set_{K_a}[q], \boxplus_{\frac{1}{2n},i=0}^{2n-1}M[q;K_b]_i, c_Q(q), c_P(B_b),\\ c_P(B_a), cmp(K_{a,b},K_a,K_b,B_a,B_b)\}$.

Then we get the following conclusion.

\begin{theorem}
The basic SARG04 protocol $\tau_I(\partial_H(\Theta(A\between B)))$ can exhibit desired external behaviors.
\end{theorem}

\begin{proof}
We can get $\tau_I(\partial_H(\Theta(A\between B)))=\sum_{D_i\in \Delta_i}\sum_{D_o\in\Delta_o}loc_A::receive_A(D_i)\leftmerge loc_B::send_B(D_o)\leftmerge \tau_I(\partial_H(\Theta(A\between B)))$.
So, the basic SARG04 protocol $\tau_I(\partial_H(\Theta(A\between B)))$ can exhibit desired external behaviors.
\end{proof}

\subsection{Verification of COW Protocol}\label{VCOW6}

The famous COW protocol\cite{COW} is a quantum key distribution protocol, in which quantum information and classical information are mixed.

The COW protocol is used to create a private key between two parities, Alice and Bob. COW is a protocol of quantum key distribution (QKD) which is practical. Firstly, we introduce the basic COW protocol briefly, which is illustrated in Figure \ref{COW}.

\begin{enumerate}
  \item Alice generates a string of qubits $q$ with size $n$, and the $i$th qubit of $q$ is "0" with probability $\frac{1-f}{2}$, "1" with probability $\frac{1-f}{2}$ and the decoy sequence with probability $f$.
  \item Alice sends $q$ to Bob through a quantum channel $Q$ between Alice and Bob.
  \item Alice sends $A$ of the items corresponding to a decoy sequence through a public channel $P$.
  \item Bob removes all the detections at times $2A-1$ and $2A$ from his raw key and looks whether detector $D_{2M}$ has ever fired at time $2A$.
  \item Bob sends $B$ of the times $2A+1$ in which he had a detector in $D_{2M}$ to Alice through the public channel $P$.
  \item Alice receives $B$ and verifies if some of these items corresponding to a bit sequence "1,0".
  \item Bob sends $C$ of the items that he has detected through the public channel $P$.
  \item Alice and Bob run error correction and privacy amplification on these bits, and the private key $K_{a,b}$ is established.
\end{enumerate}

\begin{figure}
  \centering
  %\vspace{5cm}
  \includegraphics{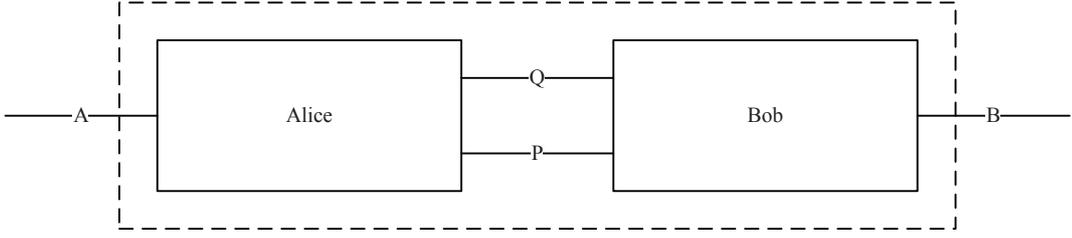}
  \caption{The COW protocol.}
  \label{COW}
\end{figure}

We re-introduce the basic COW protocol in an abstract way with more technical details as Figure \ref{COW} illustrates.

Now, we assume The generation of $n$ qubits $q$ through a unitary operator $Set[q]$. $M[q]=\sum^{2n-1}_{i=0}M[q]_i$ denotes the Bob's measurement operation of $q$.  Alice sends $q$ to Bob through the quantum channel $Q$ by quantum communicating action $send_{Q}(q)$ and Bob receives $q$ through $Q$ by quantum communicating action $receive_{Q}(q)$. Alice sends $A$ to Alice through the public channel $P$ by classical communicating action $send_{P}(A)$ and Alice receives $A$ through channel $P$ by classical communicating action $receive_{P}(A)$, and the same as $send_{P}(B)$ and $receive_{P}(B)$, and $send_{P}(C)$ and $receive_{P}(C)$. Alice and Bob generate the private key $K_{a,b}$ by a classical comparison action $cmp(K_{a,b})$. Let Alice and Bob be a system $AB$ and let interactions between Alice and Bob be internal actions. $AB$ receives external input $D_i$ through channel $A$ by communicating action $receive_A(D_i)$ and sends results $D_o$ through channel $B$ by communicating action $send_B(D_o)$.

Then the state transition of Alice can be described as follows.

\begin{eqnarray}
&&A=loc_A::(\sum_{D_i\in \Delta_i}receive_A(D_i)\cdot A_1)\nonumber\\
&&A_1=Set[q]\cdot A_2\nonumber\\
&&A_2=send_Q(q)\cdot A_3\nonumber\\
&&A_3=send_P(A)\cdot A_4\nonumber\\
&&A_4=receive_P(B)\cdot A_5\nonumber\\
&&A_5=receive_P(C)\cdot A_6\nonumber\\
&&A_6=cmp(K_{a,b})\cdot A\nonumber
\end{eqnarray}

where $\Delta_i$ is the collection of the input data.

And the state transition of Bob can be described as follows.

\begin{eqnarray}
&&B=loc_B::(receive_Q(q)\cdot B_1)\nonumber\\
&&B_1=receive_P(A)\cdot B_2\nonumber\\
&&B_2=\boxplus_{\frac{1}{2n},i=0}^{2n-1}M[q]_i\cdot B_3\nonumber\\
&&B_3=send_P(B)\cdot B_4\nonumber\\
&&B_4=send_P(C)\cdot B_5\nonumber\\
&&B_5=cmp(K_{a,b})\cdot B_6\nonumber\\
&&B_6=\sum_{D_o\in\Delta_o}send_B(D_o)\cdot B\nonumber
\end{eqnarray}

where $\Delta_o$ is the collection of the output data.

The send action and receive action of the same data through the same channel can communicate each other, otherwise, a deadlock $\delta$ will be caused. We define the following communication functions.

\begin{eqnarray}
&&\gamma(send_Q(q),receive_Q(q))\triangleq c_Q(q)\nonumber\\
&&\gamma(send_P(A),receive_P(A))\triangleq c_P(A)\nonumber\\
&&\gamma(send_P(B),receive_P(B))\triangleq c_P(B)\nonumber\\
&&\gamma(send_P(C),receive_P(C))\triangleq c_P(C)\nonumber
\end{eqnarray}

Let $A$ and $B$ in parallel, then the system $AB$ can be represented by the following process term.

$$\tau_I(\partial_H(\Theta(A\between B)))$$

where $H=\{send_Q(q),receive_Q(q),send_P(A),receive_P(A),send_P(B),receive_P(B),send_P(C),\\receive_P(C)\}$ and $I=\{Set[q], \boxplus_{\frac{1}{2n},i=0}^{2n-1}M[q]_i, c_Q(q), c_P(A),\\ c_P(B),c_P(C), cmp(K_{a,b})\}$.

Then we get the following conclusion.

\begin{theorem}
The basic COW protocol $\tau_I(\partial_H(\Theta(A\between B)))$ can exhibit desired external behaviors.
\end{theorem}

\begin{proof}
We can get $\tau_I(\partial_H(\Theta(A\between B)))=\sum_{D_i\in \Delta_i}\sum_{D_o\in\Delta_o}loc_A::receive_A(D_i)\leftmerge loc_B::send_B(D_o)\leftmerge \tau_I(\partial_H(\Theta(A\between B)))$.
So, the basic COW protocol $\tau_I(\partial_H(\Theta(A\between B)))$ can exhibit desired external behaviors.
\end{proof}

\subsection{Verification of SSP Protocol}\label{VSSP6}

The famous SSP protocol\cite{SSP} is a quantum key distribution protocol, in which quantum information and classical information are mixed.

The SSP protocol is used to create a private key between two parities, Alice and Bob. SSP is a protocol of quantum key distribution (QKD) which uses six states. Firstly, we introduce the basic SSP protocol briefly, which is illustrated in Figure \ref{SSP}.

\begin{enumerate}
  \item Alice create two string of bits with size $n$ randomly, denoted as $B_a$ and $K_a$.
  \item Alice generates a string of qubits $q$ with size $n$, and the $i$th qubit in $q$ is one of the six states $\pm x$, $\pm y$ and $\pm z$.
  \item Alice sends $q$ to Bob through a quantum channel $Q$ between Alice and Bob.
  \item Bob receives $q$ and randomly generates a string of bits $B_b$ with size $n$.
  \item Bob measures each qubit of $q$ according to a basis by bits of $B_b$, i.e., $x$, $y$ or $z$ basis. And the measurement results would be $K_b$, which is also with size $n$.
  \item Bob sends his measurement bases $B_b$ to Alice through a public channel $P$.
  \item Once receiving $B_b$, Alice sends her bases $B_a$ to Bob through channel $P$, and Bob receives $B_a$.
  \item Alice and Bob determine that at which position the bit strings $B_a$ and $B_b$ are equal, and they discard the mismatched bits of $B_a$ and $B_b$. Then the remaining bits of $K_a$ and $K_b$, denoted as $K_a'$ and $K_b'$ with $K_{a,b}=K_a'=K_b'$.
\end{enumerate}

\begin{figure}
  \centering
  %\vspace{5cm}
  \includegraphics{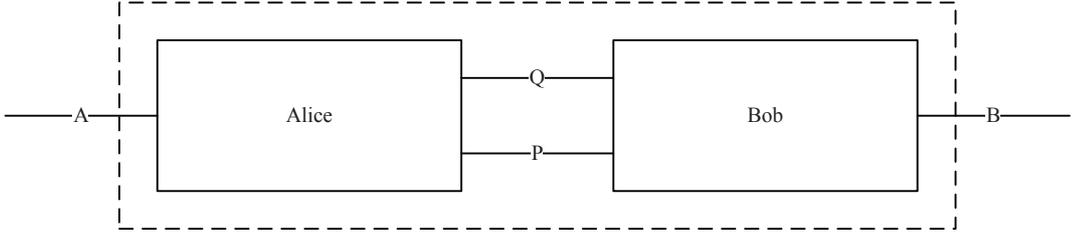}
  \caption{The SSP protocol.}
  \label{SSP}
\end{figure}

We re-introduce the basic SSP protocol in an abstract way with more technical details as Figure \ref{SSP} illustrates.

Now, we assume a special measurement operation $Rand[q;B_a]=\sum^{2n-1}_{i=0}Rand[q;B_a]_i$ which create a string of $n$ random bits $B_a$ from the $q$ quantum system, and the same as $Rand[q;K_a]=\sum^{2n-1}_{i=0}Rand[q;K_a]_i$, $Rand[q';B_b]=\sum^{2n-1}_{i=0}Rand[q';B_b]_i$. $M[q;K_b]=\sum^{2n-1}_{i=0}M[q;K_b]_i$ denotes the Bob's measurement operation of $q$. The generation of $n$ qubits $q$ through two unitary operators $Set_{K_a}[q]$ and $H_{B_a}[q]$. Alice sends $q$ to Bob through the quantum channel $Q$ by quantum communicating action $send_{Q}(q)$ and Bob receives $q$ through $Q$ by quantum communicating action $receive_{Q}(q)$. Bob sends $B_b$ to Alice through the public channel $P$ by classical communicating action $send_{P}(B_b)$ and Alice receives $B_b$ through channel $P$ by classical communicating action $receive_{P}(B_b)$, and the same as $send_{P}(B_a)$ and $receive_{P}(B_a)$. Alice and Bob generate the private key $K_{a,b}$ by a classical comparison action $cmp(K_{a,b},K_a,K_b,B_a,B_b)$. Let Alice and Bob be a system $AB$ and let interactions between Alice and Bob be internal actions. $AB$ receives external input $D_i$ through channel $A$ by communicating action $receive_A(D_i)$ and sends results $D_o$ through channel $B$ by communicating action $send_B(D_o)$.

Then the state transition of Alice can be described as follows.

\begin{eqnarray}
&&A=loc_A::(\sum_{D_i\in \Delta_i}receive_A(D_i)\cdot A_1)\nonumber\\
&&A_1=\boxplus_{\frac{1}{2n},i=0}^{2n-1}Rand[q;B_a]_i\cdot A_2\nonumber\\
&&A_2=\boxplus_{\frac{1}{2n},i=0}^{2n-1}Rand[q;K_a]_i\cdot A_3\nonumber\\
&&A_3=Set_{K_a}[q]\cdot A_4\nonumber\\
&&A_4=H_{B_a}[q]\cdot A_5\nonumber\\
&&A_5=send_Q(q)\cdot A_6\nonumber\\
&&A_6=receive_P(B_b)\cdot A_7\nonumber\\
&&A_7=send_P(B_a)\cdot A_8\nonumber\\
&&A_8=cmp(K_{a,b},K_a,K_b,B_a,B_b)\cdot A\nonumber
\end{eqnarray}

where $\Delta_i$ is the collection of the input data.

And the state transition of Bob can be described as follows.

\begin{eqnarray}
&&B=loc_B::(receive_Q(q)\cdot B_1)\nonumber\\
&&B_1=\boxplus_{\frac{1}{2n},i=0}^{2n-1}Rand[q';B_b]_i\cdot B_2\nonumber\\
&&B_2=\boxplus_{\frac{1}{2n},i=0}^{2n-1}M[q;K_b]_i\cdot B_3\nonumber\\
&&B_3=send_P(B_b)\cdot B_4\nonumber\\
&&B_4=receive_P(B_a)\cdot B_5\nonumber\\
&&B_5=cmp(K_{a,b},K_a,K_b,B_a,B_b)\cdot B_6\nonumber\\
&&B_6=\sum_{D_o\in\Delta_o}send_B(D_o)\cdot B\nonumber
\end{eqnarray}

where $\Delta_o$ is the collection of the output data.

The send action and receive action of the same data through the same channel can communicate each other, otherwise, a deadlock $\delta$ will be caused. We define the following communication functions.

\begin{eqnarray}
&&\gamma(send_Q(q),receive_Q(q))\triangleq c_Q(q)\nonumber\\
&&\gamma(send_P(B_b),receive_P(B_b))\triangleq c_P(B_b)\nonumber\\
&&\gamma(send_P(B_a),receive_P(B_a))\triangleq c_P(B_a)\nonumber
\end{eqnarray}

Let $A$ and $B$ in parallel, then the system $AB$ can be represented by the following process term.

$$\tau_I(\partial_H(\Theta(A\between B)))$$

where $H=\{send_Q(q),receive_Q(q),send_P(B_b),receive_P(B_b),send_P(B_a),receive_P(B_a)\}$ and $I=\{\boxplus_{\frac{1}{2n},i=0}^{2n-1}Rand[q;B_a]_i, \boxplus_{\frac{1}{2n},i=0}^{2n-1}Rand[q;K_a]_i, Set_{K_a}[q], \\ H_{B_a}[q], \boxplus_{\frac{1}{2n},i=0}^{2n-1}Rand[q';B_b]_i, \boxplus_{\frac{1}{2n},i=0}^{2n-1}M[q;K_b]_i, c_Q(q), c_P(B_b),\\ c_P(B_a), cmp(K_{a,b},K_a,K_b,B_a,B_b)\}$.

Then we get the following conclusion.

\begin{theorem}
The basic SSP protocol $\tau_I(\partial_H(\Theta(A\between B)))$ can exhibit desired external behaviors.
\end{theorem}

\begin{proof}
We can get $\tau_I(\partial_H(\Theta(A\between B)))=\sum_{D_i\in \Delta_i}\sum_{D_o\in\Delta_o}loc_A::receive_A(D_i)\leftmerge loc_B::send_B(D_o)\leftmerge \tau_I(\partial_H(\Theta(A\between B)))$.
So, the basic SSP protocol $\tau_I(\partial_H(\Theta(A\between B)))$ can exhibit desired external behaviors.
\end{proof}

\subsection{Verification of S09 Protocol}\label{VS096}

The famous S09 protocol\cite{S09} is a quantum key distribution protocol, in which quantum information and classical information are mixed.

The S09 protocol is used to create a private key between two parities, Alice and Bob, by use of pure quantum information. Firstly, we introduce the basic S09 protocol briefly, which is illustrated in Figure \ref{S09}.

\begin{enumerate}
  \item Alice create two string of bits with size $n$ randomly, denoted as $B_a$ and $K_a$.
  \item Alice generates a string of qubits $q$ with size $n$, and the $i$th qubit in $q$ is $|x_y\rangle$, where $x$ is the $i$th bit of $B_a$ and $y$ is the $i$th bit of $K_a$.
  \item Alice sends $q$ to Bob through a quantum channel $Q$ between Alice and Bob.
  \item Bob receives $q$ and randomly generates a string of bits $B_b$ with size $n$.
  \item Bob measures each qubit of $q$ according to a basis by bits of $B_b$. After the measurement, the state of $q$ evolves into $q'$.
  \item Bob sends $q'$ to Alice through the quantum channel $Q$.
  \item Alice measures each qubit of $q'$ to generate a string $C$.
  \item Alice sums $C_i\oplus B_{a_i}$ to get the private key $K_{a,b}=B_b$.
\end{enumerate}

\begin{figure}
  \centering
  %\vspace{5cm}
  \includegraphics{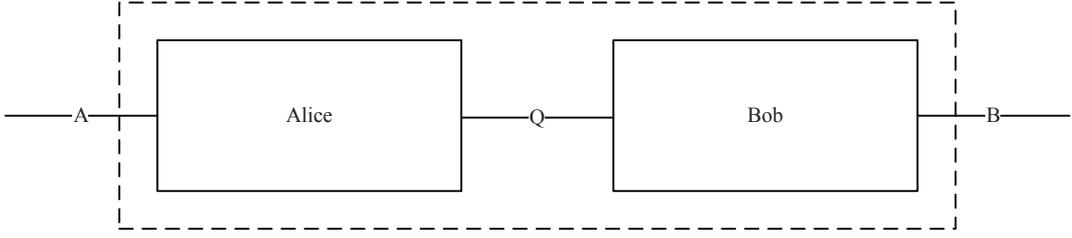}
  \caption{The S09 protocol.}
  \label{S09}
\end{figure}

We re-introduce the basic S09 protocol in an abstract way with more technical details as Figure \ref{S09} illustrates.

Now, we assume a special measurement operation $Rand[q;B_a]=\sum^{2n-1}_{i=0}Rand[q;B_a]_i$ which create a string of $n$ random bits $B_a$ from the $q$ quantum system, and the same as $Rand[q;K_a]=\sum^{2n-1}_{i=0}Rand[q;K_a]_i$, $Rand[q';B_b]=\sum^{2n-1}_{i=0}Rand[q';B_b]_i$. $M[q;B_b]=\sum^{2n-1}_{i=0}M[q;B_b]_i$ denotes the Bob's measurement operation of $q$, and the same as $M[q';C]=\sum^{2n-1}_{i=0}Rand[q';C]_i$. The generation of $n$ qubits $q$ through two unitary operators $Set_{K_a}[q]$ and $H_{B_a}[q]$. Alice sends $q$ to Bob through the quantum channel $Q$ by quantum communicating action $send_{Q}(q)$ and Bob receives $q$ through $Q$ by quantum communicating action $receive_{Q}(q)$, and the same as $send_{Q}(q')$ and $receive_{Q}(q')$. Alice and Bob generate the private key $K_{a,b}$ by a classical comparison action $cmp(K_{a,b},B_b)$. We omit the sum classical $\oplus$ actions without of loss of generality. Let Alice and Bob be a system $AB$ and let interactions between Alice and Bob be internal actions. $AB$ receives external input $D_i$ through channel $A$ by communicating action $receive_A(D_i)$ and sends results $D_o$ through channel $B$ by communicating action $send_B(D_o)$.

Then the state transition of Alice can be described as follows.

\begin{eqnarray}
&&A=loc_A::(\sum_{D_i\in \Delta_i}receive_A(D_i)\cdot A_1)\nonumber\\
&&A_1=\boxplus_{\frac{1}{2n},i=0}^{2n-1}Rand[q;B_a]_i\cdot A_2\nonumber\\
&&A_2=\boxplus_{\frac{1}{2n},i=0}^{2n-1}Rand[q;K_a]_i\cdot A_3\nonumber\\
&&A_3=Set_{K_a}[q]\cdot A_4\nonumber\\
&&A_4=H_{B_a}[q]\cdot A_5\nonumber\\
&&A_5=send_Q(q)\cdot A_6\nonumber\\
&&A_6=receive_Q(q')\cdot A_{7}\nonumber\\
&&A_7=\boxplus_{\frac{1}{2n},i=0}^{2n-1}M[q';C]_i\cdot A_8\nonumber\\
&&A_{8}=cmp(K_{a,b},B_b)\cdot A\nonumber
\end{eqnarray}

where $\Delta_i$ is the collection of the input data.

And the state transition of Bob can be described as follows.

\begin{eqnarray}
&&B=loc_B::(receive_Q(q)\cdot B_1)\nonumber\\
&&B_1=\boxplus_{\frac{1}{2n},i=0}^{2n-1}Rand[q';B_b]_i\cdot B_2\nonumber\\
&&B_2=\boxplus_{\frac{1}{2n},i=0}^{2n-1}M[q;B_b]_i\cdot B_3\nonumber\\
&&B_3=send_Q(q')\cdot B_4\nonumber\\
&&B_4=cmp(K_{a,b},B_b)\cdot B_{5}\nonumber\\
&&B_{5}=\sum_{D_o\in\Delta_o}send_B(D_o)\cdot B\nonumber
\end{eqnarray}

where $\Delta_o$ is the collection of the output data.

The send action and receive action of the same data through the same channel can communicate each other, otherwise, a deadlock $\delta$ will be caused. We define the following communication functions.

\begin{eqnarray}
&&\gamma(send_Q(q),receive_Q(q))\triangleq c_Q(q)\nonumber\\
&&\gamma(send_Q(q'),receive_Q(q'))\triangleq c_Q(q')\nonumber
\end{eqnarray}

Let $A$ and $B$ in parallel, then the system $AB$ can be represented by the following process term.

$$\tau_I(\partial_H(\Theta(A\between B)))$$

where $H=\{send_Q(q),receive_Q(q),send_Q(q'),receive_Q(q')\}$ and $I=\{\boxplus_{\frac{1}{2n},i=0}^{2n-1}Rand[q;B_a]_i, \\ \boxplus_{\frac{1}{2n},i=0}^{2n-1}Rand[q;K_a]_i, Set_{K_a}[q], H_{B_a}[q], \boxplus_{\frac{1}{2n},i=0}^{2n-1}Rand[q';B_b]_i, \boxplus_{\frac{1}{2n},i=0}^{2n-1}M[q;B_b]_i,  \\ \boxplus_{\frac{1}{2n},i=0}^{2n-1}Rand[q';C]_i, c_Q(q), c_Q(q'), cmp(K_{a,b},B_b)\}$.

Then we get the following conclusion.

\begin{theorem}
The basic S09 protocol $\tau_I(\partial_H(\Theta(A\between B)))$ can exhibit desired external behaviors.
\end{theorem}

\begin{proof}
We can get $\tau_I(\partial_H(\Theta(A\between B)))=\sum_{D_i\in \Delta_i}\sum_{D_o\in\Delta_o}loc_A::receive_A(D_i)\leftmerge loc_B::send_B(D_o)\leftmerge \tau_I(\partial_H(\Theta(A\between B)))$.
So, the basic S09 protocol $\tau_I(\partial_H(\Theta(A\between B)))$ can exhibit desired external behaviors.
\end{proof}

\subsection{Verification of KMB09 Protocol}\label{VKMB096}

The famous KMB09 protocol\cite{KMB09} is a quantum key distribution protocol, in which quantum information and classical information are mixed.

The KMB09 protocol is used to create a private key between two parities, Alice and Bob. KMB09 is a protocol of quantum key distribution (QKD) which refines the BB84 protocol against PNS (Photon Number Splitting) attacks. The main innovations are encoding bits in nonorthogonal states and the classical sifting procedure. Firstly, we introduce the basic KMB09 protocol briefly, which is illustrated in Figure \ref{KMB09}.

\begin{enumerate}
  \item Alice create a string of bits with size $n$ randomly, denoted as $K_a$, and randomly assigns each bit value a random index $i=1,2,...,N$ into $B_a$.
  \item Alice generates a string of qubits $q$ with size $n$, accordingly either in $|e_i\rangle$ or $|f_i\rangle$.
  \item Alice sends $q$ to Bob through a quantum channel $Q$ between Alice and Bob.
  \item Alice sends $B_a$ through a public channel $P$.
  \item Bob measures each qubit of $q$ by randomly switching the measurement basis between $e$ and $f$. And he records the unambiguous discriminations into $K_b$, and the unambiguous discrimination information into $B_b$.
  \item Bob sends $B_b$ to Alice through the public channel $P$.
  \item Alice and Bob determine that at which position the bit should be remained. Then the remaining bits of $K_a$ and $K_b$ is the private key $K_{a,b}$.
\end{enumerate}

\begin{figure}
  \centering
  %\vspace{5cm}
  \includegraphics{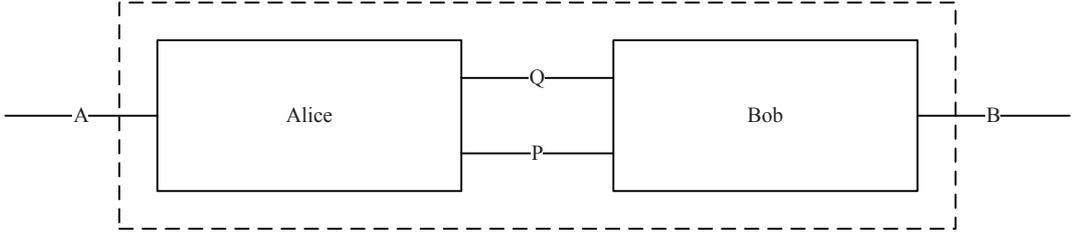}
  \caption{The KMB09 protocol.}
  \label{KMB09}
\end{figure}

We re-introduce the basic KMB09 protocol in an abstract way with more technical details as Figure \ref{KMB09} illustrates.

Now, we assume a special measurement operation $Rand[q;K_a]=\sum^{2n-1}_{i=0}Rand[q;K_a]_i$ which create a string of $n$ random bits $K_a$ from the $q$ quantum system. $M[q;K_b]=\sum^{2n-1}_{i=0}M[q;K_b]_i$ denotes the Bob's measurement operation of $q$. The generation of $n$ qubits $q$ through a unitary operator $Set_{K_a}[q]$. Alice sends $q$ to Bob through the quantum channel $Q$ by quantum communicating action $send_{Q}(q)$ and Bob receives $q$ through $Q$ by quantum communicating action $receive_{Q}(q)$. Bob sends $B_b$ to Alice through the public channel $P$ by classical communicating action $send_{P}(B_b)$ and Alice receives $B_b$ through channel $P$ by classical communicating action $receive_{P}(B_b)$, and the same as $send_{P}(B_a)$ and $receive_{P}(B_a)$. Alice and Bob generate the private key $K_{a,b}$ by a classical comparison action $cmp(K_{a,b},K_a,K_b,B_a,B_b)$. Let Alice and Bob be a system $AB$ and let interactions between Alice and Bob be internal actions. $AB$ receives external input $D_i$ through channel $A$ by communicating action $receive_A(D_i)$ and sends results $D_o$ through channel $B$ by communicating action $send_B(D_o)$.

Then the state transition of Alice can be described as follows.

\begin{eqnarray}
&&A=loc_A::(\sum_{D_i\in \Delta_i}receive_A(D_i)\cdot A_1)\nonumber\\
&&A_1=\boxplus_{\frac{1}{2n},i=0}^{2n-1}Rand[q;K_a]_i\cdot A_2\nonumber\\
&&A_2=Set_{K_a}[q]\cdot A_3\nonumber\\
&&A_3=send_Q(q)\cdot A_4\nonumber\\
&&A_4=send_P(B_a)\cdot A_5\nonumber\\
&&A_5=receive_P(B_b)\cdot A_6\nonumber\\
&&A_6=cmp(K_{a,b},K_a,K_b,B_a,B_b)\cdot A\nonumber
\end{eqnarray}

where $\Delta_i$ is the collection of the input data.

And the state transition of Bob can be described as follows.

\begin{eqnarray}
&&B=loc_B::(receive_Q(q)\cdot B_1)\nonumber\\
&&B_1=receive_P(B_a)\cdot B_2\nonumber\\
&&B_2=\boxplus_{\frac{1}{2n},i=0}^{2n-1}M[q;K_b]_i\cdot B_3\nonumber\\
&&B_3=send_P(B_b)\cdot B_4\nonumber\\
&&B_4=cmp(K_{a,b},K_a,K_b,B_a,B_b)\cdot B_5\nonumber\\
&&B_5=\sum_{D_o\in\Delta_o}send_B(D_o)\cdot B\nonumber
\end{eqnarray}

where $\Delta_o$ is the collection of the output data.

The send action and receive action of the same data through the same channel can communicate each other, otherwise, a deadlock $\delta$ will be caused. We define the following communication functions.

\begin{eqnarray}
&&\gamma(send_Q(q),receive_Q(q))\triangleq c_Q(q)\nonumber\\
&&\gamma(send_P(B_b),receive_P(B_b))\triangleq c_P(B_b)\nonumber\\
&&\gamma(send_P(B_a),receive_P(B_a))\triangleq c_P(B_a)\nonumber
\end{eqnarray}

Let $A$ and $B$ in parallel, then the system $AB$ can be represented by the following process term.

$$\tau_I(\partial_H(\Theta(A\between B)))$$

where $H=\{send_Q(q),receive_Q(q),send_P(B_b),receive_P(B_b),send_P(B_a),receive_P(B_a)\}$ and $I=\{\boxplus_{\frac{1}{2n},i=0}^{2n-1}Rand[q;K_a]_i, Set_{K_a}[q], \boxplus_{\frac{1}{2n},i=0}^{2n-1}M[q;K_b]_i, c_Q(q), c_P(B_b),\\ c_P(B_a), cmp(K_{a,b},K_a,K_b,B_a,B_b)\}$.

Then we get the following conclusion.

\begin{theorem}
The basic KMB09 protocol $\tau_I(\partial_H(\Theta(A\between B)))$ can exhibit desired external behaviors.
\end{theorem}

\begin{proof}
We can get $\tau_I(\partial_H(\Theta(A\between B)))=\sum_{D_i\in \Delta_i}\sum_{D_o\in\Delta_o}loc_A::receive_A(D_i)\leftmerge loc_B::send_B(D_o)\leftmerge \tau_I(\partial_H(\Theta(A\between B)))$.
So, the basic KMB09 protocol $\tau_I(\partial_H(\Theta(A\between B)))$ can exhibit desired external behaviors.
\end{proof}

\subsection{Verification of S13 Protocol}\label{VS136}

The famous S13 protocol\cite{S13} is a quantum key distribution protocol, in which quantum information and classical information are mixed.

The S13 protocol is used to create a private key between two parities, Alice and Bob. Firstly, we introduce the basic S13 protocol briefly, which is illustrated in Figure \ref{S13}.

\begin{enumerate}
  \item Alice create two string of bits with size $n$ randomly, denoted as $B_a$ and $K_a$.
  \item Alice generates a string of qubits $q$ with size $n$, and the $i$th qubit in $q$ is $|x_y\rangle$, where $x$ is the $i$th bit of $B_a$ and $y$ is the $i$th bit of $K_a$.
  \item Alice sends $q$ to Bob through a quantum channel $Q$ between Alice and Bob.
  \item Bob receives $q$ and randomly generates a string of bits $B_b$ with size $n$.
  \item Bob measures each qubit of $q$ according to a basis by bits of $B_b$. And the measurement results would be $K_b$, which is also with size $n$.
  \item Alice sends a random binary string $C$ to Bob through the public channel $P$.
  \item Alice sums $B_{a_i}\oplus C_i$ to obtain $T$ and generates other random string of binary values $J$. From the elements occupying a concrete position, $i$, of the preceding strings, Alice get the new states of $q'$, and sends it to Bob through the quantum channel $Q$.
  \item Bob sums $1\oplus B_{b_i}$ to obtain the string of binary basis $N$ and measures $q'$ according to these bases, and generating $D$.
  \item Alice sums $K_{a_i}\oplus J_i$ to obtain the binary string $Y$ and sends it to Bob through the public channel $P$.
  \item Bob encrypts $B_b$ to obtain $U$ and sends to Alice through the public channel $P$.
  \item Alice decrypts $U$ to obtain $B_b$. She sums $B_{a_i}\oplus B_{b_i}$ to obtain $L$ and sends $L$ to Bob through the public channel $P$.
  \item Bob sums $B_{b_i}\oplus L_i$ to get the private key $K_{a,b}$.
\end{enumerate}

\begin{figure}
  \centering
  %\vspace{5cm}
  \includegraphics{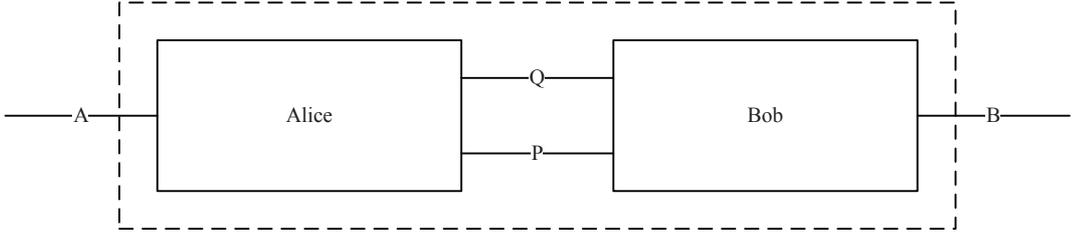}
  \caption{The S13 protocol.}
  \label{S13}
\end{figure}

We re-introduce the basic S13 protocol in an abstract way with more technical details as Figure \ref{S13} illustrates.

Now, we assume a special measurement operation $Rand[q;B_a]=\sum^{2n-1}_{i=0}Rand[q;B_a]_i$ which create a string of $n$ random bits $B_a$ from the $q$ quantum system, and the same as $Rand[q;K_a]=\sum^{2n-1}_{i=0}Rand[q;K_a]_i$, $Rand[q';B_b]=\sum^{2n-1}_{i=0}Rand[q';B_b]_i$. $M[q;K_b]=\sum^{2n-1}_{i=0}M[q;K_b]_i$ denotes the Bob's measurement operation of $q$, and the same as $M[q';D]=\sum^{2n-1}_{i=0}M[q';D]_i$. The generation of $n$ qubits $q$ through two unitary operators $Set_{K_a}[q]$ and $H_{B_a}[q]$, and the same as $Set_{T}[q']$. Alice sends $q$ to Bob through the quantum channel $Q$ by quantum communicating action $send_{Q}(q)$ and Bob receives $q$ through $Q$ by quantum communicating action $receive_{Q}(q)$, and the same as $send_{Q}(q')$ and $receive_{Q}(q')$. Bob sends $B_b$ to Alice through the public channel $P$ by classical communicating action $send_{P}(B_b)$ and Alice receives $B_b$ through channel $P$ by classical communicating action $receive_{P}(B_b)$, and the same as $send_{P}(B_a)$ and $receive_{P}(B_a)$, $send_{P}(C)$ and $receive_{P}(C)$, $send_{P}(Y)$ and $receive_{P}(Y)$, $send_{P}(U)$ and $receive_{P}(U)$, $send_{P}(L)$ and $receive_{P}(L)$. Alice and Bob generate the private key $K_{a,b}$ by a classical comparison action $cmp(K_{a,b},K_a,K_b,B_a,B_b)$. We omit the sum classical $\oplus$ actions without of loss of generality. Let Alice and Bob be a system $AB$ and let interactions between Alice and Bob be internal actions. $AB$ receives external input $D_i$ through channel $A$ by communicating action $receive_A(D_i)$ and sends results $D_o$ through channel $B$ by communicating action $send_B(D_o)$.

Then the state transition of Alice can be described as follows.

\begin{eqnarray}
&&A=loc_A::(\sum_{D_i\in \Delta_i}receive_A(D_i)\cdot A_1)\nonumber\\
&&A_1=\boxplus_{\frac{1}{2n},i=0}^{2n-1}Rand[q;B_a]_i\cdot A_2\nonumber\\
&&A_2=\boxplus_{\frac{1}{2n},i=0}^{2n-1}Rand[q;K_a]_i\cdot A_3\nonumber\\
&&A_3=Set_{K_a}[q]\cdot A_4\nonumber\\
&&A_4=H_{B_a}[q]\cdot A_5\nonumber\\
&&A_5=send_Q(q)\cdot A_6\nonumber\\
&&A_6=send_P(C)\cdot A_7\nonumber\\
&&A_7=send_Q(q')\cdot A_8\nonumber\\
&&A_8=send_P(Y)\cdot A_9\nonumber\\
&&A_9=receive_P(U)\cdot A_{10}\nonumber\\
&&A_{10}=send_P(L)\cdot A_{11}\nonumber\\
&&A_{11}=cmp(K_{a,b},K_a,K_b,B_a,B_b)\cdot A\nonumber
\end{eqnarray}

where $\Delta_i$ is the collection of the input data.

And the state transition of Bob can be described as follows.

\begin{eqnarray}
&&B=loc_B::(receive_Q(q)\cdot B_1)\nonumber\\
&&B_1=\boxplus_{\frac{1}{2n},i=0}^{2n-1}Rand[q';B_b]_i\cdot B_2\nonumber\\
&&B_2=\boxplus_{\frac{1}{2n},i=0}^{2n-1}M[q;K_b]_i\cdot B_3\nonumber\\
&&B_3=receive_P(C)\cdot B_4\nonumber\\
&&B_4=receive_Q(q')\cdot B_5\nonumber\\
&&B_5=\boxplus_{\frac{1}{2n},i=0}^{2n-1}M[q';D]_i\cdot B_6\nonumber\\
&&B_6=receive_P(Y)\cdot B_7\nonumber\\
&&B_7=send_P(U)\cdot B_8\nonumber\\
&&B_8=receive_P(L)\cdot B_9\nonumber\\
&&B_9=cmp(K_{a,b},K_a,K_b,B_a,B_b)\cdot B_{10}\nonumber\\
&&B_{10}=\sum_{D_o\in\Delta_o}send_B(D_o)\cdot B\nonumber
\end{eqnarray}

where $\Delta_o$ is the collection of the output data.

The send action and receive action of the same data through the same channel can communicate each other, otherwise, a deadlock $\delta$ will be caused. We define the following communication functions.

\begin{eqnarray}
&&\gamma(send_Q(q),receive_Q(q))\triangleq c_Q(q)\nonumber\\
&&\gamma(send_Q(q'),receive_Q(q'))\triangleq c_Q(q')\nonumber\\
&&\gamma(send_P(C),receive_P(C))\triangleq c_P(C)\nonumber\\
&&\gamma(send_P(Y),receive_P(Y))\triangleq c_P(Y)\nonumber\\
&&\gamma(send_P(U),receive_P(U))\triangleq c_P(U)\nonumber\\
&&\gamma(send_P(L),receive_P(L))\triangleq c_P(L)\nonumber
\end{eqnarray}

Let $A$ and $B$ in parallel, then the system $AB$ can be represented by the following process term.

$$\tau_I(\partial_H(\Theta(A\between B)))$$

where $H=\{send_Q(q),receive_Q(q),send_Q(q'),receive_Q(q'),send_P(C),receive_P(C),send_P(Y),\\receive_P(Y),send_P(U),receive_P(U),send_P(L),receive_P(L)\}$

 and $I=\{\boxplus_{\frac{1}{2n},i=0}^{2n-1}Rand[q;B_a]_i, \boxplus_{\frac{1}{2n},i=0}^{2n-1}Rand[q;K_a]_i, Set_{K_a}[q], \\H_{B_a}[q], \boxplus_{\frac{1}{2n},i=0}^{2n-1}Rand[q';B_b]_i, \boxplus_{\frac{1}{2n},i=0}^{2n-1}M[q;K_b]_i, \boxplus_{\frac{1}{2n},i=0}^{2n-1}M[q';D]_i, c_Q(q), c_P(C),\\c_Q(q'), c_P(Y), c_P(U), c_P(L), cmp(K_{a,b},K_a,K_b,B_a,B_b)\}$.

Then we get the following conclusion.

\begin{theorem}
The basic S13 protocol $\tau_I(\partial_H(\Theta(A\between B)))$ can exhibit desired external behaviors.
\end{theorem}

\begin{proof}
We can get $\tau_I(\partial_H(\Theta(A\between B)))=\sum_{D_i\in \Delta_i}\sum_{D_o\in\Delta_o}loc_A::receive_A(D_i)\leftmerge loc_B::send_B(D_o)\leftmerge \tau_I(\partial_H(\Theta(A\between B)))$.
So, the basic S13 protocol $\tau_I(\partial_H(\Theta(A\between B)))$ can exhibit desired external behaviors.
\end{proof}

%%\newpage\input{section7.tex}
%%\newpage\input{section8.tex}

\end{document}